\documentclass[a4paper,11pt,oneside,openany]{memoir}

\input{preamble/header}

\begin{document}

\thispagestyle{plain}
\newcommand{\HRule}{\rule{\linewidth}{0.5mm}}

\begin{titlingpage} %for the memoir document class
	\begin{center}

		\includegraphics[width=0.2\textwidth]{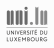}

		{\small PhD-FSTC-2019-11}\\
		The Faculty of Sciences, Technology and Communication\\

		\vspace{0.5cm}

		\begin{minipage}{0.8\textwidth}
			\begin{center}
				{\Large\textbf{Dissertation}}\\
				\vspace{0.5cm}
				{\small Defense held on 23/01/2019 in Esch-sur-Alzette\\
				to obtain the degree of}\\
				\vspace{0.5cm}
				{\Large Docteur de l'Universit\'{e} du Luxembourg\\
				\vspace{0.2cm}
				En Informatique}
			\end{center}
		\end{minipage}

		\vspace{0.2cm}
		by\\
		\vspace{0.2cm}
		\large
		\textbf{Arash Atashpendar}\\

		\vspace{0.5cm}

		\HRule \\[0.5cm]
		{\Large \bfseries \prettyprinttitle}\\[0.4cm]
		\HRule \\[0.5cm]

		\vspace{1cm}

	\end{center}

		\begin{minipage}[t]{0.8\textwidth}
			\begin{flushleft}
				% \small
				\textbf{Dissertation Defense Committee}\\
				\vspace{.5em}
				Prof.~Dr.~Peter Y. A. Ryan, dissertation supervisor \\
				\emph{University of Luxembourg} \\
				\vspace{.5em}
				Prof.~Dr.~Sjouke Mauw, Chairman \\
				\emph{University of Luxembourg} \\
				\vspace{.5em}
				Prof.~Dr.~Cas Cremers \\
				\emph{CISPA-Helmholtz Center in Saarbr\"{u}cken} \\
				\vspace{.5em}
				Prof.~Dr.~Jintai Ding \\
				\emph{University of Cincinnati, Cincinnati, USA} \\
				\vspace{.5em}
				Dr.~Peter B. R\o nne, Vice Chairman \\
				\emph{University of Luxembourg}
			\end{flushleft}
		\end{minipage}

		\vfill

\end{titlingpage}

\thispagestyle{plain}

\begin{center}
    \Large
    \textbf{\printtitle}

    \vspace{0.4cm}

    \begin{minipage}{0.8\textwidth}
			\begin{center}
				{\small A Thesis Submitted in Partial Fulfillment of the Requirements
					for the Degree of\\ Docteur de l'Universit\'{e} du Luxembourg en Informatique}
			\end{center}
		\end{minipage}

    \vspace{0.4cm}
    \small
    by

    \vspace{0.4cm}
    \textbf{\printauthor}

\end{center}

\begin{abstract}
Research questions, originally rooted in quantum key exchange (QKE), have branched off into independent lines of inquiry ranging from information theory to fundamental physics. In a similar vein, the first part of this thesis is dedicated to information theory problems in deletion channels that arose in the context of QKE. From the output produced by a memoryless deletion channel with a uniformly random input of known length $n$, one obtains a posterior distribution on the channel input. The difference between the Shannon entropy of this distribution and that of the uniform prior measures the amount of information about the channel input which is conveyed by the output of length $m$.
% , and it is natural to ask for which outputs this is extremized. In the context of formal languages, such outputs and inputs correspond to the concepts of subsequences and supersequences, respectively.
We first conjecture on the basis of experimental data that the entropy of the posterior is minimized by the constant strings $\texttt{000}\ldots$, $\texttt{111}\ldots$ and maximized by the alternating strings $\texttt{0101}\ldots$, $\texttt{1010}\ldots$. Among other things, we derive analytic expressions for minimal entropy and propose alternative approaches for tackling the entropy extremization problem. We address a series of closely related combinatorial problems involving binary (sub/super)-sequences and prove the original minimal entropy conjecture for the special cases of single and double deletions using clustering techniques and a run-length encoding of strings. The entropy analysis culminates in a fundamental characterization of the extremal entropic cases in terms of the distribution of embeddings. We confirm the minimization conjecture in the asymptotic limit using results from \textit{hidden word statistics} by showing how the analytic-combinatorial methods of Flajolet, Szpankowski and Vall\'ee, relying on generating functions, can be applied to resolve the case of fixed output length and $n\rightarrow\infty$.

In the second part, we revisit the notion of deniability in QKE, a topic that remains largely unexplored. In a work by Donald Beaver it is argued that QKE protocols are not necessarily deniable due to an eavesdropping attack that limits key equivocation. We provide more insight into the nature of this attack and discuss how it extends to other prepare-and-measure QKE schemes such as QKE obtained from uncloneable encryption. We adopt the framework for quantum authenticated key exchange developed by Mosca et al. and extend it to introduce the notion of coercer-deniable QKE, formalized in terms of the indistinguishability of real and fake coercer views. We also elaborate on the differences between our model and the standard simulation-based definition of deniable key exchange in the classical setting. We establish a connection between the concept of covert communication and deniability by applying results from a work by Arrazola and Scarani on obtaining covert quantum communication and covert QKE to propose DC-QKE, a simple construction for coercer-deniable QKE. We prove the deniability of DC-QKE via a reduction to the security of covert QKE. We relate deniability to fundamental concepts in quantum information theory and suggest a generic approach based on entanglement distillation for achieving information-theoretic deniability, followed by an analysis of other closely related results such as the relation between the impossibility of unconditionally secure quantum bit commitment and deniability. Finally, we present an efficient coercion-resistant and quantum-secure voting scheme, based on fully homomorphic encryption (FHE) and recent advances in various FHE primitives such as hashing, zero-knowledge proofs of correct decryption, verifiable shuffles and threshold FHE.
\end{abstract}

\chapter*{Declaration of Academic Honesty}
\thispagestyle{plain}

I hereby declare that the work submitted in this document is my own and based on my research, and any work that is not my own has been cited and acknowledged in the Bibliography section.

\bigskip

\bigskip

\noindent Luxembourg, 25th February 2019

\noindent Arash Atashpendar

\vfill

\newpage

\chapter*{Acknowledgments}
\thispagestyle{plain}

First and foremost, I would like to thank my supervisor, Peter Y. A. Ryan, for giving me the freedom to explore and the benefit of the doubt and the opportunity to work on challenging and interesting problems. I would also like to thank him for his patience, for being a mentor and a good friend who was always supportive throughout the duration of my research.

I would like to thank Bill (A. W.) Roscoe for his collaboration, for taking the time to explain his ideas, both in person and remotely, for inviting me to the university of Oxford, and for being a great source of inspiration with a unique and highly insightful approach to solving mathematical puzzles.

I would also like to thank my jury members, Sjouke Mauw, Cas Cremers, Jintai Ding and Peter B. R\o nne, for agreeing to evaluate my thesis. I thank both Cas and Sjouke for our yearly meetings and their constructive feedback. I thank Jintai both for agreeing to be on my jury and for the exchange of ideas and helpful discussions we have had during his visits over the past few years. Finally, I thank Peter B. R\o nne for being a good friend and collaborator.

Research endeavors rarely happen without collaboration and cross pollination of ideas. I would like to thank my collaborators and coauthors, Peter Y. A. Ryan, Bill Roscoe, David Mestel, Marc Beunardeau, Aisling Connolly, R\'{e}mi G\'{e}raud, Peter B. R\o nne, G. Vamsi Policharla and Kristian Gj\o steen. I would also like to express my gratitude to David Naccache for all the friendly chats and discussions over the years and for making the collaboration with our friends at ENS Paris possible. I also thank all my friends and colleagues for all the fun and laughs we shared together.

Finally and most importantly, I would like to thank my family and my partner for their unconditional love and support.

\newpage

% \dominitoc %Does not play well with the memoir document class
\setcounter{secnumdepth}{2}
\setcounter{tocdepth}{2}
\tableofcontents
\newpage

\listoffigures
\newpage

\listoftables
\newpage

\chapter{Introduction}\label{chp:intro}

\printminichaptoc{\minichaptocenabled}

We rely on cryptography\footnote{Cryptography and cryptology are often used interchangeably in the literature. However, the former roughly translates into \csquote{secret writing}, e.g. encryption, whereas the latter captures, more broadly, the study of secrets or hidden messages.}, among other things, for protecting the confidentiality of our information by keeping it secret from malicious agents. The use of mathematical techniques, albeit initially simple, for hiding and sending secret messages, can be traced back to ancient times \cite{kahn1996codebreakers}.

Modern cryptography is a mathematical discipline encompassing a wide\footnote{Depending on the level of abstraction, be it at the level of implementation or theoretical design and analysis, cryptography can be considered to reside at the intersection of disciplines such as mathematics, computer science, physics, electrical and software engineering.} range of techniques aimed at ensuring secure communication in the presence of adversarial interference. While umbrella terms such as encryption, authentication, message integrity and secure key exchange cover a plethora of different, specialized protocols and constructions, they do provide a coarse-grained characterization of some of the overarching themes in information security.

Due to the necessity of shared secret randomness for achieving various cryptographic tasks such as authentication and encryption, the notion of secure key exchange lies at the heart of cryptography. As such, key exchange protocols can be viewed as fundamental building blocks that enable other cryptographic primitives.

Authenticated key exchange (AKE) protocols allow two or more parties to remotely compute a shared secret, in the presence of an adversary who has complete control over their insecure communication channel. More precisely, a secure AKE allows Alice and Bob to establish a shared session key in such a way that at the end of a session, the two security guarantees of authenticity and secrecy are satisfied. This means that Alice and Bob can be sure that they share a fresh, random session key with each other, such that Eve cannot distinguish the key from a uniformly sampled key of the same length.

Quantum key exchange (QKE) schemes represent a class of AKE protocols that, given a public and untrusted quantum channel and an authenticated\footnote{Typically achieved using information-theoretically secure authentication algorithms, requiring only a logarithmic amount (in the length of the input) of pre-shared key, along with a constant size one-time pad for masking the authentication tag.} classical\footnote{Throughout, we use the term ``classical'' to refer to non-quantum schemes, i.e., constructions that depend solely on classical information and probability theory.} channel, allow two parties to agree on an information-theoretically secure key of arbitrary length, i.e., without relying on any computational hardness assumptions.

In the first part of this thesis, we focus on a series of information theory problems that originally arose as a result of some analysis in the context of QKE, which revolve around outputs in deletion channels. The second part deals with the subtle notion of deniability in quantum cryptography, and specifically in QKE, where we touch upon several aspects of QKE that are relevant for deniability.

In this chapter, we first provide some background on information theoretic and computational security, key exchange protocols and the implications of quantum computation for security in Section \ref{sec:intro-background}. We then focus on quantum key exchange in Section \ref{sec:intro-qke} and discuss its role in the realm of information-theoretically secure key agreement protocols. We then go over the most relevant results in quantum cryptography to pave the path towards a better understanding of our results and subsequent discussions. In Section \ref{sec:thesis-motivations}, we mention the original sources of inspiration and the main motivations for this work. Finally, in Section \ref{sec:contributions-and-structure}, we outline the contributions of this thesis to describe how our results are structured.

\section{Background and Context}\label{sec:intro-background}

The beginning of modern cryptography goes back to the early seminal works of Claude E. Shannon on his mathematical formulation of the theory of communication \cite{shannon2001mathematical} and the communication theory of secrecy systems \cite{shannon1949communication}. Among other things, the former introduced the fundamental notion of information entropy\footnote{The entropy of the random variable $X$ is given by $H(X) = - \sum_{x \in \mathcal{X}} P_X(x) \cdot \text{log}_2 P_X(x)$.} and the latter that of \emph{perfect secrecy}.

These developments were arguably the earliest works that gradually ushered in a new era of cryptography. Later advances made in this general area ultimately led to ground-breaking results such as the discovery of public-key cryptography in the 1970's, or the adoption of the current standard for symmetric-key encryption, namely the Advanced Encryption Standard (AES) \cite{daemen1999aes}, by the National Institute of Standards and Technology (NIST) at the beginning of the 21st century.

In this section we go through some of the most relevant discoveries to provide some background and context for subsequent discussions. An important distinction in what follows has to do with two fundamental notions of security, i.e., computational security vs. information-theoretic security. A given construction providing computational security typically depends on the assumption that a particular mathematical problem $P$ cannot be solved efficiently (in polynomial time) by relating the security of the said scheme to the difficulty of solving $P$, e.g., the RSA problem involving large integer factorization or the discrete logarithm problem. In contrast, an information-theoretically secure scheme roughly translates into a setting wherein the adversary simply does not possess enough information to break the security of the system, regardless of their computational power, e.g., one-time pad encryption.

Finally, note that although the notion of \emph{unconditional security} may crop up quite often in related contexts, there is a subtle distinction that is worth pointing out. Unconditional security refers to the fact that the provided security guarantees do not rely on unproven computational hardness assumptions, but it does not imply information-theoretic security. Simply put, a computationally secure scheme could in principle be proved to be unconditionally secure if the underlying problem that is assumed to be hard were actually proved to be so, e.g., a proof showing that factoring large integers cannot be done using a classical Turing machine in polynomial time.

\subsection{Information-Theoretic and Computational Security}

Perfect secrecy captures the intuitive idea that a ciphertext $c$ provides no information about its underlying plaintext $m$, or equivalently, that $m$ and $c$ are statistically independent, meaning that their mutual information is zero, $I(m,c)=0$\footnote{$I(X;Y) \equiv H(X) - H(X|Y) = H(Y) - H(Y|X)$.}. Yet another statement of the same property can be made in terms of the conditional entropy of the message, given an observation of the corresponding ciphertext, being the same as the entropy of the message alone, that is, $H(m) = H(m|c)$. This amounts to the posterior distribution remaining unchanged w.r.t. the prior distribution despite having knowledge of the ciphertext, that is,
\[
\forall (p \in \mathcal{P} \wedge c \in \mathcal{C}): \prob{P=p|C=c} = \prob{P=p}.
\]

This property can be also expressed in terms of perfect \emph{key equivocation}, an important feature that we will revisit in the second part of this thesis where we focus on the notion of deniability in quantum cryptography, more specifically in the context of quantum key exchange.

Shannon provided an example of a perfect cipher using the so-called one-time pad (OTP), proposed earlier by Vernam \cite{vernam1926cipher}. The idea simply consists of performing a bit-wise XOR operation, which masks each bit $m_i$ with a bit of a random binary secret key of the same length: $c_i = m_i \oplus k_i$. However, there is an important caveat to this approach, namely the strong requirement that the secret key should be at least as long as the message, i.e., $H(k) \ge H(m)$, and that it cannot be reused. Despite the promise of perfect secrecy, obtaining a truly random sequence is expensive and the difficulty of secure and efficient key distribution makes one-time pad encryption a highly impractical solution.

In the total span of cryptography's known history, a relatively recent revolution was the invention of public-key or asymmetric cryptography by Diffie and Hellman in 1976. More precisely, their well-known Diffie-Hellman key exchange protocol \cite{diffie1976new} introduced the idea of two parties agreeing on a secret key over an insecure channel by relying on the difficulty of solving a number theoretic problem by the adversary, namely the discrete logarithm problem over carefully chosen groups. Traditionally, up until that point, key exchange was thought possible only using a secure physical channel, e.g., by meeting in person or via a trusted courier, allowing both parties to share the same secret key that would then for example be used for both encryption and decryption. In effect, before the advent of public-key cryptography, a sender and receiver wanting to exchange secret messages over an insecure channel had to resort to a symmetric cipher that would allow them to encrypt and decrypt message using the same pre-shared key.

It should be pointed out that the original source of inspiration for the Diffie-Hellman key exchange protocol was a construction discovered by Ralph Merkle \cite{merkle1978secure} in 1974 with a provable quadratic security guarantee, generally known as \emph{Merkle puzzles}. His solution allows two parties to agree on a secret key by communicating over an authenticated, but otherwise insecure channel. The key insight revolves around generating and sharing puzzles, cryptograms designed to be solved\footnote{Ideally, the cryptanalysis of a secure cryptogram should be intractable, but here the idea is to introduce a computational gap between the receiver and the adversary.} or broken in an asymmetric manner such that the adversary, faced with a higher computational complexity, would have to expend more computational resources than the legitimate receiver to solve the problem.

Merkle conceived the generation of a puzzle using an encryption function $f$ such that the difficulty of breaking $f$ could be controlled by adjusting the size of the key space, while making sure that the only way to solve the puzzle would be via an exhaustive search of the key space.

In its original form, the main idea consists of the sender generating and sending $n$ puzzles - typically modelled as a one-way function\footnote{This can be instantiated using a hash function.} in a black box model - to the receiver. The size of the key space is chosen such that solving each puzzle requires $\bigO{n}$ calls to the encryption function, i.e., a key space of size $C \cdot n$, for a constant $C$. The receiver selects a random puzzle $p := f(k_p,id,k,R)$ and solves it. She then announces the $id$ and uses the decrypted puzzle key $k$ for securing subsequent communications, where $k_p$ is the key chosen from the restricted key space to generate the puzzle, and $k$ is the actual secret key with $R$ a constant allowing the solver to verify that they have correctly solved $p$.

Since the adversary does not know this mapping, her best strategy is to try puzzles at random until she hits the correct one. Given that on average she would have to try $(\frac{1}{2}n)$ puzzles, with each instance requiring $\bigO{n}$ queries to break, she must, on average, make $\bigO{n^2}$ queries to determine the key. However, both the sender and the receiver require only $\bigO{n}$ calls: the former to generate them, the latter to break one selected at random.

Compared to the conjectured exponential security of public-key schemes such as the Diffie-Hellman key exchange and RSA \cite{rivest1978method}, the quadratic security of Merkle's construction is largely viewed as a source of theoretical interest due to the fact that it offers provable security, at least in the black box query model. In retrospect, it is perhaps somewhat curious that in his original paper, Merkle concludes by conjecturing that Merkle puzzles with exponential security exist. It took more than 30 years for the optimality of the quadratic security offered by Merkle's construction to be proved by Barak and Mahmoody-Ghidari \cite{barak2009merkle} in the random oracle model, where the creation of each puzzle is considered to be done in unit or constant time via a query call. The security of Merkle puzzles against quantum adversaries has also been considered in a series of works by Brassard et al., see \cite{belovs2017provably} and references therein for more information.

\subsection{The Dawn of the Quantum Age}

The advantages offered by public-key cryptography come at the cost of introducing computational hardness assumptions such as the difficulty of solving discrete logarithms, or factoring large integers as relied upon in the famous RSA cryptosystem \cite{rivest1978method}. To a large extent, much of modern cryptography relies on similar computational assumptions, a state of affairs wherein the possibility of the emergence of an efficient algorithm for solving the underlying hard problems could put the vast majority of secure information systems at risk. While problems such as integer factoring have undergone intense scrutiny, as long as one relies on the conjectured difficulty of solving a particular hard problem, the possibility that someday an efficient solution might be discovered holds interest both for purely theoretical as well as practical reasons.

The challenges for computationally secure cryptosystems gradually took on a new dimension ever since the possibility of building quantum computers capable of harnessing the properties of quantum mechanics was suggested by Richard Feynman \cite{feynman1982simulating} in 1982\footnote{The Soviet mathematician Yuri Manin suggested, independently, similar ideas in 1980 \cite{manin1980computable}.}. Feynman's original motivation revolved around the idea that simulating natural processes using a classical Turing machine would be computationally intractable and that it can only be achieved using a quantum computer. In short, the central question entertained the idea of being able to simulate\footnote{Feynman was interested in an exact simulation of the quantum evolution of physical systems such that the computer would do exactly the same as nature, as opposed to classical approximations, e.g. differential equations.} a physical process by a universal computer such that for example doubling the number of particles would not result in an exponential blowup in terms of computational resources, i.e. memory and time. Instead, the requirement is that the number of computing elements be proportional to the space-time volume of the physical system. The primary interest of Feynman was to discover something new about physics by learning about the fundamental limitations of computing. This seminal paper marks the beginning of the field of quantum computing and quantum information processing.

For a while, this area of research was regarded as being purely theoretical until the idea of sufficiently stable quantum computers capable of efficiently solving computationally hard problems, beyond the reach of classical digital computers, became reality. More precisely, the ground-breaking quantum algorithm by Peter W. Shor \cite{shor1999polynomial} provides a solution for prime factorization and computing discrete logarithms in polynomial time using a quantum computer. The potential repercussions of this possibility spawned a new field of research focusing on what is often referred to as post-quantum cryptography, which is largely concerned with addressing security concerns that become relevant in the presence of adversaries capable of running quantum algorithms using cryptographically relevant quantum computers. This discovery was the beginning of a series of new developments in the design and analysis of quantum algorithms. A very recent work by Roetteler et al. \cite{roetteler2017quantum}, on estimating quantum resources needed for computing elliptic curve discrete logarithms, shows that running Shor's algorithm for a $3072$-bit modulus requires $6146$ logical\footnote{A larger number of physical qubits is needed to effectively yield a given number of logical qubits. Roughly speaking, this is prompted by the need for performing quantum error correction to deal with the \emph{decoherence problem}, which refers to the fragile nature of quantum information and the tendency of quantum systems to interact with their environment, thereby becoming noisy, potentially to the point of rendering their encoding irrecoverable.} qubits and $1.86 \cdot 10^{13}$ Toffoli gates.

Within the realm of public-key cryptography, almost any construction that in some way involves the Abelian hidden subgroup problem is potentially vulnerable to Shor's algorithm, if executed by a cryptographically relevant quantum computer. Regarding symmetric schemes, the most well-known threat is posed by Grover's algorithm \cite{grover1996fast}, which in its original form was designed to find the unique input $a$ to a black box function $f: \{0, \ldots, N\} \rightarrow \{0,1\}$ such that $f(a)=1$, using $\bigO{\sqrt{N}}$ calls to the oracle. The best classical solutions require $\bigO{N}$ calls, thus constituting a prime example that demonstrates a provable separation between quantum and classical computing. Roughly speaking, this improvement in the time complexity of finding an item in an unstructured database has repercussions for brute-force strategies, which is less dramatic than the potential impact of Shor's algorithm given that Grover-type attacks can be dealt with by increasing the key size. Note that while an efficient quantum algorithm for factoring large integers and solving discrete logarithm is known, the existence of an efficient classical algorithm simply remains unknown. We encourage the reader to refer to a recent survey by Montanaro \cite{montanaro2016quantum} on quantum algorithms and a classification of how they apply to breaking various cryptographic primitives.

\section{Quantum Cryptography and Key Exchange}\label{sec:intro-qke}

Faced with the emergence of quantum computers, coupled with further developments in the design of efficient quantum algorithms, and the threats they pose to the security of classical cryptosystems, one can either resort to constructions that are not known to be vulnerable to quantum algorithms such as lattice-based cryptography \cite{regev2004quantum} and learning with errors (LWE) \cite{regev2009lattices}, or consider information-theoretically secure solutions that do not depend on unproven assumptions. The latter case is where quantum key exchange (QKE)\footnote{Often referred to as quantum key distribution (QKD).} enters the picture, a solution that allows two parties to agree on a random secret key with information-theoretic security, that is, without relying on any computational assumptions.

The potential repercussions of the possibility of stable and scalable quantum computers spawned a new field of research focusing on what is often referred to as post-quantum cryptography, which tries to address security concerns that become relevant in the presence of adversaries capable of running quantum algorithms.

\subsection{QKE and Information-Theoretic Key Agreement}

Information-theoretic key agreement enables two or more parties to establish a common secret key over an insecure broadcast channel in the presence of an unbounded adversary Eve, such that the total amount of information Eve gains about the final key can be made arbitrarily small. This is in contrast with computationally secure schemes that rely on computational complexity hardness assumptions. Simply put, the security guarantee of achieving negligibly small information leakage to the adversary holds in an information-theoretic sense.

Given an authenticated classical channel and an insecure quantum channel, quantum key exchange (QKE) allows two parties to agree on a random secret key with information-theoretic security. The existence of inexpensive\footnote{In terms of computational resources and preshared randomness.} and information-theoretically secure authentication algorithms for obtaining an authenticated classical channel was first confirmed by Carter and Wegman \cite{CarterWegman77}, a work that led to a long list of follow-up work \cite{CarterWegman77,stinson1991universal,gemmell1993codes,krawczyk1994lfsr,krawczyk1995new}. In short, these schemes enable information-theoretic authentication, requiring a short preshared initial key, where the length of the key grows logarithmically in the length of the input message that is to be authenticated. In fact, as shown by Renner and Wolf, even only weakly correlated and partially secret information would suffice \cite{renner2003unconditional,renner2004exact}. Due to the requirement of a preshared key for authentication, QKE is sometimes also referred to as quantum key \emph{expansion}.

The possibility of QKE with public-key (computational) authentication has also been considered. Although doing so would introduce computational assumptions into the mix, and thereby reduce the overall information-theoretic security of QKE to the security of the underlying hardness assumption of the authentication mechanism, it still gives rise to a unique and useful property known as \emph{everlasting security}. This feature captures the idea that the adversary has a limited window of opportunity for breaking the authentication, namely during a given session, and if she does not succeed, the resulting key will retain its information-theoretic security. This property is mainly due to the notion of non-attributability \cite{ioannou2011new} in QKE, which is related to the mathematical independence of the final secret key from the classical communication, a property that we will further elaborate on in Chapter \ref{chp:coercer-deniable-qke} in the context of deniable QKE.

The most-well known QKE scheme is the BB84 protocol due to Charles Bennett and Gilles Brassard \cite{bennett1984quantum}, which makes use of conjugate coding, a primitive that was proposed earlier in a work by Stephen Wiesner \cite{wiesner1983conjugate} on unforgeable bank notes\footnote{Wiesner's paper was initially rejected.}. Conjugate coding works by encoding a sequence of random classical bits into quantum states prepared in one of two orthogonal bases chosen at random. This taps into crucial properties of quantum mechanics such as Heisenberg's \cite{heisenberg1927uber} well-known \emph{uncertainty principle} as well as the no-cloning theorem \cite{wootters1982single}. The former captures the fact that measuring in one basis irrevocably destroys information about the encoding in its conjugate basis, that is, the impossibility of precisely measuring two non-commuting self-adjoint operators (or complementary variables). The latter states that an arbitrary unknown quantum state cannot be cloned or copied perfectly.

The BB84 protocol falls within the realm of so-called prepare-and-measure schemes, which have the advantage of not requiring quantum computation and storage. Another equally important contribution was a quantum key exchange protocol using entanglement proposed independently by Artur Ekert \cite{ekert1991quantum}, where the security of QKE rests upon Bell's theorem presented in his seminal work \cite{bell2001einstein} published in 1964 on the famous EPR paradox by Einstein, Podolsky and Rosen \cite{einstein1935can}. Here the core idea is to relate the security of QKE to the violation of a Bell-type inequality, which is based on the unique property that certain statistical correlations can be violated only by systems that exhibit a certain degree of quantum entanglement, thus drawing a clear line between classical and quantum correlations. Einstein\footnote{Einstein passed away in 1955 almost a decade before Bell's result.}, Podolsky and Rosen raised the question of whether or not the theory of quantum mechanics could be considered to be complete, and if a phenomenon such as quantum entanglement can be explained by a lack of knowledge, the so-called hidden variable theory. This paradox was settled by John Bell \cite{bell2001einstein} by showing that there exist correlations that are simply impossible to achieve classically, thereby invalidating Einstein's local hidden variable theory.

In a nutshell, QKE protocols provide information-theoretically secure key agreement by leveraging these fundamental properties of quantum mechanics. In the absence of prior knowledge about the preparation configuration of a state, the impossibility of observing or measuring an unknown quantum system without disturbing its state represents a unique and inherently quantum property that allows the detection of eavesdropping, a feat that is classically impossible to achieve. This property is explained by the fact that an adversary trying to eavesdrop\footnote{Observing, reading, measuring and even trying to copy a quantum state are fundamentally the same thing in quantum information theory.} on a quantum state is bound to introduce errors (or noise) into the measurement results obtained by the legitimate parties, allowing them to detect the adversary's presence\footnote{We cannot tell the difference between inherent channel noise and noise caused by the adversary, which is why we attribute all the detected errors to the adversary.}. Naturally there is a limit on the amount of noise that can be tolerated. If the detected error rate exceeds a certain predefined threshold, the protocol aborts, otherwise, with the help of error correction and randomness distillation\footnote{In the parlance of QKE, these are traditionally referred to as information reconciliation and privacy amplification.} the parties can convert their post-measurement partially correlated bit strings into mutually shared secret keys.

It is worth pointing out that within the realm of information-theoretic security, another restrictive assumption made in Shannon's model is that the adversary is considered to have access to exactly the same information as the legitimate receiver. Although it may seem like a fair assumption given the natural and justifiably pessimistic view traditionally adopted in cryptography, it was shown by Maurer \cite{maurer1993secret} that by relaxing and modifying the model such that the adversary cannot obtain exactly the same information as the legitimate receiver, it is still possible to agree on a secret key.

In some earlier works by Wyner \cite{wyner1975wire} and Csisz{\'a}r and K\"{o}rner \cite{csiszar1978broadcast}, a similar consideration was made in that the adversary is assumed to receive messages over a noisier channel than that of the legitimate receiver, an assumption that may be deemed rather unrealistic. However, Maurer \cite{maurer1993secret} resolves this issue by allowing the parties access to a public and insecure, yet authenticated, broadcast channel. The need for guaranteeing authenticity and data integrity is identical to the requirement for QKE.

While the requirements for information-theoretic key agreement, in its modern form, are considered standard, a noteworthy subtlety in terms of early definitions was uncovered by Maurer and Wolf \cite{maurer2000information}. The authors suggested a replacement for the hitherto definition of information-theoretic key agreement with a slightly modified version, referred to as strong secrecy. The main difference between the old definition, labelled \emph{weak secrecy}, and strong secrecy lies in the fact that the former simply requires that the adversary's side information be arbitrarily small, not in an absolute sense, but rather in terms of an information rate, defined as the ratio between the information quantity of interest and the number of independent repetitions of the random experiment.

More formally, let $X$, $Y$ and $Z$ denote the random variables associated with strings of length $N$ belonging to Alice, Bob and the adversary, Eve, respectively, obtained over an insecure, but authenticated, public communication channel. Moreover, let $S$ and $S'$ be the final secret keys computed by Alice and Bob, such that we have $S=S'$ with probability at least $1-\epsilon$ and
\[
\frac{1}{N}I(S;CZ) \le \epsilon
\]
where $C$ denotes the messages exchanged over the insecure channel by Alice and Bob, and $\epsilon$ the security parameter.

The problem with this definition is that although the adversary's side information may be arbitrarily small in terms of the information rate, it is neither necessarily bounded, nor negligibly small in an absolute sense, thus allowing a potentially substantial amount of information leakage. To rectify this, in a more natural definition, namely that of strong secrecy, the privacy requirements are strengthened in that the negligibility of the adversary's knowledge is no longer expressed w.r.t. an information rate, but rather in an absolute sense as follows
\[
I(S;CZ) \le \epsilon,
\]
and additionally, the final key is required to be perfectly-uniformly distributed.

\subsection{The BB84 Protocol: an Overview}

Here we provide a rough, high-level description of the BB84 protocol \cite{bennett1984quantum} and defer a formal version to Section \ref{sec:qke-and-ue} in Chapter \ref{chp:qip}. Abstractly speaking, the core idea is that a sender prepares a sequence of random quantum states\footnote{For simplicity, assume logical qubits (two-level quantum systems) physically implemented using the polarization of a photon.} using conjugate coding and sends them over to the receiver, who measures each incoming qubit in a random basis, where the encoding and measurement are done according to two orthogonal bases, referred to as the \emph{rectilinear} and the \emph{diagonal} basis. More precisely, the interactions involved in the BB84 protocol play out as follows.

During the initial quantum phase of the protocol, the parties establish a pair of \emph{raw keys}: the sender, Alice, generates $n$ random classical bits $X = (x_1,\ldots,x_n)$ and encodes each of them into a qubit by choosing either the rectilinear or the diagonal basis at random and sends it to the receiver, Bob. Upon receiving each qubit, Bob chooses between the rectilinear and the diagonal basis at random to perform a measurement, resulting in $n$ random classical bits $Y = (y_1,\ldots,y_n)$.

The second phase of the protocol purely consists of classical post-processing to obtain a fresh secret key from the partially correlated variables $X$ and $Y$. First, Alice and Bob announce and compare their basis choices to discard incompatible measurements\footnote{An incompatible measurement simply means measuring in the conjugate basis and thus obtaining a different classical bit with probability $\frac{1}{2}$.}, on average ending up with a \emph{sifted key} of length $\ell = n/2$. Next, they perform \emph{error estimation} to get an estimate\footnote{The exact amount will be discovered in the error correction step.} for the error rate or the amount of discrepancy between their sifted keys, which translates into the level of noise. If the fraction of indices for which their partially correlated keys disagree ($|\{i \in [1,\ell] : x^{\mathrm{sifted}}_i \neq y^{\mathrm{sifted}}_i\}|$) exceeds a predefined threshold (tolerated error rate), they abort the protocol. A high error rate might indicate the presence of an eavesdropper.

If the estimated error rate lies below the threshold, error correction\footnote{The error correction can be interactive or non-interactive using forward error-correcting codes.} will be applied to obtain two keys of equal length, about which the adversary may have gained some partial knowledge via eavesdropping. Finally, given an upper bound on the amount of information leakage to the adversary, a final secret key distillation step will take place to convert the error corrected keys into shorter but secure keys about which the adversary has no knowledge. This last step is referred to as privacy amplification, involving the use of a two-universal hash function. The two parties perform key confirmation to ensure their final secret keys are equal with high probability.

\subsection{QKE and Proofs of Security}

Although one of the main selling points of QKE is the fact that its security relies on fundamental laws of quantum mechanics\footnote{This was simply claimed in the original paper by Bennett and Brassard and remained a folklore theorem for quite some time thereafter.}, as opposed to conjectured computational hardness assumptions, a formal proof of security remained an open problem for quite a while.

One of the earlier contributions was a work by Jeroen van de Graaf \cite{van1997towards} on the formalization of definitions for the security of quantum protocols. Attempts at providing a security proof for QKE resulted in an extensive series of works, e.g. by Lo and Chau \cite{lo1999unconditional}, Mayers \cite{mayers2001unconditional} and Biham et al. \cite{biham2006proof}, all of which were relatively complex and in one way or another depended on a reduction to an entanglement-based variant, and required the use of quantum computation.

This situation changed when Shor and Preskill \cite{shor2000simple} presented their \csquote{simple proof of security of the BB84 quantum key distribution protocol}, which among other things, offered a crucial insight by establishing a link between the CSS \cite{steane1996multiple,calderbank1996good} quantum error correcting codes and entanglement-based QKE. Their idea was to first give an entanglement-based QKE protocol that is proven secure using methods developed by Lo and Chau \cite{lo1999unconditional}, and to then show that this implies the security of BB84 via a reduction using CSS codes. Due to the irrelevance of phase errors and their decoupling from bit flips in CSS codes, the analysis is reduced to classical error correction and as a result, removes the use of quantum computation.

However, it should be pointed out that this proof did have some noteworthy requirements such as the fact that it does not hold for imperfect devices and requires that the sources be single-photon sources, shortcomings that have been dealt with in the meantime, e.g. see the work of Gottesman et al. \cite{gottesman2004security} on the security of QKE with imperfect devices.

Another important detail worth pointing out is that the proof of security of Shor and Preskill relies on the theoretical existence of classical error correcting codes that satisfy the dual-containing property, as shown in Section \ref{subsec:bb84-shor-preskill}, thus requiring explicit and efficiently decodable codes. In a work by Luo and Devetak \cite{luo2007efficiently}, it is shown that by using efficiently decodable non-dual-containing modern classical codes such as LDPC \cite{gallager1962low} and turbo codes \cite{berrou1993near}, this constraint can be relaxed at the cost of turning QKE into a key expansion construction that increases the size of a pre-shared key by a constant factor.

An important contribution in this general area was by Renato Renner \cite{Renner2008} on the security of QKE and relaxing common independence conditions in quantum information theory such as repeating an experiment independently many times or considering systems comprised of independent parts. By introducing new uncertainty measures such as \emph{smooth min-entropy} and a quantum version of the de Finetti's representation theorem, new techniques for a generic proof of the security of QKE were developed, which rely on breaking down the proof into information-theoretic considerations involving error-correction and privacy amplification. Finally, the same work also considers the \emph{universally composable} security of QKE, meaning that keys generated by QKE can be used in an arbitrary context or application. Follow-up works along these lines were an information-theoretic proof for QKE by Renner et al. \cite{renner2005information} and a comprehensive security analysis by Tomamichel et al. \cite{tomamichel2015rigorous}.

These efforts continued and QKE security proofs were gradually extended to less restrictive models that, among other things, did not assume perfect devices. These include the works of L\"{u}tkenhaus \cite{lutkenhaus1999estimates} and Scarani et al. \cite{scarani2009security} on the security of practical QKE. The security of finite-key QKE has also been considered by Scarani et al. \cite{scarani2008quantum,scarani2008security} and Tomamichel et al. \cite{tomamichel2012tight}. As briefly mentioned before, the requirement of an initial shared key for achieving information-theoretic authentication can be lifted and replaced with a computational assumption. This can be done by using a public-key authentication scheme to obtain everlasting security such that information-theoretic security is guaranteed as long as the computational assumption remains unbroken during a limited time (duration of a session). Everlasting security of QKE, against a bounded adversary, was formally proved by Unruh \cite{unruh2013everlasting} and Mosca et al. \cite{mosca2013quantum}.

We return to some of these works in the second part of the thesis where we make use of some of their key insights to ground the notion of deniability in quantum information theory.

\subsection{Quantum Cryptography and Bell-type Violations}

Quantum information assurance encapsulates the various ways we can exploit quantum mechanical effects to achieve security goals. Quantum key exchange protocols represent a subset of computational and cryptographic tasks that can be achieved using quantum information processing (QIP). QKE is arguably the most widely-known application of QIP to cryptography. Moreover, prepare-and-measure QKE variants can be implemented using current technology. Due to its popularity, QKE is often erroneously referred to as quantum cryptography. There is a considerable number of surveys on BB84 and other variants of QKE such as the one by Gisin et al. \cite{gisin2002quantum}.

Quantum cryptography encapsulates a wide variety of primitives ranging from bit commitment\footnote{Quantum bit commitment with unconditional security was shown to be impossible \cite{mayers1997unconditionally}.}, to randomness amplification \cite{colbeck2012free}, device-independent quantum cryptography starting with the work of Mayers and Yao \cite{mayers98deviceindependence} and later on in a work by Barrett et al. \cite{barrett2005no}, which even additionally considers the security of QKE if certain principles in quantum mechanics ceased to be valid. Similar, fundamental ideas related to the violation of Bell-type inequalities permeate the field of device-independent cryptography, randomness amplification, and \emph{self-testing} quantum devices.

In essence, in the realm of device-independent cryptography, the natural setting assumes that devices are untrusted, and may have even been prepared by the adversary. The common underlying idea is to test for the degree of quantum randomness (correlations) present in the devices by subjecting them to an experiment that violates a Bell-type inequality, which means exhibiting correlations that are strictly stronger than the limit of what can be achieved using classical information. The CHSH game \cite{clauser1969proposed} is arguably the most well-known example of such an experiment, wherein a minimal amount of quantum randomness must be present in the outputs provided by the parties in order to violate the CHSH inequality. In a work by Pironio et al. \cite{pironio2010random}, the authors provide a quantification of the relation between the violation of the CHSH inequality and the amount of entropy obtained in the measurements of quantum states, followed by another important result by Reichardt et al. \cite{reichardt2013classical} showing that winning the CHSH game with optimal probability indicates the presence of EPR pairs (maximally entangled quantum states) in the entanglement between the parties' devices. We will briefly revisit the CHSH game in Chapter \ref{chp:qip}, Part \ref{part:two} and elaborate on the relevance of Bell-type violations for deniability in Chapter \ref{chp:entanglement-distillation}.

Finally, we suggest the survey by Broadbent and Schaffner \cite{broadbent2016quantum} for a comprehensive overview of quantum constructions beyond QKE.

\section{Motivations}\label{sec:thesis-motivations}

Although the motivations for Part \ref{part:one} and \ref{part:two} of this thesis are quite different, quantum key exchange and (quantum) information theory constitute the common denominator in both cases. We now elaborate on the motivations for each part separately.

\subsection{QKE and Offshoots in Information Theory}

Research in QKE spawned different avenues of research that branched off in various directions leading to new insight in domains such as information theory and coding theory. These include studying problems such as entropy measures in the realm of unconditional security in cryptography \cite{Cachin97}, secret key agreement over public channels \cite{maurer1993secret}, generalized privacy amplification \cite{Bennett95}, linking information reconciliation and privacy amplification \cite{CachinMaurer97}, and the relation between quantum entanglement and secret key distillation \cite{devetak2005distillation} to name but a few. In a similar vein, the main focus of the first part of this thesis will be on a collection of combinatorial and information theory problems that originally arose in the context of some analysis in quantum key exchange.

While the original motivation was to obtain an upper bound on the amount of leakage resulting from some modifications to how error estimation in QKE is done \cite{ryan2013enhancements}, the analysis gave rise to a multitude of new problems that we will address in the first part of this thesis. In fact, the original problem is treated as an independent information theory problem with a concrete formulation in terms of open problems in the context of deletion channels. More specifically, the original problem of upper-bounding information leakage will be studied in terms of entropy extremizing outputs in deletion channels.

The underlying combinatorial problem that governs the shape of the probability distribution, in turn responsible for the trend in entropy, is the source of a long list of well-known open problems in coding theory, combinatorics of subsequences and supersequences, DNA sequencing and intrusion detection, to name a few. Apart from the original motivation for addressing information leakage in terms of entropy extremization, the same problem gives rise to a series of closely related combinatorial objects that shed light into the combinatorial structure of the weight distribution behind the computation of entropy.

\subsection{Deniability in Quantum Information Theory}

Deniability is a fundamental privacy-related notion in cryptography that can be roughly defined as the ability for the sender of a message to deny either its content or the fact that they have sent that message. The ability to deny a message or an action is of central importance in various scenarios such as off-the-record communication, anonymous reporting, whistleblowing and coercion-resistant secure electronic voting. Apart from its direct use cases, there are deep connections between deniability and fundamental concepts such as secure multiparty computation \cite{goldwasser1997multi} and incoercible multiparty computation \cite{canetti1996incoercible}.

Surprisingly, this fundamental notion has been largely ignored by the quantum cryptography and quantum information processing community at large. The motivation for investigating deniability in quantum cryptography came from a joint analysis of quantum protocols, model checking techniques for the analysis of quantum processes, and the notion of coercion-resistance in the context of secure e-voting protocols. This was prompted by the fact that many cryptographic tasks can be achieved only via the application of quantum phenomena and quantum information processing. Fundamentally, most of these tasks rely upon principles dictated by the main postulates of quantum mechanics, including the measurement principle and the no-cloning theorem. These postulates enable feats that are simply impossible using classical means, e.g., detection of eavesdropping in quantum channels or the impossibility of making perfect copies of unknown quantum states due to the no-cloning theorem. Achieving coercion resistance in the context of voting protocols is closely related to the notion of deniability. Yet, the role of existing quantum primitives and their effectiveness for attaining coercion resistance are still poorly understood.

In short, despite being an important concept in cryptography, it has received very little attention from the quantum cryptography community. To put things into perspective, while the seminal works of Canetti et al. on deniable encryption \cite{canetti1997deniable} and Dwork et al. on deniable authentication \cite{dwork2004concurrent} have been cited by hundreds of papers focusing on various aspects of deniability, there exists a single paper by Donald Beaver \cite{beaver2002deniability} dealing with deniability in quantum key exchange, which seems to have gone practically unnoticed by the QIP community. In this paper, Beaver shows that there exists an eavesdropping attack that can detect attempts at denial in the most well-known QKE protocol, i.e. BB84, and claims that BB84 is binding and thus undeniable. It is further claimed that existing QKE schemes, including other variants of BB84 as well as entanglement-based QKE protocols, are not necessarily deniable. In the context of deniable quantum key exchange, apart from this paper, there exists virtually no other studies on this topic. Hence, deniability in the quantum regime represents a rich and promising, yet almost entirely unexplored avenue of research.

Given the unconditional security provided by QKE against an adversary with unbounded computational power, it is natural to consider the possibility of deriving similar results for deniability using quantum information. This can range from overcoming known negative results for deniability in the classical literature to showing that information-theoretic deniability can be obtained using quantum information processing. Be it in the form of quantum protocols that achieve various forms of deniability, or in terms of proving the existence of fundamental obstructions to achieving deniability, regardless of the outcome, the possibility of gaining deeper insight into some of the fundamental properties of quantum information via studying deniability in the quantum setting is of independent theoretical interest.

The central idea in the field of post-quantum cryptography is to devise protocols that remain secure against quantum adversaries. This means considering an adversary that can perform probabilistic-polynomial-time computations on a quantum computer. As pointed out in \cite{ioannou2011new}, assuming that we live in a quantum universe, at a minimum, we should require that the secret key generated by an authenticated key exchange (AKE) protocol be secure against a quantum adversary. Consequently, it is only natural for us to expect that deniability enjoys the same considerations. Indeed, deniability has not been considered within the framework of post-quantum cryptography either.

From a theoretical point of view, studying deniability in quantum cryptography could lead to new fundamental insights in physics and information theory, with historical examples such as Shannon’s work that started the field of information theory developed in parallel with his work on cryptography or the development of quantum key exchange leading to the discovery of information reconciliation and privacy amplification protocols, which gave new insights in network information theory. Similarly, the search for security proofs for quantum key distribution against general attacks led to the development of new techniques based on symmetry, the Quantum De Finetti theorem, and the entropic formulation of the uncertainty principle in quantum mechanics. Finally, work on device independent quantum cryptography \cite{reichardt2013classical} led to the discovery of the strong mathematical constraints that are imposed on any system attempting to achieve a maximal violation of the CHSH inequality \cite{clauser1969proposed} and other related Bell-type inequalities.

From a practical point of view, research in this area could lead to the development of novel deniable QKE schemes and secure, coercion-resistant e-voting schemes.

In summary, the lack of research on the notion of deniability in the quantum regime is in stark contrast to its classical counterpart, which has been an important object of study for the past two decades.

\section{Contributions and Outline of Thesis}\label{sec:contributions-and-structure}

We now describe the structure of this thesis and provide a high-level overview of the main contributions. The contributions of this thesis are twofold. The first part, largely based on our results in \cite{atashpendar2015information, atashpendar2018clustering, atashpendar2018proof}, primarily focuses on a number of combinatorial and information theory problems that were originally encountered in the context of quantum key exchange. The second part, largely based on our results in \cite{atashpendar2018deniability,atashpendar2018voting}, deals with the notion of deniability in quantum cryptography, and more specifically with deniability in quantum key exchange, as well as a coercion-resistant and quantum-secure e-voting scheme, which uses fully homomorphic encryption to achieve linear-time tallying and quantum resistance. The contributions detailed in this section are largely based on our results in the works listed below.
\begin{enumerate}
  \item Arash Atashpendar, A. W. Roscoe and Peter Y. A. Ryan. ``Information Leakage Due to Revealing Randomly Selected Bits''. In: Security Protocols XXIII. Springer, 2015, pp. 325–341.
  \item Arash Atashpendar, Marc Beunardeau, Aisling Connolly, R\'{e}mi G\'{e}raud, David Mestel, A. W. Roscoe and Peter Y. A. Ryan. ``From Clustering Supersequences to Entropy Minimizing Subsequences for Single and Double Deletions''. In: arXiv preprint arXiv:1802.00703 (2018). (Submitted to the journal of IEEE Transactions on Information Theory on Sep. 17th 2017, first review received on Sep. 22nd 2018, revised version submitted on Dec. 20th 2018)
  \item Arash Atashpendar, David Mestel, A. W. Roscoe and Peter Y. A. Ryan. ``A Proof of Entropy Minimization for Outputs in Deletion Channels via Hidden Word Statistics''. In: arXiv preprint arXiv:1807.11609 (2018).
  \item Arash Atashpendar, G. Vamsi Policharla, Peter B. R\o nne and Peter Y. A. Ryan. ``Revisiting Deniability in Quantum Key Exchange via Covert Communication and Entanglement Distillation''. In: NordSec 2018, 23rd Nordic Conference on Secure IT Systems. Springer. 2018, pp. 104–120.
  \item Peter B. R\o nne, Arash Atashpendar, Kristian Gj\o steen and Peter Y. A. Ryan. ``Coercion- Resistant Voting in Linear Time via Fully Homomorphic Encryption - Towards a Quantum-Safe Scheme''. In: 23rd International Conference on Financial Cryptography and Data Security 2019, FC 2019, International Workshops, CIW, VOTING, and WTSC, 2019. Springer.
\end{enumerate}

\subsection{Information Theory Problems in Deletion Channels}

The first part of the thesis is dedicated to a series of studies focusing on combinatorial and information theory problems that arose as a result of an analysis of prepare-and-measure quantum key exchange protocols \cite{ryan2013enhancements}, which suggested some simple changes aimed at reducing leakage of key material and improving the final key rate. These changes included a modification of the quantum bit error rate (QBER) estimation that gave rise to an information theory problem in which the object of study was to account for bit strings leading to minimal and maximal information leakage, as described below.

From the output produced by a memoryless deletion channel from a uniformly random input of known length $n$, one obtains a posterior distribution on the channel input.  The difference between the Shannon entropy of this distribution and that of the uniform prior measures the amount of information about the channel input which is conveyed by the output of length $m$, and it is natural to ask for which outputs this is extremized. This question was posed in a previous work, where it was conjectured on the basis of experimental data that the entropy of the posterior is minimized and maximized by the constant strings $\texttt{000\ldots}$ and $\texttt{111\ldots}$ and the alternating strings $\texttt{0101\ldots}$ and $\texttt{1010\ldots}$ respectively.

In the first part of this thesis, we focus on this entropy extremization problem and a series of related combinatorial and information theory problems, largely based on our results in \cite{atashpendar2015information, atashpendar2018clustering, atashpendar2018proof}. For the most part, we will abstract away from the details of the original context and simply state the problem as an analysis of entropy extremizing outputs in deletion channels.

\subsubsection{Binary Sequences and Information Leakage}

The original problem statement seems very natural and is easy to state but has not to our knowledge been addressed before in the information theory literature: suppose that we have a random bit string $y$ of length $n$ and we reveal $m$ bits at random positions, preserving the order but without revealing the positions, how much information about $y$ is revealed? in other words, the quantity of interest is the conditional entropy of $Y$ given an observation $x$, that is, $H(Y|X=x)$.

Chapter \ref{chp:deletion-channel-framework} first provides the problem statement for the content of Part \ref{part:one}. We then introduce our framework, which includes the necessary background theory, along with the main concepts and building blocks needed throughout. We then provide a survey of the literature on some areas of research that either directly, or indirectly, depend on the same underlying combinatorial problem, namely coding theory and deletion channels, study of (sub/super)-sequences, the distribution of subsequence embeddings, and efficient dynamic programming algorithms used in DNA sequencing.

In Chapter \ref{chp:information-leakage}, we show that while the cardinality of the set of compatible $y$ strings depends only on $n$ and $m$, the amount of leakage does depend on the exact revealed $x$ string. We observe that the maximal leakage, measured as decrease in the Shannon entropy of the space of possible bit strings, corresponds to the $x$ string being all zeros or all ones and that the minimum leakage corresponds to the alternating $x$ strings. We derive a formula for the maximum leakage (minimal entropy) in terms of $n$ and $m$. We discuss the relevance of other measures of information, in particular min-entropy, in a cryptographic context. Finally, we describe a simulation tool to explore these results.

\subsubsection{Combinatorial Structures}

In Chapter \ref{chp:combinatorial-structures}, we focus on combinatorial objects encountered as a result of studying the entropy problem. We present an algorithm for counting the number of subsequence embeddings using a run-length encoding of strings. We then describe two different ways of clustering the space of supersequences and prove that their cardinality depends only on the length of the received subsequence and its Hamming weight, but not its exact form. Then, we consider supersequences that contain a single embedding of a fixed subsequence, referred to as singletons, and provide a closed form expression for enumerating them using the same run-length encoding. We prove an analogous result for the minimization and maximization of the number of singletons, by the alternating and the uniform strings, respectively.

\subsubsection{Entropy Extremization for Outputs in Deletion Channels}

In Chapter \ref{chp:finite-deletions}, we prove the original minimal entropy conjecture for the special cases of single and double deletions using similar clustering techniques and the same run-length encoding, which allow us to characterize the distribution of the number of subsequence embeddings in the space of compatible supersequences to demonstrate the effect of an entropy decreasing operation.

\subsubsection{Characterizing Entropy Extremization via Analytic Combinatorics}

The entropy analysis culminates in Chapter \ref{chp:hws}, where we confirm the minimization conjecture in the asymptotic limit using results from \textit{hidden word statistics}. We show how the analytic-combinatorial methods of Flajolet, Szpankowski and Vall\'ee \cite{flajolet2001hidden,flajolet2006hidden} for dealing with the \textit{hidden pattern matching} problem can be applied to resolve the case of fixed output length and $n\rightarrow\infty$, by obtaining estimates for the entropy in terms of the moments of the posterior distribution and establishing its minimization via a measure of autocorrelation.

\subsubsection{A Software Library for the Analysis of Binary Sequences}

In the course of exploring the various mathematical problems encountered in the context of the entropy extremization problem, we developed an extensive software library with a wide range of tools.

This data analysis toolkit has been developed not only to confirm and validate our analytic results, but also to help us gain a better understanding of numerous, otherwise poorly understood, combinatorial objects to discover new properties and results. The source code and its documentation, along with post-processed data sets, will be available in the appendix. We provide a brief overview of the main utilities provided by our toolkit in Section \ref{sec:software} of Chapter \ref{chp:deletion-channel-framework}.

\subsection{Deniability in Quantum Cryptography}

In the second part of the thesis, we shift our focus to the notion of deniability in quantum cryptography. More specifically, we consider deniability in quantum key exchange, a topic that remains largely unexplored. As mentioned earlier, in the only work on this subject by Donald Beaver, it is argued that QKE is not necessarily deniable due to an eavesdropping attack that limits key equivocation. In addition to studying deniability in QKE, we also consider the relation between covert quantum communication and deniability. Moreover, we investigate the feasibility of information-theoretic deniability via quantum entanglement. We then discuss the relation between the impossibility of quantum bit commitment and deniability, as first explicitly pointed out by Beaver. Finally, we go beyond QKE and consider deniability in the context of other quantum cryptography protocols. The content of Part \ref{part:two} is largely based on our results in \cite{atashpendar2018deniability} and \cite{atashpendar2018voting}.

\subsubsection{Deniability in Classical Cryptography}

In Chapter \ref{chp:deniability-intro}, in addition to providing a brief introduction to the notion of deniability and its evolution in cryptography, we also survey the literature on deniability and focus on the most relevant results in classical cryptography in Section \ref{sec:deniability-related-work}, followed by an overview of the state-of-the-art in quantum e-voting in Section \ref{sec:quantum-e-voting}.

\subsubsection{Preliminaries in Quantum Information Processing and Cryptography}

In Chapter \ref{chp:qip}, we will provide some background knowledge by reviewing some of the most relevant concepts in quantum information theory that will be needed throughout Part \ref{part:two}. We then focus on quantum key exchange, in particular on the BB84 protocol, and uncloneable encryption, followed by an overview of important concepts in authenticated key exchange protocols. We then close by giving a quick primer on fully homomorphic encryption.

\subsubsection{Coercer-Deniable Quantum Key Exchange}

\textbf{Analysis, Modelling and Definitions}: In Chapter \ref{chp:coercer-deniable-qke} we revisit the notion of deniability in QKE and provide more insight into the eavesdropping attack aimed at detecting attempts at denial described in \cite{beaver2002deniability}. Having shed light on the nature of this attack, we show that while coercer-deniability can be achieved by uncloneable encryption (UE) \cite{gottesman2002uncloneable}, QKE obtained from UE remains vulnerable to the same attack. We briefly elaborate on the differences between our model and simulation-based deniability \cite{di2006deniable}. To provide a firm foundation, we adopt the framework and security model for quantum authenticated key exchange (Q-AKE) developed by Mosca et al. \cite{mosca2013quantum} and extend it to introduce the notion of coercer-deniable QKE formalized in terms of the indistinguishability of real and fake coercer views.

\subsubsection{Covert Quantum Communication}

We establish a connection between the concept of covert communication and deniability in Chapter \ref{chp:dcqke}, which to the best of our knowledge has not been formally considered before. More precisely, we apply results from a recent work by Arrazola and Scarani on obtaining covert quantum communication and covert QKE via noise injection \cite{AS16} to propose DC-QKE, a simple construction for coercer-deniable QKE. We prove the deniability of DC-QKE via a reduction to the security of covert QKE. Compared to the candidate PQECC protocol suggested in \cite{beaver2002deniability} that is claimed to be deniable, our construction does not require quantum computation and falls within the more practical realm of prepare-and-measure protocols.

\subsubsection{Perfect Deniability via Quantum Entanglement Distillation}

In Chapter \ref{chp:entanglement-distillation} we consider how quantum entanglement distillation can be used not only to counter eavesdropping attacks, but also to achieve information-theoretic deniability. We relate deniability to fundamental concepts in quantum information theory and suggest a generic approach to show how entanglement distillation can be used to achieve information-theoretic deniability, followed by a discussion of the relevance of other closely related results such as the relation between the impossibility of unconditional quantum bit commitment and deniability.

\subsubsection{Coercion-Resistant and Quantum-Secure Voting in Linear Time via Fully Homomorphic Encryption}

Finally, in Chapter \ref{chp:cr-and-qsafe-voting} we present an approach for performing the tallying work in the coercion-resistant JCJ voting protocol, introduced by Juels, Catalano, and Jakobsson, in linear time using fully homomorphic encryption (FHE). The suggested enhancement also paves the path towards making JCJ quantum-resistant, while leaving the underlying structure of JCJ intact. The pairwise comparison-based approach of JCJ using plaintext equivalence tests leads to a quadratic blow-up in the number of votes, which makes the tallying process rather impractical in realistic settings with a large number of voters. We show how the removal of invalid votes can be done in linear time via a solution based on recent advances in various FHE primitives such as hashing, zero-knowledge proofs of correct decryption, verifiable shuffles and threshold FHE. We conclude by touching upon some of the advantages and challenges of such an approach, followed by a discussion of further security and post-quantum considerations.

\subsubsection{Future Work and Open Questions}

We conclude by presenting some open questions in Chapter \ref{chp:conclusions}. It is our hope that this work will rekindle interest, more broadly, in the notion of deniable communication in the quantum setting, a topic that has received very little attention from the quantum cryptography community.

\stopminichaptoc{\minichaptocenabled}

\part{Information Theory Puzzles in Deletion Channels}\label{part:one}

\chapter{Problem Statement, Framework and Related Work}\label{chp:deletion-channel-framework}

\printminichaptoc{\minichaptocenabled}

The combinatorial and information theory problems addressed in Chapters \ref{chp:information-leakage}, \ref{chp:combinatorial-structures}, \ref{chp:finite-deletions}, and \ref{chp:hws}, were originally motivated by an analysis of prepare-and-measure based quantum key exchange (QKE) protocols \cite{ryan2013enhancements}, which suggested some simple changes aimed at reducing leakage of key material and improving the final key rate. These changes also included a modification of the quantum bit error rate (QBER) estimation that gave rise to an independent information theory problem and a series of related mathematical problems that we studied in \cite{atashpendar2015information, atashpendar2018clustering, atashpendar2018proof}. After having provided a short description of this change, for the remainder of this thesis, we abstract away from the details of the original context and simply state the problem as an analysis of entropy extremizing outputs in deletion channels. We now briefly describe the suggested modification, but we do not go into the details of the motivating context here, more detail can be found at \cite{ryan2013enhancements}.

For the moment we just remark that in QKE protocols it is typical for the parties, after the quantum phase, to compare bits of the fresh session key at randomly sampled positions in order to obtain an estimate of the Quantum Bit Error Rate (QBER). This indicates the proportion of bits that have been flipped as the result of either noise or eavesdropping on the quantum channel. This serves to bound the amount of information leakage to any eavesdropper, and as long as this falls below an appropriate threshold the parties continue with the key reconciliation and privacy amplification steps.

Usually, the sample set is agreed and the bits compared using unencrypted but authenticated exchanges over a classical channel, hence the positions of the compared bits are known to a potential eavesdropper and these bits are discarded. In \cite{ryan2013enhancements}, it is suggested that the sample set be computed secretly by the parties based on prior shared secrets. They still compare the bits over an unencrypted channel, but now an eavesdropper does not learn where the bits lie in the key stream. This prompts the possibility of retaining these bits, but now we must be careful to bound the information leakage and ensure that later privacy amplification takes account of this leakage.

In \cite{ryan2013enhancements}, it is suggested that further advantages of the above approach are that it provides implicit authentication at a very early stage and it ensures fairness in the selection of the sampling, i.e. neither party controls the selection.

In practice it would probably be judged too risky to retain these bits on forward secrecy grounds: leakage of the prior secret string at a later time would compromise these bits. Nonetheless, the possibility does present the rather intriguing mathematical challenge that we address in Part \ref{part:one} of this thesis.

\section{Problem Statement}\label{sec:deletion-channel-prob-statement}

Given an alphabet $\Sigma = \{ 0, 1 \}$, $\Sigma^n$ denotes the set of all $\Sigma$-strings of length $n$. Consider a bit string $y$ of length $n$ chosen at random from the space of all possible bit strings of length $n$, i.e. $y \in \Sigma^n$. More precisely, we assume that the probability distribution over the n-bit strings is flat. We assume that the bits are indexed 1 through $n$ and a subset $\pi$ of $\{1,....,n\}$ of size $m$ $(m \le n)$ is chosen at random and we reveal the bits of $y$ at these indices, preserving the order of the bits but without revealing $\pi$. Call the resulting, revealed string $x$. We assume that $\pi$ is chosen with a flat distribution over the set of subsets of $\{1,\ldots,n\}$ of size $m$, thus every subset of size $m$ is equally probable. As an example, suppose that for $n=12$ and $m=4$ we have:
\[
y=\texttt{011000011001}
\]
and we choose $\pi=\{2, 4, 5, 8\}$, then $x=\texttt{1001}$.

The question now is, what is the resulting information leakage about $y$? We assume that the ``adversary'' knows the rules of the game, i.e. she knows $n$ and she knows that the leaked string preserves the order but she does not know the chosen $\pi$ mask. In particular, can we write the leakage as a function purely of $m$ and $n$ or does it depend on the exact form of $x$? If it does depend on $x$, can we bound this?

To illustrate: if you reveal 0 bits then obviously you reveal nothing about the full string. If you reveal just one bit ($m=1$) and suppose that it is a 0, then essentially all you have revealed about the full string is that the all 1 string is not possible. At the other extreme, if you reveal all the bits $(m=n)$ then obviously you reveal all $n$ bits of the original string. For $m=n/2$, we see that from Theorem \ref{theorem:upsilon} the number of possible $y$ strings is $(2^n)/2$, which for a flat distribution would correspond to exactly 1 bit of leakage. However, in our problem the posterior distribution departs from flat so the leakage is in fact a little more than 1 bit. So intuitively the function starts off very shallow but rises very fast as $m$ approaches $n$.

In terms of terminology, the concepts presented here are closely related to the notions of subsequences, here denoted by $x$ strings, and supersequences ($y$ strings), in formal languages and combinatorics on words.

More formally, the mathematical problem encountered in the aforementioned analysis can be described as follows. From the output produced by a memoryless deletion channel from a uniformly random input of known length $n$, one obtains a posterior distribution on the channel input.  The difference between the Shannon entropy of this distribution and that of the uniform prior measures the amount of information about the channel input which is conveyed by the output of length $m$, and it is natural to ask for which outputs this is extremized. In this chapter, we conjecture on the basis of experimental data that the entropy of the posterior is minimized and maximized by the constant strings $\texttt{000\ldots}$ and $\texttt{111\ldots}$, and the alternating strings $\texttt{0101\ldots}$ and $\texttt{1010\ldots}$, respectively.

Thus, a random bit string $y$ of length $n$ emitted from a memoryless source is transmitted via an i.i.d. deletion channel such that a shorter bit string $x$ of length $m$ ($m \le n$) is received as a subsequence of $y$, after having been subject to $n-m$ deletions. Consequently, the order in which the remaining bits are revealed is preserved, but the exact positions of the bits are not known. Given a subsequence $x$, the question is to find out how much information about $y$ is revealed. More specifically, the quantity that we are interested in is the conditional entropy \cite{cover2012elements} computed over the set of candidate supersequences upon observing $x$, i.e., $H(Y|X=x)$ where $Y$ is restricted to the set of compatible supersequences as explained below.

The said information leakage is quantified as the drop in entropy \cite{shannon2001mathematical} for a fixed $x$ according to a weighted set of its compatible supersequences, referred to as the \emph{uncertainty set}. The uncertainty set, denoted by $\Upsilon_{n,x}$, contains all the supersequences that could have given rise to $x$ upon $n-m$ deletions. In an alternative proof, we show that this set's cardinality is independent of the details of $x$ and that it is only a function of $n$ and $m$. Thus, for a fixed subsequence $x$ of length $m$, we consider the set of $y$ strings of length $n$ ($n \ge m$) that can contain $x$ as a subsequence embedding. The weight distribution used in the computation of entropy is given by the number of occurrences or embeddings of a fixed subsequence in its compatible supersequences, i.e., the number of distinct ways $x$ can be extracted from $y$ upon a fixed number of deletions, denoted by $\omega_x(y)$.

Despite the specific context in which the problem was first encountered, the underlying mathematical puzzle is a close relative of several well-known challenging problems in formal languages, DNA sequencing and coding theory. In fact, the distribution of the number of times a string $x$ appears as a subsequence of $y$, lies at the center of the long-standing problem of determining the capacity of deletion channels. More precisely, knowing this distribution would give us a maximum likelihood decoding algorithm for the deletion channel \cite{mitzenmacher2009survey}. In effect, upon receiving $x$, every set of $n-m$ symbols is equally likely to have been deleted. Thus, for a received sequence, the probability that it arose from a given codeword is proportional to the number of times it is contained as a subsequence in the originally transmitted codeword. More specifically, we have $p(y|x) = p(x|y)\frac{p(y)}{p(x)} = \omega_x(y)d^{n-m}(1-d)^m \frac{p(y)}{p(x)}$, with $d$ denoting the deletion probability. Thus, as inputs are assumed to be a priori equally likely to be sent, we restrict our analysis to $\omega_x(y)$ for simplicity.

\section{Framework, Definitions and Notation}\label{sec:deletion-channel-framework}

We now provide some notation and describe our framework, which includes a set of common definitions for all the concepts that will be needed throughout Part \ref{part:one}. In terms of completeness, the building blocks presented here will suffice for all Chapters in Part \ref{part:one}, except that we leave a description of the main concepts used in hidden word statistics for Chapter \ref{chp:hws}, as they are required only for that segment of our analysis.

\subsection{Binary Subsequences and Deletion Channels}

We consider a memoryless source that emits symbols of the supersequence, drawn independently from the binary alphabet $\Sigma = \{0, 1\}$. Given an alphabet $\Sigma=\{0,1\}$, $\Sigma^n$ denotes the set of all $\Sigma$-strings of length $n$. Let $p_\alpha$ denote the probability of the symbol $\alpha \in \Sigma$ being emitted, which in the binary case simplifies to $p_\alpha = 0.5$. This means that the probability of occurrence of a random supersequence $y$ is given by $P(y) = \prod_{i=1}^n p_{y_i}$. The probability of a subsequence of length $m$ is defined in a similar manner. Throughout, we use $h(s)$ to denote the Hamming weight of the binary string $s$.

\paragraph{Notation} We use the notation $[n] = \{1, 2, \dotsc, n\}$ and $[n_1,n_2]$ to denote the set of integers between $n_1$ and $n_2$; individual bits from a string are indicated by a subscript denoting their position, starting at $1$, i.e., $y = (y_{i})_{i \in [n]} = (y_1, \dotsc, y_n) $. We denote by $|S|$ the size of a set $S$, which for binary strings also corresponds to their length in bits. We also introduce the following notation: when dealing with binary strings, $\alpha^k$ means $k$ consecutive repetitions of $\alpha \in \{0, 1\}$. Throughout, we use $\sigma$ to refer to the constant strings $x=\texttt{1}^m$ and $x=\texttt{0}^m$ for succinctness.

\paragraph{Subsequences and Supersequences} Given $x \in \Sigma^m$ and $y \in \Sigma^n$, let $x = x_1 x_2 \cdots x_m$ denote a subsequence obtained from $y = y_1 y_2 \cdots y_n$ with a set of indexes $1 \le i_1 < i_2 < \cdots < i_m \le n$ such that $y_{i_1} = x_1, y_{i_2} = x_2, \dotsc, y_{i_m} = x_m$. Subsequences are obtained by deleting characters from the original string and thus adjacent characters in a given subsequence are not necessarily adjacent in the original string.

\paragraph{Projection Masks} We define $y_\pi = (y_i)_{i \in \pi} = x$ to mean that the string $y$ filtered by the mask $\pi$ gives the string $x$. Let $\pi$ denote a set of indexes $\{j_1, \dotsc, j_m\}$ of increasing order that when applied to $y$, yields $x$, i.e., $x = y_{j_1} y_{j_2} \cdots y_{j_m}$ and $1 \le j_1 < j_2 \cdots j_m \le n$.

\paragraph{Deletion Masks} A deletion mask $\delta$ represents the set of indexes that are deleted from $y$ to obtain $x$, i.e., $\delta_i \in [n] \setminus \pi$ and $|\delta| = n-m$, whereas a projection mask $\pi$ denotes indexes that are preserved. Thus, similarly, $\delta$ is a subset of $[n]$ and the result of applying a mask $\delta$ on $y$ is denoted by $y_\delta = x$.

\paragraph{Compatible Supersequences} We define the \emph{uncertainty set}, $\Upsilon_{n,x}$, as follows. Given $x$ and $n$, this is the set of $y$ strings that could project to $x$ for some projection mask $\pi$.
\begin{align*}
\Upsilon_{n,x}:=\{y \in \{0,1\}^n: (\exists \pi) [y_\pi=x] \} = \{y \in \{0,1\}^n: (\exists \delta) [y_{\delta}=x] \}
\end{align*}
It was shown by Levenshtein \cite{levenshtein1974elements} that that the cardinality of $\Upsilon_{n,x}$ is independent of the form of $x$ and that it is only a function of $n$ and $m$, i.e.,
\begin{equation}\label{eq:upsilon-cardinality}
|\Upsilon_{n,x}| = \sum_{r=m}^n \binom{n}{r}.
\end{equation}

\paragraph{Number of Masks or Embeddings} Let $\omega_x(y)$ denote the number of distinct ways that $y$ can project to $x$:
\[
\omega_x(y) := | \{ \pi \in \mathcal{P}([n]): y_\pi = x \} |= |\{\delta \in \mathcal{P}([n]): y_{\delta} = x \}|
\]
we refer to the number of masks associated with a pair $(y, x)$ as the weight of $y$, i.e., the number of times $x$ can be embedded in $y$ as a subsequence.
Moreover, let $\mu_{n,x}$ denote the number of configurations for $n$ and $x$, i.e. the number of pairs $\{y,\pi\}$ such that $y_\pi=x$. It is easy to see that this is given by:
\begin{equation}
\mu_{n,x}=\binom{n}{m}\cdot 2^{n-m}
\end{equation}

\paragraph{Initial Projection Masks or Canonical Embeddings} Given $y_\pi =x$, we define $\pi$ to be initial if there is no lexicographically earlier mask $\pi'$ such that $y_{\pi'} = x$. $\pi'$ is a lexicographically earlier mask than $\pi$ if, for some $r$, the smallest $r$ members of $\pi$ and $\pi'$ are the same, but the $(r+1)$-th of $\pi'$ is strictly smaller than that of $\pi$. Throughout, we will use $\tilde{\pi}$ to denote an initial projection mask. The first embedding of a subsequence $x$ in $y$ is also often referred to as the \emph{canonical} embedding in the literature. Note that for a fixed mask or embedding $\pi$, the members of $y$ up to the last member of $\pi$ are completely determined if $\pi$ is initial.

\paragraph{Run-Length Encodings} A substring $T$ of a string $Y=y_1 y_2 \ldots y_n$ over $\Sigma$ is called a \emph{run} of $Y$ if $T$ is a consecutive sequence of the same character (i.e., $T \in \alpha^{+}$ for an $\alpha \in \Sigma$). Let $\mathcal{R}_{x, \alpha}$ denote the set of runs of $\alpha$ in $x$. The notion of run-length encoding will be central to our analysis. Given an $n$-bit binary string $y$, its \emph{run-length encoding} (RLE) is the sequence $r_j = (a_j, b_j)$, $1 \leq j \leq m$, such that
\begin{equation*}
y = a_1^{b_1}a_2^{b_2} \cdots a_m^{b_m}, \qquad m \leq n.
\end{equation*}
with $a_j \in \{0,1\}$ and $b_j \in \{1, \dotsc, n\}$. This encoding is unique if we assume that $a_{i} \neq a_{i+1}$, at which point we only need to specify a single $a_i$ (e.g., the first one) to deduce all the others.
Thus the RLE for a string $y$ is denoted by
\begin{equation*}
y = (a_1; b_1, b_2, \dotsc, b_m).
\end{equation*}
When the value of $a_1$ is irrelevant, which will often be the case later on\footnote{Indeed, if $x=y_\pi$, then $\overline{x}=\overline{y}_\pi$ and $\omega_x(y)= \omega_{\overline{x}}(\overline{y})$.}, we will drop it form the notation. Consecutive zeros or ones in a binary string will be referred to as \emph{blocks} or \emph{runs}.
\begin{example}
Let $y = \texttt{0011010001}$; then we have $y = (0; 2, 2, 1, 1, 3, 1)$ as the first bit is zero; and we have 2 zeros, 2 ones, 1 zero, 1 one, 3 zeros, 1 one. Alternatively, $(2, 2, 1, 1, 3, 1)$ designates simultaneously $\texttt{0011010001}$ and $\texttt{1100101110}$.
\end{example}

\subsection{Measures of Information and Entropy}\label{subsec:classical-information-entropy}

Since it is not the purpose of this thesis to provide an extensive discussion on the history and development of entropy in information theory, we briefly touch upon a number of key points and encourage the reader to refer to Shannon's original paper \cite{shannon2001mathematical} and standard textbooks \cite{cover2012elements,mackay2003information,wilde2013quantum} for a more thorough coverage.

Intuitively, information entropy can be viewed as a measure of surprise upon learning the outcome of a random variable, or the expected amount of surprise that a random variable possesses. This intuitive view is consistent with everyday experience in that events with lower probability of occurrence surprise us more, while events that are more likely to occur surprise us less. This intuitive view in terms of surprise can be naturally translated into how much information one gains upon learning the value of a random variable.

While the notion of \emph{information} can be rather elusive\footnote{This simply refers to its meanings in different contexts revolving around data, knowledge, understanding, uncertainty, cognitive observer, etc.}, in his seminal work on a mathematical theory of communication \cite{shannon2001mathematical}, Shannon formalized information via the concept of entropy\footnote{Inspired by Boltzman's work in statistical mechanics, while the term \csquote{entropy} was chosen based on a suggestion by John von Neumann.}, which refers to a weighted sum of the information content\footnote{Basically equating information with counting possibilities in terms of bits of information needed for representing a variable, e.g. $\mathrm{log}_2(n)$ bits for a system taking on one of $n$ states.} of each realization $x$ of random variable $X$, namely $-\mathrm{log}_2(p_X(x))$. Thus, given a discrete random variable $X$ with probability distribution $p_X(x)$, the entropy $H(X)$ measures information in units of bits, defined as
\[
H(X) \equiv - \sum_{x} p_X(x) \cdot \mathrm{log}_2 (p_X(x)).
\]

Another equivalent way of viewing information is in terms of the amount of communication needed to convey it to a receiver, essentially a restatement of Shannon's source coding theorem equating the limit of compression of an information source to the Shannon entropy of the source.

This also provides an operational meaning for the Shannon entropy. In other words, one cannot compress a source down to a code rate that is lower than the Shannon entropy of the source. This roughly translates into the fact that $n$ independent and identically distributed (i.i.d.) random variables, each with entropy $H(X)$, cannot be compressed\footnote{This statement pertains to lossless compression.} into less than $n \cdot H(X)$ bits, as $n \rightarrow \infty$. In terms of properties, the Shannon entropy is non-negative, concave and permutation invariant\footnote{Its value is invariant w.r.t. permutations of event realizations.}, and its minimum value is attained if $X$ is deterministic (event with probability 1), while its maximum value is satisfied by a uniformly random $X$, i.e., a uniform distribution or all events being equiprobable.

Shannon's work gave rise to two fundamental results, namely the source coding and the channel coding theorem, ideas that jump-started the field of information and coding theory. The former captures the limit of data compression in terms of the entropy of an information source, while the latter deals with the limits of optimal error correction and the construction of error correcting codes for the purpose of faithfully transmitting data over a noisy channel in an efficient manner. Both ideas depend on redundancy in that the former is aimed at removing it, while the latter adds carefully prepared redundancy.

Regarding our analysis, the obvious follow-on question to the problem posed at the start of this chapter is: what is the appropriate measure of information to use? Perhaps the simplest measure is the Hartley measure, the log of the cardinality of the uncertainty set. This coincides with the Shannon measure if the probability distribution is uniform. In this case the solution is simple as we will see below: the cardinality of the uncertainty set is a simple function of $n$ and $m$. However, the probability distribution turns out to be rather far from uniform, so the Hartley measure does not seem appropriate here.

Thought of purely as an information theory puzzle, the standard commonly used measure is Shannon's \cite{shannon2001mathematical}. For this we have a number of interesting results and observations. In particular, our observations suggest that the maximum leakage for all $n$ and $m$ occurs for the all zero or all one $x$ strings and we have a formula for the leakage in these cases. However, we have not yet been able to prove this conjecture, although we do have intuitions as to why this appears to be the case.

Given the cryptographic motivation for the problem, it is worth considering whether alternative information measures are in fact
more appropriate. The Shannon measure has a very specific interpretation: the expected number of binary questions required to
identify the exact value of the variable. In various cryptographic contexts, this might not be the most appropriate interpretation. For example, in some situations it might not be necessary to pin down the exact value and a good approximation may be damaging. In our context, the session key derived from the key reconciliation phase will be subjected to privacy amplification to reduce the adversary's knowledge of the key to a negligible amount. What we really need therefore is a measure of the leakage that can be used to control the degree of amplification required. This question has been extensively studied in \cite{MaurerWolf97,RennerWolf05,maurer1993secret,Cachin97,CachinMaurer97}, and below we summarize the key results.

Various measures of entropy may be applicable depending on the parameters of the context in question, such as the scheme used for privacy amplification, e.g. universal hashing vs. randomness extractors or whether a distinction is made between passive adversaries and active adversaries \cite{MaurerWolf97}. As noted in the works of Bennett et al. \cite{CachinMaurer97,Bennett95}, the R\'{e}nyi entropy \cite{renyi1961measures,mackay2003information} provides a lower bound on the size of the secret key $s^{\prime}$ distillable from the partially secret key $s$ initially shared by Alice and Bob. Moreover, it is shown in \cite{MaurerWolf97}, that the min-entropy provides an upper bound on the amount of permissible leakage and specific constraints are derived as a function of the min-entropy of $s$ and the length of the partially secret string. More recently, Renner and Wolf show in \cite{RennerWolf05} that the Shannon entropy $H$ can be generalized and extended to two simple quantities, $H_{0}^{\varepsilon}$ and $H_{\infty}^{\varepsilon}$, called smooth R\'{e}nyi entropy values, which provide tight bounds for privacy amplification and information reconciliation in contexts such as QKE, where the assumption of having independent repetitions of a random experiment is generally not satisfied.

For the purpose of our study, we consider the following measures of information, which can be considered as special cases of the R\'{e}nyi Entropy.

\paragraph{\textbf{R\'{e}nyi Entropy of order $\alpha$.}} For $\alpha \ge 0$ and $\alpha \ne 1$, the R\'{e}nyi entropy of order $\alpha$ of a random variable $X$ is

\begin{equation}
H_{\alpha}(X) = \frac{1}{1-\alpha} \text{log}_2 \sum_{x \in \mathcal{X}} P_X(x)^{\alpha}.
\end{equation}

\paragraph{\textbf{Hartley Entropy.}} The Hartley measure corresponds to R\'{e}nyi entropy of order zero and is defined as

\begin{equation}\label{eq:hartley-entropy}
H_{0}(X) := - \text{log}_2 \left\vert{\mathcal{X}}\right\vert.
\end{equation}

\paragraph{\textbf{Second-order R\'{e}nyi Entropy.}} For $\alpha=2$, we get the collision entropy, also simply referred to as the R\'{e}nyi entropy

\begin{equation}\label{eq:renyi_two}
R(x) = H_{2}(X) := - \text{log}_2 \sum_{x \in \mathcal{X}} P_X(x)^2.
\end{equation}

\paragraph{\textbf{Shannon Entropy.}} As $\alpha \rightarrow 1$, in the limit we get the Shannon entropy of a random variable $X$

\begin{equation}\label{eq:shannon-h}
H(X) = - \sum_{x \in \mathcal{X}} P_X(x) \cdot \text{log}_2 P_X(x).
\end{equation}

\paragraph{\textbf{Min-Entropy.}} In the limit, as $\alpha \rightarrow \infty$, $H_\alpha$ converges to the min-entropy of a random variable $X$

\begin{equation}\label{eq:min-entropy}
H_{\infty}(X) := - \text{log}_2 \max_{x \in \mathcal{X}} (P_X(x)).
\end{equation}
As noted in \cite{MaurerWolf97}, the entropy measures given above satisfy

\begin{equation}\label{eq:h_relation}
H(X) \ge H_2(X) \ge H_{\infty}(X)
\end{equation}

\paragraph{\textbf{Conditional Entropy.}} Finally, another entropic measure that is central to our work is that of conditional entropy. Given two random variables $X$ and $Y$ that are not statistically independent (i.e. correlated), let $i(y|x)$ denote the conditional information content
\begin{equation}
	i(y|x) \equiv - \mathrm{log}_2\left(p_{Y|X}(y|x)\right)
\end{equation}
and the entropy $H(Y|X=x)$ of random variable $Y$ conditioned on a particular event or realization $x$ of the random variable $X$ corresponds to the expected conditional information content w.r.t. $(Y|X=x)$
\begin{equation}
	H(Y|X=x) \equiv \E_{Y|X=x}{\{ i(Y|x) \}} = - \sum_{y} p_{Y|X}(y|x)\mathrm{log}_2(p_{Y|X}(y|x)).
\end{equation}
If the expectation is with respect to both $X$ and $Y$, we have
\begin{align}
	H(Y|X) &\equiv \sum_{x} p_X(x)H(Y|X=x) \\
	&= -\sum_{x}p_X(x)\sum_{y}p_{Y|X}(y|x)\mathrm{log}_2(p_{Y|X}(y|x)) \\
	&= -\sum_{y,x}p_{Y,X}(y,x)\mathrm{log}_2(p_{Y|X}(y|x)).
\end{align}
where $p_{Y,X}(y,x)$ denotes the joint probability distribution of $Y$ and $X$. Finally, the entropy $H(X)$ is greater than or equal to the conditional entropy $H(X|Y)$, i.e. conditioning does not increase entropy
\[
H(X) \ge H(X|Y)
\]
with equality iff $X$ and $Y$ are independent.

\subsection{Entropy of Embeddings}

For a fixed subsequence $x$ of length $m$, the underlying weight distribution used in the computation of the entropy is defined as follows. Upon receiving a subsequence $x$, we consider the set of compatible supersequences $y$ of length $n$ (denoted by $\Upsilon_{n,x}$) that can project to $x$ upon $n-m$ deletions. Every $y \in \Upsilon_{n,x}$ is assigned a weight given by its number of masks $\omega_x(y)$, i.e., the number of times $x$ can be embedded in $y$ as a subsequence. We consider the conditional Shannon entropy $H(Y|X=x)$ where $Y$ is confined to the space of compatible supersequences $\Upsilon_{n,x}$. The total number of masks in $\Upsilon_{n,x}$ is given by
\begin{equation}
\label{eq:numbmask}
\mu_{n,m}=\binom{n}{m} \cdot 2^{n-m}
\end{equation}
Let $P_x$ denote the normalized weight distribution given below
\begin{equation*}\label{eq:subseq-prob-distribution1}
P_x=\{P(Y = y | X = x)\text{ for } y \in \Upsilon_{n,x} \}.
\end{equation*}
where $P(Y = y | X = x)$ is given by
\begin{align*}
P(Y = y | X = x) &= \frac{P(Y = y \wedge X = x)}{P(X = x)} = \frac{P(X = x| Y = y) \cdot P(Y = y)}{P(X = x)}
= \frac{\frac{|\{ \pi: \pi(y)=x \}|}{\binom{n}{m}}2^{-n}}{P(X =x)} \\
& = \frac{\omega_x(y)2^{-n}}{\binom{n}{m}P(X =x)} = \frac{\omega_x(y)2^{-n}}{\binom{n}{m}\sum_{y'}P(Y = y')P(X = x | Y = y')} \\
& = \frac{\omega_x(y)2^{-n}}{\binom{n}{m}\sum_{y'}\frac{\omega_x(y')}{\binom{n}{m}}2^{-n}} = \frac{\omega_x(y)}{\sum_{y'}\omega_x(y')} = \frac{\omega_x(y)}{\mu_{n,m}}
\end{align*}
We therefore have \begin{equation}\label{eq:subseq-prob-distribution}
P_x=\Bigg\{\frac{\omega_x(y_1)}{\mu_{n,m}}, \ldots, \frac{\omega_x(y_n)}{\mu_{n,m}} \Bigg\}.
\end{equation}
Finally, for simplicity we use $H_n(x)$ throughout this work to refer to the Shannon entropy of a distribution $P$ corresponding to a subsequence $x$ as defined below
\begin{equation}\label{eq:shannon-entropy}
H_n(x) = -\sum_i p_i \cdot \log_2(p_i)
\end{equation}
where $p_i$ is given by
\begin{equation*}
p_i = \frac{\omega_x(y_i)}{\mu_{n,m}}.
\end{equation*}
An example illustrating these concepts is given in \Cref{table:upsilon-distribution}. In addition to the distribution of weights, i.e., number of masks per $y$, Hamming-weight groupings of supersequences are indicated by horizontal separators.
\begin{table}[t]
	\tiny
    \caption{Clusters of Supersequences and Distribution of Subsequence Embeddings}
    \label{table:upsilon-distribution}
	\centering
		\begin{tabular}{|c|c|c|}
		\hline
		 \multicolumn{3}{|c|}{$x=\texttt{110}$}  \\ \hline
		$y$ & $\pi_i$ &$\omega_x(y)$ \\ \hline
		$\texttt{00110}$ & $\{3, 4, 5 \}$ & 1                     \\ \hline
		$\texttt{01010}$ & $\{ 2, 4, 5 \}$ & 1                     \\ \hline
		$\texttt{01100}$ & $\{2, 3, 4\}, \{2, 3, 5\}$ & 2                     \\ \hline
		$\texttt{10010}$ & $\{ 1, 4, 5 \}$ & 1                     \\ \hline
		$\texttt{10100}$ & $\{1, 3, 4\}, \{1, 3, 5\}$ & 2                     \\ \hline
		$\texttt{ 11000 }$ & $\{1, 2, 3\}, \{1, 2, 4\}, \{1, 2, 5\}$ & 3                     \\ \hline \hline \hline
		$\texttt{ 01101 }$ & $\{ 2, 3, 4 \}$  & 1                    \\ \hline
		$\texttt{ 01110 }$ & $\{2, 3, 5\}, \{2, 4, 5\}, \{3, 4, 5\}$ & 3                     \\ \hline
		$\texttt{ 10101 }$ & $\{ 1, 3, 4 \}$ & 1                     \\ \hline
		$\texttt{ 10110 }$ & $\{1, 3, 5\}, \{1, 4, 5\}, \{3, 4, 5\}$ & 3                     \\ \hline
		$\texttt{ 11001 }$ & $\{1, 2, 3\}, \{1, 2, 4\}$ & 2                     \\ \hline
		$\texttt{ 11100 }$ & $\{1, 2, 4\}, \{1, 2, 5\}, \{1, 3, 4\}, \{1, 3, 5\}, \{2, 3, 4\}, \{2, 3, 5\}$ & 6                     \\ \hline \hline \hline
		$\texttt{ 11011 }$ &  $\{ 1, 2, 3 \}$ & 1                    \\ \hline
		$\texttt{ 11101 }$ &  $\{1, 2, 4\}, \{1, 3, 4\}, \{2, 3, 4\}$ & 3                    \\ \hline
		$\texttt{ 11110 }$ &  $\{1, 2, 5 \},\{1, 3, 5 \},\{1, 4, 5 \},\{2, 3, 5 \},\{2, 4, 5 \},\{3, 4, 5 \}$ & 6                    \\ \hline
	\end{tabular}
	\begin{tabular}{|c|c|c|}
		\hline
		\multicolumn{3}{|c|}{$x=\texttt{101}$}  \\ \hline
		$y$ & $\pi_i$ & $\omega_x(y)$  \\ \hline
	  $\texttt{ 00101 }$ & $\{3, 4, 5 \}$ & 1                     \\ \hline
		$\texttt{ 01001 }$ & $\{2, 3, 5\}, \{2, 4, 5 \}$ & 2                     \\ \hline
		$\texttt{ 01010 }$ & $\{ 2, 3, 4 \}$ & 1                     \\ \hline
		$\texttt{ 10001 }$ & $\{1, 2, 5\}, \{1, 3, 5\}, \{1, 4, 5\}$ & 3                     \\ \hline
		$\texttt{ 10010 }$ & $\{1, 2, 4\}, \{1, 3, 4\}$  & 2                    \\ \hline
		$\texttt{ 10100 }$ & $\{ 1, 2, 3 \}$   & 1                   \\ \hline \hline \hline
		$\texttt{ 01011 }$ & $\{ 2, 3, 4 \}, \{ 2,3,5\}$ & 2                     \\ \hline
		$\texttt{ 01101 }$ & $\{ 2, 4, 5 \}, \{3,4,5\}$ & 2                     \\ \hline
		$\texttt{ 10011 }$ & $\{1, 2, 4\}, \{1, 2, 5\}, \{1, 3, 4\}, \{1, 3, 5\}$ & 4                      \\ \hline
		$\texttt{ 10101 }$ & $\{1, 2, 3\}, \{1, 2, 5\}, \{1, 4, 5\}, \{3, 4, 5\}$ & 4                     \\ \hline
		$\texttt{ 11001 }$ & $\{1, 3, 5\}, \{1, 4, 5\}, \{2, 3, 5\}, \{2, 4, 5\}$  & 4                    \\ \hline
		$\texttt{ 11010 }$ & $\{1, 3, 4\}, \{2, 3, 4\}$  & 2                    \\ \hline \hline \hline
		$\texttt{ 10111 }$ &  $\{1, 2, 3\}, \{1, 2, 4\}, \{1, 2, 5\}$ &  3                   \\ \hline
		$\texttt{ 11011 }$ &  $\{1, 3, 4\}, \{1, 3, 5\}, \{2, 3, 4\}, \{2, 3, 5\}$ & 4                    \\ \hline
		$\texttt{ 11101 }$ &  $\{1, 4, 5\}, \{2, 4, 5\}, \{3, 4, 5\}$ & 3                    \\ \hline
	\end{tabular}
\end{table}

\section{Related Work}

The sort of problems addressed in our work are close relatives of long-standing open problems mainly found in studies dealing with the combinatorics of subsequences and supersequences, and in the context of deletion channels. Here we give a brief survey of the most relevant and closely related works in the literature.

\subsection{Subsequences and Supersequences}

Studies involving subsequences and supersequences encompass a wide variety of problems that arise in various contexts such as formal languages, coding theory, computer intrusion detection and DNA sequencing to name a few. Despite their prevalence in such a wide range of disciplines, they remain largely unexplored and still present a considerable wealth of unanswered questions.

In the realm of stringology and formal languages, the problem of determining the number of distinct subsequences obtainable from a fixed number of deletions, and closely related problems, have been studied extensively in  \cite{chase1976subsequence,flaxman2004strings,hirschberg1999bounds,hirschberg2000tight}. It is worth noting that the same entropy minimizing and maximizing strings conjectured in the present work, have been shown to lead to the minimum and maximum number of distinct subsequences, respectively. The problems of finding shortest common supersequences (SCS) and longest common subsequences (LCS) represent two well-known NP-hard problems \cite{jiang1995approximation,middendorf1995finding,middendorf2004combined} that involve subproblems similar to our work. Finally, devising efficient algorithms for subsequence combinatorics based on dynamic programming for counting the number of occurrences of a subsequence in DNA sequencing is yet another important and closely related line of research \cite{rahmann2006subsequence,elzinga2008algorithms}.

Fundamental results can be found in the works of Levenshtein, Hirschberg and Calabi \cite{hirschberg1999bounds,hirschberg2000tight,calabi1969some,levenshtein2001efficient} who provide tight upper and lower bounds on the number of distinct subsequences. Furthermore, it was proved by Chase \cite{chase1976subsequence} that the number of distinct $m$-long subsequences is maximized by repeated permutations of an alphabet $\Sigma$, i.e. no letter appears twice without all of the other letters of $\Sigma$ intervening. Flaxman et al. \cite{flaxman2004strings} also provide a probabilistic method for determining the string that maximizes the number of distinct subsequences.

For a thorough presentation of efficient algorithms for computing the number of distinct subsequences, e.g. using dynamic programming, and related problems in the realm of DNA sequencing, we refer the reader to \cite{middendorf1994supersequences,lothaire2005applied,rahmann2006subsequence,elzinga2008algorithms}.

\subsection{Coding Theory and Deletion Channels}

In coding theory, and more specifically in the context of insertion and deletions channels, similar long-standing problems have been studied extensively, and yet many problems still remain elusive in the context of insertion and deletions channels. This includes designing optimal coding schemes and determining the capacity of deletion channels, both of which incorporate the same underlying combinatorial problem addressed in the present work.

In a deletion channel \cite{mitzenmacher2009survey}, for a received sequence, the probability that it arose from a given codeword is proportional to the number of different ways it is contained as a subsequence in that codeword. This translates into a maximum likelihood decoding for deletion channels as follows: For a received sequence, we count the number of times it appears as a subsequence of each codeword and we choose the codeword that admits the largest count. The problem of determining and bounding these particular distributions remains unexplored and presents a considerable number of open questions. Case-specific results for double insertion/deletion channels can be found in \cite{swart2003note}. Moreover, improved bounds for the number of subsequences obtained via the deletion channel and proofs for how balanced and unbalanced strings lead to the highest and lowest number of distinct subsequences are given in \cite{liron2012characterization}.

The studies in \cite{ullman1967capabilities,swart2003note,kanoria2013optimal} consider a finite number of insertions and deletions for designing correcting codes for synchronization errors and Graham \cite{graham2015binary} studies the problem of reconstructing the original string from a fixed subsequence. More recent works on the characterization of the number of subsequences obtained via the deletion channel can be found in \cite{sala2013counting,sala2015three,cullina2014improvement,liron2015characterization}. Another important body of research in this area is dedicated to deriving tight bounds on the capacity of deletion channels \cite{diggavi2007capacity,kalai2010tight,rahmati2013bounds,cullina2014improvement} and developing bounding techniques \cite{ordentlich2014bounding}.

In terms of more directly related combinatorial objects, Cullina, Kiyavash and Kulkarni \cite{cullina2012coloring} provide a graph-theoretic approach for deletion correcting codes, which among other things, extends Levenshtein's \cite{levenshtein1974elements} result on the size of $\Upsilon_{n,x}$ being only a function of $n$ and $m$, to supersequences of a particular length and Hamming weight. In another more recent work by the same authors \cite{cullina2016restricted}, this result for binary strings is extended to $q$-ary strings of a particular composition, where the composition of a $q$-ary string $x$ refers to a vector of $q$ nonnegative integers, which denote the number of times each symbol in the alphabet appears in the string. The authors \cite{cullina2016restricted} show that the number of distinct supersequences of a particular composition depends only on the composition of the original string, from which the distinct supersequences can be obtained via $n-m$ insertions.

Perhaps rather surprisingly, the problem of determining the number of occurrences of a fixed subsequence in random sequences has not received the same amount and level of attention from the various communities. The state-of-the-art in the finite-length regime remains rather limited in scope. More precisely, the distribution of the number of occurrences constitutes a central problem in coding theory, with a maximum likelihood decoding argument, which represents the holy grail in the study of deletion channels. A comprehensive survey, which among other things, outlines the significance of figuring out this particular distribution is given by Mitzenmacher in \cite{mitzenmacher2009survey}.

\subsection{Distribution of Subsequence Embeddings}

Despite being of interest to various disciplines, the problem of determining the number of occurrences or embeddings of a fixed subsequence in random sequences had not been comprehensively studied until Flajolet, Szpankowski and Vall\'ee gave a complete characterization of the statistics of this problem in the asymptotic limit \cite{flajolet2001hidden,flajolet2006hidden}. However, the state-of-the-art in the finite-length domain remains rather limited in scope, as mentioned above in the context of maximum likelihood decoding in deletion channels.

Another highly relevant area of research worth mentioning corresponds to the work of Gentleman and Mullin \cite{gentleman1989distribution} in the context of DNA sequencing, which seems to have gone largely unnoticed by the other communities. Their analysis revolves around the characterization of the distribution of the frequency of occurrence of nucleotide subsequences based on their overlap capabilities. The overlap capability of a subsequence is central to their approach for deriving the expectation and the variance of the distribution. This is very much in line with the notion of autocorrelation used by Flajolet et al. in \cite{flajolet2001hidden,flajolet2006hidden} almost fifteen years later. A similar study based on \cite{gentleman1989distribution}, also related to nucleotide subsequences, is available at \cite{wu2005distributions}.

Although the finite-length domain still remains quite elusive, here we make use of an asymptotic description of the statistics of hidden patterns given by Flajolet et al. in \cite{flajolet2001hidden,flajolet2006hidden} to establish the minimal entropy conjecture. To the best of our knowledge, an analysis focusing on a characterization of the mutual information for the deletion channel \cite{drmota2012mutual} is the only study that directly applies results from hidden word statistics to an information-theoretic analysis.

\section{A Software Toolkit for Binary Sequences}\label{sec:software}

The theoretical findings presented in this thesis were complemented by an extensive software package with a wide range of functionalities, mainly implemented in Python, the R programming language for statistical computing and $\mathrm{CSP}_M$ (tested with FDR3). This data analysis toolkit has been developed not only to confirm and validate our analytic results, but it was also aimed at helping us gain a better understanding of numerous, otherwise poorly understood, combinatorial objects and to discover new properties and results.

The source code and its corresponding documentation, along with post-processed data sets, can be found in the appendix. Here we limit ourselves to a brief overview of the main modules and their respective capabilities, as shown in Table \ref{tab:binseqpy}. Throughout Part \ref{part:one}, we will present plots and experimental results that were generated by our toolkit.

\begin{table}
	\caption{The main functionalities available in the BinSeqPy toolkit.}
	\label{tab:binseqpy}
	\centering
	\begin{tabular}{|l|l|}
	\hline
	\multicolumn{1}{|c|}{\textbf{Module}} & \multicolumn{1}{c|}{\textbf{Functionality}} \\ \hline
	Combinatorics                   & Reusable functions for combinatorial analysis.  \\ \hline
	Clustering                      & Customizable clustering techniques for (sub/super)-sequences. \\ \hline
	Entropy                               & Analysis of shifts in entropy. \\ \hline
	Distribution                     & Analysis of change in weight distribution. \\ \hline
	DynamicProg               & Efficient algorithms for computationally intensive problems. \\ \hline
	DataAnalysis                 & Generation and analysis of combinatorial structures.   \\ \hline
	HWS                                   & Algorithms for hidden word statistics. \\ \hline
	Validation              & Validation and verification of analytical results.  \\ \hline
	Plotting                              & Plotting for data visualization and statistical analysis.  \\ \hline
	\end{tabular}
\end{table}

\stopminichaptoc{\minichaptocenabled}

\chapter{Binary Subsequences and Entropy Extremization}\label{chp:information-leakage}

\printminichaptoc{\minichaptocenabled}

\section{Introduction}

In this chapter, we present our initial approach for solving the entropy problem as well as the obtained results in Section \ref{sec:info-leakage}, followed by a few discussions on privacy amplification and alternative approaches in Section \ref{sec:notes}. Section \ref{sec:simulations} describes how we use simulations to obtain experimental results and to tackle some of the problems addressed in this paper for which deriving analytic expressions proved to be difficult. Finally, we conclude by summarizing our contributions in Section \ref{sec:chp-leakage-conclusions}.

\section{Information Leakage}\label{sec:info-leakage}

In this section we show that the size of the uncertainty set only depends on $n$ and $m$ and provide an expression for computing its cardinality, followed by a proof. We then analyze the amount of information leakage and observe that the maximal leakage corresponds to the $x$ string being all zeros or all ones and that the minimum leakage corresponds to the alternating $x$ strings. We also derive closed form expressions for the maximum leakage (minimal entropy) in terms of $n$ and $m$ for the measures of entropy introduced in Section \ref{sec:deletion-channel-framework}.

\subsection{Cardinality of the Uncertainty Set}\label{subsec:upsilon}

\begin{theorem}\label{theorem:upsilon}

For given $n$ and $m$ the cardinality of $\Upsilon_{n,x}$ is independent of the exact $x$ string. Furthermore, $|\Upsilon_{n,x}|$ is given by:

\begin{equation}
	|\Upsilon_{n,x}| = \sum\limits_{r=m}^{n} \binom{n}{r}
\end{equation}

\end{theorem}

\begin{proof}

	$\gamma_{n,m}$ satisfies the following recursion:

	\begin{equation}
	\gamma_{n,m}=\gamma_{n-1, m}+\gamma_{n-1, m-1}
	\end{equation}
	with base cases: $\gamma_{n,n}=1$ and $\gamma_{n,0}=2^n$.

	The base cases are immediate. To see how the recursion arises, consider the following cases:

	\begin{itemize}
		\item Partition the y strings into those that have a mask overlapping the first bit of y and those that do not.

		\item For the former, we can enumerate them simply as the number of y strings of length $n-1$ with $\geq 1$ projections to the tail of x, i.e. $\gamma_{n-1,m-1}$.

		\item For the latter, the number is just that of the set of $y$ strings of length $n-1$ with $\geq1$ projection to $x$, which has length $m$, i.e. $\gamma_{n-1, m}$.
	\end{itemize}

	The solution to this recursion with the given base cases is:

	\begin{equation}
	\gamma_{n,m} = \sum\limits_{r=m}^{n} \binom{n}{r}
	\end{equation}

	This is most simply seen by observing that the recursion is independent of the exact $x$, hence we can choose the $x$ string comprising $m$ $0$s. Now we see that $|\Upsilon_{n,x}|$ is simply the number of distinct $y$ strings with at least $m$ $0$s, and the result follows immediately.

\end{proof}

If the conditional distribution over $\Upsilon_{n,x}$ given the observation of $x$ were flat, we would be done: we could compute the entropy immediately. However, it turns out the distribution is far from flat, and indeed its shape depends on the exact $x$ string. This is due to the fact that given an observed $x$, the probability that a $y$ gave rise to it is proportional to the weight of $y$, i.e. the number of ways that $y$ could project to $x$, i.e. $| \{ \pi \in \mathcal{P}([n]): y_\pi = x \} |$. This can vary between 1 and $\binom{n}{m}$.

\subsection{Shannon Entropy}

Here we will assume that the leakage is measured as the drop in the Shannon entropy of the space of possible $y$ strings. Clearly, before any observation the entropy is $n$ bits. Based on experimental data, we observe that the maximal leakage occurs when $x$ is either the all 0 or the all 1 string and we derive an expression for the corresponding entropy of $\Upsilon_{n,x}$.

\subsection{Minimal Shannon Entropy}

Assuming that the maximal leakage occurs for the all zero (or all one) $x$ string we derive the formula for the maximal leakage (minimum entropy of
$\Upsilon_{n,x}$) as follows: observe that the number of elements of $\Upsilon_{n,x}$ with $j$ 1's is $\binom{n}{j}$. Note further that for given $j$ the number of ways that a $y$ string with $j$ 1's can yield $x$ is $\binom{n-j}{m}$. Consequently, as shown in Eq. \ref{eq:subseq-prob-distribution} the probability that $y$ was a given string with $j$ 1's given the observation of $x$ is:
\begin{equation}\label{eq:max_prob}
P(y_j | x) = \frac{\binom{n-j}{m}}{\mu_{n,m}}
\end{equation}
where $\mu_{n,m}$ is the normalization, i.e. the total number of configurations that could give rise to a given $x$:
\[
\mu_{n,m}=\binom{n}{m} \times 2^{n-m}
\]
Now, inserting these terms into the formula for the Shannon entropy given in Eq. \ref{eq:shannon-entropy}, we get:

\begin{equation}
H_{n,m} = - \sum\limits_{j=0}^{n-m} \binom{n}{j} \times \frac{\binom{n-j}{m}}{\mu_{n,m}} \times \text{log}_2 \left( \frac{\binom{n-j}{m}}{\mu_{n,m}} \right)
\end{equation}

For the original cryptographic motivation of this problem, more specifically in the context of privacy amplification, it is arguably an upper bound on the maximum leakage or the amount of information that Eve has gained that we are after \cite{Bennett88}. However, it is also interesting to better understand the mean and range of the entropy for given $n$ and $m$, but coming up with analytic forms for these appears to be much harder. We switch therefore to simulations to give us a better feel for these functions.

\subsection{Minimal R\'{e}nyi Entropy}

The expression provided here is also based on the empirical results that conjecture that the minimal R\'{e}nyi entropy is attained by $0^m$ or $1^m$.

Inserting the derived expression given in Eq. \ref{eq:max_prob} corresponding to the maximal leakage into the formula of the second-order R\'{e}nyi entropy given in Eq. \ref{eq:renyi_two}, we obtain the following expression for the minimal R\'{e}nyi entropy:

\begin{equation}
R(X) = H_{2}(X) := - \text{log}_2 \sum_{j=0}^{n-m} \binom{n}{j} \cdot \left( \frac{\binom{n-j}{m}}{\mu_{n,m}} \right)^2
\end{equation}

The derived expression agrees with the experiments driven by the computer simulations presented in Section \ref{sec:simulations}.

\subsection{Min-Entropy}

The most conservative measure of information in the R\'{e}nyi family is the min-entropy, and this is of interest when it comes to privacy amplification.

This turns out to be more tractable than the Shannon entropy. In particular it is immediate that the smallest Min-Entropy is attained by the all zero or
all one $x$ strings: the largest weight of a $y$ string, and hence probability, is $\binom{n}{m}$ and this is attained by $x=0^m$ and $y=0^n$.
Thus we can derive an analytic form for the minimum Min-Entropy $H_{\infty}(X)$ by inserting the derived term for maximal probability given in Eq. \ref{eq:max_prob} into Eq. \ref{eq:min-entropy}, and thus we get:

\begin{equation}
Min(H_{\infty}(X)) := - \text{log}_2 \left( \frac{\binom{n}{m}}{\mu_{n,m}} \right)
\end{equation}

and this immediately simplifies to:

\begin{equation}
Min(H_{\infty}(X)) := n-m
\end{equation}

It is clear that this indeed corresponds to the most pessimistic bound of the leakage and can be thought of as assuming that the adversary
gets to know the exact positions of the leaked bits.

The Min-Entropy, $H_\infty(X)$, is based on the most likely event of a random variable $X$. Therefore, this term sets an upper bound on the number of leaked bits, which can be then used in the parameterization of the compression function used in privacy amplification as described in \cite{Bennett95}.

Using Eq. \ref{eq:h_relation} and  the analytic forms given above for the lower bound on the Shannon entropy as well as the min-entropy, we can effectively set loose upper and lower bounds on the R\'{e}nyi entropy.

\subsection{Maximum Entropy}

Another observation derived from empirical results obtained by simulation is that it also appears that the minimal leakage (max H) occurs when $x$ comprises alternating $0$s and $1$s, e.g. $x=101010...$, as shown in Fig. \ref{fig:renyi_orders}. We have seen that for a given $n$ and $m$, the total number of masks and the number of compatible $y$ strings are constant for all $x$ strings. Therefore, the change in entropy of the $\Upsilon$ space for different $x$ strings is solely dictated by how the masks are distributed among the compatible $y$ strings, i.e. the contribution of each $y \in \Upsilon_{n, x}$ to the total number of masks.

\section{Repercussions and Alternative Approaches}\label{sec:notes}

This section gives a brief overview of the context to which this study applies and also analyzes the presented problem from a Kolmogorov complexity point of view. We then propose an approach for estimating the expected leakage, and finally we point out a duality between our findings and similar results in the literature.

\subsection{Privacy Amplification}

PA involves a setting in which Alice and Bob start out by having a partially secret key denoted by the random variable $W$, e.g. a random $n$-bit string, about which Eve gains some partial information, denoted by a correlated random variable $V$. This leakage can be in the form of some bits or parities of blocks of bits of $W$ or some function of $W$ \cite{Bennett95}. Provided that Eve's knowledge is at most $t < n$ bits of information about $W$, i.e. $R(W|V = v) \ge n - t$, with $R$ denoting the second-order R\'{e}nyi entropy, Alice and Bob can distill a secret key of length $r = n - t - s$ with $s$ being a security parameter such that $s < n - t$. The security parameter $s$ can be used to reduce Eve's knowledge to an arbitrarily small amount, e.g. in the context of universal hash functions, it can be used to adjust the reduction size of the chosen compression function $g: \{0, 1\}^n \mapsto \{0, 1\}^r$.

The function $g$ is publicly chosen by Alice and Bob at random from a family of universal hash functions, here denoted by the random variable $G$, to obtain $K = g(W)$, such that Eve's partial information on $W$ and her complete information on $g$ give her arbitrarily little information about $K$. The resulting secret key $K$ is uniformly distributed given all her information. It is also shown by Bennett et al. in \cite{Bennett95} that $H(K | G, V = v) \geq r - 2^{-s} / \text{ln}\,2$, provided only that $R(W|V = v) \ge n - t$. The value of $s$ can be considered a fixed value and comparatively small, typically not larger than 30, as the key length increases.

It is worth noting that the measure of information used in privacy amplification for defining the bound on leakage or the minimum length of the secret key that can be extracted, may vary depending on criteria such as the algorithms used in the amplification scheme and the channel being authenticated or not. For instance, as shown in \cite{MaurerWolf97}, when randomness extractors are used instead of universal hash functions, the bound for secure PA against an active adversary is defined by the adversary's min-entropy about $W$. Various schemes for performing PA over authenticated and non-authenticated channels have been extensively studied in \cite{MaurerWolf97},\cite{CachinMaurer97},\cite{CarterWegman77}.

In QKE, privacy amplification constitutes the last sub-protocol that is run in a session and thus it takes place after the information reconciliation phase. The leakage studied in this paper deals with reduced entropy before the information reconciliation phase. However, this simply means that the leakage quantified here would in fact contribute to the $t$ bits leaked to Eve.

\subsection{Kolmogorov-Chaitin Complexity}

From a purely information theoretical point of view, quantifying the amount of information leakage in terms of various measures of entropy such as the Shannon entropy is arguably what interests us. However, from a cryptographic standpoint, a complexity analysis of exploring the search space by considering the Kolmogorov complexity, provides another perspective in terms of the amount of resources needed for describing an algorithm that reproduces a given string.

In such a context, what matters for an attacker is how efficiently a program can enumerate the elements of the search space. In other words, whether it can enumerate the space in the optimal way, to minimize the expected time to terminate successfully. To illustrate this point, consider the case of the all 0 $x$ string for which we can start with the all 0 $y$ string, then move to $y$ strings with one $1$, then two $1$s, and so on and so forth. For other generic $x$ strings, carrying out this procedure in an efficient manner becomes more involved.

\subsection{Estimating Expected Leakage}\label{subsec:estimating-expected-leakage}

Our primary goal was to compute the leakage for a given $x$ and the maximum leakage for given $n$ and $m$, however, estimating the average leakage might also be of some interest. Since an exact computation depends on a rigorous understanding of the $\Upsilon$ space and its governing probability distribution, we suggest an approach that moves the problem from the space of supersequences to that of subsequences such that further developments in the latter can enable a more fine-grained estimation of the expected leakage.

Let $Y$ be the random variable denoting the original random sequence of $n$ bits and $X$ the random variable denoting the $m$ bits of leakage from $Y$. The average leakage can be expressed in terms of the entropy of $Y$ minus the conditional entropy of $Y$ conditioned on the knowledge of $X$, i.e. $H(Y) - H(Y | X)$. While $H(Y | X)$ may seem hard to compute without the joint probability mass function of $X$ and $Y$, we can use Bayes' rule for conditional entropy \cite{cover2012elements} to reformulate the expression as follows.

\[
H(Y) - H(Y | X) = H(X) - H(X | Y)
\]

With random $Y$, $X$ is a uniformly distributed $m$-bit string and thus we have $H(X) = m$. This leaves us with $H(X | Y)$, and this reformulation allows us to define the entropy space in terms of projection weights, $\omega_{x}(y)$, assigned to the subsequences of each $Y = y$. Currently, as shown in \cite{hirschberg2000tight}, we only know the expected number of distinct subsequences given an $n$ and $t$:

\begin{equation}
E_t(n) = \sum\limits_{i=0}^{t} \binom{n-t-1+i}{i} \lambda^i.
\end{equation}
with $t$ being the number of deleted bits from the $n$-long $y$ string, i.e. $t = n-m$, and $\lambda = 1 - \frac{1}{|\Sigma|}$, which in the binary case, $\Sigma = \{0, 1\}$, would simply be $\lambda = 1 - \frac{1}{2}$. With this measure, we can get a rough estimate on the expected weight, which can then be used to estimate the average entropy, but this only gives us a very coarse-grained estimation of the expected leakage. Therefore, a better understanding of the exact number of distinct subsequences would lead to a more fine-grained estimation of the expected leakage.

\subsection{Duality: Subsequences vs. Supersequences}

An interesting observation resulting from our findings is that the two $x$ strings of interest in the space of supersequences, i.e. the all zero or all one strings (single run) $0^+|1^+$, denoted here by $\sigma$ and the alternating $x$ strings: $(\epsilon | 1)(01)^+(\epsilon | 0)$, denoted here by $\alpha$ also represent the most interesting strings in the space of distinct subsequences.

More precisely, in our study we observe that single run sequences $\sigma$ lead to the least uniform distribution of masks over the compatible supersequences, whereas the alternating sequences $\alpha$ yield the distribution of masks closest to the uniform distribution. Similarly, in the space of subsequences, $\sigma$ lead to the minimum number of $m$-long distinct subsequences and $\alpha$ generate the maximum number of $m$-long distinct subsequences.

\section{Simulations}\label{sec:simulations}

In this section we first give a brief description of how our simulator \cite{atashpendar2014qkdsimulator} carries out the experiments and then we discuss the obtained results with the help of a few plots that are aimed at describing the structure of the $\Upsilon$ spaces. We will refrain from elaborating on all the functionalities of the simulator as this would be beyond the scope of this paper. Instead, we focus on a select few sets of empirical results that were obtained from our experiments. We refer to \cite{atashpendar2014qkdsimulator} for more information and details.

The main motivation behind the computational approach driven by simulations lies in the rather complicated structure of the $\Upsilon$ spaces. As deriving analytic forms for describing the entire space seems to be hard, we rely on simulating the spaces of interest in order to explore their structure.

\subsection{Methodology}

The simulator relies on parallel computations for generating, sampling and exploring the search spaces. The experiments are carried out in two phases: First the simulator generates the $\Upsilon$ spaces that have various structures satisfying predefined constraints and then it proceeds to performing computations on the generated data sets.

The pseudo-code given in Alg. \ref{alg:entropy} provides an example that illustrates one of the main tasks accomplished by the simulator: Given an $n$ and an $x$ string, we generate the corresponding $\Upsilon$ space containing the compatible $y$ strings, compute the projection count $\omega_x(y)$ of its members, and compute its exact entropy.

\begin{algorithm}
	\caption{Compute $H_{\alpha}(\Upsilon_{n, x})$}\label{alg:entropy}
	\begin{algorithmic}[1]
		\STATE $SN \gets \text{Generate the space of bit strings of length }n$
		\STATE $\Upsilon_{n, x} \gets \text{Filter }SN \text{ and reduce it to }\{y \in \{0,1\}^n: (\exists \pi) [y_\pi=x] \}$
		\STATE $probArray \gets []$
		\FOR{$y_i$ \TO $\Upsilon_{n, x}$}
		\STATE $\omega_x(y_i) \gets computeProjectionCount(y_i, x)$
		\STATE $probArray[i] \gets \omega_x(y_i) / N$
		\ENDFOR
		\STATE $H_{\alpha} \gets computeH_{\alpha}(probArray)$
		\RETURN $H_{\alpha}$
	\end{algorithmic}
\end{algorithm}

\subsection{Results Discussion}

In this section, we present and discuss a select subset of our results with the help of plots generated by the simulator that provide a better insight into the structure of the $\Upsilon$ spaces.

As mentioned before, one of the main observations resulting from computational experiments is that the shape of the probability distributions leading to the entropy values of $x$ strings for a given $n$ and $m$, is mainly determined by how evenly the number of projecting $y$ strings are distributed across the possible projection counts for a given $n$ and $m$. This observation is illustrated in Fig. \ref{fig:projection_distribution}.

\begin{figure}
	\centering
	\includegraphics[scale=0.6]{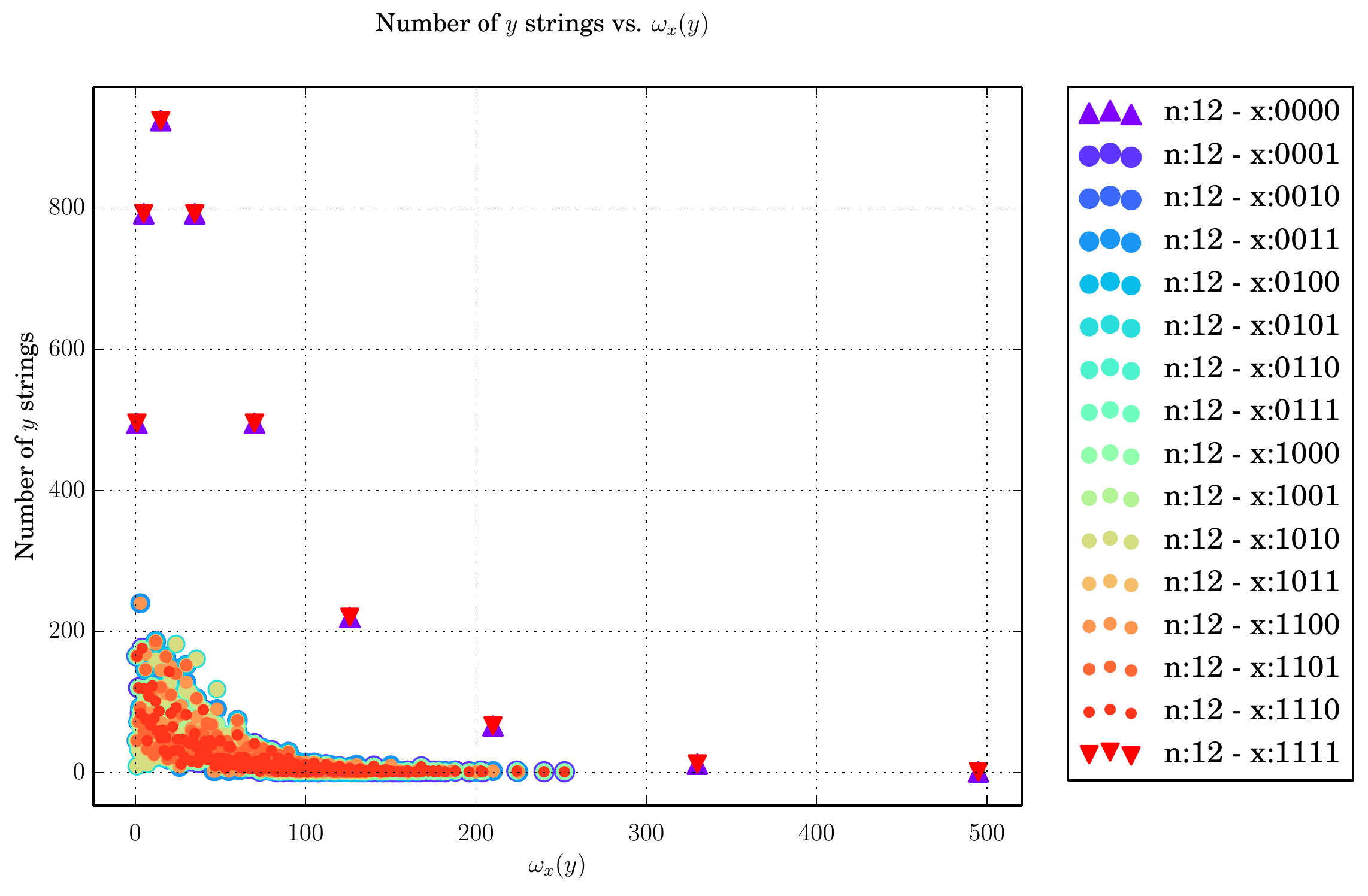}
	\caption{Count of distinct $y$ strings admitting the same $\omega_{x}(y)$}
	\label{fig:projection_distribution}
\end{figure}

Following from Theorem \ref{theorem:upsilon}, for a given $n$ and $m$, the observables computed and plotted in Fig. \ref{fig:projection_distribution} for any $x$ string satisfy the following

\begin{equation}
\sum_{i = 1}^{g(n,x)}  c_i = |\Upsilon_{n, m}|
\end{equation}

Furthermore, the sum of the product of $c_i$ and $\omega_x(y_i)$ is equal to a constant for all $x$ strings:

\begin{equation}
\sum_{i = 1}^{g(n,x)}  c_i \cdot \omega_x(y_i) = \eta_{n, m}
\end{equation}

\begin{figure}
	\centering
	\includegraphics[scale=0.6]{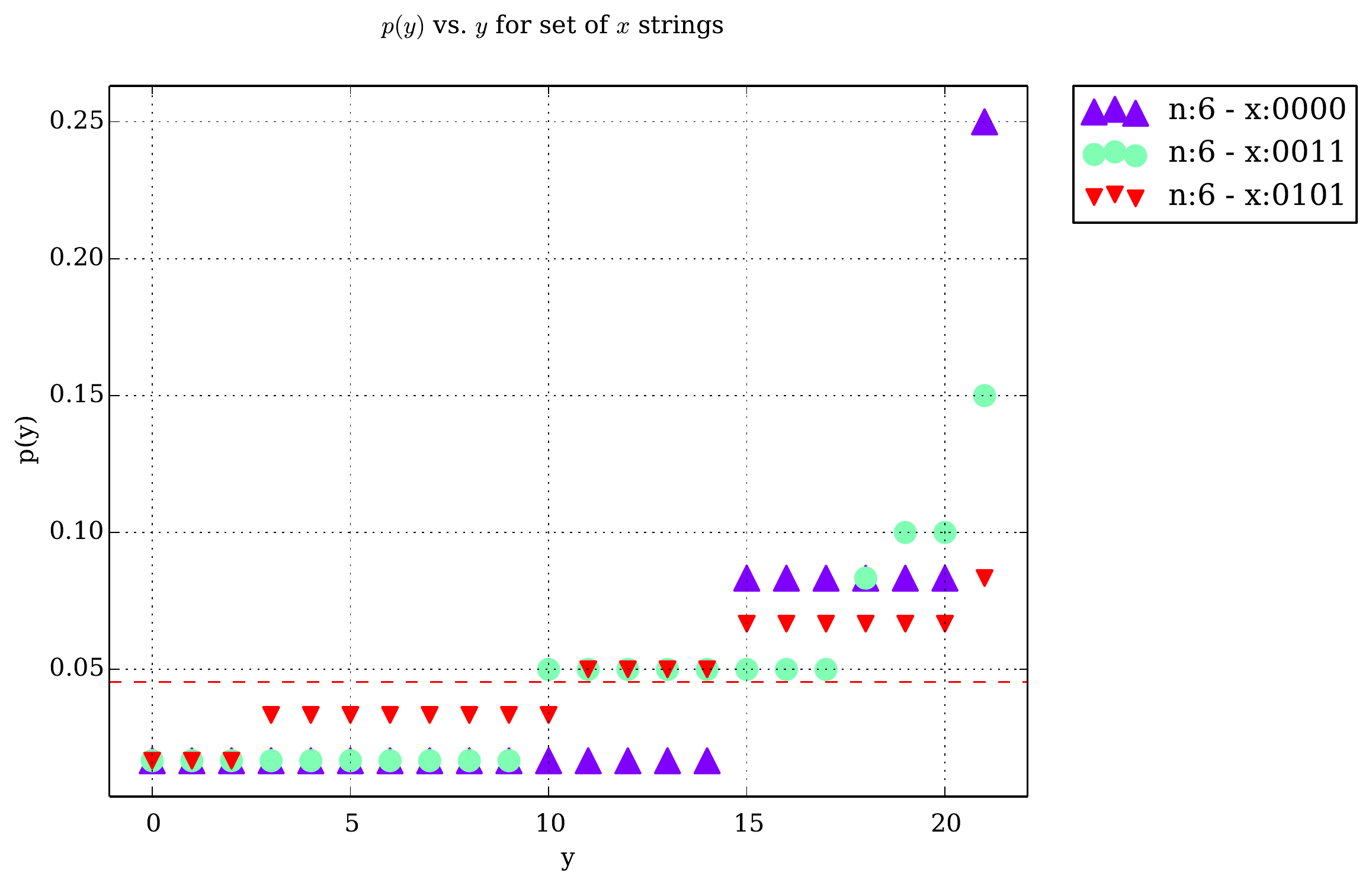}
	\caption{Probability distributions of $\Upsilon$ spaces for given $n$ and $x$ with $y$ strings enumerated by indices and the red dotted line showing the uniform distribution.}
	\label{fig:prob-dist}
\end{figure}

\begin{figure}
	\centering
	\includegraphics[width=0.8\textwidth]{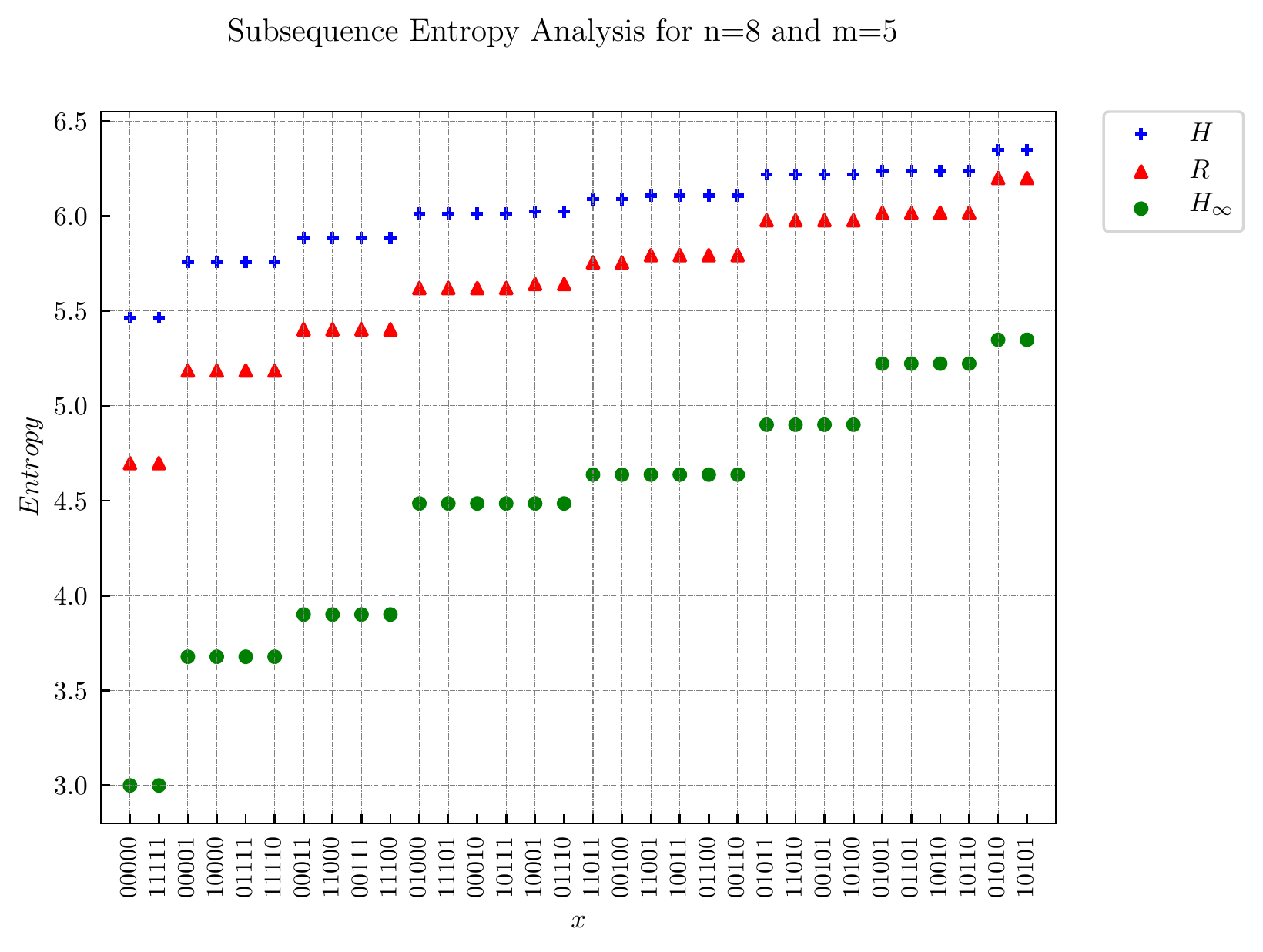}
	\caption{$H$ (Shannon), $R_2$ (second-order R\'{e}nyi entropy) and $H_\infty$ (min-entropy) vs. $x$ strings for $n=8$ and $m=5$.}
	\label{fig:renyi_orders}
\end{figure}

With $c$ denoting the values on the y-axis, i.e. the count of $y$ strings projecting $\omega_x(y)$ times, and with $\omega_x(y)$ denoting the number of distinct ways that $y$ can project onto $x$, and finally with $\eta_{n, m}$ being a constant for any given $n$ and $m$ and $g(n,x)$ being a function of $n$ and $x$ that denotes the number of data points corresponding to the distinct count of $y$ strings that have the same $\omega_{x}(y)$.

This means that for a given $n$ and $m$, the total number of projection counts in the corresponding $\Upsilon$ space is independent of the $x$ strings. We can see that for the $x$ strings yielding the maximum amount of leakage, i.e. $x= 1^m|0^m$, the lower number of data points is compensated by larger values for the distinct number of $y$ strings admitting larger projection count values, hence showing a much more biased structure in the distribution with respect to generic $x$ strings. Conversely, the distributions for the remainder of the $x$ strings are considerably dampened and noticeably closer to a flat distribution and are thus less biased compared to $x= 1^m|0^m$ strings, which in part explains the correspondingly higher entropy values. In particular, the alternating ones and zeros string admits the highest degree of dispersion in terms of the distribution of the masks and thus yields the lowest entropy.

The resulting probability distributions leading to the computed entropy values are illustrated in Fig. \ref{fig:prob-dist}. An immediate observation is that the distribution of the projecting $y$ strings for the $0^m$ or $1^m$ strings has the largest outliers. However, this alone does not capture the role of the shape of the probability distribution. Therefore, one could argue that the probability distribution that admits the largest Kullback-Leibler distance from the uniform distribution, i.e. the most biased distribution, yields the lowest entropy, and the conjecture that we put forth is that this distribution is given by the all 0 or all 1 $x$ strings.

The plot shown in Fig. \ref{fig:renyi_orders}, illustrates three measures of entropy, namely the Shannon entropy ($H$), the second-order R\'{e}nyi entropy ($R$) and the min-entropy ($H_\infty$) as a function of $n$ and $m$ for all the $2^m$ $x$ strings for $n=8$ and $m=5$. The presented empirical results validate our conjecture that the all zero or the all one strings yield the minimum entropy and that the alternating zeros and ones string gives the maximum entropy.

\section{Conclusions}\label{sec:chp-leakage-conclusions}

We have described an information theory problem that arose from some investigations into quantum key establishment
protocols. As far as we are aware, the problem, despite its seeming to be very natural and simple to state, has not been
investigated in the mathematical literature. We have shown that the maximum leakage, measured in terms of the drop
in the entropy of the space of compatible $y$ strings, corresponds to the all zero or all one observed strings.

We have presented analytic forms for the Shannon entropy, the second-order R\'{e}nyi entropy, and the min-entropy for these cases. Moreover, we have discussed the relevance of these measures specifically in the context of privacy amplification in QKE protocols. We have also noted that the simulations suggest that the minimal leakage corresponds to the $x$ strings comprising alternating zeros and ones. Moreover, we pointed out an interesting duality between our results and existing results in the literature for the space of subsequences. We have also described a simulation program to explore these results, which is available at \cite{atashpendar2014qkdsimulator}.

\stopminichaptoc{\minichaptocenabled}

\chapter{Combinatorial Structures for Binary Sequences}\label{chp:combinatorial-structures}

\printminichaptoc{\minichaptocenabled}

\section{Introduction}

In this chapter, we study several closely-related counting problems involving (super/sub)-sequences. In the context of the original entropy extremization analysis, the counting problems studied in this section are motivated by the need for gaining a better understanding of the combinatorial objects and structures involving supersequences that exhibit specific properties with respect to a fixed subsequence. Indeed, the quantities of interest in the entropy problem are precisely determined by the number of supersequences that admit a certain embedding weight for a fixed subsequence. Thus, the results in this section are aimed at providing more insight into related combinatorial objects, with similar techniques used in \Cref{sec:entropy-minimization} to cluster supersequences admitting specific weights in order to establish the entropy minimization case.

To this end, we first provide a characterization of the number of subsequence embeddings based on a run-length encoding of strings used for identifying deletion patterns that simplify the counting problem to a sequential mapping of runs from $x$ strings to $y$ strings. We then describe two different ways of clustering the space of supersequences and prove that their cardinality depends only on the length of the received subsequence and its Hamming weight, but not its exact form. We then consider the problem of counting \emph{singletons}, that is, supersequences that admit only a single embedding of $x$. We provide a closed form expression for enumerating singletons using the same run-length encoding and prove an analogous result for the minimization and maximization of the number of singletons, by the alternating and the uniform strings, respectively.

\subsection{Results and Structure}

In \Cref{sec:subseq-count-formalism}, we first present an algorithm based on a run-length encoding of strings for counting the number of embeddings of $x$ into $y$ as a subsequence. We then explore counting problems and clustering techniques in the space of supersequences including an analysis of a class of supersequences, referred to as singletons, that admit exactly a single embedding of a given subsequence and prove similar extremization results for their count.

More precisely, in Section \ref{sec:max-initials-to-hamming-clusters}, we show how similar to the way the cardinality of the set of supersequences that can project to a given subsequence, i.e., $|\Upsilon_{n,x}|$, depends only on their respective lengths, we prove that the number of supersequences that admit an initial embedding of a subsequence such that the last index of their initial embedding overlaps with their last bit, also depends only on $|y|=n$ and $|x|=m$. We then describe two clustering techniques that give rise to subspaces in $\Upsilon_{n,x}$ whose sizes depend only on $n, m$ and the Hamming weight of $x$, but not the exact form of $x$. We derive analytic expressions in Section \ref{sec:closed-form-cluster}, as well as a recurrence in Section \ref{sec:recurrence-cluster}, for the cardinality of these sets. The approach and methodology used for deriving our clustering results depend heavily on the notion of initial or canonical embeddings of subsequences in their compatible supersequences, which provide further insight into the importance of initial embeddings.

Finally, in Section \ref{sec:singletons}, we consider the problem of enumerating supersequences that admit exactly a single occurrence of a subsequence, referred to as singletons, and give an analytic expression for their count. Furthermore, we prove a similar result for the maximization and minimization of the number of singletons by the constant and alternating strings, respectively.

\section{Clustering and Counting Binary Sequences}\label{sec:counting-supersequences}

\textbf{Counting multisets:} Throughout, we use the combinatorics of counting multisets, also referred to as the method of stars and bars, to enumerate all possibilities for placing $n$ indistinguishable objects into bins marked by $m$ distinguishable separators such that the resulting configurations are distinguished only by the number of objects present in each bin, which is given by $\binom{n+m-1}{n}$.

\begin{itemize}
\item The number of $m$-tuples of non-negative integers that sum to $n$ is equal to the number of multisets of cardinality $m-1$ taken from a set of size $n+1$, which is given by
\[
\binom{n+m-1}{n}
\]
\item Graphical illustration:
\[
\bullet \bullet \bullet | \bullet | \bullet |
\]
\[
\bullet | \bullet \bullet | \bullet \bullet |
\]
\end{itemize}
The stars and bars technique computes the combinations of placing $n$ indistinguishable stars into bins separated by $m$ distinguishable bars. The resulting configurations are distinguished by the number of stars present in each bin.

\subsection{Compatible Supersequences and Subsequence Embeddings}

Recall that for a fixed subsequence $x$ of length $m$, we consider the set of $y$ strings of length $n$ ($n \ge m$), referred to as compatible supersequences, that can contain $x$ as a subsequence embedding. The set of compatible supersequences is denoted by $\Upsilon_{n,x}$. It is known that the cardinality of $\Upsilon_{n,x}$ is independent of the form of $x$ and that it is only a function of $n$ and $m$.
\begin{equation}
|\Upsilon_{n,x}| = \sum_{r=m}^n \binom{n}{r}
\end{equation}
We provided an alternative proof for this based on a simple recursion in Theorem \ref{theorem:upsilon} in Chapter \ref{chp:information-leakage} and \cite{atashpendar2015information}. The original motivation for the clustering scheme presented here was to have a more fine-grained view of the distribution of masks in the space of supersequences. This approach led to the discovery of similar structures in $\Upsilon_{n,x}$, in that their cardinality does not depend on the form of $x$, analogous to how $|\Upsilon_{n,x}|$ depends only on $n$ and $m$.

\section{Counting Embeddings via Runs}\label{sec:subseq-count-formalism}

Efficient dynamic programming algorithms for computing the number of subsequence embeddings are known in the literature, e.g., a recursive algorithm requiring $\Theta(n \times m)$ operations \cite{elzinga2008algorithms}. Here we provide an alternative algorithm, which is primarily based on the run-length encoding of strings.

Using the RLE notation, there are a few cases in which this question is easy to answer. For instance, if $y = (a;k_1, \dotsc, k_\ell)$ and $x = (a;k_1', \dotsc, k_\ell')$, with the same value of $\ell$, i.e., we have the same number of blocks in $x$ and $y$, then it is easy to see that there is a one-to-one sequential mapping of blocks between $x$ and $y$. This allows us to enumerate the different masks depending on how they map the blocks to each other as follows:
\begin{equation}\label{eq:embedding-simple-case}
\omega_x(y) = \prod_{i=1}^\ell \binom{k_i}{k'_i}.
\end{equation}
However, in the general case, the number of blocks in $x$ and $y$ can be different. If $y= (a;k_1, \dotsc, k_\ell)$ and $x= (\overline{a};k'_1, \dotsc, k'_{\ell'})$ do not start with the same character, we have to delete the first block to recover the case $y= (a;k_2, \dotsc, k_{\ell})$, and $x= (a;k'_1, \dotsc, k'_{\ell'})$. We will now suppose that $x$ and $y$ start with the same character.

Here we describe an algorithm wherein for a fixed pair of $x$ and $y$ strings, we structure and enumerate the corresponding space of masks by accounting for the number of different ways we can delete characters in order to merge blocks/runs such that we can recover the simple case given in \Cref{eq:embedding-simple-case}. In the more general case, let $y=(k_1,  \dotsc, k_{\ell})$ and $x=(k'_1, \dotsc, k'_{\ell'})$.

\begin{definition}
Let $S$ be the set of maps $f: [\ell'] \to [\ell]$ that satisfy the following properties: $f$ is strictly increasing and $f(i) \equiv i \bmod 2$. A function $f$ will define a subset of masks, by specifying blocks that will have to be completely deleted. We group the masks according to a set of functions $f$ that map indexes of blocks of $x$ to indexes of blocks of $y$. Intuitively, $f$ maps the $i$-th block of $x$ to the block of $y$ that contains the last character of the $i$-th block of $x$.
Therefore, all blocks of $y$ between $f(i) +1 $ and $f(i)$ that are not composed of matching characters have to be deleted such that we can recover the simple case in \Cref{eq:embedding-simple-case}.
\end{definition}
For the subsequent analysis, recall that $k_i$ denotes the length of the run at index $i$, whereas $k^*_i$ refers to the actual set of indexes of the $i$-th run.
\begin{definition}
Let $k^*_{i}$ denote the set of indexes belonging to the $i$-th block of $y$, i.e.,
\[\{ \sum_{j=1}^{i-1} k_j, \sum_{j=1}^{i-1} k_j + 1, \dotsc, \sum_{j=1}^{i} k_j \},
\]
and $F^{*}(i) = \{f(i-1) + 2, f(i-1) + 4, \dotsc, f(i) - 1\}$ and $\overline{F^*(i)} = \{f(i-1) + 1, f(i-1) + 3, \dotsc, f(i)\}$, then a deletion mask $\delta$ corresponds to $f$ if:
\begin{itemize}
\item $\forall i: \cup_{j \in F^*(i)} k^*_{j} \subset \delta$
\item $\forall i: k^*_{f(i)} \not\subset \delta$ (this allows us to have a partition)
\end{itemize}
We call $\omega_f$ the set of masks corresponding to $f$.
\end{definition}
\begin{theorem}
The family $(\omega_f)_{f \in S}$ defines a partition on the set of masks from $y$ to $x$.
\end{theorem}
\begin{proof}
We first show that for $f \neq f' \in S$, every deletion mask $\delta$ corresponding to $f$ is different from every mask $\delta'$ associated with $f'$ (i.e., $\omega_f \cap \omega_{f'} = \emptyset$).
Since $f \neq f'$, we have a \emph{smallest} integer $i \in [\ell]$ such that $f(i) \neq f'(i)$. We assume without loss of generality that $ f'(i-1)=f(i-1)<f(i) < f'(i)$. Due to the condition on parity, $f(i+1) \neq f'(i)$.
We distinguish between two cases:
\begin{itemize}
\item If $f'(i-1)<f(i+1)<f'(i)$, then $k_{f(i+1)} \not\subset \delta$, and $k_{f(i+1)} \subset \delta'$ since $f(i+1) \in F'^*(i-1)$.

\item Conversely if $f(i-1)<f'(i)< f(i+1)$, then $k_{f'(i)} \not\subset \delta'$, and $k_{f'(i)} \subset \delta$ since $f'(i) \in F^*(i+1)$.
\end{itemize}
Therefore, we have $\delta \neq \delta'$. We now show that $\cup_{f \in S}\omega_f$ is the set of masks from $y$ to $x$. We will use projection masks here as they are more suitable for this proof. Let $\pi$ be a projection mask such that $y_{\pi} = x$. We let $\pi = \{\pi_1, \dotsc, \pi_m\}$, where the $\pi_i$ are in increasing order. Therefore, we have for all $i$, $y_{\pi_i} = x_i$. We define $\phi :[n] \rightarrow [\ell]$ to be the mapping that takes an index of $y$ and returns the index of the block/run it belongs to, i.e., $\phi(a)$ returns the smallest $i$ such that $\sum_{j=1}^i k_j \geq a$.
We define $f$ such that $f \in S$ and $\pi$ is in $\omega_f$, by $f(i) = \phi( \pi_{\sum_{j=1}^i k'_j})$. To prove that $f$ is in $S$, note that given $i$:
\begin{itemize}
\item We have that $f(i) \leq f(i+1)$ since the $\pi_i$ are in increasing order.
\item Moreover, $\pi_{\sum_{j=1}^i k'_j}$ and $\pi_{\sum_{j=1}^{i+1} k'_j}$ correspond to indexes (of $y$) of opposite letter (if the first one is a $1$, the second is a $0$ and vice versa) since $\sum_{j=1}^i k'_j$  and $\sum_{j=1}^{i+1} k'_j$ correspond to indexes (of $x$) of opposite letter.
\end{itemize}
Therefore, $f(i)$ and $f(i+1)$ are of opposite parity and $f(i) < f(i+1)$.

We now prove that $\pi$ corresponds to $f$. For a fixed $i \in [\ell]$, let $k^*_{f(i-1)} = b^{k_{f(i-1)}}$, i.e., the $f(i-1)$-th block of $y$ is made of letters $b$. Therefore, $k^*_{t} = b^{k_t}$ for $t \in F^*(i)$, since $t$ has the same parity as $f(i-1)$. Moreover, we have $b = x_{\sum_{j=1}^{i-1} k'_j}$ according to the definition of $f$. So for every index $h$ between $\sum_{j=1}^{i-1} k_j +1$ and $\sum_{j=1}^{i} k_j$, $x_h = y_{\pi_h} = \overline{b}$, and for $t \in F^*(i)$, we have $k_t^* \cap \pi = \emptyset$ (equivalently with the deletion mask $\delta$, $k_t^* \subset \delta$). By definition, $\pi_{\sum_{j=1}^{i} k_j} \in k_{f(i)}^*$ so $\pi \cap k_{f(i)}^* \neq \emptyset$ (equivalently with the deletion mask $\delta$, $k_{f(i)}^* \not\subset \delta$).
\end{proof}

\begin{theorem}
We have for $f\in S$, given by
\begin{equation}
	|\omega_f| = \prod_{i=0}^{\ell} \binom{\sum_{j \in \overline{F^*(i)}}k_j }{k'_i} - \binom{\sum_{j \in \overline{F^*(i)}\setminus \{f(i)\}}k_j }{k'_i}
\end{equation}
\end{theorem}
\begin{proof}
Upon the deletion induced by $f$, we obtain a string of the form
\[
(\sum_{j \in \overline{F^*(1)}}k_j, \dotsc , \sum_{j \in \overline{F^*(\ell)}}k_j).
\]
Therefore, we have the same number of blocks in both the $y$ string as well as the $x$ string, and the number of masks can be computed easily as shown in \Cref{eq:embedding-simple-case}. We first count the number of ways to choose $k_i'$ elements from $k_{\overline{F^*(i)}}$ and then subtract the number of combinations not using any of the $k_{f(i)}$.
\end{proof}
\begin{remark}
We can note that $\overline{F^*(i)}$ and $F^*(i)$ form a partition of $[\ell]$.
\end{remark}
Following from the preceding theorems, the total number of masks can be computed as follows
\begin{equation}
	\omega_x(y) = \sum_{f \in S} |\omega_f|.
\end{equation}
By summing over all $f \in S$, we get the total number $|\omega|$ of compatible masks. Note that it may happen that $|\omega_f| = 0$; this happens when we try to trace a large block of $x$ from a smaller block of $y$.

\textbf{The Set $S$:} We now determine the size of $S$, as a function of $\ell$ and $\ell'$. Let this size be denoted by $\sigma(\ell',\ell)$. We denote $u = \lfloor (\ell-\ell')/2 \rfloor$.
If $f(1) = 1$, then we get $\sigma(\ell'-1,\ell-1)$; if $f(1) =3$, we get $\sigma(\ell'-1,\ell-3)$, etc. We also know that $\sigma(x,x) = 1$ for all $x$, and that $\sigma(x,y) = 0$ for all $x$, $y$ such that $y<x$.
We therefore get the following recurrence:
\begin{equation*}
	\sigma(\ell',\ell) = \sum_{i=0}^{u} \sigma(\ell'-1,\ell-1-2i)
\end{equation*}
Iterating this recursion, we get
\begin{equation*}
\sigma(\ell',\ell) = \sum_{i=0}^{u}\sum_{j=0}^{u -i} \sigma(\ell'-2,\ell-2-2i-2j),
\end{equation*}
and grouping the terms yields
\begin{equation*}
	\sigma(\ell',\ell) = \sum_{i=0}^{u} (i+1)\sigma(\ell'-2,\ell-2-2i).
\end{equation*}
We now describe a direct combinatorial argument which gives a closed form formula for $\sigma(\ell',\ell) = |S|$. First note that if $\ell\not\equiv \ell'$ mod 2 then $\ell$ cannot be in the image of $f$. So let $\tilde{\ell}=\ell$ if $\ell\equiv \ell'$ mod 2 and $\tilde{\ell}=\ell-1$ if not. Now the problem is to choose $[\ell']$ elements from $[\tilde{\ell}]$ such that all the gaps have even width. Equivalently, we are interleaving the $\ell'$ chosen elements with $u=(\tilde{\ell}-\ell')/2$ gap-segments of width 2. The number of ways to do this is plainly
\begin{equation}
|S|=\sigma(\ell', \ell)=\binom{\ell'+u}{u}.
\end{equation}

\begin{example}
For	$y=	\texttt{0000111100001111}$ and $x=\texttt{0011}$, we obtain $\omega_y(x)=300$. We now compute the number of embeddings using the run-based algorithm described above. We have $\ell=4$, $\ell'=2$ and $u=(\tilde{\ell}-\ell')/2$, which means the size of $S$ is $|S|=\sigma(\ell',\ell)=\binom{l'+u}{u}=\binom{2+1}{1}=3$. The three deletions $S=\{f_1, f_2, f_3\}$ are computed as follows: $y_{f_1} = (k_1, k_2) = \texttt{00001111}$, which amounts to $\omega_{f_1}=\binom{4}{2}\binom{4}{2}=36$. Similarly, for $f_2$ and $f_3$, we get $y_{f_1} = (k_1 + k_3, k_4) = \texttt{000000001111}$ and $y_{f_1} = (k_1, k_2 + k_4) = \texttt{000011111111}$, the two of which add up to $2\times\big(\binom{8}{2}\binom{4}{2}-\binom{4}{2} \big) = 2 \times 132 = 264$. So the total is $\omega_y(x) = \sum_{f \in S} \Omega_f = 36+132+132=300$.
\end{example}

\section{From Maximal Initials to Hamming Clusters}\label{sec:max-initials-to-hamming-clusters}
\begin{definition}
	Let $\Upsilon_{n,x}^c$ be the cluster of supersequences that have $c$ extra \texttt{1}'s with respect to $x$, where $0 \le c \le n-m$.
	\[
	\Upsilon^c_{n,x} = \{ y \in \Upsilon_{n,x} \mid h(y) - h(x) = c \}.
	\]
\end{definition}
The set of compatible supersequences is thus broken down into $n-m+1$ disjoint sets indexed from 0 to $n-m$ such that strings in cluster $c$ contain $h(x)+c$ \texttt{1}'s:
\begin{equation*}
	\Upsilon_{n,x} = \bigcup\limits_{c=0}^{n-m} \Upsilon^c_{n,x}.
\end{equation*}
\begin{definition}
	Maximal initials represent $y$ strings for which the largest index of their initial mask, $\tilde{\pi}$, overlaps with the last bit of $y$. In other words, the last index of the canonical embedding of $x$ in $y$ overlaps with the last bit of $y$. Recall that we use $\tilde{\pi}$ to denote a mask $\pi$ that is initial.
	\[
	\mathcal{M}_{n,x} = \left\{ y \in \Upsilon_{n,x} \mid (\exists \tilde{\pi})[ y_{\tilde{\pi}}=x \wedge \max(\tilde{\pi}) = |y| = n ] \right\}.
	\]
	Similarly, we define a clustering for maximal initials based on the Hamming weight of the $y$ strings
	\[
	\mathcal{M}_{n,x}^c = \{ y \in \mathcal{M}_{n,x} \mid h(y) = h(x)+c \}.
	\]
\end{definition}
\begin{example}
	For example, the initial embedding of  $x=\texttt{1011}$ in $y=\texttt{110011}$ given by $\tilde{\pi}=\{1,3,5,6\}$ is maximal, whereas its initial embedding in $y'=\texttt{101011}$ given by $\tilde{\pi}'=\{1,2,3,5\}$ is not maximal as the last index of $\tilde{\pi}'$ does not overlap with the position of the last bit of $y'$.
\end{example}
A more exhaustive example illustrating these concepts is given in \Cref{table:upsilon}. In addition to the distribution of weights, i.e., number of masks per $y$, clusters and maximal initials are indicated by horizontal separators and bold font, respectively.
\begin{table}[t]
	\caption{Clusters, Maximal Initial Masks and Distribution of Embeddings}
	\label{table:upsilon}
	\centering
	\begin{tabular}{|c|c|c|}
		\hline
		\multicolumn{3}{|c|}{$x=\texttt{110}$}  \\ \hline
		$y$ & $\tilde{\pi}$ &$\omega$ \\ \hline
		$\texttt{00110}$ & $\mathbf{\{3, 4, 5 \}}$ & 1                     \\ \hline
		$\texttt{01010}$ & $\mathbf{\{ 2, 4, 5 \}}$ & 1                     \\ \hline
		$\texttt{01100}$ & $\{ 2, 3, 4 \}$ & 2                     \\ \hline
		$\texttt{10010}$ & $\mathbf{\{ 1, 4, 5 \}}$ & 1                     \\ \hline
		$\texttt{10100}$ & $\{ 1, 3, 4 \}$ & 2                     \\ \hline
		$\texttt{ 11000 }$ & $\{ 1, 2, 3 \}$ & 3                     \\ \hline \hline \hline
		$\texttt{ 01101 }$ & $\{ 2, 3, 4 \}$  & 1                    \\ \hline
		$\texttt{ 01110 }$ & $\mathbf{\{ 2, 3, 5 \}}$ & 3                     \\ \hline
		$\texttt{ 10101 }$ & $\{ 1, 3, 4 \}$ & 1                     \\ \hline
		$\texttt{ 10110 }$ & $\mathbf{\{ 1, 3, 5\}}$ & 3                     \\ \hline
		$\texttt{ 11001 }$ & $\{ 1, 2, 3 \}$ & 2                     \\ \hline
		$\texttt{ 11100 }$ & $\{ 1, 2, 4\}$ & 6                     \\ \hline \hline \hline
		$\texttt{ 11011 }$ &  $\{ 1, 2, 3 \}$ & 1                    \\ \hline
		$\texttt{ 11101 }$ &  $\{ 1, 2, 4 \}$ & 3                    \\ \hline
		$\texttt{ 11110 }$ &  $\mathbf{\{1, 2, 5 \}}$ & 6                    \\ \hline
	\end{tabular}
	\begin{tabular}{|c|c|c|}
		\hline
		\multicolumn{3}{|c|}{$x=\texttt{101}$}  \\ \hline
		$y$ & $\tilde{\pi}$ & $\omega$  \\ \hline
		$\texttt{ 00101 }$ & $\mathbf{\{3, 4, 5 \}}$ & 1                     \\ \hline
		$\texttt{ 01001 }$ & $\mathbf{\{ 2, 3, 5 \}}$ & 2                     \\ \hline
		$\texttt{ 01010 }$ & $\{ 2, 3, 4 \}$ & 1                     \\ \hline
		$\texttt{ 10001 }$ & $\mathbf{\{ 1, 2, 5 \}}$ & 3                     \\ \hline
		$\texttt{ 10010 }$ & $\{ 1, 2, 4 \}$  & 2                    \\ \hline
		$\texttt{ 10100 }$ & $\{ 1, 2, 3 \}$   & 1                   \\ \hline \hline \hline
		$\texttt{ 01011 }$ & $\{ 2, 3, 4 \}$ & 2                     \\ \hline
		$\texttt{ 01101 }$ & $\mathbf{\{ 2, 4, 5 \}}$ & 2                     \\ \hline
		$\texttt{ 10011 }$ & $\{ 1, 2, 4 \}$ & 4                      \\ \hline
		$\texttt{ 10101 }$ & $\{ 1, 2, 3\}$ & 4                     \\ \hline
		$\texttt{ 11001 }$ & $\mathbf{\{ 1, 3, 5 \}}$  & 4                    \\ \hline
		$\texttt{ 11010 }$ & $\{ 1, 3, 4\}$  & 2                    \\ \hline \hline \hline
		$\texttt{ 10111 }$ &  $\{ 1, 2, 3 \}$ &  3                   \\ \hline
		$\texttt{ 11011 }$ &  $\{ 1, 3, 4 \}$ & 4                    \\ \hline
		$\texttt{ 11101 }$ &  $\mathbf{\{1, 4, 5 \}}$ & 3                    \\ \hline
	\end{tabular}
\end{table}
\begin{theorem}\label{theorem:total-maximal-initials}
	For given $n$, the cardinality of $\mathcal{M}_{n,x}$ is independent of the exact $x$.
\end{theorem}
\begin{proof}
	It is clear that every $n$-element sequence that has $x$ as an $m$-element subsequence has a unique initial mask $\tilde{\pi}$ that gives $x$. Furthermore, if we fix $\pi$, then the members of $y$ up to the last member of $\pi$ are completely determined if $\pi$ is initial. To see this, consider the case $i \in \tilde{\pi}$, then $y_i$ (the $i$-th member of $y$) must correspond to $x_j$, where $i$ is the $j$-th smallest member of $\tilde{\pi}$. If $i \notin \tilde{\pi}$, but smaller than $max(\tilde{\pi})$, then the $i$-th member of $y$ must correspond to $x_{j+1}$, where $j$ is the number of members of $\tilde{\pi}$ smaller than $i$. The latter follows because if this bit were $x_{j+1}$, then the given $\pi$ would not be initial.

	We also need to observe that for a given $\tilde{\pi}$, there always exists a $y$ that has $x$ initially in $\tilde{\pi}$: suppose that $x$ starts with a \texttt{0}, we set all the bits of $y$ before $\tilde{\pi}$ to be \texttt{1}. For a given value $\ell$ of $max(\tilde{\pi})$ - which can range from $m$ to $n$ - there are exactly $\binom{\ell-1}{m-1}$ $\tilde{\pi}$'s, one for each selection of the other $m-1$ members of $\tilde{\pi}$ among the $\ell-1$ values less than $\ell$.

	Moreover, here we have an additional constraint, namely that the initial masks should be maximal as well, i.e., $max(\tilde{\pi}) = n$. This means that $\ell = n$ and so we can count the number of distinct initials for the remaining $m-1$ elements of $x$ in the remaining $(n-1)$-long elements of $y$ strings, which is simply given by
	\begin{equation}
	|\mathcal{M}_{n,x}|=|\mathcal{M}_{n,m}|=\binom{n-1}{m-1}
	\end{equation}
	Clearly the cardinality of the set of maximal initials is independent of the form of $x$ and depends only on $n$ and $m$.
\end{proof}

\begin{remark}
	Note that if we extend the analysis in the proof of \Cref{theorem:total-maximal-initials} and let $\ell$ run over the range $[m,n]$, we can count all the distinct initial embeddings in $\Upsilon_{n,x}$, given by $\sum_{\ell=m}^n\binom{\ell-1}{m-1}$.

	Moreover, since the bits beyond $max(\tilde{\pi})$ are completely undetermined, for a given $\tilde{\pi}$, there are exactly $2^{n-max(\tilde{\pi})}$ $y$'s that have $\tilde{\pi}$ in common, which, incidentally, provides yet another proof for the fact that $|\Upsilon_{n,x}|$ is a function of only $n$ and $m$ since $|\Upsilon_{n,x}|= \sum_{\ell=m}^n\binom{\ell-1}{m-1}2^{n-\ell}$. This allows us to choose the $x$ comprising $m$ \texttt{0}'s and the result in Eq. \ref{eq:upsilon-cardinality} follows immediately.
\end{remark}

\begin{theorem}\label{theorem:maximal-initials}
	All $x$ strings of length $m$ that have the same Hamming weight, give rise to the same number of maximal initials in each cluster.
	\[
	\forall x, x' \in \Sigma^m, h(x) = h(x') \implies |\mathcal{M}^c_{n,x}| = |\mathcal{M}^c_{n, x'}|.
	\]
\end{theorem}
\begin{proof}
	We now describe a simple combinatorial argument for counting the number of maximal initials in each cluster indexed by $c$, i.e., a grouping of all $y \in \Upsilon^c_{n,x}$ such that $h(y) = h(x)+c$. Let $p$ and $q$ denote the number of additional \texttt{0}'s and \texttt{1}'s contributed by each cluster, respectively. Furthermore, let $a$ and $b$ denote the number of $\texttt{1}$'s and $\texttt{0}$'s in $x$, respectively.

	Similar to the method used in the proof of Theorem \ref{theorem:total-maximal-initials}, due to maximality we fix the last bit of $y$ and $x$, and consider $y'=y-tail(y)$ and $x'=x-tail(x)$ where $tail(s)$ denotes the last bit of $s$. Now the problem amounts to counting distinct initials of length $m-1$ in $(n-1)$-long elements in each cluster by counting the number of ways distinct configurations can be formed as a result of distributing $c$ \texttt{1}'s and $(n-m-c)$ \texttt{0}'s around the bars/separators formed by the $b$ $\texttt{0}$'s and $a$ $\texttt{1}$'s in $x$, respectively.

	We now need to observe that to count such strings with distinct initials, we can fix the $m-1$ elements of $x'$ as distinguished elements and count all the unique configurations formed by distributing $p$ indistinguishable $\texttt{0}$'s and $q$ indistinguishable $\texttt{1}$'s among bins formed by the fixed \texttt{1}'s and \texttt{0}'s of $x'$ such that each such configuration is distinguished by a unique initial.

	Equivalently, we are counting the number of ways we can place the members of $x'$ among $n-1$ positions comprising $p$ \texttt{0}'s and $q$ \texttt{1}'s without changing the relative order of the elements of $x'$ such that these configurations are uniquely distinguished by the positions of the $m-1$ elements.

	Intuitively, the arrangements are determined by choosing the positions of the $m-1$ bits of $x'$: by counting all the unique distributions of bits of opposite value around the elements of $x'$, we are simply displacing the elements of $x'$ in the $n-1$ positions, thereby ensuring that each configuration corresponds to a unique initial.

	Note that this coincides exactly with the multiset coefficient (computed via the method of stars and bars) as we can consider the elements of the runs of $x$ to be distinguished elements forming bins among which we can distribute indistinguishable bits of opposite value to count the number of configurations that are distinguished only by the number of $\texttt{1}$'s and $\texttt{0}$'s present in the said bins.

	Thus we count the number of unique configurations formed by distributing $p$ \texttt{0}'s and $q$ \texttt{1}'s among the $a$ \texttt{1}'s and $b$ \texttt{0}'s of $x$, respectively. The total count for each cluster $c$ is given by: $\binom{p+a-1}{p} \binom{q+b-1}{q}$, which expressed in terms of the Hamming weight of $x$ gives
	\begin{equation}\label{eq:exp-maximal-initial}
	|\mathcal{M}^c_{n,x}| = \binom{(n-m-c)+h(x)-1}{n-m-c} \binom{c+(m-h(x))-1}{c}
	\end{equation}
	With the total number of maximal initials in $\Upsilon_{n,x}$ given by
	\[
	|\mathcal{M}_{n,x}| = \sum_{c=0}^{n-m} |\mathcal{M}^c_{n,x}|=\binom{n-1}{m-1}.
	\]
\end{proof}

\begin{theorem}\label{theorem:hamming-cluster}
	The size of a cluster is purely a function of $n, m, c$ and $h(x)$
	\[
	\forall x, x' \in \Sigma^m, h(x) = h(x') \implies |\Upsilon^c_{n,x}| = |\Upsilon^c_{n, x'}|
	\]
\end{theorem}
\begin{proof}
	Let $\ell$ denote the position of the last bit of $y$ ranging from $|x|=m$ to $|y|=n$. Starting from a fixed $x$ string, we enumerate all $y$ strings in cluster $c$ by considering maximal initials within the range of $\ell$, i.e., $\ell \in [m, \ldots, n]$.

	Let $g$ denote the number of \texttt{1}'s belonging to the surplus bits in cluster $c$ constrained within the range of the maximal initial, $[1, \dotsc, \ell]$. For each $\ell$, compute $|\mathcal{M}^g_{\ell,x}|$ and count the combinations of choosing the remaining $c-g$ additional bits in the remaining $n-\ell$ bits. Let $UB = \min(c, \ell-m)$ and $LB = \max(0, c-(n-\ell))$ and thus we get the following:
	\begin{equation}\label{eq:exp-cluster-raw}
	|\Upsilon^c_{n,x}| = \sum_{\ell=m}^{n} \sum_{g=\max\left(0, c-\left(n-\ell\right)\right)}^{\min(c, \ell-m)} |\mathcal{M}^g_{\ell,m}|
	\binom{n-\ell}{c-g}
	\end{equation}
	Finally, inserting \Cref{eq:exp-maximal-initial} into \Cref{eq:exp-cluster-raw} gives
	\begin{equation}\label{eq:exp-cluster}
	|\Upsilon^c_{n,x}| = \sum_{\ell=m}^{n} \sum_{g=LB}^{UB} \binom{(\ell-m-g)+h(x)-1}{\ell-m-g}
	\binom{g+(m-h(x))-1}{g}
	\binom{n-\ell}{c-g}.
	\end{equation}
	As shown in \Cref{eq:exp-cluster}, $|\Upsilon^c_{n,x}|$ depends on the length and the Hamming weight of $x$, but it is independent of the exact form of $x$.
\end{proof}

\section{Simple closed form expression for the size of a cluster}\label{sec:closed-form-cluster}
We have shown that $|\Upsilon^c_{n,x}|$ is independent of the form of $x$. We can now derive a more simplified analytic expression for this count by considering an $x$ string of the following form: $x =\mathtt{11...11}_a\mathtt{00...0}_{m}$, i.e., $a$ \texttt{1}'s followed by $b$ \texttt{0}'s, with $a > 0$ and $b=m-a$.

The $y$ strings in each cluster are precisely the strings of length $n$ that have $a+c$ \texttt{1}'s in them (and $n-a-c$ \texttt{0}'s) where the $a$-th \texttt{1} (i.e., the last one in an initial choice for $x$) occurs before at least $b$ \texttt{0}'s. Clearly there are $\binom{n}{a+c}$ strings with exactly $a+c$ \texttt{1}'s, but some of these will violate the second principle. To find an expression for counting the valid instances, we sum over the positions of the $a$-th \texttt{1}, which must be between $a$ and $a+z$, where $z = n-a-b-c$ is the number of added \texttt{0}'s. Thus we get the following expression
\begin{equation}
|\Upsilon^c_{n,x}|=\sum_{p=h(x)}^{h(x)+z} \binom{p-1}{h(x)-1} \binom{n-p}{c}.
\end{equation}
With $z=n-m-c$ and $p$ denoting the index of the $a$-th 1, we thus count the number of ways of picking \texttt{1}'s before $p$ and the $c$ \texttt{1}'s after $p$. Note that for $h(x)=0$, the cardinality of cluster $c$ is simply given by $\binom{n}{c}$.

\section{Recursive expression for the size of a cluster}\label{sec:recurrence-cluster}
We present a recurrence for computing the size of a cluster by considering overlaps between the first bits of $x$ and $y$, respectively. Let $\bullet$ and $\varepsilon$ denote concatenation and the empty string, respectively. Moreover, let $x'$ be the tail of $x$ (resp. $y'$ the tail of $y$).
\begin{itemize}
	\item $\Upsilon^c_{n,\texttt{0}\bullet x} = \Upsilon^c_{n-1,x} + \Upsilon^{c-1}_{n-1,\texttt{0}\bullet x}$
	\begin{itemize}
		\item First term: first bit of $y$ is 0, find $x'$ in $y'$
		\item Second term: first bit of $y$ is 1 (part of cluster), so we reduce $c$ and find $x$ in $y'$
	\end{itemize}
	\item $\Upsilon^c_{n,\texttt{1}\bullet x} = \Upsilon^c_{n-1,x} + \Upsilon^{c}_{n-1,\texttt{1}\bullet x}$
	\begin{itemize}
		\item Same arguments as above, but for $x$ starting with 1
	\end{itemize}
	\item Base cases:
	\begin{itemize}
		\item $\Upsilon^0_{n,\texttt{0}\bullet x} = \Upsilon^0_{n-1,x}$
		\item $\Upsilon^0_{n,\texttt{1}\bullet x} = \Upsilon^0_{n-1,x} + \Upsilon^{0}_{n-1,\texttt{1}\bullet x}$
		\item $\Upsilon^c_{n,\varepsilon} = \binom{n}{c}$
		\item if $c+|x| > n$ then return 0 else $\Upsilon^c_{n,x}$
	\end{itemize}
\end{itemize}
It is worth pointing out that since this recursion depends on the form of $x$, i.e., whether or not $x$ starts with a 0 or 1, it does not explicitly capture the bijection between clusters of $x$ strings that have the same Hamming weight, as proved in \Cref{theorem:hamming-cluster}. An implementation of this recurrence and other related algorithms, in $\mathrm{CSP}_{\mathrm{M}}$, can be found in Section \ref{app:csp-code} of the Appendix. The code has been tested using version 3.0 of the FDR\footnote{\url{https://www.cs.ox.ac.uk/projects/fdr/}} tool, a refinement checker for the process algebra CSP.
\section{Enumerating Singletons via Runs}\label{sec:singletons}

Let \emph{singletons} define supersequences in $\Upsilon_{n,x}$ that admit exactly a single mask for a fixed subsequence $x$ of length $m$, i.e., they give rise to exactly a single occurrence of $x$ upon $n-m$ deletions. We use $\mathcal{S}_{n,x}$ to denote this set.
\begin{equation*}
	\mathcal{S}_{n,x} = \{ y \in \Upsilon_{n,x} | \; \omega_x(y) = 1 \}.
\end{equation*}
To compute the cardinality of $\mathcal{S}_{n,x}$, we describe a counting technique based on splitting runs of \texttt{1}'s and \texttt{0}'s in $x$ according to the following observations: $(i)$ inserting bits of opposite value to either side of the framing bits in $x$, i.e., before the first or after the last bit of $x$, does not alter the number of masks. $(ii)$ splitting runs of $0$'s and $1$'s in $x$, i.e., insertion of bits of opposite value in between two identical bits, does not modify the count. This amounts to counting the number of ways that singletons can be obtained from a fixed $x$ string via weight preserving insertions.

The number of possible run splittings corresponds to the number of distinct ways that $c$ \texttt{1}'s and $(n-m-c)$ \texttt{0}'s can be placed in between the bits of the runs of \texttt{0}'s and \texttt{1}'s in $x$, respectively. Again, this count is given by the multiset number $\binom{a+b-1}{a}$, where we count the number of ways $a$ indistinguishable objects can be placed into $b$ distinguishable bins. Note that the number of singletons depends heavily on the number of runs in $x$ and their corresponding lengths. The counting is done by summing over all $n-m$ clusters and computing the configurations that lead to singletons as a function of the runs in $x$ and the number of additional \texttt{1}'s and \texttt{0}'s contributed by each cluster at index $c$.

In order to do this computation, we first count the number of insertions slots in $x$ as a function of its runs of \texttt{1}'s and \texttt{0}'s, given by $\rho_0(x)$ and $\rho_1(x)$, respectively. Let $r_i^j$ be a run with $i$ and $j$ denoting its first and last index and let $\rho_\alpha(x)$ denote the number of insertion slots in $x$ as a function of its runs of $\alpha$. To compute $\rho_\alpha(x)$, we iterate through the runs of $\alpha$ and in $x$ and count the number of indexes at which we can split runs as follows
\begin{align}
	\rho_\alpha(x) & = \sum_{r \in \mathcal{R}_{x,\alpha}} f(r)
\end{align}
where
\begin{equation}
f(r)=
\begin{cases}
|r_i^j|+1,& \text{if } i=1 \wedge j=n\\
|r_i^j|,& \text{if } (i=1 \wedge j<n) \vee (i>1 \wedge j=n)\\
|r_i^j| - 1,              & \text{otherwise}
\end{cases}
\end{equation}
Note that if either the first bit or the last bit of a run overlaps with the first or last bit of $x$, the number of bars is equal to the length of the run. If the said indexes do not overlap with neither the first nor the last bit of $x$, the count is equal to the length of the run minus 1, and finally if both indexes overlap with the first and last bit of $x$ the count is equal to the length of the run plus 1.

We can now count the total number of singletons for given $n$ and $x$ as follows. Let $c$ and $b$ ($b=n-m-c$) denote the number of \texttt{1}'s and \texttt{0}'s contributed by the $c$-th cluster, and the total number of singletons is given by
\begin{align}
|\mathcal{S}_{n,x}| = \binom{n-m + \rho_1(x) - 1}{n-m} &+ \sum_{c=1}^{n-m-1} \binom{b+\rho_1(x)-1}{b} \binom{c+\rho_0(x)-1}{c} \\
&+ \binom{n-m + \rho_0(x)-1}{n-m} \nonumber
\end{align}
The first and last terms correspond to the number of singletons in the first and last cluster, respectively, where we insert either \texttt{1}'s or \texttt{0}'s, but not both. The summation over the remaining clusters counts the configurations that incorporate both additional \texttt{1}'s and \texttt{0}'s.
The final result can be simplified to the identity below
\begin{equation}
   |\mathcal{S}_{n,x}| = \binom{n-m+\rho(x)_1+\rho(x)_0-1}{n-m}.
\end{equation}
\begin{theorem}\label{thm:singleton-extremization}
	The constant (i.e., $x=\texttt{11...1}$ or $x=\texttt{00...0}$) and the alternating $x$ strings ($x=\texttt{1010...}$ or $x=\texttt{0101...}$) maximize and minimize the number of singletons, respectively.
\end{theorem}
\begin{proof}
	This follows immediately from a maximization and minimization of the number of runs in $x$, i.e., $\rho_\alpha(x)$. In the case of the all \texttt{1}'s $x$ string, which comprises a single run, every index in $x$ can be used for splitting. Conversely, the alternating $x$ has the maximum number of runs $|\mathcal{R}|=m$, where $\forall. r \in \mathcal{R}_{x}: |r|=1$, thus splittings are not possible, i.e., no operations of type $(ii)$, and the insertions are confined to pre-pending and appending bits of opposite values to the first and last bit of $x$, respectively.
\end{proof}

\subsection{Further Characterizations of Singletons}

Suppose that for a given string $x$ (length $m$) there is a $y$ string of size 1 (i.e., there is a long string $y$ (length $n$) which has exactly one mask that gives $x$). Then the initial mask is unique. This tells us a great deal about $y$ and $x$.

Firstly, all bits of $y$ following the final bit of the mask must be of opposite value to the last bit of $x$. Secondly, for every $\texttt{10}$ or $\texttt{01}$ consecutive bits in $x$, the corresponding mask bits must also be consecutive.

Thirdly, for every $\texttt{11}$ or $\texttt{00}$ consecutive bits in $x$, if the corresponding mask bits are not consecutive then all the bits of $y$ between the two mask bits must be of the opposite value (this does not use the assumption of uniqueness).

Fourthly, every bit preceding the mask must be opposite to the first bit of $x$ (Again not using uniqueness).

It follows that for a given $x$, the singleton $y^\prime$ can be constructed by taking the number $d=n-m$ (i.e. the difference in length of the two strings) and partitioning it into a sum of however many pairs of consecutive bits there are in $x$ (say $cons(x)$), plus 2, integers (the two comes from the beginning and the end).

So for example if $x=\texttt{101001}$, and $d=6$ then it would be the number of ways in which you can give a sequence of three natural numbers summing to $6:\langle 6,0,0 \rangle, \langle 5,1,0 \rangle, \langle 5,0,1 \rangle, \langle 4,2,0 \rangle, \langle 4,1,1 \rangle, \langle 4,0,2 \rangle$, etc. The singleton corresponding to $\langle 4,1,1 \rangle$ would be $\langle 0,0,0,0,*1,*0,*1,*0,1.*0,*1,0 \rangle$ (*s are in the mask) where four $\texttt{0}$'s are placed before the first bit in the mask, one $\texttt{1}$ is placed between the consecutive $\texttt{0}$s of $x$, and one $\texttt{0}$ is after the mask.

Basically this means that the more equal consecutive bits there are, the more singletons there are.  This is very positive evidence for both conjectures.

Note that the number of these sums is just the $d+1$th $t$-dimensional simplex number (i.e. triangular numbers for $t=2$, tetrahedral number for $t=3$, where $t=cons(x)+1$.

Thus every possible $x$ has at least $d+1$ singletons, with this rising to $(d+2)(d+1)/2$ when there is one pair of consecutive equal bits in $x$, and if there are $k$ pairs of consecutive bits this becomes $\binom{d+t+1}{t+1}$ singletons.
Note that this is minimized for $x$ being our conjectured max entropy case, and minimized for our max entropy case, and maximized for a given weight of $x$ by $\texttt{00000011111}$. One can consider similar but more complex results for doubletons and triples.

\section{Distribution of Embeddings in Clusters}

The number of configurations/beans with $h(x) + a$ ones (where $h(x)$ is the number of ones in the short string $x$) is exactly
\[
\binom{n}{m} \cdot \binom{n-m}{a}
\]
and this depends solely on the number of ones outside $x$, which is intuitively another reason why we can partition the masks into levels $a=0 \ldots n-m$ whose sizes do not depend on anything other than $n,m$ and $a$.

\begin{example}
For $n=7$, $m=5$ and $x=\texttt{00000}$ we get 84 masks in total partitioned into levels as follows:

\begin{itemize}

\item For $a=0$, we get $\binom{7}{5} \cdot \binom{2}{0} = 21$ masks.

\item For $a=1$, we get $\binom{7}{5} \cdot \binom{2}{1} = 42$ masks.

\item For $a=2$, we get $\binom{7}{5} \cdot \binom{2}{2} = 21$ masks.

\end{itemize}
\end{example}

Each $y$ string is a subset of a single level/cluster. Transforming the masks of $x$ to the beans of $x^\prime$ by complementing an $x$-masks selected bits where the corresponding bits of $x$ and $x^\prime$ differ exactly preserves the level of the mask, thus any level $r$ masks of $x$ is transformed to a level $r$ mask of $x^\prime$.  So in seeing the effect of such transformations (e.g. the ones in the conjecture above) we are effectively investigating separate restructuring of the $y$ string at each level individually.

Looking at our conjectured extreme cases and how the different levels' $y$ strings are structured is interesting. If $x$ consists of all zeros then the levels are

\begin{itemize}

\item $a=0$: one $y$ string of size $\binom{n}{m}$.
%a=0  -- one bucket of size C(N,k)

\item In general $\binom{n}{a}$ $y$ strings of size $\binom{n-a}{m}$.
%In general C(N,a) buckets of size C(N-a,k)

\item $a=n-m$: $\binom{m}{m}$ $y$ strings of size 1.
%$a=N-k$: C(N,k) buckets of size 1

So you get extremal behaviour in terms of supersequence structure for the two end levels.

\end{itemize}

\begin{example}

For $n=7$, $m=5$ and $x=\texttt{00000}$ we get:

\begin{itemize}

\item For $a=0$, i.e. all $\texttt{0}$ long string, we get 1 $y$ string of size 21.

\item For $a=1$, i.e. long ($y$) strings with a single 1 in them, we get 7 $y$ strings of size 6

\item For $a=2$, i.e. long ($y$) strings with 2 $\texttt{1}$s in them, we get 21 $y$ strings of size 1.

\end{itemize}
\end{example}

Clearly $x$ all ones is symmetric. If $x = \texttt{0101010101}$, then the supersequence structure for adding $a$ $\texttt{1}$ is exactly the same as that for adding all $n-m-a$ ones (i.e, $m$ $\texttt{0}$s).

\section{Weight Distribution for Finite Deletions}

Here we provide a characterization of the weight distribution in terms of the number of supersequences with a given projection count, for single and double deletions, i.e. $n=m+1$ and $n=m+2$. We do so by identifying various types of insertions that can take a subsequence of length $m$ into a supersequence of length $n$ using $n-m$ insertions of $\texttt{1}$'s and $\texttt{0}$'s.

\subsection{Classifying Supersequences via Single Insertions}

In the case of a single deletion, there are $2n$ projection masks (i.e. $\mu_{n,x} = 2n$), shared among $n+1$ supersequences or $y$ strings.

\subsubsection{Constant Subsequences}

In the case where the short string are all 0 ($x=\texttt{0}^m$) or all 1 ($\texttt{1}^m$), then one $y$ string (namely the string of all 0 or 1 as appropriate) has $n$ masks in, and there are $n$ supersequences with one mask each, namely the strings with a single $\texttt{1/0}$ among $m$ $\texttt{0/1}s$. In other words, this is initial and consequently has the minimum entropy possible.

\subsubsection{Alternating Subsequences}

In the case where the short string $x$ alternates, then the $y$ strings, all obtained via a single bit inserted into $x$, take two forms:

\begin{itemize}

\item There are $m$ supersequences with one of the members doubled. Each of these supersequences admits two masks.

\item There are two supersequences in which the alternating sequence has been extended (with one starting with $\texttt{0}$ and the other with $\texttt{1}$), which just have one mask.

\end{itemize}

\subsection{Classifying Supersequences via Double Insertions}

There are $n \times (n-1)/2 + n + 1$ supersequences and $n \times (n-1) \times 2$ masks.

\subsubsection{Constant Subsequences}

For $x$ all $\texttt{1}$s, there is one supersequence with $n \times (n-1)/2$ masks, that number of supersequences (the ones labelled with a string with two $\texttt{0}$s) with one mask, and $n$ buckets (ones with one $\texttt{0}$) with $n-1$ masks each.

\subsubsection{Alternating Subsequences}

For the alternating subsequence, the $y$ strings are formed by adding 2 bits into the $m$-bit long alternating string. There are five general types of $y$ strings:

\begin{itemize}

\item The unique one which is alternating and formed by adding a bit at each end. There is only one mask in that.

\item The unique $y$ string that is alternating and formed by adding an alternating pair at one end or somewhere in the middle (same result). There are $n-1$ masks in that, as any consecutive pair can be dropped.

\item The $y$ strings formed by adding a single bit at one end or the other to keep alternation, but then destroy alternation by adding a bit somewhere which is equal to an adjacent one.  There are $2 \times (n-1)$ such $y$ strings. Two of these (where it is the just-added bit that is doubled) have one masks, and all the rest have two.

\item The $y$ strings formed by turning one of the originals into a run of three. These each have three masks. There are $m$ of them.

\item The $y$ strings formed by doubling two of the original bits. $m \times (m-1)/2$ of these and each has 4 masks.

\end{itemize}

\stopminichaptoc{\minichaptocenabled}

\chapter{Towards Entropy Minimizing Subsequences for Finite Deletions}\label{chp:finite-deletions}

\printminichaptoc{\minichaptocenabled}

In this chapter, we revisit the original entropy extremization question and prove the minimal entropy conjecture for the special cases of single and double deletions, i.e., for $n=m+1$ and $n=m+2$. We also prove the maximal entropy conjecture for single deletions using a combinatorial counting of long strings (supersequences) and the corresponding weight distribution. In the case of minimization, the entropy result is obtained via a characterization of the number of strings with specific weights, along with an entropy decreasing operation. This is achieved using clustering techniques and a run-length encoding of strings: we identify groupings of supersequences with specific weights by studying how they can be constructed from a given subsequence using different insertion operations, which are in turn based on analyzing how runs of \texttt{1}'s and \texttt{0}'s can be extended or split. The methods used in the analysis of the underlying combinatorial problems, based on clustering techniques and the run-length encoding of strings may be of interest in their own right.

It is worth pointing out that while questions on the combinatorics of random subsequences requiring closed-form expressions are already quite challenging, the problem tackled in this work and first raised in
\cite{atashpendar2015information}, is further complicated by the dependence of entropy on the distribution of subsequence embeddings, i.e., the number of supersequences having specific embedding weights. To put this in contrast, in a related work \cite{swart2003note}, a closed-form expression is provided for computing the number of distinct subsequences that can be obtained from a fixed supersequence for the special case of two deletions, whereas here we need to account for the entire space of supersequences and characterize the number of times a given subsequence can be embedded in them in order to address the entropy question. Moreover, one would have to work out how these weights (number of embeddings) get shifted across their compatible supersequences when we move from one subsequence to another. To the best of our knowledge, other than our original statement of the problem \cite{atashpendar2015information} and the conjectured limiting entropic cases, proving the entropy extremization conjecture has not been addressed before.

\section{Towards Entropy Minimization}\label{sec:entropy-minimization}

We now prove the minimal entropy conjecture for the special cases of one and two deletions. Our approach incorporates two key steps: first we work out a characterization of the number of $y$ strings that have specific weights $\omega_x(y)$. We then consider the impact of applying an entropy decreasing transformation to $x$, denoted by $g(x)$, and prove that this operation shifts the weights in the space of supersequences such that it results in a lowering of the corresponding entropy. This is achieved using clustering techniques and a run-length encoding of strings: we identify groupings of supersequences with specific weights by studying how they can be constructed from a given subsequence using different insertion operations, which are in turn based on analyzing how runs of \texttt{1}'s and \texttt{0}'s can be extended or split. The methods used in the analysis of the underlying combinatorial problems, based on clustering techniques and the run-length encoding of strings may be of interest in their own right. It is thus our hope that our results will also be of independent interest for analyzing estimation and coding problems involving deletion channels.

\subsection{Entropy Decreasing Transformation}

\begin{definition}\label{def:transformation-g}
We now define the transformation $g$ on strings of length $m$ as follows:
\begin{equation}
g((k_1, \dotsc, k_\ell)) =
\begin{cases}
(k_1+k_2, k_3, \dotsc, k_\ell) & \text{if $\ell > 1$} \\
g(\sigma) = \sigma
\end{cases}
\end{equation}
\end{definition}

Hence $g$ is a ``merging'' operation, that connects the two first blocks together. As we shall see, $g$ decreases the entropy. Thus, one can start from any subsequence $x$ and apply the transformation $g$ until the string becomes $\sigma$, i.e., $x=\texttt{0}^m$ or $x=\texttt{1}^m$. As a result, $\sigma$ exhibits minimal entropy and thus the highest amount of leakage in the original key exchange problem.
Note that, as indicated implicitly in the definition above, this transformation always reduces the number of runs by one by flipping the first run to its complement.

Thus we avoid cases where merging two runs would lead to connecting to a third neighboring run, thereby resulting in a reduction of runs by two. For example, $g$ transforms the string $x = \texttt{1001110} = (1;1,2,3,1)$ into $x = \texttt{0001110} = (0;3,3,1)$, as opposed to $x = \texttt{1111110} = (1;6,1)$.

The plots shown in \Cref{figure:g_transform_impact} illustrate the impact of the transformation $g$ on the weight distribution as we move from $x=\texttt{101010}$ to $x'=\texttt{000000}$, i.e., $\texttt{101010} \rightarrow \texttt{001010} \rightarrow \texttt{111010} \rightarrow \texttt{000010} \rightarrow \texttt{111110} \rightarrow \texttt{000000}$.
\begin{figure}[tp]
	\centering
	\includegraphics[width=0.8\textwidth]{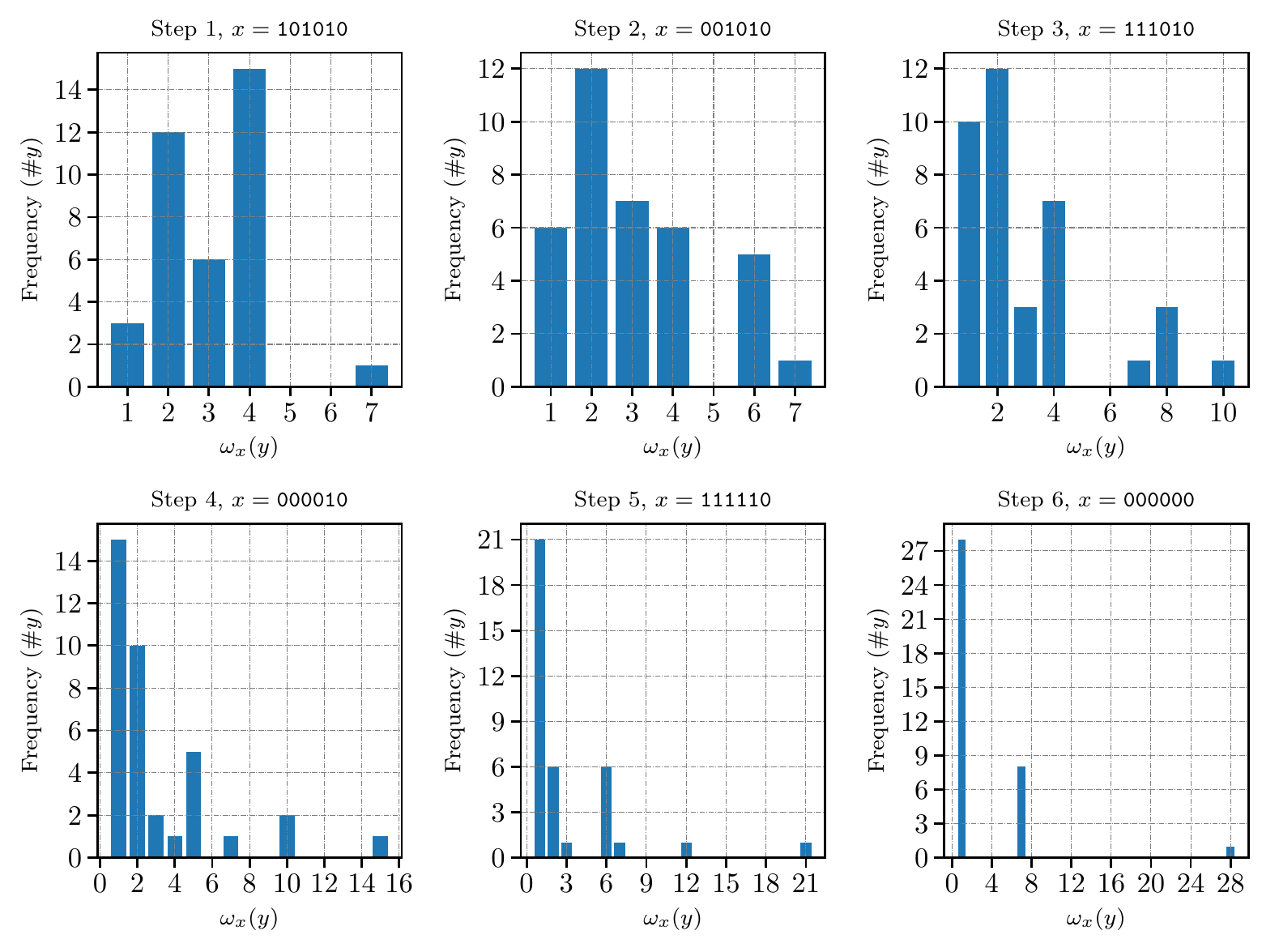}
	\caption{Impact of the transformation $g$ on the weight distribution for converting $x=\texttt{101010}$ to $x'=\texttt{000000}$, with $n=8, m=6$.}
	\label{figure:g_transform_impact}
\end{figure}

\section{Entropy Extremization for Single Deletions}\label{sec:single-deletion}

In this section we consider the case of a single deletion. Let $x$ be a fixed string of length $m$. We study the space of $y$ strings of length $n = m+1$ that can be masked to yield $x$, i.e., $Y_1 = \{y \in \{0,1\}^n \mid \exists \delta \in \mathcal{P}([n]), y_\delta = x \text{ and } |\delta| = 1\}$. Recall that we associate a weight $\omega_x(y)$ to each $y \in Y_1$, defined as the number of ways that $y$ can be masked into $x$. Finally, we define the entropy associated to $x$ as the Shannon entropy of the variable $Z \in \{0,1\}^n$ having distribution
\begin{equation*}
\Pr[Z=y] = \frac{1}{\mu_1} \omega_x(y).
\end{equation*}
where $\mu_1 = \sum_{y \in \Upsilon_{n,x}}$, which for the case $m=n-1$ gives $\binom{n}{m}2^{n-m} = \binom{n}{n-1}2^{n-(n-1)} = 2n$.

\subsubsection{Clustering Supersequences via Single Insertions}
Let $x = (k_1, \dotsc, k_\ell)$. A string $y \in Y_1$ can take only one of the following forms:
\begin{enumerate}
\item $y = (k_1, \dotsc, k_{i-1}, k_i+1, k_{i+1}, \dotsc  k_\ell) $ for some $i \in [\ell]$;
\item $y= (k_1, \dotsc, k_{i-1}, k'_i, 1, k''_i, k_{i+1}, \dotsc, k_\ell)$, for some $i \in [\ell]$ and where $k'_i + k''_i = k_i$ and $k'_i \neq 0$ and $k''_i \neq 0$;
\item $y=(1, k_1, \dotsc k_\ell)$;
\item $y= (k_1, \dotsc, k_\ell, 1)$.
\end{enumerate}
The first case will be referred to as a ``block-lengthening insertion'', denoted by $1/0$, which corresponds to extending runs/blocks. The last three cases will be referred to as ``block-splitting insertions'', and denoted by $0/1$, corresponding to splitting runs or adding a new run of length $1$. For the remainder of our discussion, $a/b$ means: ``$a$ block-lengthening insertions and $b$ block-splitting insertions''.

\begin{lemma}\label{lem:single-deletion-clustering}
$Y_1$ is composed of\footnote{A sanity check can be done to verify that we do not miss any strings, since $\binom{m+1}{m} + \binom{m+1}{m+1} = \ell + m - \ell +2$.}:
\begin{itemize}
	\item $\ell$ block-lengthening insertions, resulting in strings of respective weights $ k_1 + 1, k_2 + 1, k_3+1, \dotsc, k_\ell+1$; and
	\item $m - \ell +2$ block-splitting insertions, i.e., strings of weight 1.
\end{itemize}
\end{lemma}

\subsection{Proof of Minimal Entropy For Single Deletions}

\begin{lemma}\label{lemma:g-single-deletion}
The transformation $g$ decreases the entropy $H_n(x)$ for single deletions, i.e., $m=n-1$.
\end{lemma}

\begin{proof}
The proof consists of computing the difference between the entropy before and after applying $g$, i.e., $\Delta_1 = H_{n}(x) - H_n(g(x))$, and showing that this difference is positive.
From Lemma \ref{lem:single-deletion-clustering}, after applying $g$,
\begin{itemize}
\item The block-lengthening insertions give $\ell - 1$ strings of respective weights $ k_1 + 1 + k_2 + 1 -1, k_3+1, \dotsc, k_\ell+1$.
\item The block-splitting insertions give $m+2 - (\ell - 1)$ strings of weight 1.
\end{itemize}
We now compute the difference of the entropy thanks to the analyses of
\[
(k_1, k_2, k_3, \dotsc, k_\ell)
\]
and
\[
(k_1+k_2, k_3, \dotsc, k_\ell),
\]
which after simplification gives
\begin{equation*}
\Delta_1(k_1, \dotsc, k_\ell) = (k_1+1) \log \frac{1}{k_1+1} + (k_2+1) \log \frac{1}{k_2+1} - (k_1+k_2+1) \log \frac{1}{k_1+k_2 +1}.
\end{equation*}
This is positive, since
\[
\log\frac{1}{k_1+1},\log\frac{1}{k_2+1} > \log\frac{1}{k_1+k_2+1} \; \text{and} \; (k_1+1)+(k_2+1) > k_1+k_2+1.
\]
\end{proof}

\begin{corollary}\label{cor:entropy-single-deletion}
For all $n$ and any subsequence $x$ of length $m = n-1$, we have
\[H_n(x) \geq H_n\left( \sigma \right),\]
with equality only if $x\in \{\texttt{0}^m, \texttt{1}^m\}$.
\end{corollary}
\begin{proof}
Given any $x \neq \sigma$ of length $m = n-1$, it can be transformed into the string $\sigma$ by a series of consecutive $g$ operations, as defined in \Cref{def:transformation-g}. Each such operation can only decrease the entropy, as shown in Lemma \ref{lemma:g-single-deletion}, and thus we get a proof for the fact that $H_n(x) \geq H_n\left(\texttt{0}^m\right)$.
\end{proof}
\begin{remark}
It is worth pointing out that for the special case of single deletions, the minimization of entropy by the constant string, $x=[m]$, can also be proved using a simple combinatorial argument as follows. For $m=n-1$, in cluster $c=1$ we get a single $y$ string with maximum weight, $\omega_y(x)=\binom{n}{m}$, corresponding to $y=[n]$ and $x=[m]$, and the remaining strings in cluster $c=0$ are all singletons, $\omega_x(y)=1$. This is clearly the most concentrated distribution and hence the least entropic one. However our analysis shows how we will deal with the more complicated case of double deletions. In the case of single deletions, we can also illustrate the utility  of our approach by deriving a stronger result for the R\'enyi entropy.
\end{remark}

\begin{definition}[R\'enyi Entropy]
For any $\alpha>0$ and $\alpha \neq 1$, the R\'{e}nyi entropy of order $\alpha$ of a distribution $P$ is defined by
\begin{equation*}
 H_{\alpha} = \frac{1}{1-\alpha} \log_2
 \sum_{i=1}^n p_i^{\alpha}
\end{equation*}
\end{definition}

\begin{theorem}
For all $n$ and any subsequence $x$ of length $m = n-1$, and $\alpha>0$, $\alpha \neq 1$, we have that $\sigma$ exhibits the lowest R\'enyi entropy $H_\alpha$
\[H_{\alpha}(x) \geq H_{\alpha}\left( \sigma\right),\]
with equality only if $x\in \{\texttt{0}^m, \texttt{1}^m\}$.
\end{theorem}

\begin{proof}
Similar to the proof of the Shannon entropy minimization, we just have to show that $g$ decreases the R\'enyi entropy as well such that here almost all the terms also disappear and we end up with
\begin{equation*}
H_{\alpha}(x) - H_{\alpha}\left( g(x) \right) =
\frac{\alpha}{1-\alpha}\left( (k_1+1)^{\alpha} + (k_2+1)^{\alpha}  - (k_1+k_2+1)^{\alpha} \right)
\end{equation*}
This is positive since $(k_1+1)^{\alpha} + (k_2+1)^{\alpha}  - (k_1+k_2+1)^{\alpha} \geq 0 \Leftrightarrow
\alpha < 1$.
\end{proof}

\subsection{Maximal Entropy For Single Deletions}

We now show that the case of maximal entropy for single deletions follows from a simple combinatorial breakdown of the weights and their corresponding $y$ strings.

\begin{theorem}\label{thm:max-entropy-single-deletion}
	For all $n$ and any subsequence $x$ of length $m = n-1$, we have
	\[H_n(x) \leq H_n\left( \texttt{1010}... \right),\]
	with equality only if $x\in \{\texttt{1010}..., \texttt{0101}...\}$.
\end{theorem}
\begin{proof}
In the case where the short string $x$ alternates, then the $y$ strings (all characterized/labelled with a single bit inserted into $x$) take two forms:
\begin{itemize}

\item There are 2 $y$ strings in which the alternating sequence has been extended (with one starting $\texttt{0}$ and the other $\texttt{1}$). These just have one projection, i.e., singletons.

\item There are $m$ supersequences ($y$ strings) with one of the members doubled. Each of these $y$ strings yields two projections, i.e., a flat distribution of weights across the $m$ strings.

\end{itemize}
We also know from Theorem \ref{thm:singleton-extremization} that the alternating strings yield the minimum number of singletons. Thus this configuration is final and hence of maximal possible entropy.
\end{proof}

\section{Entropy Minimization for Double Deletions}

This section will follow the same structure as \Cref{sec:single-deletion}. We will enumerate the different supersequences and their corresponding weights. This is summed up in Lemma \ref{lem:double-deletion-clustering}. We then apply our analysis to a string $x$, and to its image by the function the merging operator $g$ \cref{def:transformation-g}.
\smallskip

In the case of two deletions, there are three types of insertions to consider;
using the notation introduced in the previous section, these are $2/0$, $1/1$, and $0/2$ insertions. For a fixed string $x = (k_1, \dotsc, k_\ell)$, we now analyze each case to account for the corresponding number of supersequences and their respective weights in each cluster. We will then study how this distribution changes when we go from $x$ to $g(x)$ in order to prove the following lemma:
\begin{lemma}\label{lemma:g-double-deletion}
The transformation $g$ decreases the entropy $H_n(x)$ for double deletions, i.e., $m=n-2$.

\end{lemma}

Note that while this technique could be applied to a higher number of insertions, the complexity of the analysis blows up already for two deletions, as the next section will show.

\subsection{Clustering Supersequences via Double Insertions}\label{sec:clustering-supersequences-double-insertion}

\paragraph{Case $2/0$}
The case $2/0$ corresponds to the situation where the insertions do not create new blocks. This happens when both bits are added to the same block, or when they are added to two different blocks, as follows.

The former corresponds to
\[
y = (k_1, \dotsc, k_{i-1},k_i+2, k_{i+1}, \dotsc, k_\ell)
\]
for some $i \in [\ell]$, which has weight $\omega_x(y) = \binom{k_i+2}{2}$. There are $\ell$ strings of this type.

The latter corresponds to
\[
y = (k_1, \dotsc, k_{i-1},k_i+1, k_{i+1}, \dotsc, k_{j-1},k_j+1, k_{j+1}, \dotsc, k_\ell)
\]
for $1 \leq i < j \leq \ell$, and has weight $\omega_x(y) = (k_i+1)(k_j+1)$. There are $\frac{\ell(\ell -1)}{2}$ strings with this weight. In total, there are $\frac{\ell(\ell+1)}{2}$ strings for the case $2/0$.

\paragraph{Case $0/2$}
In the $0/2$ case, there are only block-splitting insertions, hence all strings have weight 1. block-splitting insertions may happen in a single block, or in two separate blocks. To ease notation, we introduce
\begin{equation*}
\widetilde k_i =
\begin{cases}
k_i - 1 & \text{if $i \in [2, \ell-1]$} \\
k_i & \text{if $i = 1$ or $i = \ell$}
\end{cases}
\end{equation*}
The different treatments for ``endpoints'' $1$ and $\ell$ correspond to cases $(1, k_1, \dotsc, k_{\ell})$ and $(k_1, \dotsc, k_{\ell}, 1)$, whereas a block-splitting insertion in the $i$-th block can happen at only $k_i-1$ places.
\begin{itemize}
\item If we insert into the first or the last block, we choose respectively $k_1$ and $k_\ell$ positions, i.e., there are respectively $k_1$ and $k_\ell$ different strings.

\item If we insert into any other block $i$, we choose amongst $k_i-1$ positions, which yields $k_i-1$ different strings.
\item If we insert in different blocks, we apply the same analysis twice, independently, which gives $\widetilde{k_i}\widetilde{k_j}$ different strings.
\item If we insert twice in the same block, we get $\binom{\widetilde{k_i} + 1}{2}$ different strings.
\end{itemize}
In the end, the total number of $0/2$ insertions is
\begin{align*}
\sum_{1 \leq i < j \leq \ell}\widetilde{k_i}\widetilde{k_j} + \sum_{i=1}^\ell \binom{\widetilde{k_i} + 1}{2}
\end{align*}
\begin{example}
If $k_1 = \cdots = k_\ell = 1$, so that $\widetilde k_1 = \widetilde k_\ell = 1$ and $\widetilde k_2 = \cdots = \widetilde k_{\ell-1} = 0$, we count 3 strings.
\end{example}
\begin{example}[]
For example, for $1<i<j<l$ we get for all $a_1, a_2, b_1, b_2 > 0$ such that $a_1 + a_2 = k_i$ and $b_1 + b_2 = k_j$, the string
\[
k_1 \dots k_{i-1} a_1 1 a_2 k_{i+1} \dots k_{j-1} b_1 1 b_2 k_{j+1} \dots k_l.
\]
The number of such strings is $(\widetilde{k_i})(\widetilde{k_j})$.

Another example: for the particular cases $i=j=1$ we get for all $a_1,a_2,a_3 \in \mathbb{N}$ with $a_2$ strictly positive such that $a_1 + a_2 + a_3 = k_1$, the string $a_1 1 a_2 1 a_3 k_2 k_l$ or (case $a_2 = 0$) $a_1 2 a_3 k_2 k_l$. The number of such strings is $\binom{\widetilde{k_1} + 1}{2}$.
\end{example}

\paragraph{Case $1/1$}
As in the previous case, we choose a block in which we apply a block-lengthening insertion, yielding $k_i + 1$ masks; then we choose a block for a block-splitting insertion, yielding $\widetilde k_i$ strings. However, one must be careful:
to see why, consider the following string $x = \texttt{000111} = (\texttt{0};3,3)$.
\begin{itemize}
\item If we insert a block-lengthening \texttt{\good{0}} in the first block, and then a block-splitting \texttt{\bad{1}} in the last-but-one position of the first block, we get the string $y = \texttt{000\bad{1}\good{0}111} = (\texttt{0};3,1,1,3)$.
This string is of weight $(3+1) + (3+1) -1$, since we can delete the $\texttt{\good{0}}$ then one of the four $\texttt{1}$, or the $\texttt{\bad{1}}$ then one of the four $\texttt{0}$, and we remove one so that we do not double count the deletion of $\texttt{\bad{1}\good{0}}$.
\item If we insert a block-lengthening \texttt{\good{1}} in the second block, followed by a block-splitting \texttt{\bad{0}} in the second position of the second block, we obtain the same string $y = \texttt{000\good{1}\bad{0}111} = (\texttt{0}; 3,1,1,3)$.
\end{itemize}
Hence there are two ways to get each $y$. We will therefore exercise a preference toward the first situation, where we perform a block-lengthening insertion in the first block, followed by a block-splitting insertion in the first block's last-but-one position. Let $i \in [\ell]$.
\begin{itemize}
\item If $i = 1$, we get $\sum_{j=1}^{\ell}\widetilde{k_j} (=m -\ell +2)$ strings of weight $k_1 + 1$, as well as a string of weight $k_1+1+k_2$. In total, we get $m-\ell + 3$ strings.
\item If $1 < i < \ell$, we perform a block-lengthening insertion in the block $i$, the number of strings we will get is $(\sum_i \widetilde{k_i}) $. Indeed the string
\[
(k_1, \dotsc, k_{i-1}, 1, 1, k_i, k_{i+1}, \dotsc, k_\ell)
\]
will be counted for the case $i-1$. Each of these strings has weight $k_i + 1 $, except one $(k_1, \dotsc, k_i, 1, 1, k_{i+1}, \dotsc, k_\ell)$ which has weight $k_i + 1 + k_{i+1}$ (the string that we will not count for $i+1$).
\item If $i=\ell$, we can keep the same formula by introducing $k_{\ell +1} = 0$ for the weight of the string $(k_1, \dotsc,  k_\ell,1,1)$.
\end{itemize}

\begin{lemma}\label{lem:double-deletion-clustering}
$Y_2$ is composed of:
\begin{itemize}
	\item case $2/0$: for all $i$ in $[\ell]$ we have one supersequence of weight $\binom{k_i+2}{2}$ and $\forall 1 \leq i < j \leq [\ell]$ we have one supersequence of weight $(k_i +1)(k_j+1)$
	\item case $0/2$: we have $\sum_{1\leq i < j \leq \ell}\widetilde{k_i}\widetilde{k_j} + \sum_{i=1}^{\ell} \binom{\widetilde{k_i}+1}{2}$ supersequences of weight $1$.
	\item case $1/1$: for all $i$ in $[\ell]$ we have $m-\ell+2$ supersequence of weight $k_i+1$ and one of weight $k_i+k_{i+1}+1$ with the convention that $k_{\ell+1}=0$.
\end{itemize}
\end{lemma}
Since the analysis is quite convoluted, we make two sanity checks on the number of supersequences and the sum of all the weights.

\begin{remark}[Sanity check for the number of supersequences]
\label{REMARK:NUMSTRING-DOUBLE}
We check the result of \Cref{lem:double-deletion-clustering} against Eq. \ref{eq:upsilon-cardinality}.

We give an algebraic proof in Appendix \ref{app:2} that if $(k_i)_{i \in \{1, \dotsc, \ell \} }$ are
positive integers such that $m = \sum_{i=1}^{\ell} k_i$, then we have
\[
\frac{\ell(\ell+1)}{2} + \sum_{1 \leq i < j \leq \ell}\widetilde{k_i}\widetilde{k_j} + \sum_{i=1}^\ell \binom{\widetilde{k_i} + 1}{2} + 1 + \ell(m-\ell -2) = \binom{m+2}{m} + \binom{m+2}{m+1} + \binom{m+2}{m+2},
\]
to make sure we have not missed or double-counted any strings.
\end{remark}

\begin{remark}[Sanity check for the sum of all weights]
\label{REMARK:NUMWEIGHT:DOUBLE}
We check the result of \Cref{lem:double-deletion-clustering} against \Cref{eq:numbmask}.
Similarly, to ensure that we have not missed or double-counted any weights, we give an algebraic proof in \Cref{app:3} showing that if there exist positive integers $(k_i)_{i \in \{1, \dotsc, \ell \} }$ such that $m = \sum_{i=1}^{\ell} k_i$, then we have
\begin{align*}
 &\sum_{i=1}^{\ell}\binom{k_i+2}{2} + \sum_{1 \leq i < j \leq \ell}(k_i +1)(k_j+1) + \sum_{1\leq i < j \leq \ell}\widetilde{k_i}\widetilde{k_j} + \sum_{i=1}^{\ell} \binom{\widetilde{k_i}+1}{2} \\ &+ \sum_{i=1}^{\ell}\left[(m-\ell+2)\times (k_i+1)+k_i+k_{i+1}+1\right] = 4 \binom{m+2}{m}.
\end{align*}
\end{remark}

\subsection{Proof of Minimal Entropy For Double Deletions}\label{sec:double-deletion}
As in \Cref{sec:single-deletion}, we analyze the effects of the merging operation $g$ on entropy. For this, we consider the impact of $g(x) = (k_1+k_2,k_3, \dotsc, k_{\ell})$ on the clustering results developed in \Cref{sec:clustering-supersequences-double-insertion}. We will omit the analyses when no insertions are made in the first or second block, since we will get the same weight and this will disappear in the difference.

\paragraph{Case $2/0$}
For $x$, we had $\frac{\ell(\ell+1)}{2}$ strings of this type, we now have $\frac{\ell(\ell-1)}{2}$, there are $\ell$ less strings and $\ell-1$ that grow bigger. The rest remains the same.

\paragraph{Case $0/2$}
Similar to $x$, we have a certain number of strings with weight 1 counted as before \[\sum_{3 \leq i \leq j \leq \ell} \widetilde{k_i}\widetilde{k_j} + \sum_{3 \leq i \leq \ell} \binom{\widetilde{k_i} + 1}{2}\]
However, a part of the formula changes:
\begin{equation}\label{eq:1}
\binom{k_1+k_2 + 1}{2} + (k_1+k_2) \sum_{3 \leq i \leq \ell} \widetilde{k_i}
\end{equation}
Then, for the part of the analysis of $g(x)$ equivalent with that of $x$ we get
\begin{equation}\label{eq:2}
\binom{k_1 + 1}{2} + \binom{k_2}{2} + (k_1 + k_2-1) \times \sum_{3 \leq i \leq \ell} \widetilde{k_i} + k_1(k_2-1)
\end{equation}
now we take the difference between \Cref{eq:1} and \Cref{eq:2}
\begin{equation*}
\sum_{3 \leq i \leq \ell} \widetilde{k_i} + \binom{k_1 + k_2 + 1}{2} - \left(\binom{k_1 + 1}{2} + \binom{k_2}{2} + k_1(k_2-1)\right)
\end{equation*}
After simplifications, we obtain $\sum_{1 \leq i \leq \ell} \widetilde{k_i} + 1$.

\paragraph{Case $1/1$}
In the case of $x$, we had
\[(\ell-1) \sum_{1 \leq i \leq \ell}(\widetilde{k_i} - 1) + \sum_{1 \leq i \leq \ell}\widetilde{k_i}.\]
We now have
\[(\ell-2) (\sum_{1 \leq i \leq \ell}(\widetilde{k_i}-1)+1) + \sum_{1 \leq i \leq \ell}\widetilde{k_i}+1.\]
Taking the difference between now and before we get $\sum_{1 \leq i \leq \ell}{\widetilde{k_i}} + 1 - l$. We have $(\sum_{1 \leq i \leq l}(\widetilde{k_i}-1)+1)$ weights (the block-lengthening insertion in the first block) that grow bigger, the rest stays the same.

\begin{remark}[Sanity check]
We can check that the numbers of strings is constant:
\begin{itemize}
\item Case 0/2: $\sum_{1 \leq i \leq \ell} \widetilde{k_i} + 1$ more strings
\item Case 1/1: $(\sum_{1 \leq i \leq \ell}{\widetilde{k_i}} + 1 - \ell)$ less strings
\item Case 2/0: $\ell$ less strings.
\end{itemize}
and
$  \sum_{1 \leq i \leq \ell} \widetilde{k_i} + 1 -  (\sum_{1 \leq i \leq \ell}{\widetilde{k_i}} + 1 - \ell) -  \ell = 0$.
\end{remark}
We can now compute the difference of the two entropies.
Note that instead of working with the probabilities, we will multiply everything by $4\binom{m+2}{m}$ (i.e., the total number of masks). We can focus on the very few strings that show a change in weight (when an insertion is made in the first or second block).
\\
\textbf{Case $2/0$:} For $x$, we have $1$ string for each of the weights
\begin{align*}
& (k_1+1)(k_2+1), (k_1+1)(k_3+1), \dotsc, \\
& (k_1+1)(k_l+1),(k_2+1)(k_3+1), (k_2+1)(k_4+1), \dotsc, \\
& (k_2+1)(k_l+1), \binom{k_1+2}{2}\binom{k_2+2}{2}
\end{align*}
For $g(x)$, we still have $1$ string for each of the following weights:
\begin{equation*}
(k_1+k_2+1)(k_3+1), (k_1+k_2+1)(k_4+1), \dotsc, (k_1+k_2+1)(k_l+1), \binom{k_1+k_2+2}{2}
\end{equation*}
\\
\textbf{Case $0/2$:} For $g(x)$, we have $\sum_{1 \leq i \leq \ell}\widetilde{k_i}+1$.
\\
\textbf{Case $1/1$:} For $x$, the remaining strings are:
\begin{center}
\small
\begin{tabular}{cc}\toprule
Multiplicity & Weight \\\midrule
$\sum_{i=1}^{\ell}\widetilde{k_i}$ & $k_1+1$ \\
$\sum_{i=1}^{\ell}\widetilde{k_i}-1$ & $k_2+1$\\
$1$ & $k_1+k_2+1$\\
$1$ & $k_2+k_3+1$
\\\bottomrule
\end{tabular}
\end{center}
There remains, for $g(x)$, one string for each of the following weights $k_3+1,k_4+1, \dotsc, k_l+1$
and $\sum_{1 \leq i \leq \ell}\widetilde{k_i}+1$ strings of weight $k_1+k_2+1$ along with $1$ string of weight $k_1+k_2+k_3+1$. The difference of entropies is equal to the difference between $A$ and $B$ defined in the following equations:
\begin{align*}
A ={}
& \sum_{2 \leq i \leq \ell} (k_1+1)(k_i+1) \log\frac{1}{(k_1+1)(k_i+1)} +
 \binom{k_1+2}{2}\log \frac{1}{ \binom{k_1+2}{2}} + \binom{k_2+2}{2}\log \frac{1}{ \binom{k_2+2}{2}} \\
& + \sum_{1 \leq i \leq \ell}\widetilde{k_i}(k_1+1)\log\frac{1}{(k_1+1)} +(k_1+k_2+1)\log\frac{1}{(k_1 +k_2 +1)} \\
& + (\sum_{1 \leq i \leq \ell}\widetilde{k_i}-1)(k_2+1)\log\frac{1}{(k_2 +1)}\\
& + (k_2 + k_3+1)\log\frac{1}{(k_2 +k_3+1)}\\
B ={}
& \sum_{3 \leq i \leq l}(k_i+1) \log\frac{1}{(k_i+1)} + (\sum_{1 \leq i \leq l}\widetilde{k_i}+1)(k_1+k_2+1)\log\frac{1}{(k_1+k_2 +1)}\\
& + (k_1+ k_2 + k_3+1)\log\frac{1}{(k_1 + k_2 +k_3+1)}\\
&  + \sum_{3 \leq i \leq l} (k_1+k_2+1)(k_i+1) \log\frac{1}{(k_1+k_2+1)(k_i+1)} \\
& + \binom{k_1 +k_2+2}{2}\log \frac{1}{ \binom{k_1+k_2+2}{2}}
\end{align*}
where $A$ corresponds to $x$, and $B$ corresponds to $g(x)$.
We are now in a position to conclude the proof of Lemma \ref{lemma:g-double-deletion}.
\begin{lemma}\label{TH:6.3}
The transformation $g$ decreases the entropy $H_n(x)$ for double deletions, i.e., $m=n-2$.
\end{lemma}
\begin{proof}
To prove this, it suffices to show that for $\ell \geq 2$, $k_i \geq 1$, $A - B > 0$. The proof mostly consists of computing partial derivatives to show that the function is increasing. We refer the reader to \Cref{app:deletions} for details.
\end{proof}

\begin{corollary}\label{thm:entropy-two-deletions}
For all $n$ and any subsequence $x$ of length $m = n-2$, we have
\[H_n(x) \geq H_n\left( \sigma \right),\]
with equality only if $x\in \{\texttt{0}^m, \texttt{1}^m\}$.
\end{corollary}
\begin{proof}
Given any $x \neq \sigma$ of length $m = n - 2$, it can be transformed into the string $\sigma$ by a series of consecutive $g$ operations (cf. \Cref{def:transformation-g}). Each such operation can only decrease the entropy, as proved in Lemma \ref{lemma:g-double-deletion}, and thus we get a proof for the fact that $H_n(x) \geq H_n\left(\texttt{0}^m\right)$.
\end{proof}

\section{Concluding Remarks}\label{sec:chp-finite-length-deletion-conclusion}

From the original cryptographic motivation of the problem, the minimal entropy case corresponding to maximal information leakage is arguably the case that interests us the most. While our results shed more light on various properties of the space of supersequences and the combinatorial problem of counting the number of embeddings of a given subsequence in the set of its compatible supersequences, the original entropy maximization conjecture remains an open problem. Finally, proving the entropy minimization conjecture for an arbitrary number of deletions as well as a more general characterization of the distribution of the number of subsequence embeddings in supersequences of finite-length present some further open problems.

\stopminichaptoc{\minichaptocenabled}

\chapter{Entropy Minimization via Hidden Word Statistics}\label{chp:hws}

\printminichaptoc{\minichaptocenabled}

\section{Introduction}

In the previous chapter, we provided an analysis of the same information theory problem proving the entropy minimization conjecture for the special cases of single and double deletions, i.e., $m=n-1$ and $m=n-2$. While the methodology used in \cite{atashpendar2018clustering} and described in Chapter \ref{chp:finite-deletions} depended on showing that any bit string can be transformed into the uniform bit string by successively applying an operation that strictly decreases the entropy, here we adopt an entirely different approach based on some key theorems proven in the works of Flajolet, Szpankowski and Vall\'ee \cite{flajolet2006hidden} on \textit{hidden word statistics}. More precisely, we rely on the fact that the distribution of subsequence embeddings asymptotically tends to a Gaussian to obtain estimates for the entropy based on the moments of the posterior distribution. A crucial quantity for establishing the limiting case of entropy minimization is a measure of autocorrelation that is used in estimating the variance.

The entropy minimization result ultimately follows from a maximization of this autocorrelation coefficient by the uniform string. The number of runs and their respective lengths in $x$ strings play a central role in the distribution of subsequence embeddings, and in turn, in the corresponding entropy. While this property was already hinted at in \cite{atashpendar2015information}, and directly used in the entropy minimization proof for single and double deletions in \cite{atashpendar2018clustering}, our experimental data in this work indicates that the autocorrelation coefficient captures this run-dependent entropy ordering perfectly.

The common thread shared between the present work and the previous papers in this series \cite{atashpendar2015information,atashpendar2018clustering} can be described as a characterization of the limiting entropic cases of the distribution of subsequence embeddings over candidate bit strings transmitted via a deletion channel. This problem is directly linked to that of enumerating the occurrences of a fixed pattern as a subsequence in a random text, also known as the \textit{hidden pattern matching} problem \cite{flajolet2006hidden}. Moreover, the distribution of the number of times a string $x$ appears as a subsequence of $y$, lies at the center of the long-standing problem of determining the capacity of deletion channels: knowing this distribution would give us a maximum likelihood decoding algorithm for the deletion channel \cite{mitzenmacher2009survey}. In effect, upon receiving $x$, every set of $n-m$ symbols is equally likely to have been deleted. Thus, for a received sequence, the probability that it arose from a given codeword is proportional to the number of times it is contained as a subsequence in the originally transmitted codeword. More specifically, we have $p(y|x) = p(x|y)\frac{p(y)}{p(x)} = \omega_x(y)d^{n-m}(1-d)^m \frac{p(y)}{p(x)}$, with $d$ denoting the deletion probability. Thus, as inputs are assumed to be a priori equally likely to be sent, we restrict our analysis to $\omega_x(y)$ for simplicity.

In this chapter, we confirm the entropy minimization conjecture in the asymptotic limit using results from hidden word statistics. To do so, we relate our study to the hidden pattern matching problem investigated in the works of Flajolet et al. \cite{flajolet2006hidden}. We show how their analytic-combinatorial methods can be applied to resolve the case of fixed output length and $n\rightarrow\infty$, by obtaining estimates for the entropy in terms of the moments of the posterior distribution.

\subsection{Results}

We consider the random variable $\Omega_n$, the number of ways of embedding a given output string into a uniformly random input string. Results from hidden word statistics derived by Flajolet et al. \cite{flajolet2006hidden} establish a Gaussian limit law for $\Omega_n$ by showing that the moments of $\Omega_n$ converge to the appropriate moments of the standard normal distribution and determine the mean and variance of the number of embedding occurrences. We use these results to establish the limiting case of the random variable $\Omega_n$ in terms of its variance via an approach that depends intricately on the form of $x$ by incorporating a measure of autocorrelation of $x$. We then relate these results to the original entropy problem to prove the case of maximal information leakage for large $n$.

\subsection{Structure}

In Section \ref{sec:framework}, we introduce some notation and describe the main definitions, models, and building blocks used in our study. We then relate our work to the hidden pattern matching problem in Section \ref{sec:hws-entropy-estimation} and use results from hidden word statistics to prove the entropy minimization conjecture. Finally, we conclude by presenting some open problems in Section \ref{sec:hws-conclusions}.

\section{Framework and Hidden Word Statistics}\label{sec:framework}

In this section, while we rely on the same fundamental ingredients introduced in Chapter \ref{chp:deletion-channel-framework}, we briefly elaborate on a few additional concepts and definitions, along with notation and terminology that will be needed throughout. We will then review some of the building blocks used in hidden word statistics that will be required for obtaining our results.

\subsection{Subsequence Embeddings and Entropy}

\paragraph{Notation} We use the notation $[n] = \{1, 2, \dotsc, n\}$ and $[n_1,n_2]$ to denote the set of integers between $n_1$ and $n_2$; individual bits from a string are indicated by a subscript denoting their position, starting at $1$, i.e., $y = (y_{i})_{i \in [n]} = (y_1, \dotsc, y_n) $. We denote by $|S|$ the size of a set $S$ and the length of a binary string. We also introduce the following notation: when dealing with binary strings, $[a]^k$ means $k$ consecutive repetitions of $a \in \{0, 1\}$. Throughout, we use $h(s)$ to denote the Hamming weight of the binary string $s$.

\paragraph{Probabilistic Model and Alphabet} We consider a memoryless i.i.d. source that emits symbols of the input string (supersequence), drawn independently from the binary alphabet $\Sigma = \{0, 1\}$. Let $\Sigma^n$ denote the set of all $\Sigma$-strings of length $n$ and $p_\alpha$ the probability of the symbol $\alpha \in \Sigma$ being emitted. For a given input length $n$, a random text is drawn from the binary alphabet according to the product probability on $\Sigma^n$: $p(y) \equiv p(y_1 \ldots y_ n) = \prod_{i=1}^n p_{y_i} = 2^{-n}$. The probability of a subsequence of length $m$ is defined in a similar manner.

\paragraph{Number of Masks or Subsequence Embeddings} Recall that $\omega_x(y)$ denotes the number of distinct ways that $y$ can project to $x$:
\[
\omega_x(y) := | \{ \pi \in \mathcal{P}([n]): y_\pi = x \} |
\]
we refer to the number of masks associated with a pair $(y, x)$ as the weight of $y$, i.e., the number of times $x$ can be embedded in $y$ as a subsequence. Moreover, we use $\Omega_n(x)$ to denote the number of occurrences of a given subsequence $x$ in a random text of length $n$ generated by a memoryless source.

\subsection{Generating Functions}

Generating functions constitute an essential component in the following analysis as they play a central role in determining the combinatorial and statistical properties of hidden patterns. Here we briefly discuss the utility of generating functions and how combinatorial problems can be translated into generating functions. For more details, we refer the reader to standard textbooks such as ``generatingfunctionology`` by Herbert S. Wilf \cite{wilf2005generatingfunctionology} and ``Analytic Combinatorics'' by Philippe Flajolet and Robert Sedgewick \cite{flajolet2009analytic}.

Roughly speaking, generating functions turn problems about sequences into problems about functions and thus enable access to the mathematical machinery available for manipulating functions. In other words, a generating function is a power series that encodes an infinite sequence $(a_n)$ into its coefficients, e.g., the infinite sequence $(a_0, a_1, a_2, a_3, \ldots)$ is given by the \emph{ordinary generating function}
\[
A(x) = a_0 + a_1 x + a_2 x^2 + a_3 x^3 + \cdots
\]
Generating functions are often referred to as ``formal'' power series to highlight the fact that the variable $x$ is actually treated as an indeterminate, which serves as a placeholder rather than a number. As such, except for certain cases, in general convergence is not a concern and the formal power series is allowed to diverge.

\subsubsection{Deriving Closed Form Formulas}

As already mentioned, an important property of generating functions is that they allow us to perform operations on sequences via their associated functions, e.g., scaling, addition and differentiation. For instance, adding two generating functions amounts to adding two sequences term by term.

Once a generating function for a sequence has been established, in many cases we can derive a closed form formula for complicated sequences such as that of the Fibonacci numbers. For instance, using the symbolic method, in particular in the context of analytic combinatorics and thanks to techniques developed by Philippe Flajolet \cite{flajolet2009analytic}, it can be shown that the Fibonacci numbers can be represented by the following simple generating function
\[
(0, 1, 1, 2, 3, 5, 8, 13, 21, \ldots) \leftrightarrow \frac{x}{1-x-x^2}.
\]
A crucial feature of generating functions lies in their effectiveness for deriving closed form expressions. More precisely, given a generating function for a sequence, we can often derive a closed form formula for the $n$-th coefficient. For example, using the generating function $F(x) = \frac{x}{1-x-x^2}$, we can derive the famous closed-form for the Fibonacci numbers, known as Binet's formula, by taking the coefficient of $x^n$ in the power series above\footnote{One way we can extract coefficients from a generating function that is a ratio of polynomials is through the method of \emph{partial fraction expansion/decomposition}.}, to compute the $n$-th Fibonacci number.

\subsection{Hidden Patterns, Constraints and Blocks}

In the terminology of hidden word/pattern statistics, the same problem of determining the number of distinct embeddings of a subsequence in a supersequence is referred to as the ``hidden pattern matching'' problem. Here we review the most relevant concepts introduced in the work of Flajolet, Szpankowski and Vall\'ee \cite{flajolet2006hidden}.

\paragraph{Hidden Patterns and Constraints} Let $\mathcal{W} = w_1, w_2, \ldots, w_m$ denote the pattern or subsequence obtained from the text $T_n = t_1, t_2, \ldots, t_n$, and let $\mathcal{D} = (d_1, \ldots , d_{m-1})$ be an element of $(\mathbb{N}^+ \cup \{ \infty \})^{m-1}$. The pattern matching problem is determined by a pair $(\mathcal{W}, \mathcal{D})$, called a ``hidden pattern'' specification, i.e., a subsequence pattern $\mathcal{W}$ along with an additional set of constraints $\mathcal{D}$ on the indices $i_1, i_2, \ldots, i_m$. If an occurrence in the form of an $m$-tuple $S=(i_1, i_2, \ldots, i_m)$ with $(1 \le i_1 < i_2 < \ldots < i_m)$ satisfies the constraint $\mathcal{D}$, i.e., $i_{j+1} - i_j \le d_j$, it is then considered to be a valid mask or a position. In essence, the notion of constraints models the existence of gaps between the embeddings of the symbols of a subsequence in a random text. In other words, the analysis considers the number of occurrences of a subsequence as embeddings that satisfy a specific set of distance constraints.

Moreover, let $\mathcal{P}_n(D)$ be the set of all positions subject to the separation constraint $\mathcal{D}$, satisfying $i_m \le n$. Let also $\mathcal{P}(D) = \bigcup_n \mathcal{P}_n(D)$. This allows us to view the number of occurrences $\Omega$ of a subsequence $w$ in text $T$ subject to the constraint $\mathcal{D}$ as a sum of characteristic variables
\begin{equation}
\Omega(T) = \sum_{I \in \mathcal{P}_{|T|}(\mathcal{D})} X_I(T), \quad \text{with} \quad X_I(T) := [[ \mathrm{w}\, \text{occurs at position } I \text{ in } T ]].
\end{equation}
with $[[B]]$ being 1 if the property $B$ holds and 0, otherwise.

The two ends of the spectrum in this model are given by the following. The \emph{fully unconstrained case} is modelled by $\mathcal{D} = (\infty, \ldots, \infty)$; whereas the \emph{constrained problem} is modelled by the case where all $d_j$ are finite. Our study is only concerned with the former, namely the fully unconstrained problem, as we allow an arbitrary number of symbols in between the gaps.

\paragraph{Blocks} A given pattern $x$ is broken down into $b$ independent subpatterns that are called blocks, $x_1, x_2, \ldots, x_b$. The quantity denoted by $b$ is defined as the number of unbounded gaps (the number of indices $j$ for which $d_j = \infty$) plus 1, which is also referred to as the number of blocks. The two extreme cases, namely the fully unconstrained and the fully constrained problem, are thus described by $b=m$ and $b=1$, respectively. For the purpose of our study, we always assume $b=m$. Collections of blocks are then used to form an \emph{aggregate}, which describes the interval of indices that marks a block, the first and last index in an interval. One of the main uses of blocks and aggregates is to model the fact that masks and occurrences of a subsequence can overlap with each other by quantifying the extent to which such overlaps can occur. However, as we are only interested in the fully unconstrained case, covering the notion of aggregates goes beyond the scope of our work. The reader is encouraged to refer to \cite{flajolet2006hidden} for a more complete and detailed presentation of these concepts.

\section{Estimating Entropy via Hidden Word Statistics}\label{sec:hws-entropy-estimation}

We now revisit the original entropy problem and provide an analysis in the asymptotic limit by considering the case of fixed output length $m$ and $n \rightarrow \infty$. This allows us to apply results from hidden pattern statistics to establish the limiting case of minimal entropy. The probabilistic aspects of the statistics of hidden patterns were quantified by Flajolet et al. in an extensive study \cite{flajolet2006hidden}, which was originally motivated by intrusion detection in computer security. Among other things, they showed that the random variable $\Omega_n$ asymptotically tends to a Gaussian. We relate our work to their study and incorporate two key theorems related to hidden patterns to establish the limiting case of minimal entropy via a notion of autocorrelation associated with subsequences.

\subsection{Hidden Word Statistics}

In \cite{flajolet2006hidden}, it is shown that for fixed short strings of length $m$ as $n \to \infty$, the dominant contribution to the moments comes from configurations where the positions of the short strings are minimally intersecting. We will briefly describe the approach used in \cite{flajolet2006hidden}.

For a position $I$ (that is, a subset of $[n]$ of size $m$), let $X_I$ denote the indicator of the event that the long string restricted to $I$ matches the short string. Let $Y_I = X_I - \mathbb{E}(X_I) = X_I - 2^{-m}$. Then $X = \Omega - E = \sum_I Y_I$, and so

\begin{equation}
\mathbb{E}(X^r) = \sum_{I_1, \ldots, I_r} \mathbb{E}(Y_{I_1} \ldots Y_{I_r}).
\end{equation}

Now let $\mathcal{O}_r$ be the combinatorial class consisting of pairs $((I_1, \ldots, I_r), T)$, where the $I_j$ are positions and $T$ is a ``text'' (i.e. a string of length $> m$), taken with weight $Y_{I_1}(T) \ldots Y_{I_r}(T)2^{-|T|}$. Now if $O_r (z)$ is the generating function of $O_r$ with this weighting, we have

\begin{equation}
[z^n]O_r(z) = \sum_{|T|=n}\sum_{I_1,\ldots,I_r} Y_{I_1}(T) \ldots Y_{I_r}(T) 2^{-n} = \mathbb{E}(X^r).
\end{equation}
Note that $[z^n]O_r(z)$ means the coefficient of $z^n$ in $O_r(z)$.

We can partition $\mathcal{O}_r$ according to the number of points covered by some $I_j$. Let $\mathcal{O}_r^{[p]}$ denote the class of elements in which the number of points covered is $rm-p$. Note that if $I_1$ does not intersect with any other $I_j$, then $Y_1$ is independent of $Y_2, \ldots, Y_r$ and so $\mathbb{E}(Y_1 \ldots Y_r) = 0$, so contributions only come from families where each position intersects some other position. Such families are called ``friendly'' and require in particular $p \ge \ceil*{r/2}$.

To obtain the generating function for $\mathcal{O}_r^{[p]}$, we apply a combinatorial isomorphism to group together all the covered points of intersection, so that we have
\begin{equation}
\mathcal{O}_r^{[p]} \cong (\{ 0,1 \}^*)^{rm-p+1} \times \mathcal{B}_r^{[p]},
\end{equation}
where $\mathcal{B}_r^{[p]}$ is the subset of $\mathcal{O}_r^{[p]}$ which is \emph{full}, that is, for which the set of covered points is contiguous. We thus have
\begin{equation}
O_r^{[p]}(z) = \Big( \frac{1}{1-z} \Big)^{rm-p+1} \times B_r^{[p]}(z).
\end{equation}
Since the analysis in \cite{flajolet2006hidden} considers a fixed short string of length $m$ as $n \to \infty$, it is enough to observe that $B_r^{[p]}(z)$ is some \emph{fixed} polynomial, because one can then easily show that the coefficient $[z^n]O_r^{[p]} = O(n^{rm-p})$. This means that however fast the coefficients of $B_r^{[p]}(z)$ grow as $p$ grows, for large enough $n$, the minimal-$p$ term will dominate.

\subsection{Entropy Minimization via Hidden Word Statistics}

We will rely on the fact that the distribution of $\Omega_n$ asymptotically tends to a Gaussian and use a measure of autocorrelation defined for subsequences to obtain estimates for the entropy in terms of the moments of the posterior distribution. Indeed, the underlying probability distribution in our original entropy analysis coincides with that of the so-called \emph{hidden pattern matching} problem in which one searches for the number of occurrences of a given \emph{pattern}\footnote{The words ``subsequence'' and ``pattern'' are used interchangeably.} $\mathcal{W}$, as a subsequence in a random text $T$ of length $n$ generated by a memoryless source. More precisely, given that $\Omega_n \sim \mathcal{N}(\mu, \sigma^2)$, we will analyze how the mean and the variance of the distribution change for different $x$ strings in order to resolve the limiting case of minimal entropy exhibited by the uniform string $[0]^m$.

The probabilistic analysis done in \cite{flajolet2006hidden} relies on a description of the structures of interest in formal languages, involving a joint use of combinatorial-enumerative techniques and analytic-probabilistic methods. This approach enables a systematic translation of the combinatorial problem into generating functions. The essential combinatorial-probabilistic features of the problem, such as variance coefficients and a notion of autocorrelation, are derived by using an asymptotic simplification made possible by the use of the singular forms of generating functions. For an extensive and complete coverage of these techniques, we refer the reader to \cite{flajolet2009analytic,sedgewick2013introduction}.

In our work, we will mainly make use of two fundamental theorems presented in \cite{flajolet2006hidden}. The first theorem states that $\Omega_n$ asymptotically tends to a Gaussian, while the second theorem provides analytic expressions for its moments, i.e., the expectation and the variance of $\Omega_n$. Another equally important result that we will use to distinguish between two different subsequences of length $m$ is a measure of autocorrelation that depends intricately on the exact form of $x$. Given that the mean (Eq. \ref{eq:omega-mean}) is constant for all $x$ strings of equal length, the autocorrelation factor, incorporated in the variance coefficient, allows us to differentiate between two subsequences in that it is the only term  that depends on the form of $x$, with all other terms in Eq. \ref{eq:omega-variance} being only a function of $n$ and $m$.

\subsection{Distribution of Embeddings in the Asymptotic Limit}

The plots given in Fig.\ref{figure:gaussian-omega} illustrate the convergence of the distribution of $\Omega_n$ to a Gaussian for the subsequence $x=\texttt{01}$ and increasing values of $n$. As already mentioned, the distribution of subsequence embeddings tending to a Gaussian in the asymptotic limit is of particular significance for our work given that $\Omega_n$ is precisely the random variable associated with the weights of the supersequences in $\Upsilon_{n,x}$ for the computation of entropy.

\begin{figure}[ht!]
	\centering
	\includegraphics[scale=0.7]{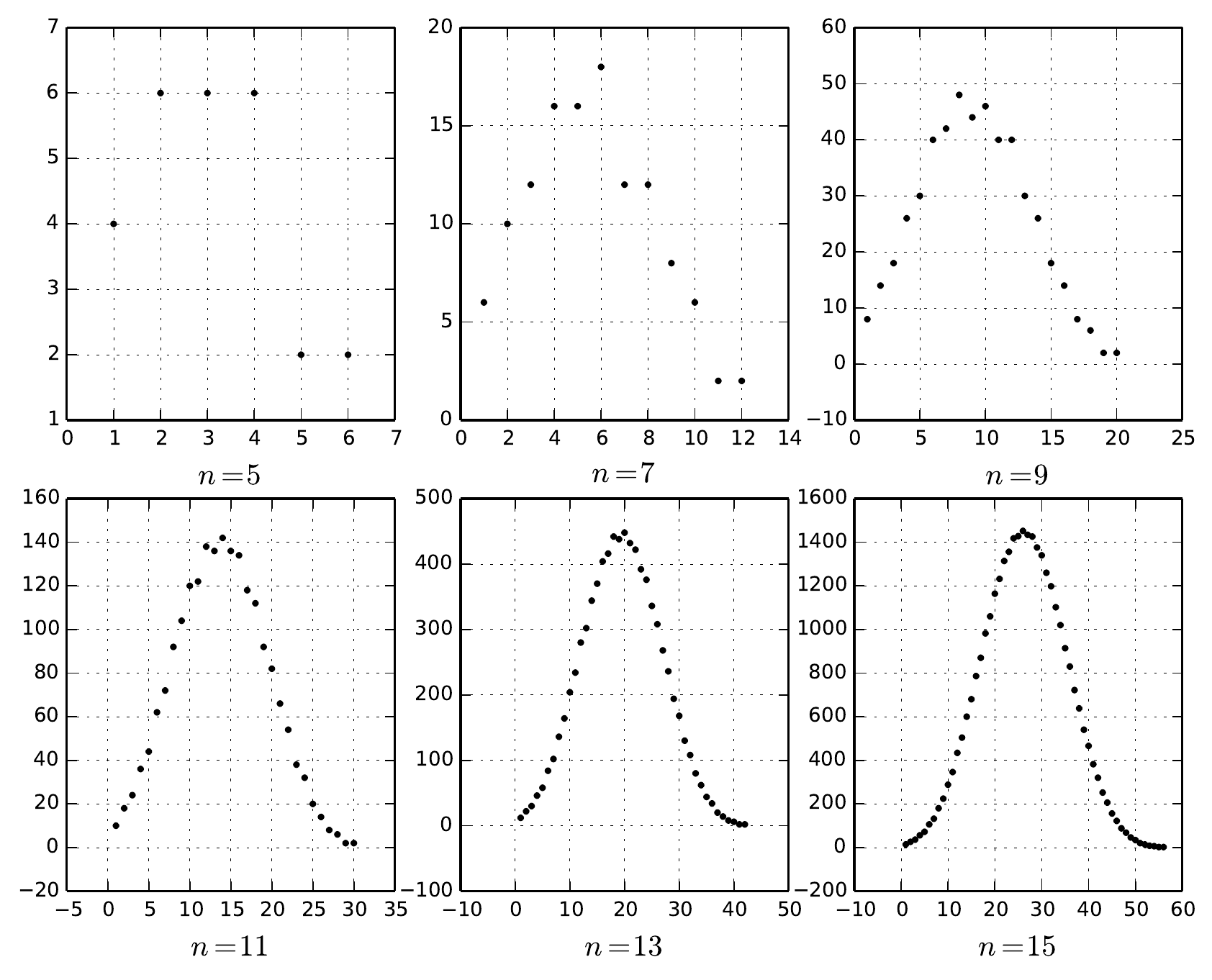}
	\caption{Frequency distribution of $\Omega_n$ converging to a Gaussian for $x=\texttt{01}$ and $n= 5\ldots15$}
	\label{figure:gaussian-omega}
\end{figure}

In the following, we first present the analytic expressions satisfying the mean and the variance of the number of occurrences $\Omega_n$ and adapt them to the parameters of our problem. We then characterize the limiting case of minimal entropy exhibited by the uniform string, i.e. $x = [0/1]^m$, via a notion of autocorrelation coefficient incorporated in the variance.

\subsubsection{Moments and Convergence}

The results provided here have been sourced from \cite{flajolet2006hidden} and adapted to the specific parameters of our problem, i.e., we consider the fully unconstrained setting, restricted to the binary alphabet. For all $x$ strings of length $m$, the mean is constant and therefore, we mainly focus on the variance.
\begin{theorem}\label{thm:hws_moments}\cite{flajolet2006hidden}
	The mean and the variance of the number of occurrences $\Omega_n$ of a subsequence $x$ for $p_\alpha = 0.5$, subject to constraint $\mathcal{D}=(\infty, \ldots, \infty)$, and thus $b=m$, are given by
	\begin{equation}\label{eq:omega-mean}
	\mathbb{E}[\Omega_n] = \frac{2^{-m}}{m!} n^m \left(1 + O\left(\frac{1}{n}\right)\right)
	\end{equation}
	\begin{equation}\label{eq:omega-variance}
	\mathbb{V}[\Omega_n] = \frac{2^{-2m}}{(2m-1)!} \kappa^2(x) n^{2m-1} \Bigg(1 + O\Bigg(\frac{1}{n}\Bigg)\Bigg),
	\end{equation}
	where the autocorrelation $\kappa^2(x)$ is defined by
	\begin{equation}\label{eq:kappa}
	\kappa^2(x) := \sum_{1 \le r,s \le m} \binom{r+s-2}{r-1} \binom{2m-r-s}{m-r} [[x_r = x_s]].
	\end{equation}
	Note that $[[P]]$ denotes the indicator function of the property $P$ (so $[[P]] = 1$ if $P$ holds and $0$ otherwise).
\end{theorem}
\begin{theorem}\label{thm:hws_gaussian}\cite{flajolet2006hidden}
	\begin{equation}
	X_n:=\frac{\Omega_n-\mathbb{E}(\Omega_n)}{\sqrt{\mathbb{V}(\Omega_n)}}
	\end{equation}
	converges in measure to a standard normal distribution.
\end{theorem}

We encapsulate the multiplicands in the definition of $\kappa^2(x)$ into matrices, viewing the indicator function as a mask on the matrix of binomial coefficients.  Let $\mathcal{B}$ be the matrix representing the indicator function $\mathcal{B}_{r,s}:=[[x_r=x_s]]$, and let $\mathcal{M}$ be the matrix of binomial coefficients
\[
\mathcal{M}_{r,s} = \binom{r+s-2}{r-1} \binom{2m-r-s}{m-r}.
\]
Write $\mathcal{R}=\mathcal{B}\circ\mathcal{M}$, the Hadamard or elementwise product of $\mathcal{B}$ and $\mathcal{M}$, for the result of applying the mask $\mathcal{B}$ to the matrix $\mathcal{M}$.  We then have an equivalent formulation of equation \eqref{eq:kappa}, namely
\[
\kappa^2(x) = \sum_{r=1}^m \sum_{s=1}^m \mathcal{R}_{r,s}.
\]

\subsubsection{Autocorrelation}

It is worthwhile to provide some explanation of the combinatorial meaning of the autocorrelation coefficient $\kappa^2$ derived in \cite{flajolet2006hidden}, in view of its significance in the analysis that follows.

The coefficient $\kappa^2$ is related to a generalization of the autocorrelation polynomial originally introduced for classical string matching by Guibas and Odlyzko \cite{guibas1981periods,guibas1981string}.  The variance of $\Omega_n$ is determined by the probability that a random pair of $m$-subsets of a random long string are \emph{both} matches for the short string, and how this compares to the square of the corresponding probability for a single $m$-subset.

Analytic-combinatorial methods show that the dominant contribution for large $n$ comes from pairs which overlap in only a single position, so computing the variance amounts to counting the number of triples consisting of a long string and a pair of $m$-subsets intersecting in precisely one location such that both are matches for the short string.  Grouping the chosen locations together introduces a constant factor of $\binom{n}{2m-1}2^{n-(2m-1)}$, and so it suffices to count the number of ways to interleave two copies of the short string, with a single intersection.  This quantity is the autocorrelation coefficient $\kappa^2(x)$.

Explicitly, $\mathcal{M}_{r,s}$ is the number of combinations with the $r^{th}$ location of the first set meeting the $s^{th}$ location of the second: $\binom{r+s-2}{r-1}$ is the number of interleavings of the $r-1$ and $s-1$ locations before this, and $\binom{2m-r-s}{m-r}$ the number of interleavings of the $m-r$ and $m-s$ locations after.

\subsection{Maximal Autocorrelation}

We now study the extremization of the variance of $\Omega_n$ by analyzing the extreme values of the autocorrelation $\kappa^2$. Here we consider the all-0s and all-1s strings (x=$[0]^m$ and $[1]^m$), for which the autocorrelation matrix contains $m^2$ 1's: $\forall i,j \in \{1 \ldots m\}: x_i = x_j$.
\begin{theorem}\label{theorem:maximal-variance}
	Let $x$ be a string of length $m$. Then
	\begin{equation}\label{eq:kappa-max}
	\kappa^2(x) \leq \kappa^2\left([0]^m \right)
	= \kappa^2\left([1]^m \right) = m\binom{2m-1}{m}.
	\end{equation}
\end{theorem}
\begin{proof}
	Since $\mathcal{M}$ is independent of the form of $x$, we focus only on the indicator matrix $\mathcal{B}$. It is clear that the constant $x$ strings comprising all 0's and all 1's are the unique strings that result in an all-ones masking matrix $\mathcal{B}$. Consequently,  $\kappa^2([0/1]^m)$ includes all of the $m^2$ terms involved in $\mathcal{M}$ and thus attains its maximal value, i.e., $\kappa^2([0/1]^m) = \sum_{1 \le r,s \le m} \mathcal{R}_{r,s} =  \sum_{1 \le r,s \le m} \mathcal{M}_{r,s}$, hence Eq. \ref{eq:kappa-max}.
\end{proof}

The alternating $x$ string $x=\texttt{1010...}$ appears to lie at the other end of the entropy spectrum. While the proof for the maximization of the autocorrelation coefficient by the all 0's string was rather straightforward, showing its minimization still escapes us. We simply state the minimization as a conjecture.

\begin{conjecture}\label{theorem:minimal-variance}
	The alternating subsequence of length $m$, i.e., $x = \texttt{1010...}$, minimizes the autocorrelation coefficient $\kappa^2$.
\end{conjecture}

\subsection{Proof of Entropy Minimization}\label{sec:entropy}

We briefly review the results of the entropy analysis in which it is conjectured that the all 0's and the alternating $x$ string, minimize and maximize the entropy, respectively.

\subsubsection{Calculating Entropy From Moments of Distribution}

An equivalent formulation of Theorem \ref{thm:hws_gaussian} (and in fact the form in which it is proved) is that the moments of the (normalized) converge to the corresponding moments of the standard normal distribution:

\begin{lemma}\label{lem:momconv}The moments of the normalized version of $\Omega_n$ converge to the corresponding moments of the standard normal distribution.  That is,
	\[\mathbb{E}\left(\left(\frac{\Omega_n - \mathbb{E}(\Omega_n)}{\sqrt{\mathbb{V}(\Omega_n)}} \right)^r \right) \rightarrow \begin{cases}0 &r \mbox{ odd} \\
	(r-1)\times (r-3)\times \ldots \times 1 &r \mbox{ even}.\end{cases}\]
\end{lemma}

We have a distribution $\Omega$ and the goal is to estimate $\mathbb{E}(\Omega \mathrm{log} \Omega)$, given the moments of $\Omega$. Let $E = \mathbb{E}(\Omega)$, and let the pdf of $\Omega$ be $f$. By Taylor's theorem, we have

\begin{align}\label{eq:taylor}
\Omega \log\Omega = E\log E + (\log E + 1)(\Omega - E) + \frac{(\Omega - E)^2}{2E} - \frac{(\Omega-E)^3}{6E^2} +  \mathcal{R}(\Omega),
\end{align}
where $\mathcal{R}(\Omega)$ is the integral form of remainder, i.e.
\[\mathcal{R}(\Omega) = \int_E^\Omega \frac{(\Omega-t)^3}{6t^3} dt.\]

Note that $\mathcal{R}(\Omega)$ is non-negative for all $\Omega$.  Now whenever $\Omega \geq \frac{1}{2}E$, we have
\[\mathcal{R}(\Omega) \leq \int_E^\Omega \frac{8(\Omega-t)^3}{6E^3} dt = \frac{8(\Omega-E)^4}{24E^3}.\]

On the other hand if $\Omega < \frac{1}{2} E$ then
\begin{align*}\mathcal{R}(\Omega) &= \int^\Omega_{\frac{E}{2}} \frac{(\Omega-t)^3}{6t^3} dt + \int^{\frac{E}{2}}_E \frac{(\Omega-t)^3}{6t^3} dt \\
&\leq \int^\Omega_{\frac{E}{2}} \frac{1}{6} dt + \int_E^{\frac{E}{2}} \frac{8(\Omega-t)^3}{6E^3} dt \\
&\leq \frac{E}{12} + \int_E^\Omega \frac{8(\Omega-t)^3}{6E^3} dt \\
&= \frac{E}{12} + \frac{8(\Omega-E)^4}{24E^3}.
\end{align*}

Hence we have that
\begin{equation}\label{eq:rembound1}
|\E(\mathcal{R})| \leq \frac{1}{12}E\PP\left(\Omega < \frac{1}{2}E\right) + \frac{\E((\Omega-E)^4)}{3E^3}.
\end{equation}

We obtain a Chebychev bound on the first term:
\begin{align*}
\PP\left(\Omega < \frac{1}{2}E\right) &\leq \PP\left(|\Omega-E| > \frac{1}{2}{E}\right) \\
&\leq \frac{\E\left((\Omega-E)^4\right)}{\left(\frac{1}{2}E\right)^4} = \frac{16\E\left((\Omega-E)^4\right)}{E^4}
\end{align*}

Substituting this into (\ref{eq:rembound1}) gives
\begin{equation}\label{eq:rembound}|\E(\mathcal{R})| \leq \frac{5\E((\Omega-E)^4)}{3E^3}.\end{equation}

Hence taking expectations of (\ref{eq:taylor}) gives
\begin{equation}\label{eq:est}\E(\Omega\log\Omega) = E\log E + \frac{\mathbb{V}(\Omega)}{2E} - \frac{E((\Omega-E)^3)}{6E^2} + \epsilon\left(\frac{5}{3}\frac{\E((\Omega-E)^4)}{E^3}\right),
\end{equation}
where the notation $\epsilon(x)$ means an error term of magnitude at most $x$.

\subsubsection{Minimal Entropy}

We are now in a position to prove the main theorem of this Section, that (for sufficiently large $n$), the entropy is minimized uniquely by the constant strings $[0]^m, [1]^m$.

\begin{theorem}\label{theorem:fixed-k}
	For all $m$, there is some $N$ such that for all $n>N$, and any string $x$ of length $m$, we have
	\[H_n(x) \geq H_n\left([0]^m\right),\]
	with equality only if $x\in \{[0]^m, [1]^m\}$.
\end{theorem}

\begin{proof}
	From Eq. \ref{eq:subseq-prob-distribution} and Eq. \ref{eq:shannon-entropy}, we have
	\begin{align}\label{eq:omegaest}
	\begin{split}
	H_n(x) = H(P) &= -\sum_i p_i \log p_i \\
	&= -\sum_i \left( \frac{\omega(i)}{\mu} \right) \log \left( \frac{\omega(i)}{\mu} \right) \\
	&= \frac{\log\mu}{\mu}\sum_i \omega(i) - \frac{1}{\mu}\sum_i \omega(i)\log\omega(i) \\
	&= \log\mu - \E(\Omega\log\Omega)
	\end{split}
	\end{align}
	Hence it suffices to prove that for sufficiently large n the constant strings maximize $\mathbb{E}\left(\Omega_n \log \Omega_n\right)$.

	Note that $E$ depends only on $n$, and not on the form of $x$; by Theorem \ref{thm:hws_moments} we have $E=\Theta(n^m)$.  On the other hand, by the same Theorem we have $\mathbb{V}(\Omega_n) = \frac{2^{-2m}}{(2m-1)!}\kappa^2(x)n^{2m-1}\left(1+O(1/n)\right)$.

	Now $\kappa^2$ depends only on the form of $x$ and not on $n$, and by Theorem \ref{theorem:maximal-variance}
	it is uniquely maximized by the all-1s/0s strings.  Because $\kappa^2$ is independent of $n$, we therefore also have that the change in $\mathbb{V}(\Omega_n)$ induced by moving away from these strings is $\Theta(n^{2m-1})$, and so it suffices to prove that all of the error terms in \eqref{eq:est} are $o\left(\frac{n^{2m-1}}{E}\right) = o\left(n^{m-1}\right)$.

	Now by Lemma \ref{lem:momconv} (combined with the fact that by Theorem \ref{thm:hws_moments} $\mathbb{V}(\Omega_n) = \Theta\left(n^{2m-1}\right)$), we have that $\E((\Omega-E)^3)=o(n^{3m-3/2})$, and $\E((\Omega-E)^4)=O(n^{4m-2})$.  Combining this with the fact that $E=\Theta(n^m)$ yields the required bounds on the errors, and hence the result.
\end{proof}

\subsubsection{Entropy Ordering based on Autocorrelation}\label{sec:numerical-experiments}

Although we have proved the extremal case of minimal entropy in the asymptotic limit for $n \rightarrow \infty$ and fixed output length $m$ via the autocorrelation coefficient $\kappa^2(x)$, it is worth pointing out that our computer experiments indicate that $\kappa^2(x)$ predicts the entropy ordering perfectly in the finite length domain as well, i.e., for small and comparable fixed values of $n$ and $m$. An example obtained from empirical data is presented in Table \ref{tab:kappa-entropy} to illustrate the correlation between $H_n(x)$ and $\kappa^2(x)$ for $n=8$ and $m=5$.
\begin{table}[ht!]
	\caption{Entropy Ordering Prediction via Autocorrelation Sorting of Subsequences}
	\label{tab:kappa-entropy}
	\centering
	\begin{tabular}{l*{6}{c}r}
		$x$ & $\kappa^2(x)$ $\downarrow$ & $H(x)$ \\
		\hline
		$11111$ & 630 & 5.4649 \\
		\hline
		$00000$ & 630 & 5.4649 \\
		\hline
		$00001$ & 518 & 5.7581 \\
		$\ldots$ & $\ldots$ & $\ldots$ \\
		\hline
		$11000$ & 486 & 5.8838 \\
		$\ldots$ & $\ldots$ & $\ldots$ \\
		$00010$ & 458 & 6.0132 \\
		$\ldots$ & $\ldots$ & $\ldots$ \\
		$10011$ & 398 & 6.1076 \\
		$\ldots$ & $\ldots$ & $\ldots$ \\
		$01101$ & 366 & 6.2375 \\
		$\ldots$ & $\ldots$ & $\ldots$ \\
		$01010$ & 350 & 6.3498  \\
		\hline
	\end{tabular}
\end{table}

\section{Concluding Remarks}\label{sec:hws-conclusions}

We have provided a proof for the minimization of entropy by the uniform string in the asymptotic limit, i.e., $n\rightarrow\infty$ and fixed output length $m$, using results from hidden word statistics. However, showing the entropy maximization by the alternating string remains an open problem given that a proof establishing the minimization of the autocorrelation coefficient $\kappa^2(\texttt{1010\ldots})$ still escapes us. Beyond establishing this maximization, proving the entropy ordering of $x$ strings determined by $\kappa^2(x)$ for finite $n$ and $m$ represents another open problem.

\stopminichaptoc{\minichaptocenabled}

\part{Deniability in Quantum Cryptography}\label{part:two}

\chapter{Deniability in Cryptography}\label{chp:deniability-intro}

\printminichaptoc{\minichaptocenabled}

This chapter sets the context for the following chapters in Part \ref{part:two}. We start by giving a high level description of what deniability means, in what specific scenarios it has practical relevance, followed by a few words on its history and the evolution of its formalization in classical cryptography. We then overview related work in this area, focusing on the most relevant of contributions in classical cryptography, which are closely related to our work. Finally, we close this chapter by taking some distance from abstract definitions and discussing one of the most well-known and practical applications of deniability in Section \ref{sec:quantum-e-voting}, namely the notion of \emph{coercion-resistance} in the context of secure voting protocols.

\section{Introduction}

Deniability represents a fundamental privacy-related notion in cryptography. The ability to deny a message or an action is a desired property in many contexts such as off-the-record communication, anonymous reporting, whistle-blowing and coercion-resistant secure electronic voting.

The concept of non-repudiation is closely related to deniability in that the former is aimed at associating specific actions with legitimate parties and thereby preventing them from denying that they have performed a certain task, whereas the latter achieves the opposite property by allowing legitimate parties to deny having performed a particular action. For this reason, deniability is sometimes referred to as \emph{repudiability}. It is worth pointing out that deniability is a stronger notion than privacy as it should hold against an adversarial model wherein the attacker can demand that secret information be revealed or issue instructions. Moreover, deniability is intimately related to fundamental concepts such as secure multiparty computation \cite{goldwasser1997multi} and incoercible multiparty computation \cite{canetti1996incoercible}.

The definitions and requirements for deniable exchange can vary depending on the cryptographic task in question, e.g., encryption, authentication or key exchange. Roughly speaking, the common underlying idea for a deniable scheme can be understood as the impossibility for an adversary to produce cryptographic proofs, using only algorithmic evidence, that would allow a third-party, often referred to as a judge, to decide if a particular entity has either taken part in a given exchange or exchanged a certain message, which can be a secret key, a digital signature, or a plaintext message. In the context of key exchange, this can be also formulated in terms of a corrupt party (receiver) proving to a judge that a message can be traced back to the other party \cite{di2006deniable}.

In the public-key setting, an immediate challenge for achieving deniability is posed by the need for
remote authentication as it typically gives rise to binding evidence,  e.g., digital signatures, see \cite{di2006deniable,dodis2009composability}. The formal analysis of deniability in classical cryptography can be traced back to the original works of Canetti et al. and Dwork et al. on deniable encryption \cite{canetti1997deniable} and deniable authentication \cite{dwork2004concurrent}, respectively. These led to a series of papers on this topic covering a relatively wide array of applications. Deniable key exchange was first formalized by Di Raimondo et al. in \cite{di2006deniable} using a framework based on the simulation paradigm, which is closely related to that of zero-knowledge proofs.

Despite being a well-known and fundamental concept in classical cryptography, rather surprisingly, deniability has been largely ignored by the quantum cryptography community. To put things into perspective, with the exception of a single paper by Donald Beaver \cite{beaver2002deniability}, and a footnote in \cite{ioannou2011new} commenting on the former, there are no other works that directly tackle deniable QKE.

In the adversarial setting described in \cite{beaver2002deniability}, it is assumed that the honest parties are approached by the adversary after the termination of a QKE session and demanded to reveal their private randomness, i.e., the raw key bits encoded in their quantum states. It is then claimed that QKE schemes, despite having perfect and unconditional security, are not necessarily deniable due to an eavesdropping attack. In the case of the BB84 protocol, this attack introduces a binding between the parties' inputs and the final key, thus constraining the space of the final secret key such that key equivocation is no longer possible.

Note that since Beaver's work \cite{beaver2002deniability} appeared a few years before a formal analysis of deniability for key exchange was published \cite{di2006deniable}, its analysis is partly based on the adversarial model formulated earlier in \cite{canetti1997deniable} for deniable encryption. For this reason, the setting corresponds more closely to scenarios wherein the honest parties try to deceive a coercer by presenting fake messages and randomness, e.g., deceiving a coercer who tries to verify a voter's claimed choice using an intercepted ciphertext of a ballot in the context of secure e-voting, see \cite{ryan2016} for an example involving a compromising interaction between a coercer and voter.

\section{Related Work}\label{sec:deniability-related-work}

We focus on some of the most prominent works in the extensive body of work on deniability in classical cryptography. Probably the first study on this topic was a work by Donald Beaver on plausible deniability \cite{beaver1996plausible}.

\subsection{Deniable Encryption}

The notion of deniable encryption was considered by Canetti et al. \cite{canetti1997deniable} in a setting where an adversary demands that parties reveal private coins used for generating a ciphertext. This motivated the need for schemes equipped with a faking algorithm that can produce fake randomness with distributions indistinguishable from that of the real encryption.

The original motivation for deniable encryption was to consider a hostile setting wherein the privacy of transmitted data should still remain protected if we assume that the adversary/eavesdropper has acquired a ciphertext and that she has the additional power to require that Alice and/or Bob reveal their corresponding private choices such as randomness and plaintext data used for the generation of the obtained ciphertext.

For this to be possible, the requirements were modelled as follows. For a message $m$, a fake message $m'$ used for denying, random input $r$, and an encryption algorithm $E(m, r)$ that generates a ciphertext $c$, there should be a corresponding faking algorithm $\phi$ that generates a fake random input $r' = \phi(m, r, c)$ such that $E(m', r')$ and $E(m, r)$ are computationally indistinguishable.

More precisely, Canetti et al. \cite{canetti1997deniable} define $\pi$ with sender $S$ and receiver $R$ to be a shared-key $\delta(n)$-sender-deniable encryption protocol, with security parameter $n$ if it satisfies correctness ($\prob{k_R \neq k_S} \leq \negl{n}$), security ($\forall: m_1, m_2 \in M$, and for a shared key $k$ chosen at random, $\mathrm{com}_{\pi}(m_1, k) \overset{c}{\approx} \mathrm{com}_{\pi}(m_2, k)$), and finally that there exists an efficient ``faking'' algorithm $\phi$ having the following property with respect to any $m_1, m_2 \in M$. Let $k, r_S, r_R$ be uniformly chosen shared-key and random inputs of $S$ and $R$, respectively. Let $c = \mathrm{com}_{\pi}(m_1, k, r_S, r_R)$, and let $(\tilde{k}, \tilde{r_S}) = \phi(m_1, k, r_S, c, m_2)$, then the random variables
\[(m_2, \tilde{k}, \tilde{r_S}, c) \; \text{and} \; (m_2, k, r_S, \mathrm{COM}_{\pi}(m_2, k, r_S, r_R))
\]
are $\delta(n)$-close.

The indistinguishability was formalized in terms of the computational distance between two probability distributions: for $\delta: N \rightarrow [0,1]$, two probability distributions $\mathcal{A}$ and $\mathcal{B}$ are said to be $\delta(n)$-close if for all PPT distinguisher $\dist$ and for large enough $n$, we have that $|\prob{\dist(\mathcal{A})=1} - \prob{\dist(\mathcal{B})=1}| < \delta(n)$. If $\delta(n)$ is negligible, then $\mathcal{A}$ and $\mathcal{B}$ are said to be computationally indistinguishable.

\subsection{Deniability in the Simulation Paradigm}

In a framework based on the simulation paradigm, Dwork et al. introduced the notion of deniable authentication \cite{dwork2004concurrent}, followed by the work of Di Raimondo et al. on the formalization of deniable key exchange \cite{di2006deniable}. Both works rely on the formalism of zero-knowledge (ZK) proofs, with definitions formalized in terms of a simulator that can produce a simulated view that is indistinguishable from the real one. In a subsequent work, Di Raimondo and Gennaro gave a formal definition of forward deniability \cite{di2009new}, requiring that indistinguishability remain intact even when a (corrupted) party reveals real coins after a session. Among other things, they showed that statistical ZK protocols are forward deniable. Roughly speaking, the core idea of the model used in these works is that the transcript of the protocol does not give rise to any evidence of interaction.

In an interactive zero-knowledge proof (ZKP), introduced by Goldwasser et al. \cite{goldwasser1989knowledge}, the prover reveals no knowledge other than the validity of the assertion she proves to the verifier, where knowledge amounts to having the ability to perform a task. The intuition behind this idea is that the verifier of a ZKP should not be able to perform a task after the interaction that she would not have been able to before her interaction in the ZKP. This property is captured by showing the existence of a simulator\footnote{The simulator can simply be thought of as an algorithm.} that can simulate a malicious verifier's output without having access to the prover. This implies that the adversary could have simply generated the transcript by running the simulator to obtain the same results. To be more precise, we want the distribution of the simulated transcript to be indistinguishable from the real one.

The deniability of a protocol in the simulation paradigm, as considered in \cite{di2009new} depends on whether or not the view or the transcript of the receiver (verifier) can be simulated by a machine that does not know the secret key of the sender (prover). The natural intuition behind this definition is that since the transcript of a deniable protocol cannot be distinguished from one sampled from a random distribution, it does not provide any algorithmic evidence that associates the identity of a specific sender with a transcript. In other words, in the same vein as ZKPs described above, the transcript could have been generated as a result of the adversary interacting with a simulator, thereby rendering the transcript void of any identity-binding evidence.

Formal definitions for full and partial deniability for key exchange are proposed in \cite{di2009new}, where the former enables a party to deny having participated in a given session, while the latter is limited to preventing a recipient from proving that a particular sender has communicated with them (e.g. to prove that Alice communicated with someone but not with Bob), thus resulting in peer-independent transcripts.

Let $\Sigma$ be a key-exchange protocol consisting of the following tuple $(KG, \Sigma_I, \Sigma_R)$, where $KG$ defines a key generation algorithm, with $\Sigma_I$ and $\Sigma_R$ denoting interactive machines specifying the roles of the honest initiator and responder, respectively.

\begin{definition}\cite{di2006deniable}
   ($KG, \Sigma_I, \Sigma_R$) is said to be a concurrently deniable key exchange protocol with respect to the class $\text{AUX}$ of auxiliary inputs if for any adversary $E$ for any input of public keys $\vec{pk} = (pk_1, \ldots ,pk_{\ell})$ and any auxiliary input $aux \in \text{AUX}$, there exists a simulator $\text{SIM}_{E}$ that running on the same inputs as $E$ produces a simulated view which is indistinguishable from the real view of $E$. That is, consider the following probability distributions where $\vec{pk}=(pk_1, \ldots, pk_{\ell})$ is the set of public keys of the honest parties:
	\begin{equation}
	\begin{split}
	   Real(n, aux) = [(sk_i, pk_i) \leftarrow KG(1^n); (aux, \vec{pk}, \text{View}_{E}(\vec{pk}, aux)] \\
	   Sim(n, aux) = [(sk_i, pk_i) \leftarrow KG(1^n); (aux, \vec{pk}, \text{SIM}_{E}(\vec{pk}, aux)]
	\end{split}
\end{equation}

Then, for all probabilistic polynomial time machines $\dist$ and all $aux \in \text{AUX}$
\[
	|\mathrm{P}_{x \in Real(n, aux)} [\dist(x)=1] - \mathrm{P}_{x \in Sim(n, aux)} [\dist(x) = 1] | \le \negl{\kappa}
\]

\end{definition}

Note that $E$'s interaction with $B$ does not constitute impersonation as she is not trying to impersonate $A$. Instead, her interaction is aimed at obtaining a proof that she herself interacted with $B$ and established a key with him. Therefore, intuitively, the goal of deniability is to prevent $E$ from proving to a third party that this was the case. Thus, when $E$ interacts with $B$ we assume that she will do so using a public key $pk_{E}$, which may or may not be associated with $E$'s identity. Indeed, the idea is that since the definition guarantees that even when the attacker $E$ runs the key generation algorithm to generate a public key (thus knowing the corresponding secret key), she cannot prove that $B$ talked to her, then $E$ is certainly not able to prove that $B$ talked to any other party $A$. In particular, this implies deniability with respect to eavesdroppers. In addition, while $E$ may decide to reveal her secret key $sk_{E}$ to help in proving that $B$ talked to her, she does not have to do so. In fact, $E$ may even use a public key for which she does not know the corresponding private key.

\subsection{Deniability in the CRS and RO Model}

Pass \cite{pass2003deniability} formally defines the notion of deniable zero-knowledge and presents positive and negative results in the common reference string (CRS) and random oracle (RO) model.

While ZKP in the standard model can be quite challenging, sometimes even impossible, one can resort to the CRS or the RO model. The former refers to a setting wherein a random string is accessible to the parties, while the latter provides them with an ideal random function through oracle calls.

An observation made in \cite{pass2003deniability} is that although the standard definition of ZK in the plain model satisfies deniability, the same result does not necessarily carry over to the CRS and RO models. This is due to the fact that in a real world instantiation, the CRS string or the random oracle is fixed once at the onset, which is problematic as the simulator can in principle choose the public information in any way that allows it to make the simulated transcript look indistinguishable from the real one, a property that is hampered by a fixed CRS or random oracle. This means that a party is no longer necessarily capable of simulating a transcript using predefined public information.

Pass goes on to show a series of negative results for deniable ZKP in the CRS model, including simulatability without rewinding and non-interactive ZKP to name a few. Although deniable ZK is ruled out in a number of natural settings, a weaker form of deniability in the CRS model is shown to be achievable if the communication is restricted to only arguments of knowledge that are zero-knowledge in a class of protocols where the CRS may be reused. Despite several impossibility results in the CRS model, deniable ZK protocols in the RO model are shown to be possible.

\subsection{On-line Deniability}

In \cite{dodis2009composability}, Dodis et al. establish a link between deniability and ideal authentication and further model a situation in which deniability should hold even when a corrupted party colludes with the adversary during the execution of a protocol, a stronger form of deniability referred to as \emph{on-line deniability}.

Roughly speaking, in terms of parties and interactions, their model considers a sender $S$ who may have sent a message $m$ to a receiver $R$, a judge $J$ who decides whether or not this exchange has taken place, an informant $\mathcal{I}$ who witness the exchange and tries to convince the judge, and finally a misinformant $\mathcal{M}$ who has not witnessed any message exchanges but still tries to convince the judge that one has occurred. Thus, a protocol satisfies their definition if the judge cannot distinguish a true informant - who interacts with $S$ and $R$ during the run of a protocol - from a misinformant.

The authors show that a protocol satisfies their definition of on-line deniability if and only if it achieves message authentication in the generalized universal composability (GUC) framework of Canetti et al. \cite{canetti2007universally}, which is to say that such a protocol would naturally inherit composability guarantees. It is however also shown that the same definition is impossible to satisfy in the PKI model if adaptive corruptions are allowed, even under the assumption of secure erasure. Moreover, via reductions from deniable authentication to deniable key exchange, and vice versa, the same impossibility results are shown to imply that on-line deniable key exchange w.r.t. adaptive corruptions is also impossible in this model.

Finally, in terms of feasibility results, several positive results are presented for relaxed definitions, such as restricting the model to static adversaries or allowing adaptive corruptions but in the symmetric-key setting where all parties share a key. Since the latter may not be as appealing as a construction in the PKI setting, the authors suggest a hybrid solution that establishes the symmetric keys using a weaker variant referred to as deniable key exchange with incriminating abort (KEIA). In this approach, on-line deniability is satisfied as long as the protocol terminates successfully, i.e., it is not aborted by a malicious party.

\subsection{Perfect Forward Secrecy and Peer-and-Time Deniability}

Cremers and Feltz introduce another variant for deniable key exchange referred to as peer and time deniability \cite{cremers2011one}, while also capturing perfect forward secrecy. The latter describes a post-compromise security property that is aimed at preventing the compromise of long-term secret keys from compromising past session keys.

One of the motivations of this work is the observation that secure one-round key exchange protocols in the PKI setting achieve either perfect forward secrecy or some form of deniability, but not both. This is primarily due to the interplay between explicit authentication and deniability resulting in trade-offs wherein the former typically achieves perfect forward secrecy at the cost of giving up on deniability. The authors introduce the notion of peer-and-time deniability, which is based on a stronger variant of the extended-CK model \cite{lamacchia2007stronger}, modelled by allowing the adversary to corrupt parties even after the completion of the test-session.

Peer-and-time deniability allows parties to deny that they were alive during a certain interval of time such that while a judge can can be convinced that a party has signed some self-generated data at some point, there can be no proof of the party's role, their intended peer, or if the session has terminated successfully or the time at which the signature was issued.

\subsection{Deniable AKE for Secure Messaging}

More recently, Unger and Goldberg studied deniable authenticated key exchange (DAKE) in the context of secure messaging \cite{unger2015deniable}. Secure messaging is an area of research that has been witnessing a surge of interest in recent years, with the prime example being the Signal\footnote{Signal was preceded by TextSecure, which was eventually merged with RedPhone and renamed Signal, with the source code available at \url{https://github.com/signalapp}.} protocol, which is the most widely-used secure messaging protocol providing end-to-end encryption and known for its ratcheting construction \cite{marlinspike2013ratcheting}. See a recent work by Cohn-Gordon et al. \cite{cohn2017formal} for a thorough, formal security analysis the Signal messaging protocol.

While early attempts at devising secure messaging systems either did not meet basic requirements or relies on a trusted service provider, research into systems with end-to-end confidentiality and integrity can be traced back to one of the first security protocols of this kind, namely Off-the-Record (OTR) messaging by Borisov et al. \cite{borisov2004off}. A year after its publication, a security analysis by Di Raimondo et al. \cite{di2005secure} revealed a number of vulnerabilities, which were partly due to the use of an insecure key-exchange protocol. The authors suggested some changes for improving the security of the system as well as providing features such as deniability.

In \cite{unger2015deniable}, the authors make some indirect claims about deniability in Signal (referred to as TextSecure in the paper), and present two DAKEs that are analyzed in the generalized UC framework. The first protocol is a non-interactive DAKE, called Spawn, claiming to offer forward secrecy and deniability against both offline and online judges. It is also suggested that the construction can be used to improve the deniability properties of the Signal protocol.

\subsection{Deniability in the Quantum Setting}

To the best of our knowledge, the only work related to deniability in QKE is a single paper by Donald Beaver \cite{beaver2002deniability}, in which the author suggests a negative result arguing that existing QKE schemes are not necessarily deniable.

\section{Coercion-Resistance and Deniability in Secure E-Voting}\label{sec:quantum-e-voting}

Coercion-resistance in secure voting protocols represents a concrete case in which deniability is considered to be a desired property. Roughly speaking, this corresponds to scenarios in which a coercer or vote-buyer demands that a voter reveal their cast ballot (typically in encrypted form) so that they can verify if the voter has complied with their request, e.g. to vote for a particular candidate.

However, having to provide complete or even partial coercion-resistance in addition to other properties such as individual and universal verifiability is a notoriously difficult task. In many instances, it is possible to achieve one at the cost of either completely compromising the other, or making strong assumptions. To address some of these long-standing problems, one may consider turning to the field of QIP to either overcome known obstacles in the classical setting or to confirm no-go theorems in the quantum setting.

Despite some of the inherent advantages of quantum information, developments in quantum information processing have also been met with a number of setbacks in the form of no-go theorems with results such as the impossibility of unconditional quantum bit commitment \cite{mayers1997unconditionally} and oblivious transfer \cite{lo1997insecurity}. Regarding deniability, as far as classical results are concerned, it is not clear whether existing no-go theorems from the classical literature (see e.g. \cite{dodis2009composability}) can be transformed into feasibility results. Perhaps more importantly, the lack of research in this area makes it hard to determine what other uniquely quantum tasks, relevant for deniability, might be possible.

\subsection{Quantum E-Voting}

Voting protocols represent a suitable candidate for deniable communication schemes. Since the publication of \cite{christandl2005quantum}, a number of quantum voting schemes such as \cite{hillery2006towards,vaccaro2007quantum,bonanome2011toward} have been proposed that are either directly based on or inspired by the primitives developed in \cite{christandl2005quantum}. However, most of these schemes suffer from various vulnerabilities \cite{arapinis2018comprehensive}. In \cite{christandl2005quantum}, Christiandl and Wehner show that it is possible to obtain anonymous transfer of both classical as well as quantum information, followed by further works on anonymous quantum communication \cite{bouda2007anonymous,brassard2007anonymous}. Furthermore, in \cite{christandl2005quantum}, it is proved that traceless exchange of messages is only possible using quantum information. To the best of our knowledge, quantum anonymous transfer and more specifically, the property of traceless exchange, along with constructions such as counterfactual QKD \cite{noh2009counterfactual} and uncloneable encryption \cite{gottesman2002uncloneable}, all of which with potential applications for deniability, have not been considered in the context of deniability before.

In a recent work, Arapinis et al. \cite{arapinis2018comprehensive} provide a comprehensive analysis of quantum e-voting protocols. The authors systematically classify all existing solutions in terms of their main primitives and analyze their claimed security properties. Their analysis also provides a systematization of knowledge of the body of work on quantum e-voting schemes. Moreover, it presents non-trivial attacks on almost all the major constructions, which clearly demonstrates the fragile nature of security claims with ad-hoc definitions and models and a lack of security proofs. The classification of the existing quantum e-voting schemes is done primarily in terms of four classes that involve an election authority, a tallier and a set of voters. The first three classes rely on sharing entangled states with the voters, whereas the last category makes use of conjugate coding (BB84 states).

The authors show how various security properties ranging from the correctness of the protocol to more subtle properties such as verifiability and ballot privacy can be undermined by exploiting various flaws that were either entirely unknown or if they were known, no concrete attacks had been presented yet. Finally, this work highlights the importance of providing security proofs within well-established formal models and points out the vital importance of precise definitions and formal adversarial models, areas of work that still present a considerable gap between the quantum cryptography literature and its classical counterpart.

It is worth pointing out that, despite its relative degree of maturity, the current state-of-the-art in classical cryptography when it comes to secure e-voting is such that coming up with precise and non-conflicting definitions for subtle requirements such as ballot privacy and universal/individual verifiability is still ongoing research and presents a wide range of open and challenging problems. While specific properties in classical e-voting protocols have been proved in the computational model (e.g., see \cite{khader2013proving} by Khader et al.), the classical literature contains a plethora of different constructions that lack composable proofs of security and that have not been rigorously analyzed. This difficulty, among other things, could be mitigated by the use of formal methods, more specifically by using tools as such automated provers and model checkers.

\subsection{Verifying Quantum Security Protocols}

However, when it comes to verifying quantum protocols, it remains unclear how a similar approach can be adopted in the quantum setting given that without having access to quantum computers, we are limited to a specific subset of protocols that fall within the stabilizer formalism (based on the Gottesman-Knill theorem), which for example makes it possible to perform model checking for a certain class of quantum protocols in polynomial time using a classical computer. Compared to its classical counterpart, the use of formal verification in the context of quantum information assurance has been somewhat limited in scope.

A unique aspect of quantum protocols lies in their dependence on the subtle and notoriously elusive properties of quantum information, along with the rather common juxtaposition of classical and quantum components. Given the proofs of unconditional security of the BB84 protocol, verification techniques such as model checking might seem redundant at first, yet there are other reasons in favor of using formal verification, apart from the obvious benefit of discovering attacks, flaws and fixes in constructions without any proofs of security. These motivations include the existence of a multitude of protocols with varying degrees of complexity claiming to satisfy notoriously subtle and difficult to achieve properties such as verifiability and coercion-resistance in the context of quantum e-voting schemes, along with the fact that despite there being a proof of security for a protocol, an implementation of a system based on the said protocol might still be subject to a wide variety of loopholes. Moreover, quantum protocols typically come with classical components dedicated to tasks such as error-correction and authentication, which call for a unified framework and language for reasoning about such hybrid constructions.

\subsubsection{State-Explosion meets Quantum Complexity}

Apart from the unavoidable and obvious limitation of undecidability, one of the central obstacles encountered in the design and implementation of model checkers is that of the state-explosion problem, which is due to the exponential growth of the global state space in the number of concurrent components, rooted in the combinatorial complexity inherent to exploring the state space. This problem is exacerbated in the quantum regime due to an additional inherent feature of quantum systems, namely that of the exponential growth of resources needed for simply representing quantum states using classical bits of information. In general\footnote{This refers to quantum states for which the state vector cannot be factored into a tensor product of states.}, a classical description of a quantum state consisting of $n$ qubits involves $2^n$ complex coefficients, thus requiring exponential computational resources. This feature, which in this particular case constitutes a hurdle, lies at the core of almost all approaches proposed so far. In other words, short of having a sufficiently stable and scalable quantum computer, any classical model checking approach would have to somehow account for this aspect, as otherwise the verification of quantum systems using classical computers would simply be impossible.

At its core, the stabilizer formalism is based on the Gottesman-Knill theorem by Daniel Gottesman and Emanuel Knill \cite{nielsen2002quantum,gottesman1998heisenberg,aaronson2004improved,gay2011stabilizer}, which defines a restricted model in which a subclass of quantum circuits, called stabilizer circuits, can be simulated efficiently on a classical computer\footnote{See the ``CHP: CNOT-Hadamard-Phase'' program available at \url{https://www.scottaaronson.com/chp/} for a C implementation of the stabilizer circuits based on algorithms given in \cite{aaronson2004improved}.}. This encompasses quantum computations consisting of $(i)$ state preparation in the computational basis, Clifford group operations, namely Hadamard gates, phase gates, controlled-not gates and Pauli gates $(ii)$ measurements in the computational basis, and classical conditional branching based on quantum measurement results.

\subsubsection{Prevalent Approaches}

Thus far, the attempted approaches range from relatively simple ones such as the application of classical probabilistic model checking tools \cite{nagarajan2005automated,gay2005probabilistic,papanikolaou2005techniques} such as PRISM \cite{kwiatkowska2011prism} to more sophisticated ones such as a series of works led by Rajagopal Nagarajan and Simon J. Gay, in collaboration with Papanikolaou, that resulted in the development of a tool called QMC (quantum model checker) \cite{gay2008qmc}. QMC is a verification tool for quantum protocols that to the best of our knowledge was the first of its kind, see also \cite{gay2008model,papanikolaou2009model,gay2010specification} for more details.

We note that the authors entertain the possibility of handling protocols that are not in the scope of the stabilizer formalism by using approximation techniques proposed by Bravyi and Kitaev \cite{bravyi2005universal} for performing universal quantum computation (UQC). In this approach, closely related to the Gottesman-Knill theorem, the operations are limited to ideal Clifford unitaries, and it is allowed to create a blank state $\ket{0}$, perform qubit measurements in the computational basis, and create a single (noisy) ancilla qubit in a mixed state $\rho$. By viewing $\rho$ as a parameter of the model, the goal is to determine for which $\rho$, UQC can be efficiently simulated via purification protocols that consume several copies of $\rho$ to produce, in the asymptotic limit, a pure state, referred to as a ``magic'' state. It is shown that the Clifford group operations combined with magic states are sufficient for UQC.

In a somewhat similar fashion, akin to the classical setting, the approaches employed in the quantum literature are either \textbf{process-oriented}, i.e., based on establishing process equivalence, or they are \textbf{property-oriented}, which roughly translates into verifying a set of property specifications.

The former deals with checking that a given system is equivalent in terms of its behavior to another system, whose description is provided as a specification, i.e., the ideal or intended behavior whose correctness is self-evident, see e.g., \cite{jorrand2004toward,davidson2011formal,davidson2012model,davidson2012analysis,ardeshir2013equivalence,ardeshir2014verification,larijani2018automated}. As for the latter, given a formal specification of security requirements and a description of a model, the goal is to check if the model satisfies the expected specifications expressed in a set of logical formulas, along every possible execution path, similar to some of the standard classical approaches, see e.g., \cite{gay2008qmc,gay2008model,gay2010specification}.

Another related and very recent work by Unruh \cite{unruh2018quantum} on quantum relational Hoare logic is part of a line of research dedicated to developing ``QuEasyCrypt'' for the formal verification of quantum cryptographic protocols via crypto games, i.e., the quantum equivalent of EasyCrypt \cite{barthe2011computer,barthe2014easycrypt}.

\subsection{Quantum E-Voting vs. Classical Quantum-Secure Voting}

Moreover, unless solutions based on quantum primitives offer inherently superior solutions, e.g., in terms of unconditional security or more efficient use of computational resources, they may end up being mere feasibility results without any real advantages w.r.t. existing classical solutions. Another important issue in this regard is that in constructions beyond quantum key exchange, solutions such as the ones presented in the context of quantum e-voting tend to make generous use of quantum and classical resources to the point that given such resources, one can sometimes solve the same problem classically using known solutions, e.g., large amounts of preshared classical randomness and/or preshared entanglement, anonymous channels, or the need for running a protocol for an exponential number of rounds in the number of voters.

Such constructions, based on rather ad-hoc models, are prime examples that highlight the need for formal verifications of complex protocols. We will return to coercion-resistance in voting protocols in Chapter \ref{chp:cr-and-qsafe-voting} where we propose a coercion-resistant and efficient classical voting scheme, based on fully homomorphic encryption, which we also conjecture to be quantum-secure.

\stopminichaptoc{\minichaptocenabled}

\chapter{Preliminaries in Quantum Information Processing and Cryptography}\label{chp:qip}

\printminichaptoc{\minichaptocenabled}

\section{Basics of Quantum Information}\label{sec:qi-qc}

We use standard terminology from quantum computing and cryptography. For the purpose of our work, we limit ourselves to a description of the most relevant concepts in quantum information theory. More details can be found in standard textbooks \cite{nielsen2002quantum,wilde2013quantum,rieffel2011quantum,schumacher2010quantum}.

\paragraph{Notation} For brevity, let $A$ and $B$ denote the honest parties, and $E$ the adversary. Let $s \sample \{0,1\}^n$ denote a binary string of length $n$ sampled uniformly at random. Employing the Dirac bra-ket notation, we use $\ket{\psi}$ and $\bra{\psi}$ to denote a state vector of a quantum system labelled by $\psi$, and its conjugate transpose, respectively. Moreover, let $\bra{\psi}\ket{\psi}$ denote the inner (scalar) product of $\ket{\psi}$ and $\ket{\psi}^{\dagger}$, and $\ket{\psi}\bra{\psi}$ their outer product. Let $P=\ket{\psi}\bra{\psi}$ denote the projection operator onto $\ket{\psi}$ such that for any $\ket{\phi}$, we have $P\ket{\phi} = \ket{\psi}\bra{\psi}\ket{\phi}$.

\subsection{Noiseless Qubits}

Given an orthonormal basis formed by $\ket{0}$ and $\ket{1}$ in a two-dimensional complex Hilbert space $\mathcal{H}_2$, let $(+) \equiv \{ \ket{0}, \ket{1} \}$ denote the computational basis and $(\times)$ the diagonal basis, as defined below
\[
(\times) \equiv \{ \frac{\ket{0} + \ket{1}}{\sqrt{2}}, \frac{\ket{0} - \ket{1}}{\sqrt{2}} \}.
\]

A \emph{qubit}, short for quantum bit, represents the fundamental unit of quantum information, i.e., a two-state quantum system. We say that a general noiseless qubit $\ket{\psi}$ is in a pure state expressed as a linear combination of other pure states given by
\[
\ket{\psi} = \alpha \ket{0} + \beta \ket{1},
\]
where $\alpha, \beta \in \mathbb{C}$ are \emph{probability amplitudes} with unit norm such that $|\alpha|^2 + |\beta|^2 = 1$.

Upon measuring $\ket{\psi}$ in the computational basis, we obtain $\ket{0}$ with probability $|\alpha|^2$ and $\ket{1}$ with probability $|\beta|^2$. For a system of two qubits, the state can be expressed as $\ket{\psi} = \sum_{i,j} \alpha_{ij}\ket{ij}$ where $\sum_{ij}|\alpha_{ij}|^2 = 1$, with a similar expression for a larger number of qubits, and the probability that upon measurement, the first qubit is in state $i$ and the second in $j$ is $|\alpha_{ij}|^2$. Now if we were to measure just one qubit, then the probability that for example the first qubit is 0 for state
\[
\ket{\psi} = \alpha_{00}\ket{00} + \alpha_{01}\ket{01} + \alpha_{10}\ket{10} + \alpha_{11}\ket{11},
\]
is given by
\[
\prob{\ket{\psi} = \ket{00} \vee \ket{\psi} = \ket{01}} = |\alpha_{00}|^2+|\alpha_{01}|^2,
\]
with the post-measurement state, normalized to be a unit vector, being as follows
\[
\ket{\psi} = \frac{\alpha_{00}\ket{00} + \alpha_{01}\ket{01}}{\sqrt{|\alpha_{00}|^2+|\alpha_{01}|^2}}.
\]

One may choose any orthogonal basis $\{\ket{v}, \ket{w} \}$ to measure the qubit $\ket{\psi}$ in; it suffices to rewrite the state in that basis (change of basis), $\ket{\psi}= \alpha'\ket{v}+\beta'\ket{w}$. Upon measurement, one would obtain $\ket{v}$ and $\ket{w}$ with probability $|\alpha'|^2$ and $|\beta'|^2$, respectively.

We use the tensor product to describe the quantum state of two or more qubits, e.g., $\ket{\phi} \otimes \ket{\psi}$, often denoted using the shorthand $\ket{\phi \psi}$. Finally, when dealing with multipartite quantum states shared among multiple participants, we use subscripts to refer to the qubit in possession of a given party, e.g., for the three participants Alice, Bob and Eve, denoted by $A, B$ and $E$, respectively, $\ket{0_A 0_B 0_E}$ (short for $\ket{0_A} \otimes \ket{0_B} \otimes \ket{0_E}$) tells us that Alice, Bob and Eve have access to their own separate qubits in state $\ket{0}$.

\subsection{Quantum Entanglement}\label{sec:entanglement}

If the state vector of a composite system cannot be expressed as a tensor product $\ket{\psi_1} \otimes \ket{\psi_2}$, the state of each subsystem cannot be described independently, but rather only as a whole, and we say the two qubits are \emph{entangled}. This property is best exemplified by maximally entangled qubits (\emph{ebits}), the so-called \emph{Bell states}, also referred to as EPR pairs due to the famous EPR paradox \cite{einstein1935can}:
\begin{align*}
\ket{\Phi^\pm}_{AB} = \frac{1}{\sqrt{2}}(\ket{00}_{AB} \pm \ket{11}_{AB}) \quad , \quad \ket{\Psi^\pm}_{AB} = \frac{1}{\sqrt{2}}(\ket{01}_{AB} \pm \ket{10}_{AB})
\end{align*}

Imagine a source that generates two maximally entangled particles $\ket{\Phi^+}_{AB}$ and sends one qubit to Alice and another to Bob, who can be arbitrarily spatially far apart from each other. Now suppose Alice measures her particle and observes the state $\ket{0}$. This means that the combined state will now be $\ket{00}$, and if at any point after Alice's measurement, Bob measures his particle, he will also observe $\ket{0}$.

As already mentioned, entangled states cannot be factored into tensor product states, whereas a state
\[
\frac{1}{\sqrt{2}}\ket{00}+\frac{1}{\sqrt{2}}\ket{01}
\]
can be rewritten as
\[
\ket{0} \otimes (\frac{1}{\sqrt{2}}\ket{0}+\frac{1}{\sqrt{2}}\ket{1}).
\]

If the state of a composite quantum system can be expressed as a tensor product state, then only $\textrm{dim}(\ket{\phi} + \textrm{dim}(\ket{\psi})$ complex numbers are enough to describe the composite system. However, generally speaking, most composite systems are not tensor states, but rather entangled states. Therefore, the Hilbert space of a system with $n$ qubits lives in $\mathbb{C}^{2n}$. This also explains the exponential blowup in terms of memory when simulating qubits on a classical computer. However, when we measure an $n$-qubit system, we extract only $n$ bits of classical information (rather than $2^n$ bits), a property that is captured by Holevo's theorem \cite{kholevo1973bounds}, which puts an upper bound on the amount of classical information that can be retrieved from a quantum system.

The CHSH game \cite{clauser1969proposed} is one of the most well-known examples of an application of entanglement that uses a Bell-type inequality in a two-player setting. The main takeaway is that it establishes an upper bound on the probability of winning the game using classical strategies, and then demonstrates how quantum strategies involving a maximally entangled Bell state shared between the players can surpass this limit. Thus, a variation of a violation of attainable correlations in the same vein as Bell's theorem, that clearly separates classical correlations from quantum ones.

The game consists of two players, Alice and Bob, who are spatially separated and cannot communicate during the execution of the game. Alice and Bob receive, from a referee, two random bits $x_A$ and $x_B$, drawn uniformly at random. Each party then generates a bit, say $a$ for Alice and $b$ for Bob, and sends it back to the referee. The referee decides if the players win according to the winning condition given below
\[
x_A \wedge x_B = a \oplus b
\]
Note that due to the spatial separation, Alice's response $a$ cannot depend on Bob's input bit $x_B$, with a similar constraint for Bob. The main result of the CHSH game states that the maximal winning probability with a classical deterministic algorithm is at most $\frac{3}{4}$, which is obtained if both Alice and Bob always return $a=0$ and $b=0$, regardless of the values of $x_A$ and $x_B$, whereas if they follow a quantum strategy using their shared Bell state, they can achieve a winning probability of $\mathrm{cos}^2(\pi/8) \approx 0.85$.

\subsection{Noisy Qubits and Mixed States}

A noisy qubit that cannot be expressed as a linear superposition of pure states is said to be in a \emph{mixed} state, a classical probability distribution of pure states:
\[
\{p_X(x), \ket{\psi_x}\}_{x \in X}.
\]
The \emph{density operator} $\rho$, defined as a weighted sum of projectors, captures both pure and mixed states:
\[
\rho \equiv \sum_{x \in \mathcal{X}}p_X(x) \ket{\psi_x}\bra{\psi_x}.
\]

An important state in the noisy quantum theory is the maximally mixed state $\pi$ corresponding to a uniform ensemble of orthogonal states in a $d$-dimensional Hilbert space
\[
	\pi \equiv \frac{1}{d} \sum_{x \in \mathcal{X}} \ket{\psi_x}\bra{\psi_x}
\]

Given a density matrix $\rho_{AB}$ describing the joint state of a system held by $A$ and $B$, the \emph{partial trace} allows us to compute the local state of $A$ (density operator $\rho_A$) if $B$'s system is not accessible to $A$.
To obtain $\rho_A$ from $\rho_{AB}$ (the reduced state of $\rho_{AB}$ on $A$), we trace out the system $B$: $\rho_A = \mathrm{tr}_{B}(\rho_{AB})$.

As a distance measure, we use the expected fidelity $F(\ket{\psi}, \rho)$ between a pure state $\ket{\psi}$ and a mixed state $\rho$ given by $F(\ket{\psi}, \rho) = \bra{\psi}\rho\ket{\psi}$.

Finally, another important notion is that of \emph{purification}, which allows us to view noise in a different way by modelling noisy quantum systems not in terms of our lack of information about them but rather as entanglement with another system that is simply inaccessible to us, referred to as the purifying system. For a density operator $\rho_A$, its purification is given by a pure bipartite state $\ket{\psi}_{RA} \in \mc{H}_R \otimes \mc{H}_A$, where $R$ is a reference system such that the reduced state on system $A$ is
\[
\rho_A = \mathrm{tr}_R(\ket{\psi}\bra{\psi}_{RA}),
\]
with the global state $\ket{\psi}_{RA}$ being a pure state, whereas the reduced state $\rho_A$ is not in general a pure state as it is obtained by tracing over the reference system, unless of course the global state is a pure tensor product state.

Moreover, all purifications of a given density operator $\rho_A$ are equivalent, a result that follows as a consequence of the Schmidt decomposition. The Schmidt decomposition tells us that given a bipartite pure state $\ket{\psi}_{AB} \in \mc{H}_A \otimes \mc{H}_B$ where the Hilbert spaces $\mc{H}_A$ and $\mc{H}_B$ are not necessarily of the same dimension, the state can be expressed as $\ket{\psi}_{AB} \equiv \sum_{i=0}^{d-1}\lambda_i \ket{i}_A \ket{i}_B$, where the amplitudes $\lambda_i$ (referred to as Schmidt coefficients) are real, strictly positive and $\sum_i \lambda_i^2 = 1$, and the states $\{\ket{i}_A\}$ and $\{\ket{i}_B\}$ form orthonormal bases for the systems $A$ and $B$, respectively. Thanks to the Schmidt decomposition, as long as the joint state of two systems $A$ and $B$ is a pure state, if for example $A$ is two-dimensional, regardless of how large the Hilbert space of $B$ is, one can always express the whole system using $\mc{H}_A$ and a two-dimensional subspace of $\mc{H}_B$.

\subsection{Quantum Entropy}

There are several important measures for quantifying information and correlations in quantum systems. A fundamental measure, which we will come back to in Chapter \ref{chp:entanglement-distillation}, is the von Neumann entropy, or \emph{quantum entropy}, due to John von Neumann, the discovery of which goes back to von Neumann's works in statistical physics, long before Shannon developed his information-theoretic formulation. Quantum entropy is the generalization of the Shannon entropy, which captures both classical and quantum uncertainty about a quantum state, or the extent to which a state is mixed.

For some quantum system in a state $\rho \equiv \sum_{i}p_i \ket{\psi_i}\bra{\psi_i}$, its quantum entropy $H(\rho)$, often denoted as $S(\rho)$, is defined as follows:
\begin{equation}
	S(\rho) = H(\rho) \equiv - \mathrm{tr}\{ \rho \mathrm{log} \rho  \}
\end{equation}
Suppose the density matrix $\rho$ has a matrix diagonal representation with eigenvalues $\lambda_1, \ldots, \lambda_n$
\[
	\rho \equiv \sum_j \lambda_j \ket{e_j}\bra{e_j},
\]
then, the von Neumann entropy is given by
\[
	S(\rho) = H(\rho) \equiv - \sum_j \lambda_j \mathrm{log} \lambda_j,
\]
which is analogous to the classical Shannon entropy.

Similar to the classical case, an operational meaning for quantum entropy is given by Schumacher's \cite{schumacher1995quantum} noiseless quantum coding theorem, stating that the achievable code rate is asymptotically equal to $S(\rho)$, i.e., the fundamental limit of compression for an i.i.d. quantum information source in the asymptotic limit.

Moreover, in terms of mathematical properties, similar to the ones for classical entropy discussed in Subsection \ref{subsec:classical-information-entropy}, in addition to non-negativity and concavity, its minimum value (zero) is reached when the density operator is a pure state, while its maximum is obtained by the maximally mixed state, i.e., uniformly distributed ensemble of pure states. Finally, conditional quantum entropy is defined as
\[
S(\rho | \sigma) \equiv S(\rho, \sigma) - S(\sigma)
\]
with an inequality similar to the classical case, namely $S(\rho) \ge S(\rho | \sigma)$.

While the definition of quantum entropy is analogous to its classical counterpart, it does possess certain inherently quantum characteristics such as the possibility of conditional quantum entropy being negative \cite{cerf1997negative,horodecki2005partial}, a property known as \emph{coherent information} in quantum information theory. This seemingly counterintuitive feature may be better understood when considering its operational meaning in terms of the communication cost needed for conveying partial information, similar to the cost of transmitting classical bits as explained in Subsection \ref{subsec:classical-information-entropy}: given an unknown quantum state distributed over two systems and some prior information on the receiver's end, the communication cost required for the sender to transfer the full state measures the partial information, conditioned on the prior, that is needed for the transfer.

This operational interpretation of the conditional quantum entropy is given by a task known as the \emph{quantum state merging} protocol, introduced by Horodecki et al. \cite{horodecki2005partial}. Suppose Alice and Bob share a large number $n$ of the state $\rho_{AB}$, with access to a noiseless qubit channel and a classical side channel, and the objective is to transfer Alice's shares such that by the end of their communication they are in Bob's possession. For positive $S(A|B)$, Alice uses the quantum channel $n \cdot S(A|B)$ times to share the same number of qubits, whereas for negative $S(A|B)$, they can share $\approx n \cdot S(A|B)$ ebits without using the quantum channel at all. Intuitively, from a quantum informational point of view, this captures the idea that it is possible to know less about a part of a quantum system than about it as a whole, and measure the extent to which this is true.

\subsection{Quantum Evolution, Measurements and No-Cloning}

The evolution of a closed quantum system can be described by a unitary transformation $U$, a reversible norm preserving linear operation that can be intuitively thought of as a rotation of the Hilbert space. A unitary operator over the space of a $d$-dimensional quantum system is described by a $d \times d$ matrix $U$ satisfying $U^\dagger U = UU^\dagger = I$, with $U^\dagger$ denoting the inverse of $U$, which is the complex-conjugate transpose of $U$.

Some of the most common single-qubit unitary operators are the Hadamard operator ($H$) and the Pauli matrices, i.e., $X$ for the NOT gate flipping a $\ket{0}$ to $\ket{1}$ and vice-versa, the phase flip gate $Z$ acting as a NOT gate in the Hadamard basis $+/-$, the identity operator $I$, the $Y$ operator where $Y = iXZ$, along with the two-qubit CNOT gate, which is a controlled NOT gate such that when applied to $\ket{ct}$, the target bit $t$ flips if and only if the control bit $c$ is 1, e.g., $\textrm{CNOT}(\ket{10}) \rightarrow \ket{11}$. A representation of some of these operators in matrix form is given below
\[
	    H = \frac{1}{\sqrt{2}} \begin{pmatrix} 1 & 1 \\ 1 & -1 \end{pmatrix}, \; X = \begin{pmatrix} 0 & 1 \\ 1 & 0 \end{pmatrix}, \; Z = \begin{pmatrix} 1 & 0 \\ 0 & -1 \end{pmatrix}, \; \textrm{CNOT} = \begin{pmatrix} 1 & 0 & 0 & 0 \\ 0 & 1 & 0 & 0 \\ 0 & 0 & 0 & 1 \\ 0 & 0 & 1 & 0 \end{pmatrix}
\]

\subsubsection{Quantum Measurement}

In order to access the information encoded in a quantum state, we need to measure it. A measurement is a non-unitary evolution that allows us to extract classical information from a quantum state. The measurement postulate based on the Born rule provides a probabilistic interpretation of quantum measurements. A measurement is done according to a fixed orthonormal measurement basis $\{ \ket{i_0}, \ldots, \ket{i_{d-1}} \}$ in an $d$-dimensional space.

For a given state expressed in the chosen basis $\ket{\psi} = \sum_{i=0}^{i=d-1} \alpha_i \ket{i}$, the result of the measurement is $\ket{i}$ with probability equal to the squares of the probability amplitudes $|\alpha_i|^2$, which corresponds to the overlap or projection of $\ket{\psi}$ onto $\ket{i}$. Upon measurement, the state of $\ket{\psi}$ is said to \emph{collapse} onto the basis vector $\ket{i}$ meaning that post measurement the original state vector $\ket{\psi}$ becomes $\ket{i}$. Recall the norm-constraint for probability amplitudes and we have that $\sum_{i} |\alpha_i|^2 = 1$.

More generally, any quantum measurement can be described by a Hermitian operator (called observable) $O = \sum_{i=1}^{k} \lambda_i P_i = \sum_{i=1}^{k} \lambda_i \ket{\phi_i}\bra{\phi_i}$, where the $P_i$ are projectors onto a unique subspace decomposition (eigenspaces), and the $\lambda_i$ and $\ket{\phi_i}$ the eigenvalues and eigenvectors of $O$, respectively. Measurements are not reversible and they alter the state of the quantum system such that subsequent measurements in the same basis yield the same result, thereby preventing us from gaining additional information about the probability amplitudes $\alpha_i$.

\subsubsection{The No-Cloning Theorem}

Copying bits and bytes is a natural property of classical information that is readily taken for granted. In stark contrast, a crucial distinction between quantum and classical information is captured by the well-known No-Cloning theorem \cite{wootters1982single}, which states that an arbitrary unknown quantum state cannot be copied or cloned perfectly.

More precisely, this theorem rules out the possibility of a universal copying machine capable of making an identical copy of an arbitrary input state, i.e., a device described by the unitary $U$ such that $U\ket{\psi}\ket{0} = \ket{\psi}\ket{\psi}$, where $\ket{\psi}$ denotes the input state and $\ket{0}$ a blank or ancilla state that is to become a copy of the input state.

Note that this characteristic does not entirely rule out the cloning of quantum states, it only states the impossibility of a universal copying device. It is worth pointing out that this property lies at the core of the security of quantum key exchange: an adversary trying to eavesdrop on quantum states in transit is bound to introduce detectable disturbances as reading and copying quantum information are one and the same.

Finally, another closely related result is that of the No-Deletion theorem stating the impossibility of a universal quantum deleting device $U$ capable of deleting one of two identical input states using a unitary operation.

\section{Quantum Key Exchange}\label{sec:qke-and-ue}

In this section we first describe a fundamental, and yet simple coding technique called conjugate coding, that is by far the most widely used primitive in quantum cryptography, including QKE protocols such as BB84. We then provide a detailed description of the BB84 protocol, followed by a simplified version of uncloneable encryption.

\subsection{Conjugate Coding}

Conjugate coding, also known as quantum coding, is a simple scheme for encoding classical information into conjugate quantum bases, which was originally proposed by Wiesner in his construction for unforgeable quantum bank notes \cite{wiesner1983conjugate}. Despite, or perhaps thanks to its simplicity, it is arguably the most widely used quantum primitive. The core idea consists of mapping a classical bit onto a qubit chosen uniformly at random from the set $\{ \ket{0}, \ket{1}, \ket{+}, \ket{-} \}$. More precisely, to encode a bit $b \in \{0, 1\}$, the encoder picks either the \emph{rectilinear} (computational) basis $\{ \ket{0}, \ket{1} \}$ or the \emph{diagonal} basis $\{ \ket{+}, \ket{-} \}$ and initializes their qubit in one of the basis vectors in the chosen basis. Deciding which state vector gets mapped to 1 or 0 is simply a matter of convention.

At an abstract level, conjugate coding exploits Heisenberg's uncertainty principle and the no-cloning theorem to provide the following properties: $(i)$ it allows a party with knowledge of the preparation/measurement basis to access and reliably extract the classical information encoded in quantum states. However, without knowing the encoding basis, a measurement done in the wrong basis will not only result in a random outcome, but it will also irreversibly destroy any information about the encoding in the conjugate basis, i.e., the probability amplitudes $\alpha$ and $\beta$. $(ii)$ it prevents anyone with access to a single encoded state and without knowledge of the encoding basis, i.e., in possession of an unknown quantum state, to reliably create a copy or clone of this state with high fidelity. Moreover, any such attempts are bound to introduce errors and result in disturbing the original encoding due to property $(i)$ and thus provide eavesdropping detection, a feature that is uniquely provided by quantum information.

\subsection{The BB84 Protocol}\label{subsec:bb84-shor-preskill}

QKE allows two parties to establish a common secret key with information-theoretic security using an insecure quantum channel, and a public authenticated classical channel.

In Protocol \ref{protocol:bb84} we describe the \textbf{BB84} protocol, the most well-known QKE variant due to Bennett and Brassard \cite{bennett1984quantum}. For consistency with related works, we use the well-established formalism based on error-correcting codes, developed by Shor and Preskill \cite{shor2000simple}. Let $C_1[n,k_1]$ and $C_2[n,k_2]$ be two classical linear binary codes encoding $k_1$ and $k_2$ bits in $n$ bits such that $\{0\} \subset C_2 \subset C_1 \subset \mathbf{F}^n_2$ where $\mathbf{F}^n_2$ is the binary vector space on $n$ bits. A mapping of vectors $v \in C_1$ to a set of basis states (codewords) for the Calderbank-Shor-Steane (CSS) \cite{calderbank1996good,steane1996multiple} code subspace is given by: $v \mapsto (\sfrac{1}{\sqrt{|C_2|}})\sum_{w \in C_2}\ket{v+w}$. Due to the irrelevance of phase errors and their decoupling from bit flips in CSS codes, Alice can send $\ket{v}$ along with classical error-correction information $u+v$ where $u,v \in \mathbf{F}^n_2$ and $u \in C_1$, such that Bob can decode to a codeword in $C_1$ from $(v+\epsilon)-(u+v)$ where $\epsilon$ is an error codeword, with the final key being the coset leader of $u + C_2$.

\begin{algorithm}
	\floatname{algorithm}{Protocol}
	\caption{BB84 for an $n$-bit key with protection against $\delta n$ bit errors}
	\label{protocol:bb84}
	\begin{algorithmic}[1]
		\STATE Alice generates two random bit strings $a,b \in \{0,1\}^{(4+\delta)n}$, encodes $a_i$ into $\ket{\psi_i}$ in basis $(+)$ if $b_i=0$ and in $(\times)$ otherwise, and $\forall i \in [1,|a|]$ sends $\ket{\psi_i}$ to Bob.
		\STATE Bob generates a random bit string $b' \in \{0,1\}^{(4+\delta)n}$ and upon receiving the qubits, measures $\ket{\psi_i}$ in $(+)$ or $(\times)$ according to $b'_i$ to obtain $a'_i$.

		\STATE Alice announces $b$ and Bob discards $a'_i$ where $b_i \neq b'_i$, ending up with at least $2n$ bits with high probability.

		\STATE Alice picks a set $p$ of $2n$ bits at random from $a$, and a set $q$ containing $n$ elements of $p$ chosen as check bits at random. Let $v = p \setminus q$.

		\STATE Alice and Bob compare their check bits and abort if the error exceeds a predefined threshold.

		\STATE Alice announces $u+v$, where $v$ is the string of the remaining non-check bits, and $u$ is a random codeword in $C_1$.

		\STATE Bob subtracts $u+v$ from his code qubits, $v+\epsilon$, and corrects the result, $u+\epsilon$, to a codeword in $C_1$.

		\STATE Alice and Bob use the coset of $u+C_2$ as their final secret key of length $n$.
	\end{algorithmic}
\end{algorithm}

\subsection{Uncloneable Encryption}

Uncloneable encryption (UE) enables transmission of ciphertexts that cannot be perfectly copied and stored for later decoding, by encoding carefully prepared codewords into quantum states, thereby leveraging the No-Cloning theorem. We refer to Gottesman's original work \cite{gottesman2002uncloneable} for a detailed explanation of the sketch in Protocol \ref{protocol:ue}. Alice and Bob agree on a message length $n$, a Message Authentication Code (MAC) of length $s$, an error-correcting code $C_1$ having message length $K$ and codeword length $N$ with distance $2\delta N$ for average error rate $\delta$, and another error-correcting code $C_2$ (for privacy amplification) with message length $K'$ and codeword length $N$ and distance $2(\delta+\eta)N$ to correct more errors than $C_1$, satisfying $C_2^\bot \subset C_1$, where $C_2^\bot$ is the dual code containing all vectors orthogonal to $C_2$. The pre-shared key is broken down into four pieces, all chosen uniformly at random: an authentication key $k \in \{ 0,1\}^s$, a one-time pad $e \in \{0,1\}^{n+s}$, a syndrome $c_1 \in \{0,1\}^{N-K}$, and a basis sequence $b \in \{ 0,1\}^N$.

\begin{algorithm}
	\floatname{algorithm}{Protocol}
	\caption{Uncloneable Encryption for sending a message $m\in \{0,1\}^n$}
	\label{protocol:ue}
	\begin{algorithmic}[1]
		\STATE Compute $\mathrm{MAC}(m)_k = \mu \in \{0,1\}^s$. Let $x = m || \mu \in \{0,1 \}^{n+s}$.
		\STATE Mask $x$ with the one-time pad $e$ to obtain $y = x \oplus e$.
		\STATE From the coset of $C_1$ given by the syndrome $c_1$, pick a random codeword $z \in \{0,1 \}^N$ that has syndrome bits $y$ w.r.t. $C_2^{\bot}$, where $C_2^\bot \subset C_1$.
		\STATE For $i \in [1, N]$ encode ciphertext bit $z_i$ in the basis $(+)$ if $b_i = 0$ and in the basis $(\times)$ if $b_i = 1$. The resulting state $\ket{\psi_i}$ is sent to Bob.
	\end{algorithmic}
	To perform decryption:
	\begin{algorithmic}[1]
		\STATE For $i \in [1, N]$, measure $\ket{\psi'_i}$ according to $b_i$, to obtain $z'_i \in \{0,1\}^N$.
		\STATE Perform error-correction on $z'$ using code $C_1$ and evaluate the parity checks of $C_2/C_1^{\bot}$ for privacy amplification to get an $(n+s)$-bit string $y'$.
		\STATE Invert the OTP step to obtain $x' = y' \oplus e$.
		\STATE Parse $x'$ as the concatenation $m' || \mu'$ and use $k$ to verify if $\mathrm{MAC}(m')_k = \mu'$.
	\end{algorithmic}
\end{algorithm}

\paragraph{QKE from UE.} It is known \cite{gottesman2002uncloneable} that any quantum authentication (QA) scheme can be used as a secure UE scheme, which can in turn be used to obtain QKE, with less interaction and more efficient error detection. We give a brief description of how QKE can be obtained from UE in Protocol \ref{protocol:qke-from-ue}.

\begin{algorithm}
	\floatname{algorithm}{Protocol}
	\caption{Obtaining QKE from Uncloneable Encryption}
	\label{protocol:qke-from-ue}
	\begin{algorithmic}[1]
		\STATE Alice generates random strings $k$ and $x$, and sends $x$ to Bob via UE, keyed with $k$.
		\STATE Bob announces that he has received the message, and then Alice announces $k$.
		\STATE Bob decodes the classical message $x$, and upon MAC verification, if the message is valid, he announces this to Alice and they will use $x$ as their secret key.
	\end{algorithmic}
\end{algorithm}

\section{Authenticated Key Exchange Protocols}

Here we provide an overview of some of the fundamental notions used in key exchange protocols and briefly discuss the security requirements of an authenticated key exchange protocol, but for a more thorough analysis, we refer the reader to \cite{bellare1993entity,bellare1998modular,canetti2001analysis,lamacchia2007stronger}.

Suppose two parties, Alice and Bob, wish to exchange secret messages over an insecure channel. Now if they were already in possession of a mutually shared secret key, they could simply use standard protocols such as TLS\footnote{The Transport Layer Security (TLS) protocol is the standard and widely used solution for ensuring secure communication over the internet.} that can provide them with privacy, authentication and message integrity. However, they would clearly need a shared secret key to begin with, and this is where key exchange (KE) protocols come into play. See a recent work \cite{cremers2017comprehensive} by Cremers et al. on a comprehensive security analysis of TLS 1.3.

In order for Alice and Bob to have a shared secret key, short of using an out-of-band (OOB) key establishment method, e.g., trusted couriers, which is not particularly practical, their best option is to use an \emph{authenticated key exchange} (AKE) protocol to establish a mutually authenticated shared secret key. Informally, a KE protocol allows two legitimate parties Alice and Bob, in the presence of an adversary Eve who has complete control over their communication channel, to output keys $k_A, k_B \in \{0,1\}^n$ with the correctness requirement that $k_A = k_B$. More precisely, a secure authenticated KE allows Alice and Bob to establish a shared session key in such a way that at the end of a session, the two security guarantees of authenticity and secrecy are also satisfied. This means that Alice and Bob can be sure that they share a fresh, random session key $k = k_A = k_B$ with each other, such that Eve cannot distinguish the key $k$ from a random/uniform key of length $n$.

\subsection{Key Exchange Protocols}

Key exchange protocols are message-driven protocols involving a set of parties, which are modelled as probabilistic polynomial-time (PPT) Turing machines. These parties engage in an exchange of messages with each other by running an instance of an AKE protocol (via point-to-point links) such that upon completion of a run, the output is a secret key called a \emph{session key}. A session simply refers to an instantiation of a protocol by a party Sessions are typically identified by a tuple of terms including the identity of the initiator or the party running the session and that of the responder, along with a session $id$ and a $role \in \{ \textrm{initiator}, \textrm{responder} \}$.

\subsection{Key Exchange Security}

The security of a protocol cannot be analyzed without considering a specific adversarial model, which defines the capabilities and actions of the adversary in terms of computational power and type of access to resources. The natural approach in cryptography is to adopt the most pessimistic view and to consider the most generic types of attack while remaining realistic in terms of assumptions and requirements. The latter can easily make the difference between a protocol that satisfies a certain security definition and one that clearly violates it, as we will point out for the case of deniability in this work.

The attacker is modelled as a PPT that is in complete control of the communication channel, meaning that she can eavesdrop on all exchanged messages, intercept messages to either change them or to inject new messages or to prevent them from reaching their target. In other words, the attacker can simply replace the channel with her own and act as a middle-man to deliver the messages. Furthermore, some form of information leakage to the adversary is also allowed. Here the idea is to ensure that the exposure of certain pieces of secret information has the least possible impact on the security of other secrets. For example, it stands to reason to guarantee that the leakage of some ephemeral state information of a given session will not compromise the security of other sessions.

\subsection{Session-Key Security}

In order to formalize the security of a key-exchange protocol, Canetti and Krawczyk introduced the notion of \emph{session-key security} \cite{canetti2001analysis} by building on top of earlier definitional work by Bellare et al. in \cite{bellare1993entity,bellare1998modular}. Intuitively, the idea is that the security of a given session key is guaranteed in that even by allowing the adversary to interact with the KE protocol via a set of well-defined queries or to have access to other sessions, she still cannot learn anything about the value of the session key such that she will not be able to distinguish the latter from a random key.

Another concept worth mentioning is the notion of \emph{matching sessions} \cite{lamacchia2007stronger}, which captures the idea of two corresponding sessions, i.e., two sessions between two communication partners. Two complete sessions $s$ and $s'$ are considered to be matching if all the terms belonging to their session identifier tuple are equal, except for the $role$ term.

More precisely, the security of a key exchange protocol is specified using a game-based definition in the form of a security experiment. The adversarial game consists of an interactive two player game between a probabilistic polynomial time adversary $E$, and a challenger $\mathcal{C}$ that responds to a specific set of queries made by Eve. These queries can include \textbf{send-message($s$,$m$)}, \textbf{reveal-session-state($s$)}, \textbf{reveal-ephemeral-key($s$)}, \textbf{reveal-session-key($s$)}, \textbf{corrupt($P$)} and \textbf{test-session($s$)}. As indicated by their names, invoking such queries (except for sending messages or initiating a session) results in revealing the stated secret value for the identified session, which could also be achieved via corrupting a party.

Sending any of the above ``reveal''-based queries, except for the \textbf{test-session($s$)} query, would result in a session that is considered to be exposed. Moreover, the expiration of a session implies the erasure of the session key from the owning party's memory, not necessarily the matching session's party. This is meant to model the fact that session keys can have a limited life time, after which none of the above mentioned queries can reveal the value of the session key. More details on this will be given in Section \ref{subsec:security-model}.

Note that without additional constraints, having access to these queries would allow the adversary to compromise the security of any protocol. The motivation for this is to impose a number of minimal restrictions on the set of queries the adversary can invoke for a given session in order to limit the adversary as little as possible and still leave room for protocols that can remain secure against such an adversary. In other words, we want to add just enough restrictions such that there exist protocols that can still satisfy the definition of session-key security. The notion of a \emph{fresh} session captures the idea of excluding cases that would allow the adversary to trivially win the security experiment. It is these restrictions that are formulated in terms of a so-called ``test-session'' query probing a completed session that has not been expired or exposed.

\section{A Quick Primer on Fully Homomorphic Encryption}\label{sec:fhe-primer}

A comprehensive introduction to fully homomorphic encryption would certainly go beyond the scope of this thesis. However, since in this chapter we are mainly interested in applying the machinery of FHE, we provide a concise description of its main ingredients and key ideas and refer the reader to comprehensive texts such as \cite{gentry2010computing,armknecht2015guide} for further information.

\subsection{Core Idea and Intuition}

In addition to the wealth of sources available on FHE, there exists a wide variety of intuitive explanations for what FHE is and how it works such as Gentry's analogy consisting of a distrustful owner of a jewelry store \cite{gentry2010computing} who provides their workers with an impenetrable box containing precious raw material, equipped with glove compartments that allow the worker to manipulate the content without actually ever being able to touch them. Barak Boaz\footnote{\url{https://intensecrypto.org/public/lec_15_FHE.html}} uses a somewhat similar analogy based on bags containing toxic material, that can be manipulated within a limited span of time before leaking their content and harming the user. The idea of moving the content of one bag into another using an extra bag mimics the idea of bootstrapping and lowering the noise in FHE circuits. Barak also compares the bootstrapping process to the concept of ``escape velocity'' from physics.

As mentioned in the introduction, additively or multiplicatively homomorphic cryptosystems have been known for quite a long time and are indeed used frequently in the design of secure e-voting schemes, e.g., ElGamal. For instance, consider the multiplicative property of the well-known RSA cryptosystem where given $c_1 = m_1^e \; \mathrm{mod} \; N$ and $c_2 = m_2^e \; \mathrm{mod} \; N$, we can compute the product of the plaintext messages without having access to them using only the public key, by simply working over $c_1$ and $c_2$ as follows: $c_1 \cdot c_2 = (m_1 \cdot m_2)^e \; \mathrm{mod} \; N$. Similarly, the Paillier cryptosystem is also homomorphic, but only additively.

Although being able to process data that we cannot read, may still seem rather counterintuitive, we do not take on the challenge of providing yet another intuitive explanation and instead state the main promise concisely as follows: FHE allows us to perform arbitrary computations on encrypted data.

More concretely, let $\mc{E}_{pk}(m)$ denote an FHE-encryption of a message $m \in \{0,1\}^n$ under the public key $pk$. At its core, for $b_0, b_1 \in \{0,1\}$, given $\mc{E}_{pk}(b_0)$ and $\mc{E}_{pk}(b_1)$, FHE allows us to compute $\mc{E}_{pk}(b_0 \oplus b_1)$ and $\mc{E}_{pk}(b_0 \cdot b_1)$ by working over ciphertexts alone, without having access to the secret key, thus enabling the homomorphic evaluation of any boolean circuit, i.e., computing $\mc{E}_{pk}(f(m))$ from $\mc{E}_{pk}(m)$ for any computable function $f$.

In terms of functionality, FHE is closely related to functional encryption (FE) and cryptographic obfuscation schemes, see \cite{armknecht2015guide} for a discussion surrounding this connection. Cloud computing is one of the most typical examples used for  illustrating the utility of FHE. The basic idea is that a client could hand over encrypted data to a server and the latter can perform computations on the encrypted data and return the result (still in encrypted form) back to the client without having learned anything about the client's data.

\subsection{Definition and Properties}

A public-key FHE scheme consists of a tuple of PPT algorithms
\[
\text{FHE} = (\text{FHE.Setup, FHE.Enc, FHE.Dec, FHE.Eval}).
\]
Sometimes a parameter generation $\text{FHE.KeyGen}$ algorithm is also included, which deals with generating the input/output space, key space and randomness. Let $\lambda$ and $d$ be a security parameter and circuit depth bound, respectively, and these algorithms are defined as follows.
\begin{itemize}
	\item $\text{FHE.Setup}(1^\lambda, 1^d) \rightarrow (pk, sk)$: Given the security parameter $\lambda$ and a circuit depth bound $d$ as inputs, outputs a public/private key pair $(pk,sk)$.
	\item $\text{FHE.Enc}_{pk}(m) \rightarrow c$: Encrypts a message $m \in \{0,1\}$ under $pk$ and outputs ciphertext $c$.
	\item $\text{FHE.Dec}_{sk}(c) \rightarrow m$: Decrypts a ciphertext $c \in \mc{C}$ under $sk$ and outputs either plaintext $m \in \{0,1\}$ or $\bot$.
	\item $\text{FHE.Eval}_{pk}(f,c_1,\ldots,c_n) \rightarrow c_{\text{new}}$: Evaluates a circuit (function) $f: \{0,1\}^n \rightarrow \{0,1\}$ of depth at most $d$, over ciphertexts $c_1, \ldots, c_n$ and outputs $c_{\text{new}}$.
\end{itemize}

An FHE scheme should also satisfy the properties of compactness, decryption correctness and security. Compactness captures the idea that the ciphertext size should be bounded by some fixed polynomial in the security parameter, independent of the number of inputs and the size of the evaluation circuit. Evaluation correctness requires that
\[
\prob{\text{FHE.Dec}_{sk}(\text{FHE.Eval}_{pk}(f,c_1,\ldots,c_n)) = f(m_1, \ldots, m_n)} = 1 - \negl{\lambda}
\]
and finally, security simply refers to the concept of semantic security, meaning that a secure public-key FHE scheme should satisfy ciphertext indistinguishability. This corresponds to a setting in which a challenger $\mc{C}$ flips a coin to choose a bit uniformly at random, i.e., $b \sample \{0,1\}$, such that the adversary $E$ cannot distinguish between the encryption of $m_1$ and $m_2$ given $\mc{E}_{pk}(m_b)$.

\subsection{Managing the Growth of Noise}

The security of most FHE schemes reduces to computational hardness assumptions in lattice-based cryptography such as the shortest vector problem (SVP) or the closest vector problem (CVP). Gentry's breakthrough FHE construction \cite{gentry2009fully}, based on ideal lattices, introduced the notion of \emph{bootstrapping} whose role is to limit the growth of noise\footnote{The noise serves to hide the message.} and keep it below a certain threshold, lest it corrupt the computation. In short, this is achieved by starting from a somewhat homomorphic scheme\footnote{Such a scheme would consist of applying a polynomial of small degree to ciphertexts to prevent the noise level from getting too high.} and refreshing the ciphertext using the decryption circuit homomorphically by evaluating it on the ciphertext and the encryption of the secret key. This results in a new, less noisy ciphertext for the same underlying plaintext so that FHE is obtained by continuously refreshing the ciphertext, and thus lifting the limitation on the number of homomorphic evaluations.

Brakerski et al. \cite{brakerski2014leveled} tackle the noise problem without bootstrapping via so-called ``levelled FHE'', which roughly speaking amounts to deciding on some prior parameterization in the form of reducing the modulus of the ciphertext space and the noise, thereby fixing the circuit depth. This, among other things, translates into a given number of public keys corresponding to the selected depth.

For further information, we encourage the reader to refer to Gentry's original work \cite{gentry2009fully} and subsequent works \cite{gentry2010computing,gentry2012homomorphic}, extra sources cited further below in our following discussions, as well as the guide to FHE by Armknecht et al. \cite{armknecht2015guide}.

\stopminichaptoc{\minichaptocenabled}

\chapter{Coercer-Deniable Quantum Key Exchange}\label{chp:coercer-deniable-qke}

\printminichaptoc{\minichaptocenabled}

Following the setting in \cite{beaver2002deniability}, in which it is implicitly assumed that the adversary has established a binding between the participants' identities and a given QKE session, we introduce the notion of coercer-deniability for QKE. This makes it possible to consider an adversarial setting similar to that of deniable encryption \cite{canetti1997deniable} and expect that the parties might be coerced into revealing their private coins after the termination of a session, in which case they would have to produce fake randomness such that the resulting transcript and the claimed values remain consistent with the adversary's observations.

Beaver's analysis \cite{beaver2002deniability} is briefly addressed in a footnote in a paper by Ioannou and Mosca \cite{ioannou2011new} and the issue is brushed aside based on the argument that the parties do not have to keep records of their raw key bits. It is argued that for deniability to be satisfied, it is sufficient that the adversary cannot provide binding evidence that attributes a particular key to the classical communication as their measurements on the quantum channel do not constitute a publicly verifiable proof. However, counter-arguments for this view were already raised in the motivations for deniable encryption \cite{canetti1997deniable} in terms of secure erasure being difficult and unreliable \cite{gutmann1996secure}, and that erasing cannot be externally verified. Moreover, it is also argued that if one were to make the physical security assumption that random choices made for encryption are physically unavailable, the deniability problem would disappear. We refer to \cite{canetti1997deniable} and references therein for more details.

Bindings, or lack thereof, lie at the core of deniability. Although we leave a formal comparison of our model with the one formulated in the simulation paradigm \cite{di2006deniable} as future work, a notable difference can be expressed in terms of the inputs presented to the adversary. In the simulation paradigm, deniability is modelled only according to the simulatability of the legal transcript that the adversary or a corrupt party produces naturally via a session with a party as evidence for the judge, whereas for coercer-deniability, the adversary additionally demands that the honest parties reveal their private randomness.

Finally, note that viewing deniability in terms of ``convincing'' the adversary is bound to be problematic and indeed a source of debate in the cryptographic research community as the adversary may never be convinced given their knowledge of the existence of faking algorithms.
Hence, deniability is formulated in terms of the indistinguishability of views (or their simulatability \cite{di2006deniable}) such that a judge would have no reason to believe a given transcript provided by the adversary establishes a binding as it could have been forged or simulated.

\section{Defeating Deniability in QKE via Eavesdropping in a Nutshell}\label{subsec:state-injection-attack}

We briefly review the eavesdropping attack described in \cite{beaver2002deniability} and provide further insight. Recall that $\ket{\phi_S}, \ket{\phi_R}, \ket{\phi_E}, \ket{\phi_J}, \ket{\phi_{env}}$ model the registers belonging to the sender ($S$), receiver ($R$), eavesdropper ($E$), judge ($J$), and the environment, respectively. $\ket{\phi_{env}}$ is limited to modelling the authenticated classical channel as a party broadcasting measurements to $S$, $R$ and $E$. Let $\rho(m_1)$ denote the global state for the message $m_1$ and $\rho(m_1,m_2) = \sum \ket{\phi'_S \phi'_R \phi_E \phi_{env}}$ an attempt at denial by pretending that $m_2$ was really sent, instead of the actual message $m_1$, where $\ket{\phi'_S} = D_S(m_1, m_2, \ket{\phi_S})$ and $\ket{\psi'_R} = D_R(m_1, m_2, \ket{\phi_R})$, with $D_S$ and $D_R$ being denial algorithms (not necessarily unitary).

The judge, receiving inputs for registers $\phi_S, \phi_R, \phi_E$ and $\phi_{env}$, has a final state described in registers $d$ and $J^\prime$, where $d$ is a decision bit. A coin flip $c$ in the environment determines whether denial will be attempted. If $c=0$, run $J$ on $\rho(m_1)$; if $c=1$, run $J$ on $\rho(m_1, m_2)$, i.e., apply $D_R$ and $D_S$ before submitting to the judge. The final result is of the form $\rho(m_1, m_2, c) = \sum \ket{cd} \ket{\phi_{J^\prime}} \ket{\phi_{env}}$. We trace out $J^\prime$ and the environment to get a mixture over $\ket{cd}$'s and a judge is considered to be safe if $\ket{01}$ has zero probability, i.e., $J$ never accuses the parties of attempts at denial while no such attempts were made. The author defines deniability in terms of the probability that $J$ decides $\ket{11}$ on security parameter $\kappa$, denoted by $P_{J,E}(m_1, m_2, \kappa)$, i.e., if for any $E$, any safe $J$, and for any $m_1, m_2$ we have: $P_{J,E}(m_1, m_2, \kappa) = \kappa^{-\omega(1)}$.

Now suppose Alice sends qubit $\ket{\psi}^{m,b}$ to Bob, which encodes a single-bit message $m$ prepared in a basis determined by $b \in \{+, \times\}$. Let $\Phi(E, m)$ denote the state obtained after sending $\ket{\psi}^{m,b}$, relayed and possibly modified by an adversary $E$. Moreover, let $\rho(E, m)$ denote the view presented to the judge, obtained by tracing over inaccessible systems. Now for a qubit measured correctly by Eve, if a party tries to deny by pretending to have sent $\sigma_1 = \rho(E, 1)$ instead of $\sigma_2  = \rho(E, 0)$, e.g., by using some local transformation $U_{neg}$ to simply negate a given qubit, then $F(\sigma_1, \sigma_2) = 0$, where $F$ denotes the fidelity between $\sigma_1$ and $\sigma_2$. Thus, the judge can successfully detect this attempt at denial.

Clearly, without making any further assumptions about Eve's strategy/program being known (publicly available) or about Alice and Bob deviating from the standard description of BB84, the futility of such attempts follows naturally from the properties of quantum measurement. Simply put, suppose a bit $v_i$ is encoded using a basis $b_i$ in qubit $\ket{\psi}$, and Alice tries to deny the underlying encoded value at index $i$ without knowing beforehand if Eve has prepared this particular state. Now if upon being interrogated by Eve, Alice claims to have used the basis $b'_i$ for her encoding, and $b'_i$ happens to be the same as the original one in which $\ket{\psi}$ was prepared (i.e., $b_i$), then it follows that if the measurement operator $\Pi_i$ is the same as the one claimed by Alice $\Pi'_i$, then
\[
\Pi_i = \Pi'_i \implies \frac{\Pi_i\ket{\psi}}{\sqrt{\bra{\psi}\Pi_i\ket{\psi}}} = \frac{\Pi'_i\ket{\psi}}{\sqrt{\bra{\psi}\Pi'_i\ket{\psi}}}.
\]
Therefore, claiming that $\lnot v_i$ was sent would be trivially detected.

This attack can be mounted successfully with non-negligible probability without causing the session to abort: Assume that $N$ qubits will be transmitted in a BB84 session and that the tolerable error rate is $\frac{\eta}{N}$, where clearly $\eta \sim N$. Eve measures each qubit with probability $\frac{\eta}{N}$ (choosing a basis at random) and passes on the remaining ones to Bob undisturbed, i.e., she plants a number of decoy states proportional to the tolerated error threshold. On average, $\frac{\eta}{2}$ measurements will come from matching bases, which can be used by Eve to detect attempts at denial, if Alice claims to have measured a different encoding. After discarding half the qubits in the sifting phase, this ratio will remain unchanged. Now Alice and/or Bob must flip at least one bit in order to deny without knowledge of where the decoy states lie in the transmitted sequence, thus getting caught with probability $\frac{\eta}{2N}$ upon flipping a bit at random.

Note that even in the noiseless quantum theory model, a similar attack would still succeed with non-negligible probability. In effect, as far as Eve is concerned, a single quantum state would suffice to detect attempts at denial with probability $\sfrac{1}{2}$: Eve picks a quantum state at random from the sequence of transmitted qubits and measures it in a random basis; she gets the measurement right with probability $\sfrac{1}{2}$ and if the honest parties claim to have used a different classical encoding for this particular state, they will be caught.

\section{On the Coercer-Deniability of Uncloneable Encryption}

The vulnerability described in Section \ref{subsec:state-injection-attack} is made possible by an eavesdropping attack that induces a binding in the key coming from a BB84 session. Uncloneable encryption remains immune to this attack because the quantum encoding is done for an already one-time padded classical input. More precisely, a binding established at the level of quantum states can still be perfectly denied because the actual raw information bits $m$ are not directly encoded into the sequence of qubits, instead the concatenation of $m$ and the corresponding authentication tag $\mu = \mathrm{MAC}_k(m)$, i.e., $x=m||\mu$, is masked with a one-time pad $e$ to obtain $y = x \oplus e$, which is then mapped onto a codeword $z$ that is encoded into quantum states. For this reason, in the context of coercer-deniability, regardless of a binding established on $z$ by the adversary, Alice can still deny to another input message in that she can pick a different input $x'=m'||\mu'$ to compute a fake pad $e' = y \oplus x'$, so that upon revealing $e'$ to Eve, she will simply decode $y \oplus e' = x'$, as intended.

However, note that a prepare-and-measure QKE obtained from UE still remains vulnerable to the same eavesdropping attack due to the fact that we can no longer make use of the deniability of the one-time pad in UE such that the bindings induced by Eve constrain the choice of the underlying codewords.

\section{Security Model}\label{subsec:security-model}

We adopt the framework for quantum AKEs developed by Mosca et al. \cite{mosca2013quantum}, inspired in turn by \cite{bellare1993entity,canetti2001analysis}, and mainly focus on our proposed extensions.

\subsection{Parties and Sessions}

\textbf{Parties}, including the adversary, are modelled as a pair of classical and quantum Turing machines (TM) that execute a series of interactive computations and exchange messages with each other through classical and quantum channels, collectively referred to as a \textbf{protocol}. An execution of a protocol is referred to as a \textbf{session}, identified with a unique session identifier.

An ongoing session is called an \emph{active} session, and upon completion, it either outputs an error term $\bot$ in case of an abort, or it outputs a tuple $(sk, pid, \vec{v}, \vec{u})$ in case of a successful termination. The tuple consists of a session key $sk$, a party identifier $pid$ and two vectors $\vec{u}$ and $\vec{v}$ that model public (authentication) values and secret terms, respectively.

In contrast with traditional classical AKE security models, instead of hardcoding or specifying restrictions on what secret values can be learned by the adversary into the security \emph{model} itself, in this rather more generic model, the vectors $\vec{u}$ and $\vec{v}$ in the session output specify what can and cannot be learned for a given \emph{protocol}. This avoids the need for defining a new security model every time the set of learnable values changes.

\subsection{Adversarial Model and Security Games}

We adopt an extended version of the \textbf{adversarial model} described in \cite{mosca2013quantum}, to account for coercer-deniability. Let $E$ be an efficient, i.e. (quantum) polynomial time, adversary with classical and quantum runtime bounds $t_c(k)$ and $t_q(k)$, and quantum memory bound $m_q(k)$, where bounds can be unlimited. Following standard assumptions, the adversary controls all communication between parties and carries the messages exchanged between them. We consider an authenticated classical channel and do not impose any special restrictions otherwise. Additionally, the adversary is allowed to approach either the sender or the receiver after the termination of a session and request access to a subset $
\vec{r} \subseteq \vec{v}$ of the private randomness used by the parties for a given session, i.e. set of values to be faked.

Security notions can be formulated in terms of \textbf{security experiments} in which the adversary interacts with the parties via a set of well-defined \textbf{queries}. These queries typically involve sending messages to an active session or initiating one, corrupting a party, learning their long-term secret key, revealing the ephemeral keys of an incomplete session, obtaining the computed session key for a given session, and a \textbf{test-session($id$)} query capturing the winning condition of the game that can be invoked only for a \emph{fresh} session.

Revealing secret values to the adversary is modeled via \textbf{partnering}, e.g., given a party who has a private and public value pair $(x,X)$ in memory, $\text{Partner}(X)$ returns the private value $x$ and $\text{Partner}(\Psi)$ returns the session secret key $sk$ for session $\Psi$. The notion of \emph{freshness} captures the idea of excluding cases that would allow the adversary to trivially win the security experiment. This is done by imposing minimal restrictions on the set of queries the adversary can invoke for a given session such that there exist protocols that can still satisfy the definition of session-key security. For completeness, we present the original definition of a fresh session below. We refer the reader to \cite{mosca2013quantum} for more details.

\begin{definition}[Fresh session \cite{mosca2013quantum}]\label{def:qke-fresh-session}
A session $\Psi$ owned by an honest party $P_i$ is fresh if all of the following conditions are satisfied:
\begin{enumerate}
	\item For every vector $\vec{v}_j$, $j \ge 1$, in $P_i$'s output for session $\Psi$, there is at least one element $X$ in $\vec{v}_j$ such that the adversary is not a partner to $X$.

	\item The adversary did not issue $\text{Partner}(\Psi')$ to any honest party $P_j$ for which $\Psi'$ has the same public output vector as $\Psi$ (including the case where $\Psi'$ = $\Psi$ and $P_j = P_i$).

	\item At the time of session completion, for every vector $\vec{u}_j$, $j \ge 1$, in $P_i$’s output for session $\Psi$, there was
at least one element $X$ in $\vec{u}_j$ such that the adversary was not a partner to $X$.
\end{enumerate}
\end{definition}

\begin{remark}
	Regarding the values in the output vector $\vec{v}$, note that exactly \emph{when} these values are revealed plays a crucial role in terms of preserving the freshness condition. For instance, if the adversary were to learn these values before parties' measurements, it must partner to the values, which would violate session freshness. In the case of deniability, the corruption model accounts for the coercion of honest parties after the completion of a session.
\end{remark}

\begin{remark}
	In \cite{mosca2013quantum}, Mosca et al. state the following as an open problem. Given the constraints imposed by a fresh session, which are specified by the conditions in the output vector, the authors consider the possibility of developing a QKE protocol that can retain its security properties in the short-term and long-term even if some random values were known to the adversary before the run of the protocol. Similarly, we pose a similar question in the context of coercer-deniable QKE as an open problem, namely the possibility of relaxing the freshness constraints given a suitable instantiation of the model by a particular QKE protocol. In such a case, one could potentially obtain coercer-deniability in a stronger adversarial model.
\end{remark}

\subsection{Transcript of Classical and Quantum Exchanges}

The \textbf{transcript} of a protocol consists of all publicly exchanged messages between the parties during a run or session of the protocol. The definition of ``views'' and ``outputs'' given in \cite{beaver2002deniability} coincides with that of transcripts in \cite{di2006deniable} in the sense that it allows us to model a transcript that can be obtained from observations made on the quantum channel. The \emph{view} of a party $P$ consists of their state in $\mathcal{H}_P$ along with any classical strings they produce or observe. More generally, for a two-party protocol, captured by the global density matrix $\rho_{AB}$ for the systems of $A$ and $B$, the individual system $A$ corresponds to a partial trace that yields a reduced density matrix, i.e., $\rho_A = \mathrm{Tr}_B(\rho_{AB})$, with a similar approach for any additional couplings.

\section{Coercer-Deniable QKE via Indistinguishability}

We use the security model in Section \ref{subsec:security-model} to introduce the notion of coercer-deniable QKE, formalized via the indistinguishability of real and fake views. Note that in this work we do not account for forward deniability and forward secrecy.

\subsection{Coercer-Deniability Security Experiment.}\label{sec-exp:coercer-deniable-qke}
Let $\mathrm{CoercerDenQKE}^{\Pi}_{E, \C}(\kappa)$ denote this experiment and $Q$ the same set of queries available to the adversary in a security game for session-key security, as described in Section \ref{subsec:security-model}, and \cite{mosca2013quantum}. Clearly, in addition to deniability, it is vital that the security of the session key remains intact as well. For this reason, we simply extend the requirements of the security game for a session-key secure KE by having the challenger $\C$ provide an additional piece of information to the adversary $E$ when the latter calls the \textbf{test-session()} query. This means that the definition of a fresh session remains the same as the one given in \cite{mosca2013quantum}. $E$ invokes queries from $Q \setminus \{\text{{\bfseries test-session()}}\}$ until $E$ issues \textbf{test-session()} to a fresh session of their choice. $\C$ decides on a random bit $b$ and if $b=0$, $\C$ provides $E$ with the real session key $k$ and the real vector of private randomness $\vec{r}$, and if $b=1$, with a random (fake) key $k'$ and a random (fake) vector of private randomness $\vec{r}'$.
Finally, $E$ guesses an output $b'$ and wins the game if $b = b'$. The experiment returns 1 if $E$ succeeds, and 0 otherwise. Let $Adv_{E}^{\Pi}(\kappa) = |\prob{b = b'} - \sfrac{1}{2}|$ denote the winning advantage of $E$.

\begin{definition}[Coercer-Deniable QKE]\label{def:coercer-deniable-qke}
	For adversary $E$, let there be an efficient distinguisher $\dist_E$ on security parameter $\kappa$. We say that $\Pi_{\vec{r}}$ is a coercer-deniable QKE protocol if, for any adversary $E$, transcript $\vec{t}$, and for any $k, k'$, and a vector of private random inputs $\vec{r} = (r_1, \ldots, r_{\ell})$, there exists a denial/faking program $\mathcal{F}_{A,B}$ that running on $(k, k', \vec{t}, \vec{r})$ produces $\vec{r}' = (r'_1, \ldots, r'_{\ell})$ such that the following conditions hold:
	\begin{itemize}
		\item $\Pi$ is a secure QKE protocol.
		\item The adversary $E$ cannot do better than making a random guess for winning the coercer-deniability security experiment, i.e., $Adv_{E}^{\Pi}(\kappa) \le \mathrm{negl}(\kappa)$
		\[
		\mathrm{Pr}[\mathrm{CoercerDenQKE}^{\Pi}_{E,\C}(\kappa) = 1] \le \frac{1}{2} + \mathrm{negl}(\kappa)
		\]
	\end{itemize}
	Equivalently, we require that for all efficient distinguisher $\dist_E$
	\[
	|\prob{\dist_E(\view{Real}(k, \vec{t}, \vec{r})) = 1} - \prob{\dist_E(\view{Fake}(k', \vec{t}, \vec{r'})) = 1}| \le \mathrm{negl}(\kappa),
	\]
	where the transcript $\vec{t}=(\vec{c}, \rho_E(k))$ is a tuple consisting of a vector $\vec{c}$, containing classical message exchanges of a session, along with the local view of the adversary w.r.t. the quantum channel obtained by tracing over inaccessible systems (see Section \ref{subsec:security-model}).
\end{definition}

A function $f: \mathbb{N} \rightarrow \mathbb{R}$ is negligible if for any constant $k$, there exists a $N_k$ such that $\forall  N \ge N_k$, we have $f(N) < N^{-k}$. In other words, it approaches zero faster than any polynomial in the asymptotic limit.

\begin{remark}\label{remark:randomness-compromise}
	We introduced a vector of private random inputs $\vec{r}$ to avoid being restricted to a specific set of ``fake coins'' in a coercer-deniable setting such as the raw key bits in BB84 as used in Beaver's analysis. This allows us to include other private inputs as part of the transcript that need to be forged by the denying parties without having to provide a new security model for each variant. Indeed, in \cite{mosca2013quantum}, Mosca et al. consider the security of QKE in case various secret values are compromised before or after a session. This means that these values can, in principle, be included in the set of random coins that might have to be revealed to the adversary and it should therefore be possible to generate fake alternatives using a faking algorithm.
\end{remark}

\subsection{Revealing Randomness Before and After a Run of the Protocol}

In our Definition \ref{def:coercer-deniable-qke} for coercer-deniable QKE, we introduced a vector of private random inputs $\vec{r}$ to avoid being restricted to a specific value that is expected to be forged and revealed to the adversary in a coercer-deniable setting such as the raw key bits used in Beaver's analysis. This allows us to include other private inputs as part of the transcript that need to be forged by the denying parties without having to provide a new security model for each variant.

We now highlight some observations made in \cite{stebila2009case} regarding what secret values can be learned by the adversary before and after a QKE session without compromising the security of the established session-key, which is highly relevant w.r.t. coercer-deniability.

All known QKE protocols cease to be secure if any of the following random values are leaked to the adversary before a session: the long-term secret key, the basis choices, raw data bits, information reconciliation function, or privacy amplification function (a universal-2 hash function keyed under a secret key). This means that all these values can, in principle, be included in the set of random coins that might have to be revealed to the adversary in a coercer-deniable setting and it should therefore be possible to generate fake alternatives using a faking algorithm.

Regarding secret values that can be revealed after the termination of a session, for entanglement-based QKE protocols, all random choices can be revealed. This is due to the fact that for a successfully terminated EPR-based protocol, the actual key bits are determined neither by the two parties, nor by the adversary, instead they are the result of measurements, which become uncorrelated with any of the inputs bits of any of the parties (including the adversary) after successful privacy amplification.

The same property does not hold for prepare-and-measure variants such as BB84, as the sender determines the random choice of raw data bits that will be encoded into quantum states.

\section{Deniable Key Exchange in the Simulation Paradigm}

We now compare the model for coercer-deniable QKE with the standard definition of deniability for key exchange in classical cryptography.	The notion of deniability for key exchange was formalized in a seminal work by Di Raimondo, Gennaro and Krawczyk \cite{di2006deniable}, in a framework based on the simulation paradigm, inspired by deniable authentication, which is in turn based on the simulation of transcripts used in the formalism of zero-knowledge proofs \cite{dwork2004concurrent}.

The main idea is that for a deniable authentication protocol, the receiver's view (entire transcript) of the protocol can be simulated by an efficient machine ($\ppt$) that does not have access to the sender's secret key, leading to an indistinguishability problem between the distribution of the simulated view and the real one obtained when interacting with the sender. Therefore, a possibly dishonest receiver cannot convince a judge (a third party) by presenting such a view as it will not constitute convincing evidence given that it could have just as well been generated by the receiver running the simulator.

The same definition is extended to the setting of authenticated key exchange with the additional requirement that it should be possible to simulate not only the communication during the KE session, but also the value of the session key itself. This prevents a dishonest party not only from proving to a judge that they exchanged a key with a specific party, but also from producing a proof of the contents of a communication protected with that key, be it for encryption or authentication. The latter captures the idea that a shared session key established through a deniable AKE can be used to encrypt and authenticate messages using a symmetric-key mechanism that would be deniable provided that the key cannot be associated with either party.

In short, in this model the transcript should not leave any binding traces that can give rise to algorithmic proofs of participation. Such traces could for example be due to digital signatures for authentication in the public key setting. We now recall the standard definition for deniable key exchange in the simulation paradigm and consider the possibility of extending it to capture the requirements for coercer-deniability.

\subsection{Deniable QKE in the Simulation Paradigm}

Let $aux$ denote some auxiliary input drawn from a distribution of $\text{AUX}$, which models extra (publicly obtainable) information that the adversary might have gathered in some other form via eavesdropping, over the classical and quantum channel, e.g., legal transcripts from previous runs of the protocol.

\begin{definition}\cite{di2006deniable}\label{def:sim-deniable-ake}
Let $\Pi^{(S, R)}$ be a key exchange protocol with respect to the class of auxiliary inputs $\text{AUX}$ if for any adversary $\adv$, for any input of public keys $\vec{pk} = (pk_1, \ldots, pk_{\ell})$ and any auxiliary input $aux \in \text{AUX}$, there exists a simulator $\text{SIM}_{\adv}$ that, running on the same inputs as $\adv$, produces a simulated view that is indistinguishable from the real of view of $\adv$. Consider the probability distributions
\begin{align*}
&Real(n, aux) = [(\sk_i, \pk_i)] \sample \kgen(1^n); (aux, \vec{pk}, \text{View}_{\adv}(\vec{pk}, aux))]\\
&Sim(n, aux) = [(\sk_i, \pk_i)] \sample \kgen(1^n); (aux, \vec{pk}, \text{SIM}_{\adv}(\vec{pk}, aux))]
\end{align*}
then for all $\ppt$ distinguisher $\dist$
\[
|\mathrm{Pr}_{x \in Real(n, \vec{t})}[\dist(x) = 1] - \mathrm{Pr}_{x \in Sim(n, \vec{t})}[\dist(x) = 1]| \le \mathrm{negl}(\kappa),
\]
\end{definition}

As already pointed out in \cite{di2006deniable}, in the symmetric-key setting, deniability in the simulation paradigm is satisfied as long as the secret key cannot be associated to the identities of the peers via third-party verifiable proofs, an argument that is hinted at in \cite{ioannou2011new}. In the case of QKE, this follows from two simple observations. The first has to do with the participation repudiation of the parties: note that in the inherently symmetric-key setting of QKE, MACs (authentication tags) included in all classical exchanges, are computed under a random key using information-theoretically secure symmetric-key authentication algorithms.

This means that the existence of a simulator for such transcripts is trivial given that the malicious party could have generated them by herself without interacting with the sender. Secondly, regarding the simulation of the session key itself, it follows directly from the security proofs of QKE \cite{shor2000simple, christandl2004generic, mosca2013quantum} that the resulting session key is indistinguishable from a random key. We defer an analysis of deniability for public-key authenticated QKE to future work.

A direct approach for extending Definition \ref{def:sim-deniable-ake} such that it models coercer-deniability would consist of including the fake message $m'$ and the fake vector of private randomness $\vec{r'}$ in the input. This would lead to an additional requirement, namely the existence of a denial/faking program $\phi_{S/R}$ that running on $(aux, \vec{pk}, m, \vec{r})$ generates a fake message and a vector of fake random inputs $(m', \vec{r'})$, such that the resulting distributions $Real(n, aux, m', \vec{r'})$ and $Sim(n, aux, m', \vec{r'})$ still remain indistinguishable. Simulation-based deniability would correspond to letting $\vec{r} = \vec{r'} = \vec{0}$ and $m = m' = \epsilon$. Moreover, symmetric-key authenticated QKE amounts to setting $\vec{sk} = (sk)$ (the symmetric preshared key for authentication) and null for the public key, i.e., $\vec{pk} = ()$.

Finally, it is worth stressing that fundamentally, the model for coercer-deniable QKE differs from the one framed in the simulation paradigm in that the former is specified in terms of a game-based definition, while the latter relies on a simulation-based definition. Similar to the initially unknown equivalence of IND-CPA to semantic security, whether or not the two deniability definitions can be considered to be equivalent, potentially under certain constraints, remains an open question.

\subsection{The Relevance of Non-attributability for Deniability}

Among the known key exchange schemes, one of the unique properties of QKE protocols is that of ``non-attributability'', an advantage over classical protocols that was first explicitly highlighted in a work by Ioannou and Mosca \cite{ioannou2011new}. Non-attributability captures the property that the final secret key $k$ produced by a QKE protocol is entirely independent of the classical communication and the initial pre-shared key. In other words, there is no way to mathematically link or attribute the final secret key to the publicly readable transcript of discussion between Alice and Bob. Note that this property is closely related to the definition of deniability in the simulation paradigm.

This property will be an important part of a proof of deniability in the simulation paradigm for a public-key authenticated QKE protocol, i.e., a QKE protocol in which ITS authentication is replaced by computationally secure authentication mechanisms. The reason for this is that as long as a given computationally secure authentication scheme is shown to be deniable, the decoupling of the final key from the input bits due to the non-attributability property implies the simulatability of the final key itself, thereby satisfying the two requirements for a deniable AKE protocol in the simulation paradigm.

\stopminichaptoc{\minichaptocenabled}

\chapter{Deniable QKE via Covert Communication}\label{chp:dcqke}

\printminichaptoc{\minichaptocenabled}

\section{Introduction}

In this chapter we explore the relation between covert communication and deniable exchange. To the best of our knowledge, although techniques involving some form of covert storage of information have been used in practical cases such as the now defunct disk-encryption software TrueCrypt, the use of covert communication for achieving deniability has not been formally considered in the cryptographic literature. This is precisely what we set out to address in this chapter in the quantum setting by restricting our analysis to the case of quantum communication, and more specifically, to obtaining deniable quantum key exchange.

TrueCrypt provided support for the notion of \emph{plausible deniability} via hidden volumes, the existence of which could be denied. This provides a fairly good analogy for the following discussions. The underlying idea is that by using methods of covert communication, thereby preventing the adversary to notice an exchange of information in the first place, one tackles the problem of leaving traces of binding evidence by allowing honest parties to deny having ever exchanged any messages.

We establish a connection between covert communication and deniability by providing a simple construction for coercer-deniable QKE using covert QKE. We then show that deniability is reduced to the covertness property, meaning that deniable QKE can be performed as long as covert QKE is not broken by the adversary, formalized via the security reduction given in Theorem \ref{thm:den-covert-reduction}.

Before presenting our results, we provide some background on covert communication both in the context of transmitting classical as well as quantum information.

\section{Covert Communication and Embracing Noise}\label{sec:dc-qke}

Covert communication becomes relevant when parties wish to keep the very act of communicating secret or hidden from a malicious warden. This can be motivated by various requirements such as the need for hiding one's communication with a particular entity when this act alone can be incriminating. While encryption can make it impossible for the adversary to access the contents of a message, it would not prevent them from detecting exchanges over a channel under their observation.

Steganography, the art of hiding information in an innocuous object called a \emph{covertext}, resulting in a \emph{stegotext}, constitutes one of the earliest forms of covert communication, such as the practice of communicating secret messages by hiding them on a soldier's scalp underneath their hair using tattoos, which dates back to ancient times \cite{bash2015fundamental}. There exists a considerable body of work on steganography and its properties in the asymptotic limit, which falls outside the scope of this work. For more details, we refer the reader to the works of Andrew D. Ker such as \cite{ker2006batch,ker2007capacity,ker2009square,ker2010square}, which provide fundamental results in this area.

In terms of its applicability to covert communication, steganography is typically hampered by constraints such as the assumption that stegotexts are not corrupted by a noisy channel and that they are limited to discrete finite alphabets. Moreover, for whatever reason, when communication is not allowed, for example, precisely in a situation where exchanging messages can be incriminating, using steganography becomes impossible. This is where covert communication over noisy channels gains relevance.

\subsection{Covert Classical Communication}

An example of a modern technique for achieving covert communication involves the use of spread-spectrum radiofrequency (RF) communication where the signal power is lowered below the noise floor via bandwidth expansion. Note that spread-spectrum techniques also count among the earliest methods for protecting RF communication from detection and jamming, e.g., during the two world wars.

In a series of works, Bash et al. \cite{bash2013AWGN,bash2015fundamental,bash2015bosonic,bash2015hiding,sheikholeslami2016covert} addressed covert communication in various settings such as the standard model for RF channels, namely additive white Gaussian noise (AWGN) channels. Among other things, they established a square-root law (SRL) for covert communication in the presence of an unbounded quantum adversary stating that $\bigO{\sqrt{n}}$ covert bits can be exchanged in $n$ channel uses\footnote{Note that this square root law also crops up in the context of steganography, a similarity which is due to the mathematical properties of statistical hypothesis testing. The extra $\mathrm{log}(n)$ factor in the steganographic SRL has to do with the fact that the steganographic channel between Alice and Bob is assumed to be noiseless.} with arbitrarily low probability of detection by the monitoring adversary.

An approach used in these works is to hide information in the noise of optical channels. For instance, using a quantum information-theoretic analysis, in \cite{bash2015bosonic}, Bash et al. address the fundamental limits of covert classical communication and show the possibility of covertly transmitting classical information over standard, lossy bosonic channels with thermal noise\footnote{The lossy thermal-noise bosonic channel provides the quantum mechanical description of the transmission of a single (spatio-temporal polarization) mode of the electromagnetic field at a given transmission wavelength, e.g., optical or microwave, over linear loss and additive Gaussian noise such as noise originating from blackbody radiation.}, in the presence of a quantum adversary.

The crucial role played by noise and how it enables stealthy communication can be summarized as follows. Covert communication becomes impossible \cite{bash2015bosonic} if the quantum adversary, Eve, has complete control over the noise in the channel. However, any excess noise that is not controlled by Eve (e.g., the inevitable thermal noise due to the blackbody radiation at the operating temperature) can be used to enable Alice to reliably transmit $\bigO{n}$ covert bits to Bob in $n$ bosonic modes, even assuming that Eve intercepts all the photons that do not reach Bob. The excess noise can for example come from dark counts in photon counting detectors of Eve.

While from a cryptographic point of view, here we provide sufficient details for the purpose of presenting our results, a comprehensive description of the physics behind the underlying fundamental results in this area would go well beyond the scope of our work. Therefore, we encourage the reader to refer to the main sources cited above for more details.

\subsection{Covert Quantum Communication}

Recently, Arrazola and Scarani \cite{AS16} extended covert communication to the quantum regime for transmitting qubits covertly. The authors showed the feasibility of covertly transmitting quantum information using both single photon and coherent state signals.

To achieve covert quantum communication, the authors show that sequences of qubits can be transmitted covertly in the presence of noise, originating either from the environment or from the sender's lab (i.e., the sender injects noise into the channel) and derive analytical bounds for both cases. By allowing the sender to inject noise into the channel, i.e., a model in which noise originates from the sender's lab and is thus inaccessible to Eve, security is achieved even against an adversary in complete control of the channel connecting the sender and receiver.

To explain the main idea, for simplicity consider that the sender, Alice, encodes a qubit in a single photon across two optical modes, which correspond to the polarization degrees of freedom of a single time-bin mode. The idea is that assuming that Alice and Bob have access to $N$ time-bins, each of which capable of carrying a qubit signal, covert quantum communication is performed as follows. For each of the $N$ time-bins, Alice sends a qubit signal with probability $q \ll 1$, and with probability $1-q$, she does nothing. Compared to a regular protocol, the main difference is that signals are not sent sequentially, but they are rather randomly spread out in time. Alice sends on average $Nq$ qubit signals coded in time-bins that are pre-agreed based on a shared secret key. For Eve, not knowing this secret key, each time-bin carries a signal with probability $q$, i.e., Alice sends a signal with probability $q$ for each time-bin. We state the main result below in Def. \ref{def:covert-quantum-communication}.

\begin{definition}\label{def:covert-quantum-communication}
	Covert quantum communication consists of two parties exchanging a sequence of qubits such that an adversary trying to detect this cannot succeed by doing better than making a random guess, i.e., $P_d \le \frac{1}{2} + \epsilon$ for sufficiently small $\epsilon > 0$, where $P_d$ denotes the probability of detection and $\epsilon$ the detection bias.
\end{definition}

\section{Covert Quantum Key Exchange}

We now briefly describe the original motivations for performing covert quantum key exchange and its inherent limitation, which can be mitigated using pseudo-random number generators (PRNG). We then present a security experiment capturing a game-based definition for covert QKE, which will then be used later on for relating covert QKE to deniability.

\subsection{Covert QKE using PRNGs}

Since covert communication requires pre-shared secret randomness, a natural question to ask is whether QKE can be done covertly. This was also addressed in \cite{AS16} and it was shown that covert QKE with unconditional security for the covertness property is impossible because the amount of key consumption is greater than the amount produced.

The main reason behind the impossibility of an information-theoretic bootstrapping of covert QKE can be explained as follows. Since Alice sends a signal with probability $q$ for each of the $N$ available time-bins, in the limit of large $N$, the average amount of shared bits that is needed to specify the selected time-bins is given by $N \cdot H_b(q)$, where $H_b(\cdot)$ denotes the binary entropy\footnote{The binary entropy of a variable $X$ taking on two possible values occurring with probabilities $p$ and $1-p$ (a Bernoulli process), respectively, is defined as follows: $H_b(p) = -p \mathrm{log}_2(p) - (1-p)\mathrm{log}_2(1-p)$}. However, Alice and Bob can at best obtain an average of $d = N \cdot H_b(q)$ secret bits of information from running covert QKE, which does not even account for the overhead of information reconciliation, i.e., error estimation and error correction. Thus, since $H_b(q) > q$ for $q < \frac{1}{2}$, it follows that such a scheme requires more key material than it produces.

Given this theoretical limitation, a hybrid approach involving pseudo-random number generators (PRNG) was proposed to achieve covert QKE with a positive key rate such that the resulting secret key remains information-theoretically secure, while the covertness of QKE is shown to be at least as strong as the security of the PRNG. The PRNG is used to expand a truly random pre-shared key into an exponentially larger pseudo-random output, which is then used to determine the time-bins for sending signals in covert QKE.

To show that Bob can use the weak signals sent by Alice, the channel with a parameterization that guarantees a low detection bias $\epsilon$ would have to be shown to provide a positive key rate, even if all the errors are attributed to Eve. Despite the inherent requirement of noise for achieving covert communication, the authors \cite{AS16} show that if the quantum bit error rate due to the noise is sufficiently low, performing covert QKE with a non-zero key rate while guaranteeing a low detection bias is possible.

\subsection{Covert QKE Security Experiment}\label{sec-exp:covert-qke}

Let $\mathrm{CovertQKE}^{\Pi^{cov}}_{E,\C}(\kappa)$ denote the security experiment. The main property of covert QKE, denoted by $\covpi{}$, can be expressed as a game played by the adversary $E$ against a challenger $\C$ who decides on a random bit $b$ and if $b=0$, $\C$ runs $\covpi{}$, otherwise (if $b=1$), $\C$ does not run $\covpi{}$. Finally, $E$ guesses a random bit $b'$ and wins the game if $b=b'$. The experiment outputs 1 if $E$ succeeds, and 0 otherwise.
The winning advantage of $E$ is given by $Adv_{E}^{\Pi^{cov}}(\kappa) = |\prob{b = b'} - \sfrac{1}{2}|$ and we want that $Adv^{\Pi^{cov}}_{E}(\kappa) \le \negl{\kappa}$.

\begin{definition}\label{def:covert-QKE}
	Let $G: \{0,1\}^s \rightarrow \{0,1\}^{g(s)}$ be a $(\tau,\epsilon)$-PRNG secure against all efficient distinguishers $\dist$ running in time at most $\tau$ with success probability at most $\epsilon$, where $\forall s: g(s) > s$. A QKE protocol $\covpi{G}$ is considered to be covert if the following property holds for any $\qpt$ adversary $E$:
	\begin{itemize}
		\item $\covpi{G}$ is a secure QKE protocol.
		\item The probability that $E$ guesses the bit $b$ correctly ($b' = b$), i.e., $E$ manages to distinguish between Alice and Bob running $\covpi{G}$ or not, is no more than $\frac{1}{2}$ plus a negligible function in the security parameter $\kappa$, i.e.,
		\[
		\prob{\mathrm{CovertQKE}^{\Pi^{cov}}_{E, \C}(\kappa) = 1} \le \frac{1}{2} + \negl{\kappa}
		\]
	\end{itemize}
\end{definition}
\begin{theorem}(Sourced from \cite{AS16})\label{thm:covert-QKE}
	The secret key obtained from the covert QKE protocol $\covpi{G}$ is information-theoretically secure and the covertness of $\covpi{G}$ is as secure as the underlying PRNG.
\end{theorem}

\section{Deniable Covert Quantum Key Exchange}

We are now in a position to describe DC-QKE, a simple construction shown in Protocol \ref{protocol:dc-qke}, which preserves unconditional security for the final secret key, while its deniability is as secure as the underlying PRNG used in $\covpi{\vec{r},G}$. In terms of the Security Experiment \ref{sec-exp:coercer-deniable-qke}, $\covpi{\vec{r},G}$ is run to establish a real key $k$, while non-covert QKE $\Pi_{\vec{r}'}$ is used to produce a fake key $k'$ aimed at achieving deniability, where $\vec{r}$ and $\vec{r}'$ are the respective vectors of real and fake private inputs.

\subsection{Running Covert and Non-Covert Quantum Key Exchange}

Operationally, consider a setting wherein the parties suspect\footnote{The honest parties do not have to and indeed are not expected to always run the faking component, thus the corresponding random input does not naturally crop up in the definition.} in advance that they might be coerced into revealing their private coins for a given run: their joint strategy consists of running both components in Protocol \ref{protocol:dc-qke} and claiming to have employed $\Pi_{\vec{r}'}$ to establish the fake key $k'$ using the fake private randomness $\vec{r}'$ (e.g. raw key bits in BB84) and provide these as input to the adversary upon termination of a session. Thus, for Eve to be able to produce a proof showing that the revealed values are fake, she would have to break the security of covert QKE to detect the presence of $\covpi{\vec{r},G}$, as shown in Theorem \ref{thm:den-covert-reduction}. Moreover, note that covert communication can be used for dynamically agreeing on a joint strategy for denial, further highlighting its relevance for deniability.

\begin{algorithm}
	\floatname{algorithm}{Protocol}
	\caption{DC-QKE for an $n$-bit key}
	\label{protocol:dc-qke}
	\begin{algorithmic}[1]
		\STATE \textbf{RandGen:} Let $\vec{r} = (r_1, \ldots, r_{\ell})$ be the vector of private random inputs, where $r_i \sample \{0,1\}^{|r_i|}$.
		\STATE \textbf{KeyGen:} Run $\Pi^{cov}_{\vec{r},G}$ to establish a random secret key $k \in \{0,1\}^{n}$.
	\end{algorithmic}
	Non-covert faking component $\mathcal{F}_{A,B}$:
	\begin{algorithmic}[1]
		\STATE \textbf{FakeRandGen:} Let $\vec{r}' = (r'_1, \ldots, r'_{\ell})$ be the vector of fake private random inputs, where $r'_i \sample \{0,1\}^{|r'_i|}$.
		\STATE \textbf{FakeKeyGen:} Run $\Pi_{\vec{r'}}$ to establish a separate fake key $k' \in \{0,1\}^n$.
	\end{algorithmic}
\end{algorithm}

\begin{remark}
	The original analysis in \cite{beaver2002deniability} describes an attack based solely on revealing fake raw key bits that may be inconsistent with the adversary's observations. An advantage of DC-QKE in this regard is that Alice's strategy for achieving coercer-deniability consists of revealing all the secret values of the non-covert QKE $\Pi_{\vec{r}'}$ honestly.
	This allows her to cover the full range of private randomness that could be considered in different variants of deniability as discussed in Remark \ref{remark:randomness-compromise}. A potential drawback is the extra communication cost induced by $\mathcal{F}_{A,B}$, which could, in principle, be mitigated using a less interactive solution such as QKE via UE.
\end{remark}

\begin{remark}
	If the classical channel is authenticated using an information-theoretically secure algorithm, the minimal entropy overhead in terms of pre-shared key (logarithmic in the input size) for $\Pi$ can be generated by $\covpi{\vec{r}}$.
\end{remark}

\begin{example}
	In the case of encryption, $A$ can send $c = m \oplus k$ over a covert channel to $B$, while for denying to $m'$, she can send $c' = m' \oplus k'$ over a non-covert channel. Alternatively, she can transmit a single ciphertext over a non-covert channel such that it can be opened to two different messages. To do so, given $c = m \oplus k$, Alice computes $k' = m' \oplus c = m' \oplus m \oplus k$, and she can then either encode $k'$ as a codeword, as described in Section \ref{sec:qke-and-ue}, and run $\Pi_{\vec{r}'}$ via uncloneable encryption, thus allowing her to reveal the entire transcript to Eve honestly, or she can agree with Bob on a suitable privacy amplification (PA) function (with PA being many-to-one) as part of their denying program in order to obtain $k'$.
\end{example}

\subsection{Reducing Deniability to Covert QKE}

As already mentioned, Alice and Bob can opt to make use of covert QKE to produce the real key and reveal the non-covert QKE session data to the adversary. Note that the honest parties are not expected to always run the non-covert component, which is why a priori there is no reason for the adversary to assume they have done so. Indeed, the case where the two parties do run the covert QKE component is indistinguishable from the case when they do not and only run a non-covert session.

\begin{theorem}\label{thm:den-covert-reduction}
	If $\covpi{\vec{r},G}$ is a covert QKE protocol, then DC-QKE given in Protocol \ref{protocol:dc-qke} is a coercer-deniable QKE protocol that satisfies Definition \ref{def:coercer-deniable-qke}.
\end{theorem}
\begin{proof}
	The main idea consists of showing that breaking the deniability property of DC-QKE amounts to breaking the security of covert QKE, such that coercer-deniability follows from the contrapositive of this implication, i.e., if there exists no efficient algorithm for compromising the security of covert QKE, then there exists no efficient algorithm for breaking the deniability of DC-QKE. We formalize this via a reduction, sketched as follows. Let $w' = \view{Fake}(k', \vec{t}_E, \vec{r}')$ and $w = \view{Real}(k, \vec{t}_E, \vec{r})$ denote the two views. Flip a coin $b$ for an attempt at denial: if $b=0$, then  $\vec{t}_E=(\vec{t}',\varnothing)$, else ($b=1$), $\vec{t}_E=
	(\vec{t}', \vec{t}^{cov})$, where $\vec{t}^{cov}$ and $\vec{t}'$ denote the transcripts of covert and non-covert exchanges from $\covpi{\vec{r},G}$ and $\Pi_{\vec{r}'}$.
	Now if DC-QKE is constructed from $\Pi^{cov}$, then given an efficient adversary $E$ that can distinguish $w$ from $w'$ with probability $p_1$, we can use $E$ to construct an efficient distinguisher $\dist$ to break the security of covert QKE with probability $p_2$ such that $p_1 \le p_2$. Indeed, given an instance of a DC-QKE security game and a corresponding fresh session, we construct a distinguisher $\dist$ that uses $E$ on input $w$ and $w'$, with the goal to win the game described in the Security Experiment \ref{sec-exp:coercer-deniable-qke}. The distinguisher $\dist$ would simply run $E$ (with negligible overhead) and observe whether $E$ succeeds at distinguishing $w$ from $w'$. Since the only element that is not sampled uniformly at random is in $\vec{t}^{cov}$ containing exchanges from the covert channel, which relies on a PRNG, the only way $E$ can distinguish $w$ from $w'$ is if she can distinguish $(\vec{t}', \vec{t}^{cov})$ from $(\vec{t}', \varnothing)$. If $E$ succeeds, then $\dist$ guesses that a covert QKE session has taken place, thereby winning the Security Experiment \ref{sec-exp:covert-qke} for covert QKE.
\end{proof}

\section{Further Applications of Covert Communication for Deniability}

As briefly mentioned in the previous discussions, covert communication can also be used in conjunction with standard non-covert, but deniable schemes, for agreeing on joint strategies in a dynamic fashion. In particular, one of the problems in the context of coercion-resistance or other closely related deniability settings is that when the honest parties try to decide on a specific strategy, they may need to communicate with one another via a secondary channel. Clearly, this act alone, when performed over a non-covert channel, would be detectable by the adversary, thus potentially rendering the existing deniability mechanism rather useless given that the adversary can demand that the parties reveal the private information used in their secondary channel as well.

A line of work worth pursuing would be to determine whether or not deniability with everlasting security can be obtained from covert QKE. More specifically, similar to the well-known everlasting security of QKE, given the quantum nature of the covert communication techniques considered here, it may be that once the window of opportunity for detecting message exchanges is closed for the adversary, she will not be able to discover this beyond the cutoff point of a given session, without any computational assumptions.

\subsection*{Acknowledgments}

We thank Mark M. Wilde and Ignatius William Primaatmaja for their comments.

\stopminichaptoc{\minichaptocenabled}

\chapter{Deniability via Entanglement Distillation}\label{chp:entanglement-distillation}

\printminichaptoc{\minichaptocenabled}

\section{Introduction}

We now consider the possibility of achieving information-theoretic deniability via entanglement distillation (ED). The approach that we proposed in Chapter \ref{chp:dcqke}, based on covert communication, provides deniability under the computational assumption that the security of the PRNG used for deciding the time-bins for sending qubit signals is not broken by a quantum adversary. Here we seek to move beyond that limitation and explore deniability with information-theoretic security.

In addition to our primary goal, namely achieving information-theoretic deniability, the approach we suggest here is also aimed at grounding deniability into quantum information-theoretic concepts. For instance, as explained in Sect. \ref{subsec:den-qke-via-ed-and-qt}, by using entanglement distillation for achieving deniable quantum key exchange, we can provide an operational meaning for deniability in terms of the amount of distillable entanglement, which is quantified by the von Neumann entropy or the entropy of entanglement, as detailed in Section \ref{sec:entanglement}.

Entanglement is an essential resource in quantum information that enables a wide range of classically impossible tasks such as superdense coding and quantum teleportation. As such, it is natural to ask how, and if, it can be used to achieve deniability.

The paradigm suggested here represents merely an initial attempt at taking the first steps towards formalizing the notion of deniable exchange using well-understood concepts in quantum information theory. While our analysis is limited to QKE, it is our hope that this will inspire future work dealing with other variants of deniability games, which go beyond key exchange.

\section{Entanglement Distillation}

In its most general form, ED allows two parties to distill maximally entangled pure states (\emph{ebits}) from an arbitrary sequence of entangled states at some positive rate using local operations and classical communication (LOCC), i.e. to move from a state
\[
\ket{\Phi_\theta}_{AB} \equiv cos(\theta)\ket{00}_{AB} + sin(\theta)\ket{11}_{AB}
\]
to a state
\[
\ket{\Phi^+}_{AB} = \frac{1}{\sqrt{2}} (\ket{00}_{AB} + \ket{11}_{AB}),
\]
where $0 < \theta < \pi/2$.

As explained in Section \ref{sec:entanglement} of Chapter \ref{chp:qip}, the resulting state is a Bell pair, which consists of maximally entangled qubits. The restriction to the LOCC paradigm, and thus imposing constraints on permissible operations, is motivated by the goal of discovering the fundamental limits for achieving such tasks with minimal resources and requirements, as otherwise the task would become trivial, i.e., if Alice and Bob were allowed to perform arbitrary global operations on their systems.

\subsection{Fundamental Properties and Limits}

In the noiseless model, $n$ independent identically distributed (i.i.d.) copies of the same partially entangled state $\rho$ can be converted into $\approx nH(\rho)$ Bell pairs in the limit $n \rightarrow \infty$, i.e., from $\rho_{AB}^{\otimes n}$ to $\ket{\Phi^+}_{AB}^{\otimes nH(\rho)}$, where $H(\rho) = -\mathrm{Tr}(\rho \mathrm{ln} \rho)$ denotes the von Neumann entropy of entanglement.

If the parties start out with pure states, local operations alone will suffice for distillation \cite{bennett1996concentrating,bennett1996purification}, otherwise the same task can be achieved via forward classical communication (one-way LOCC), as shown by the Devetak-Winter theorem \cite{devetak2005distillation}, to distill ebits from many copies of some bipartite entangled state. In the fundamental work of Devetak and Winter \cite{devetak2005distillation} on the distillation of secret key and entanglement from quantum states, the authors also derive information-theoretic formulas for the distillable secret key and fundamental key rate bounds in the presence of an adversary who is assumed to be in possession of a purification of Alice and Bob's joint state.

The pioneering work of Bennett et al. \cite{bennett1996mixed} on the relation between mixed state ED and quantum error correction count among the earliest developments in this area. In a relatively more recent work, Buscemi and Datta \cite{buscemi2010distilling} relax the i.i.d. assumption and provide a general formula for the optimal rate at which ebits can be distilled from a noisy and arbitrary source of entanglement via one-way and two-way LOCC.

\section{Deniability via Entanglement Distillation}\label{sec:entanglement-distillation}

Intuitively, the eavesdropping attack described in \cite{beaver2002deniability} and further detailed in Section \ref{subsec:state-injection-attack}, is enabled by the presence of noise in the channel as well as the fact that Bob cannot distinguish states sent by Alice from those prepared by Eve. As a result, attempting to deny to a different bit value encoded in a given quantum state - without knowing if this is a decoy state prepared by Eve - allows the adversary to detect such an attempt with non-negligible probability.

In terms of deniability, the intuition behind this idea is that while Alice and Bob may not be able to know which states have been prepared by Eve, they can instead remove her ``check'' decoy states from their set of shared entangled pairs by decoupling her system from theirs. Once they are in possession of maximally entangled states, they will have effectively factored out Eve's state such that the global system is given by the pure tensor product space $\ket{\Psi^+}_{AB} \otimes \ket{\phi}_E$. Thus the pure bipartite joint system between Alice and Bob cannot be correlated with any system under Eve's control, thereby foiling her cross-checking strategy. The singlet states can then be used to perform QKE via quantum teleportation \cite{bennett1993teleporting}.

\subsection{Deniable QKE via Entanglement Distillation and Teleportation}\label{subsec:den-qke-via-ed-and-qt}

We now argue why performing randomness distillation at the quantum level, thus requiring quantum computation, plays an important role w.r.t. deniability.

The subtleties alluded to in \cite{beaver2002deniability} arise from the fact that randomness distillation is performed in the classical post-processing step. This allows Eve to leverage her tampering in that she can verify the parties' claims against her decoy states. However, this attack can be countered by removing Eve's knowledge before the classical exchanges begin. Most security proofs of QKE \cite{lo1999unconditional,shor2000simple,mayers2001unconditional} are based on a reduction to an entanglement-based variant, such that the fidelity of Alice and Bob's final state with $\ket{\Psi^+}^{\otimes m}$ is shown to be exponentially close to 1.

Moreover, secret key distillation techniques involving ED and quantum teleportation \cite{bennett1996purification,devetak2005distillation} can be used to faithfully transfer qubits from $A$ to $B$ by consuming ebits. To illustrate the relevance of distillation for deniability in QKE, consider the generalized template shown in Protocol \ref{protocol:distillation-qke}, based on these well-known techniques.
\begin{algorithm}
	\floatname{algorithm}{Protocol}
	\caption{{\small Template for deniable QKE via entanglement distillation and teleportation}}
	\label{protocol:distillation-qke}
	\begin{algorithmic}[1]
		\STATE $A$ and $B$ share $n$ noisy entangled pairs (assume i.i.d. states for simplicity).
		\STATE They perform entanglement distillation to convert them into a state $\rho$ such that $F(\ket{\Psi^+}^{\otimes m},\rho)$ is arbitrarily close to 1 where $m < n$.
		\STATE Perform verification to make sure they share $m$ maximally entangled states $\ket{\Psi^+}^{\otimes m}$, and abort otherwise.
		\STATE $A$ prepares $m$ qubits (e.g. BB84 states) and performs quantum teleportation to send them to $B$ at the cost of consuming $m$ ebits and exchanging $2m$ classical bits.
		\STATE $A$ and $B$ proceed with standard classical distillation techniques to agree on a key based on their measurements.
	\end{algorithmic}
\end{algorithm}

By performing ED, Alice and Bob make sure that the resulting state cannot be correlated with anything else due to the monogamy of entanglement (see e.g. \cite{koashi2004monogamy,streltsov2012general}), thus factoring out Eve's system. The parties can open their records for steps $(2)$ and $(3)$ honestly, and open to arbitrary classical inputs for steps $(3), (4)$ and $(5)$: deniability follows from decoupling Eve's system, meaning that she is faced with a reduced density matrix on a pure bipartite maximally entangled state, i.e., a maximally mixed state $\rho_E = \mathbb{I}/2$, thus obtaining key equivocation.

Moreover, with the optimal rate of entanglement distillation being equal to the entropy of entanglement in the noiseless setting, distillation provides an operational interpretation of quantum entropy, which in this context can also be viewed as a measure of deniability.

In terms of the hierarchy of entanglement-based constructions mentioned in \cite{beaver2002deniability}, this approach mainly constitutes a generalization of such schemes. It should therefore be viewed more as a step towards a theoretical characterization of entanglement-based schemes for achieving information-theoretic deniability.

\subsection{Potential Pitfalls}

If the initial states are prepared and distributed by a malicious party (adversary), one can resort to techniques from device-independent quantum cryptography (DIQC). However, the results for secure key exchange from DIQC would not necessarily immediately translate into the deniability properties desired in our model.

Note that for a concrete instantiation, i.e., a protocol that implements this template, each step could potentially constitute an attack vector. For example, consider step $(3)$: although the verification and its statistical tests for Bell-type violations might pass, this would not immediately rule out similar attacks, for Eve could still demand that the parties reveal their inputs for a given choice of entangled states, which lies within the negligible fraction of noisy states she may have had access to. In terms of security and minimum entropy required for key material, no vital properties get violated, however when it comes to deniability, the security model is much stronger given that a single state allowing Eve to detect attempts at denial would suffice for defeating the security property.

\subsection{Deniability beyond QKE}

Going beyond QKE, note that quantum teleportation allows the transfer of an \emph{unknown} quantum state, meaning that even the sender would be oblivious as to what state is sent.

\paragraph{Anonymous Transfer and Traceless Exchange}

Additionally, ebits can enable uniquely quantum tasks (classically impossible) such as \emph{traceless exchange}, as shown by Christiandl and Wehner \cite{christandl2005quantum} in the context of quantum anonymous transmission (QAT), to achieve \emph{incoercible} protocols that allow parties to deny to any random input. More specifically, it allows a set of players to cast a yes/no vote in such a way that it is fundamentally impossible to trace a vote back to its voter.

\subsubsection{Entanglement Convertibility}

In a work by Nielsen \cite{nielsen1999conditions}, it is shown under what conditions, parties in possession of entanglement can have access to a class of transformations that allow them to transform a pure state of some composite bipartite system into another state using only LOCC techniques. These conditions establish a connection between the linear-algebraic theory of majorization and entanglement. We mention an important theorem derived in Nielsen's work that could potentially constitute a prerequisite for deniability when entanglement transformations would be of relevance, e.g., for equivocating to a different composite system.

Let $\ket{\psi}_{AB}$ denote the state of a bipartite composite system shared between Alice and Bob, and $\rho_{\ket{\psi}}$ Alice's system, where $\rho_{\ket{\psi}} \equiv \mathrm{tr}_B(\ket{\psi}\bra{\psi})$. Moreover, let $\lambda_{\ket{\psi}}$ be the vector of eigenvalues\footnote{Equivalently, Schmidt coefficients in a Schmidt decomposition of $\ket{\psi}$.} of $\rho_{\ket{\psi}}$. Finally, $\ket{\psi} \rightarrow \ket{\phi}$ means that $\ket{\psi}$ may be transformed into $\ket{\phi}$ using local operations and potentially unlimited two-way classical communication.

The theorem expresses a majorization requirement in terms of eigenvalues, stating that $\ket{\psi}$ can be transformed into $\ket{\phi}$ using LOCC if and only if $\lambda_{\psi}$ is majorized by $\lambda_{\phi}$
\[
\ket{\psi} \rightarrow \ket{\phi} \; \text{\emph{iff}} \; \lambda_{\ket{\psi}} \prec \lambda_{\ket{\phi}}.
\]
where $x \prec y$ denotes a majorization relation (partial order) between two real $d$-dimensional vectors $x=(x_1, \ldots, x_d)$ and $y = (y_1, \ldots, y_d)$ if $\forall k \in [1,d]$
\[
\sum_{i=1}^k x_{i}^{\downarrow} \le \sum_{i=1}^k y_{i}^{\downarrow},
\]
where $\downarrow$ indicates that elements are taken in descending order.

This result is of particular importance in the context of deniability. In effect, in entanglement-based schemes, parties that are in possession of shared ebits can transform their joint states into another system depending on various requirements. For instance, shared ebits can be used either in entanglement-assisted schemes, e.g. in quantum teleportation, or directly as carriers, meaning that they can be used directly for distilling shared randomness.

Note that the ability to change entangled systems to an arbitrary state can be done either prior to performing measurements (i.e. actually carrying out a transformation) or once measurements have been made in the sense that the parties can claim to have made a particular transformation. Again, the reason the latter works is that the adversary's local view of such joint systems is a maximally mixed state, thus allowing the parties to equivocate to another system.

\section{Quantum Bit Commitment and Deniability}

In a rather unexpected claim\footnote{We are not aware of any papers neither conjecturing, nor explicitly mentioning such a relation even in vague terms.} stating a folklore theorem in the quantum cryptography community regarding a relation between the impossibility of unconditionally secure quantum bit commitment by Mayers \cite{mayers1997unconditionally} and deniability, Beaver \cite{beaver2002deniability} argues why the impossibility result of Mayers for quantum bit commitment \cite{mayers1997unconditionally} does not simply carry over to QKE. Roughly speaking, the underlying belief that quantum protocols are deniable is claimed to have been due to the impossibility result for quantum bit commitment, i.e. if one is fundamentally unable to commit to a value, it is natural to expect deniability.

The belief that was claimed to have been shared by the community is that since committing to a bit is impossible in the quantum regime, then it stands to reason that deniability should follow naturally. This belief is then extended to coercer-deniable QKE and the author mentions some limitations, centered around three main criteria, to argue why the impossibility result for QBC does not carry over to deniability in QKE. While to the best of our knowledge, there are no studies or results to that effect in the printed literature, we briefly revisit this claim and address these limitations in light of our results based on covert communication and entanglement distillation.

\subsection{The Impossibility of Unconditional Quantum Bit Commitment}

The general structure of the formalism for QBC consists of three quantum systems modeled by their corresponding Hilbert spaces: $H_A$ for Alice, $H_B$ for Bob, and $H_{env}$ for the communication channel or the environment. The entire system is in a pure state and expressed in the tensor product space $H_A \otimes H_B \otimes H_{env}$. Thus, in the decoherence model, the two-party system $H_A \otimes H_B$ is a mixed state, the reduced state of $H_A \otimes H_B$ entangled with the environment $H_{env}$.

The idea for the impossibility of unconditional QBC is that given a commitment $\ket{\psi_b}$ by Alice, the entire system is in one of two states: $\ket{\psi_1}$ or $\ket{\psi_0}$ and the security of the QBC scheme demands that Bob knows nothing about $b$, i.e., the states $\ket{\psi_1}$ and $\ket{\psi_0}$ must look identical and be indistinguishable from each other (maximally mixed states). Mathematically, this means that the partial trace over $H_A$ of the global state, which is the view of Bob, remains unchanged, irrespective of the bit Alice chose during the commit phase:
\[
\mathrm{tr}_{H_A} (\ket{\psi_0}\bra{\psi_0}) = \mathrm{tr}_{H_A} (\ket{\psi_1}\bra{\psi_1}).
\]
This implies that $\ket{\psi_0}$ and $\ket{\psi_1}$ are purifications of the same state and hence have the same Schmidt Decomposition, meaning that Alice can move from $\ket{\psi_0}$ to $\ket{\psi_1}$ by using local operations. She prepares $\ket{\psi_0}$ as if she were committing to the bit 0, and upon opening the commitment, she does nothing if she wishes to open to 0, otherwise she applies a basis change operator on $H_A$, leading Bob to obtain $\ket{\psi_1}$ upon measurement.

The main idea is to ensure that Eve's view remains identical, when Alice and/or Bob try to deny (global state $\rho$), compared to when they behave honestly (global state $\sigma$). Mayers' equivocation transformation relies on the fact that the global state is pure, hence a Schmidt decomposition exists allowing Alice to perform local operations to change the bit she committed to. In the three-party case, where we want to achieve deniability there are two requirements: Eve's view should be indistinguishable in the honest and denial modes
\[
\mathrm{tr}_{A}(\mathrm{tr}_{B}(\rho_{ABE})) = \mathrm{tr}_{A}(\mathrm{tr}_{B}(\sigma_{ABE})),
\]
where $\rho$ and $\sigma$ are the global states when denying and being honest respectively, and Alice and Bob should be able to perform local operations to take $\rho_{ABE}$ to $\sigma_{ABE}$.

\subsection{Revisiting Limitations for Relating QBC to Deniability}

The three limitations mentioned in \cite{beaver2002deniability} are those of colocation, quantification and measurement, which we briefly recall below.
\begin{itemize}
	\item $(i)$ colocation: assuming the existence of a unitary $U_{S,R}$ where $(S,R)$ are considered to be the committer, the factorization of $U_{S,R}$ into local transforms $U_S$ and $U_R$ is not guaranteed. Moreover, colocation also characterizes the difficulty for the sender and receiver to dynamically agree on a joint strategy for performing a particular transformation $U_{deny}$, i.e., being able to consider the sender and receiver as a single entity.
	\item $(ii)$ quantification: the existence of a transformation for changing a commitment of $\ket{1}$ to $\ket{0}$ depends on having access to the adversary's program and knowledge on the joint system.
	\item $(iii)$ measurement: as opposed to the setting in QBC, delaying measurements is not typically expected in QKE, instead they need to take place at specific junctures of the protocol by both parties.
\end{itemize}

Limitations $(i)$ and $(ii)$ prevent us from being able to consider the global system as a pure state. This problem disappears after performing entanglement distillation as it will have decoupled Eve's system from the joint system of Alice and Bob, thus removing Eve's knowledge by denoising the originally partially entangled EPR pairs and also rendering Eve's program irrelevant. Therefore, the existence of an equivocation transformation similar to the one described by Mayers would in principle become possible.

Moreover, once we can consider the parties to be in possession of maximally entangled EPR pairs, results related to the convertibility of entanglement in the LOCC paradigm \cite{nielsen1999conditions} become applicable. Finally, regarding the problem of measurement, having access to persistent quantum memory along with using a one-shot protocol such as QKE via uncloneable encryption, one can lift the requirement for forced measurements. This relaxation will also enable performing various unitaries on the joint system.

In the attack described in Section \ref{subsec:state-injection-attack}, Eve sees an inconsistency in the transcript provided by Alice, i.e., her view $\rho_{E}$ is not the same as $\sigma_{E}$. However, in DC-QKE, Eve's view is indistinguishable as long as the covert channel is not broken. Similarly, in the entanglement distillation based approach, after the distillation procedure we know that the view of Eve is maximally mixed, allowing Alice to open to any random record.

\stopminichaptoc{\minichaptocenabled}

\chapter[Coercion and Quantum Resistant Voting via Fully Homomorphic Encryption]{Coercion-Resistant and Quantum-Secure Voting in Linear Time via Fully Homomorphic Encryption}\label{chp:cr-and-qsafe-voting}

\printminichaptoc{\minichaptocenabled}

\section{Introduction}

Over the past few decades, we have witnessed significant advances in cryptographic voting protocols. Yet, despite all the progress, see e.g., \cite{adida2008helios}, secure e-voting is still faced with a plethora of challenges and open questions, which largely arise as a result of the interplay between intricate properties such as vote privacy, individual and universal verifiability, receipt-freeness, and a notoriously difficult requirement, namely that of coercion-resistance.

Coercion-resistance can be viewed as a stronger form of privacy that should hold even against an adversary who may instruct honest parties to carry out certain computations while potentially even requiring that they reveal secrets in order to verify their behavior and ensure compliance. This property is typically enforced by providing honest parties with a mechanism that allows them to either deceive the coercer or to deny having performed a particular action. Due to limited space, we do not elaborate on the long series of works in this area and instead refer the reader to \cite{juels2005coercion,delaune2006coercion,kusters2012game,cortier2016sok}
and references therein for more details.

Since the breakthrough work of Gentry \cite{gentry2009fully} on fully homomorphic encryption (FHE), there has been a surge of interest in this line of research that remains very active to this day, with a series of recent advances including, but not limited to, a homomorphic evaluation of AES \cite{gentry2012homomorphic}.

Although the use of additively or multiplicatively homomorphic cryptosystems is common place in the e-voting literature, the relevance of FHE for potentially quantum-safe secure e-voting, with better voter verifiability, was only recently discussed by Gj\o steen and Strand \cite{gjosteen2017roadmap}. In our work, instead of designing an FHE-based protocol from scratch, we apply the machinery of FHE to a well-known, classical voting scheme, in order to improve its time complexity and to replace its reliance on the hardness assumption of solving the discrete logarithm problem with a quantum-resistant solution, namely lattice-based cryptography.

So far, no efficient quantum algorithms capable of breaking lattice-based FHE schemes have been discovered. Although provably quantum-secure solutions such as quantum key exchange do exist, here a quantum-secure construction merely captures the fact that efficient attacks based on quantum algorithms are yet to be discovered, which constitutes the central idea in post-quantum cryptography.

Although constructions with varying degrees of coercion-resistance do exist, the voting protocol introduced by Juels, Catalano, and Jakobsson \cite{juels2005coercion}, often referred to as the \emph{JCJ protocol}, is among the most well-known solutions in the context of coercion-resistant voting schemes. JCJ provides a reasonable level of coercion-resistance using a voter credential faking mechanism, and it was arguably the first proposal with a formal definition of coercion-resistance. However, JCJ suffers from a complexity problem due to the weeding steps in its tallying phase, which are required for eliminating invalid votes and duplicates. The pairwise comparison-based approach of JCJ using plaintext equivalence tests (PET) \cite{jakobsson2000mix} leads to a quadratic blow-up in the number of votes, which makes the tallying process rather impractical in realistic settings with a large number of voters or in the face of ballot-box stuffing. For instance, in the Civitas voting system \cite{DBLP:conf/sp/ClarksonCM08} based on JCJ, voters are grouped into blocks or virtual precincts to reduce the tallying time.

\subsection{Contributions and Structure}

Here we propose an enhancement of the JCJ protocol aimed at performing its tallying work in linear time, based on an approach that incorporates primitives from the realm of fully homomorphic encryption (FHE), which also paves the path towards making JCJ quantum-safe.

In Sect. \ref{sec:jcj}, we describe the JCJ protocol and cover some related work. Next, in Sect. \ref{sec:jcj-with-fhe}, we show how the weeding of ``bad'' votes can be done in linear time, with minimal change to JCJ, via an approach based on recent advances in various FHE primitives such as hashing, zero-knowledge (ZK) proofs of correct decryption, verifiable shuffles and threshold FHE. We also touch upon some of the advantages and remaining challenges of such an approach in Sect. \ref{subsec:open-questions} and in Sect. \ref{sec:security-considerations}, we discuss further security and post-quantum considerations.

\section{The JCJ Model and Voting Protocol in a Nutshell}\label{sec:jcj}

We first review the building blocks used in JCJ and then proceed to describing the protocol itself. Throughout, unless otherwise specified, it is assumed that the computations of the talliers and registrars are done in a joint, distributed threshold manner. We use $\in_U$ to denote an element that is sampled uniformly at random.

\subsection{Cryptographic Building Blocks}

JCJ relies on a modified version of ElGamal, a threshold public-key cryptosystem with re-encryption, secure under the hardness assumption of the Decisional Diffie-Hellman (DDH) problem in a multiplicative cyclic group $\mc{G}$ of order $q$. A ciphertext on message $m \in \mc{G}$ has the form $(\alpha, \beta, \gamma) = (mh^r, g_1^r, g_2^r)$ for $r \in_U \mathbb{Z}_q$, with $(g_1, g_2, h)$ being the public key where $g_1,g_2,h \in \mc{G}$, and the secret key consists of $x_1,x_2 \in \mathbb{Z}_q$ such that $h=g_1^{x_1}g_2^{x_2}$. The construction allows easy sharing of the secret key in a threshold way. We denote the probabilistic encryption of $m$ under key $k$ by $E_k(m)$ and omit the randomness for brevity.

The weeding steps make use of a plaintext equivalence test (PET), which is carried out by the secret key holders and takes as input two ciphertexts and outputs 1 if the underlying plaintexts are equal, and 0 otherwise. The PET produces publicly verifiable evidence with negligible information leakage about plaintexts.

Finally, JCJ uses non-interactive zero-knowledge (NIZK) proofs and mix-nets, which are aimed at randomly and secretly permuting and re-encrypting input ciphertexts such that output ciphertexts cannot be traced back to their corresponding ciphertexts. Throughout, it is assumed that the computations of the talliers and registrars are done in a joint, distributed threshold manner. We use $\in_U$ to denote an element that is sampled uniformly at random.

\paragraph{\textbf{Agents.}}

JCJ mainly consists of three sets of agents, described as follows.
\begin{enumerate}
	\item \textbf{Registrars}: A set $\mathcal{R} = \{R_1, R_2, \ldots, R_{n_R}\}$ of $n_R$ entities in charge of jointly generating and distributing credentials to voters.

	\item \textbf{Talliers}: A set $\mathcal{T} = \{T_1, T_2, \ldots, T_{n_T} \}$ of \emph{authorities} in charge of processing ballots, jointly counting the votes and publishing the final tally.

	\item \textbf{Voters}: A set of $n_V$ voters, $\mathcal{V} = \{V_1, V_2, \ldots, V_{n_V} \}$, participating in an election, where each voter $V_i$ is publicly identified by an index $i$.
\end{enumerate}

\paragraph{\textbf{Bulletin Board and Candidate Slate.}}

A \emph{bulletin board}, denoted by $\bb$, is an abstraction representing a publicly accessible, append-only, but otherwise immutable board, meaning that participants can only add entries to $\bb$ without overwriting or erasing existing items. A \emph{candidate slate}, $\vec{C}$, is an ordered set of $n_C$ distinct identifiers $\{c_1, c_2, \ldots, c_{n_C}\}$ capturing voter choices, which can represent a candidate or a party. A $\emph{tally}$ is defined under slate $\vec{C}$, as a vector $\vec{X} = \{x_1, x_2, \ldots, x_{n_C}\}$ of $n_C$ positive integers, where each $x_j$ indicates the number of votes cast for choice $c_j$.

\paragraph{\textbf{Assumptions for Coercion-Resistance.}}
No threshold set of agents in $\mc{T}$ should be corrupted, otherwise privacy is lost. In the registration phase, it is assumed that the distribution of voter credentials is done over an untappable channel and that no registration transcripts can be obtained, assuming that secure erasure is possible. Cast votes are transmitted via anonymous channels, which is a basic requirement for ruling out forced-abstention attacks.

\subsection{The JCJ Protocol}

\paragraph{\textbf{Setup and Registration.}}
The key pairs $(sk_{\mc{R}}, pk_{\mc{R}})$ and $(sk_{\mc{T}}, pk_{\mc{T}})$ are generated in a trustworthy manner, and the public keys, i.e., $pk_{\mc{T}}$ and $pk_{\mc{R}}$, are published with other public system parameters. The registrars $\mc{R}$ generate and transmit to eligible voter $V_i$ a random string $\sigma_i \in_U \mc{G}$ that serves as the credential of the voter. $\mc{R}$ adds an encryption of $\sigma_i$, $S_i = E_{pk_{\mc{T}}}(\sigma_i)$, to the voter roll $\vec{L}$, which is maintained on the bulletin board $\bb$ and digitally signed by $\mc{R}$.

\paragraph{\textbf{Voting.}} An integrity-protected candidate slate $\vec{C}$ containing the names and unique identifiers in $\mathcal{G}$ for $n_C$ candidates, along with a unique, random election identifier $\epsilon$ are published by the authorities. Voter $V_i$ generates a ballot in the form of a variant of ElGamal ciphertexts $(E_1, E_2)$, for candidate choice $c_j$ and voter credential $\sigma_i$, respectively, e.g., for $a_1, a_2 \in_U \mathbb{Z}_q$, we have $E_1 = (g_1^{a_1}, g_2^{a_1}, c_j h^{a_1})$ and $E_2 = (g_1^{a_2}, g_2^{a_2}, \sigma_i h^{a_2})$. $V_i$ computes NIZK proofs of knowledge and correctness of $\sigma_i$ and $c_j \in \vec{C}$, collectively denoted by $P_f$. These ensure non-malleability of ballots, also across elections by including $\epsilon$ in the hash of the Fiat-Shamir heuristic.

\paragraph{\textbf{Tallying.}} In order to compute the tally, duplicate votes and those with invalid credentials will have to be removed. The complexity problem crops up in steps \colorbox{lightgray}{2} and \colorbox{lightgray}{4} such that given $n$ votes, the tallying work has a time complexity of $\bigO{n^2}$. To tally the ballots posted to $\bb$, the authority $\mc{T}$ performs the following steps:
\begin{enumerate}
\item $\mc{T}$ verifies all proofs on $\bb$ and discards any ballots with invalid proofs.
Let $\vec{A_1}$ and $\vec{B_1}$ denote the list of remaining $E_1$ candidate choice ciphertexts, and $E_2$ credential ciphertexts, respectively.
\item $\mc{T}$ performs pairwise PETs on all ciphertexts in $\vec{B_1}$ and removes duplicates according to some fixed criterion such as the order of postings to $\bb$. For every element removed from $\vec{B_1}$, the corresponding element with the same index is also removed from $\vec{A_1}$, resulting in the ``weeded'' vectors $\vec{B'_1}$ and $\vec{A'_1}$.
\item $\mc{T}$ applies a mix-net to $\vec{A'_1}$ and $\vec{B'_1}$ using the same, secret permutation, resulting in the lists of ciphertexts $\vec{A_2}$ and $\vec{B_2}$.
\item $\mc{T}$ applies a mix-net to the encrypted list $\vec{L}$ of credentials from the voter roll and then compares each ciphertext of $\vec{B_2}$ to the ciphertexts of $\vec{L}$ using a PET. $\mathcal{T}$ keeps a vector $\vec{A_3}$ of all ciphertexts of $\vec{A_2}$ for which the corresponding elements of $\vec{B_2}$ match an element of $\vec{L}$, thus achieving the weeding of ballots with invalid voter credentials.
\item $\mc{T}$ decrypts all ciphertexts in $\mathbf{A_3}$ and tallies the final result.
\end{enumerate}

\paragraph{Properties.}

Vote \textbf{privacy} is maintained as long as neither a threshold set of talliers nor all the mixing servers are corrupted. A colluding majority of talliers can obviously decrypt everything and colluding mixing authorities could trace votes back to $\vec{L}$. Regarding \textbf{correctness}, voters can refer to $\bb$ to verify that their vote has been recorded as intended and that the tally is computed correctly. Similar attacks become possible in case of collusion by a majority of authorities. As for \textbf{verifiability}, anyone can refer to $\bb$, $P_f$ and $\vec{L}$ to verify the correctness of the tally produced by $\mc{T}$.
The \textbf{coercion-resistance} provided by JCJ essentially boils down to keeping voter credentials hidden throughout the election. A coerced voter can choose a random fake credential $\sigma'$ to cast a fake vote and present it as their real vote. Any vote cast with the fake credential will not be counted, and the voter can anonymously cast their real vote using their real credential.

\subsection{Related Work}\label{subsec:related-work}

We focus on the most closely-related works on improving the efficiency problem of the tallying work in JCJ. Smith \cite{smith2005new} and Weber et al. \cite{weber2007coercion,weber2008coercion} follow a similar approach in that they do away with comparisons using PETs, and instead, they raise the credentials to a jointly $\mc{T}$-shared secret value and store these blinded terms in a hash table such that collisions can be found in linear time.
The use of a single exponent means that a coercer can test if the voter has provided them with a fake or a real credential by submitting a ballot with the given credential and another with the credential raised to a known random value.

In \cite{araujo2008practical,araujo2010towards}, Araujo et al. move away from comparing entries in $\vec{L}$ with terms in the cast ballots to a setting in which duplicates are publicly identifiable and a majority of talliers use their private keys to identify legitimate votes, and in \cite{DBLP:conf/fc/AraujoBBT16} the authors use algebraic MACs.

Spycher et al. \cite{spycher2011new} use the same solution proposed by Smith and Weber to remove duplicates and apply targeted PETs only to terms in $\vec{L}$ and $\vec{A}$, identified via additional information provided by voters linking their vote to the right entry in $\vec{L}$.

In \cite{grontas2018towards}, publicly auditable conditional blind signatures are used to achieve coercion-resistance in linear time using a FOO-like \cite{fujioka1992practical} architecture, the downsides being the need for extra authorization requests for participation privacy and a double use of anonymous channels.

\section{JCJ in Linear Time via Fully Homomorphic Encryption}\label{sec:jcj-with-fhe}

Our proposal revolves around replacing the original cryptosystem of JCJ with a fully homomorphic one, thus allowing us to preserve the original design of JCJ. The main idea is to homomorphically evaluate hashes of the underlying plaintext of the FHE-encrypted voter credentials, perform FHE-decryption and post the hash values of the credentials to the bulletin board $\bb$. Now the elimination of invalid and duplicate entries can be done in linear time by using a hash table.

\subsection{FHE Primitives}

We provided a very brief introduction to fully homomorphic encryption in Section \ref{sec:fhe-primer}. Here we only enumerate the cryptographic primitives that will be required for the enhancement suggested below. We refer the reader to the cited sources throughout this chapter for further details.

Recall from Section \ref{sec:fhe-primer} that we use $\mc{E}_{pk}(m)$ to denote an FHE-encryption of a message $m \in \{0,1\}^n$ under the public key $pk$. For the purpose of understanding the enhancement suggested further below, it is important to remember that for $b_0, b_1 \in \{0,1\}$, given $\mc{E}_{pk}(b_0)$ and $\mc{E}_{pk}(b_1)$, FHE allows us to compute $\mc{E}_{pk}(b_0 \oplus b_1)$ and $\mc{E}_{pk}(b_0 \cdot b_1)$ by working over ciphertexts alone, without having access to the secret key, thus enabling the homomorphic evaluation of any boolean circuit, i.e., computing $\mc{E}_{pk}(f(m))$ from $\mc{E}_{pk}(m)$ for any computable function $f$.

In addition to using an FHE cryptosystem, see e.g., \cite{gentry2009fully,brakerski2014leveled}, we make use of the following FHE primitives, all of which are based on very recent advances. \colorbox{lightgray}{1} Fully homomorphic hashing: Fiore et al. \cite{fiore2014efficiently} introduce a family of universal one-way homomorphic hash functions, along with a one time use collision resistant homomorphic hash. \colorbox{lightgray}{2} Carr et al. \cite{carr2018zero} address the question of providing zero-knowledge proofs of correct decryption for FHE ciphertexts. \colorbox{lightgray}{3} Strand \cite{strand2018fhemixnet} tackles FHE mix-nets by proposing the first verifiable shuffle for FHE schemes, in particular for the GSW cryptosystem of Gentry, Sahai and Water \cite{gentry2013homomorphic}. \colorbox{lightgray}{4} Boneh et al. \cite{boneh2018threshold} provide a construction for a threshold FHE scheme based on the learning with errors (LWE) problem introduced by Regev \cite{regev2009lattices}. See Section \ref{subsec:open-questions} for more details on open questions and the state-of-the-art as the suggested primitives have indeed appeared only within the past year.

\subsection{Enhancing JCJ: FHE and Weeding in Linear Time}

We now describe how FHE primitives can be incorporated into JCJ while inducing minimal change in the original protocol. We assume threshold FHE throughout.

\paragraph{\textbf{Setup and Registration.}}

The setup and registration phases remain unchanged w.r.t. JCJ, except that $\mc{R}$ now adds an FHE-encryption of $\sigma_i$, $S_i = \mc{E}_{pk_{\mc{T}}}(\sigma_i)$, to the voter roll $\vec{L}$. We adopt the same assumptions mentioned earlier in Sect. \ref{sec:jcj}.

\paragraph{\textbf{Voting.}}

Instead of using ElGamal encryption, the credentials posted on the $\bb$ are encrypted under some FHE scheme, say BGV \cite{brakerski2014leveled}, with a key pair $(pk,sk)$. Each voter $V_i$ adds $\mc{E}_{pk_{\mc{T}}}(\sigma_i)$, along with the required NIZK proofs, to $\bb$.

\paragraph{\textbf{Tallying.}}

The tallying phase remains largely the same except that for removing duplicates and invalid votes, we leverage our use of FHE to carry out simple equality tests between hash digests of credentials. Since the concealed credentials are now stored in FHE ciphertexts, we can process them using an FHE hashing circuit. More precisely, for a jointly created $\mc{T}$-shared key $k$, published under encryption $\mc{E}_{pk}(k)$, the credentials $\sigma_i$ contained in the FHE-encrypted terms $\mc{E}_{pk}(\sigma_i)$ are homomorphically hashed (see \cite{fiore2014efficiently} by Fiore, Gennaro and Pastro and \cite{catalano2014authenticating} by Catalano et al.), under key $k$ resulting in $\mc{E}_{pk}(h_k(\sigma_i))$, such that upon decryption we obtain $h_k(\sigma_i)$.
A ZK proof of correct decryption is also posted to $\bb$ for verifiability, see \cite{carr2018zero} by Carr et al. for an approach to this. Once the hash values of the credentials are posted on the $\bb$, the weeding of duplicates can be done in $\bigO{n}$ using a simple hash table look-up, i.e., iterate, hash and check for collision in constant time, thus an overall linear-time complexity in the number of votes. Next, the registered credentials and the submitted vote/credential pairs are mixed \cite{strand2018fhemixnet} and the homomorphic hashing procedure is carried out again using a new secret key on all credential ciphertexts. Comparing the hashed registered credentials with those from the cast ballots allows us to remove invalid votes in $\bigO{n}$. Finally, the remaining valid votes are verifiably decrypted.

\subsection{Advantages, Potential Pitfalls and Open Questions}\label{subsec:open-questions}

Apart from the linear-time weeding algorithm, as already pointed out by Gj\o steen and Strand in \cite{gjosteen2017roadmap}, in addition to being a novel application of FHE to secure e-voting, obtaining better voter verifiability and a scheme believed to be quantum-resistant are among the noteworthy benefits of such an approach.

Clearly, in terms of real world FHE implementations, the state-of-the-art still suffers from efficiency issues. However, some significant progress has already been made in this area, e.g., the homomorphic evaluation of AES \cite{gentry2012homomorphic} or block ciphers designed for homomorphic evaluation \cite{DBLP:conf/eurocrypt/AlbrechtR0TZ15}. Moreover, it should be pointed out that some of the needed primitives, e.g., turning ZK proofs of correct decryption for FHE \cite{carr2018zero,luo2018verifiable} into NIZK proofs, are still not satisfactory and remain the subject of ongoing research and future improvements.

\section{Further Security Remarks}\label{sec:security-considerations}

A security analysis aimed at providing proofs of security for various properties such as correctness, verifiability and coercion-resistance will remain future work. One possibility would be to investigate whether the required security properties in our enhanced variant of JCJ hold against classical adversaries, under the same oracle access assumptions for mixing, PETs, threshold decryption and hashing. Post-quantum security will have to be proved in the quantum random oracle model.

\subsection{Eligibility Verifiability}

Assuming a majority of colluding authorities, apart from a compromise of vote privacy, another, perhaps more damaging problem with JCJ and its improved variants is that of \emph{eligibility verifiability}. A colluding majority would be able to retrieve voter credentials and submit valid votes for non-participating voters, i.e., perform ballot stuffing.

A solution in \cite{RoenneJCJ16} suggests performing the registration phase in such a way that only the voter would know the discrete logarithm of their credential. Votes are then cast with an anonymous signature in the form of a ZK proof of knowledge of the discrete logarithm of the encrypted credential, thus preventing ballot stuffing. A similar approach could be used here, with the potential downside of having inefficient proofs and a discrete logarithm hardness assumption, thus not being quantum secure.

\subsection{Post-Quantum Considerations}

For a relaxation of the trustworthiness assumption of $\mc{R}$, without assuming secure erasure, quantum-resistant designated verifier proofs \cite{sun2012toward,jao2014isogeny} could replace the classical ones suggested in the original JCJ \cite{juels2005coercion}.

To obtain post-quantum security for eligibility verifiability, future research will investigate the use of a quantum-resistant signature scheme that can be evaluated under FHE to preserve ballot anonymity. As a naive, but illustrative example that is one-time only and non-distributive, consider that the voter creates their credential as $\sigma_i = h(x)$, and that only the voter knows the preimage $x$. The voter now submits both $\mc{E}_{pk}(x)$ and $\mc{E}_{pk}(\sigma_i)$ to $\bb$. Before weeding, the hash is homomorphically evaluated on the ciphertext of the preimage, i.e., $\mc{E}_{pk}(h(x))$, followed by an equality test against the ciphertext of the credential $\mc{E}_{pk}(\sigma_i)$. A malicious authority can now cast only a valid ballot with a registered credential after the corresponding voter has cast a ballot, and an attempt to vote on their behalf is detectable in the weeding phase.

\stopminichaptoc{\minichaptocenabled}

\chapter{Concluding Remarks and Open Questions}\label{chp:conclusions}

\printminichaptoc{\minichaptocenabled}

We now conclude by stating some open questions, along with suggestions for approaching them in a systematic way.

\section{Entropy Extremizing Outputs in Deletion Channels}

\subsection{Estimating Expected Conditional Entropy}

In terms of estimating the expected leakage, as discussed in Subsection \ref{subsec:estimating-expected-leakage} of Chapter \ref{chp:information-leakage}, further developments in a characterization of the number of distinct subsequences can enable a more fine-grained estimation of the expected leakage.

\subsection{Finite Length Analysis}

In our finite length analysis presented in Chapter \ref{sec:chp-finite-length-deletion-conclusion}, we proved the entropy minimization case for single and double deletions. However, the other end of the spectrum, namely that of entropy maximization still remains unresolved. Moreover, it is not clear whether or not the techniques used in that approach lend themselves to a natural generalization. In fact, in its current form, since we are relying on an explicit clustering and enumeration of supersequences and their corresponding embedding weights, a simplistic extension of the same approach to higher or an arbitrary number of deletions does not seem viable. Yet, it may be possible to characterize the way the weights shift across the space of supersequences in such a way that would allow for a simpler generalization.

\subsection{Asymptotic Analysis}

In stark contrast, our approach in Chapter \ref{chp:hws} for resolving the entropy minimization case in the asymptotic limit, using methods from hidden word statistics, is completely different from the one used in our finite length analysis. We derived more general results for the case of minimal entropy for fixed output length $m$ and large input length $(n \rightarrow \infty)$. Deriving more precise results, for example for dealing with the case of the input length $n \rightarrow \infty$ and $n \sim m^2$ would be a natural continuation of this approach.

Such a solution, for the case where we let $m$ grow w.r.t $n$, would require a result similar to the original results for hidden word statistics, namely showing that the higher moments of the distribution converge to the corresponding moments of the normal distribution. Secondly, the rate of convergence would have to be accounted for, and we would have to obtain an estimate for the variance, showing that it is extremized by the intended strings when even $m$ is allowed to grow. Finally, one would need a calculation analogous to that in the proof of Theorem \ref{theorem:fixed-k}, showing that the errors are small even when terms depending on $m$ are tracked.

Clearly, in both cases, a proof for the minimization of the autocorrelation coefficient by the alternating strings remains open.

\section{Deniability in Quantum Cryptography}

\subsection{Deniability and Forward Secrecy in QKE}

Studying the deniability of public-key authenticated QKE both in our model and in the simulation paradigm, and the existence of an equivalence relation between our indistinguishability-based definition and a simulation-based one would be a natural continuation of this work.

Other lines of inquiry include forward deniability, deniable QKE in conjunction with forward secrecy, deniability using covert communication in stronger adversarial models, a further analysis of the relation between the impossibility of unconditional quantum bit commitment and deniability mentioned in \cite{beaver2002deniability}, and deniable QKE via uncloneable encryption.

\subsection{Going beyond Key Exchange}

Gaining a better understanding of entanglement distillation w.r.t. potential pitfalls in various adversarial settings and proposing concrete deniable protocols for QKE and other tasks beyond key exchange represent further research avenues.

\subsection{Coercion-Resistance in Quantum E-Voting}

In Section \ref{sec:quantum-e-voting} we briefly discussed the state-of-the-art in quantum voting protocols. A comprehensive overview of the classical and quantum literature would be a necessary step towards a systematization of knowledge in the general area of quantum e-voting protocols. This would consist of exploring and documenting the most relevant solutions that have been developed in the classical literature, covering feasibility and impossibility results, along with a formal classification of known solutions in terms of their requirements, functionalities and efficiency.

Although a considerable amount of definitional work has been done for classical deniability, its quantum counterpart suffers from a lack of rigorous formulations and definitions for various threat models, under varying computational and adversarial assumptions. The work presented in this thesis presents a first step towards developing rigorous definitions, paving the path for further formal definitions that capture the subtleties of quantum protocols. To capture the hybrid nature of QKE protocols, which is due to the juxtaposition of quantum and classical primitives, one would also have to rely on the body of knowledge that has been developed for classical solutions. This approach would make it possible to adopt a methodology that builds on existing and well-tested theories, which will form the foundation of further definitional work.

Finally, one can envisage a systematic analysis of quantum crypto primitives specifically aimed at identifying and classifying functionalities that are known to be possible thanks to uniquely quantum properties and as a result, known to be impossible to achieve using classical solutions. Such a classification of knowledge would group these quantum primitives in terms of their requirements, costs, computational and adversarial assumptions, functionalities and known use cases. This would make it less likely to reinvent primitives that may have already been discovered and it would also make it possible to systematically make use of known solutions in the design and analysis of deniability for QKE and more generally, for deniable quantum communication.

\subsection{FHE Primitives}

We introduced a classical coercion-resistant voting scheme, based on fully homomorphic encryption, in Chapter \ref{chp:cr-and-qsafe-voting} that is conjectured to be quantum-resistant. Our proposal makes use of FHE primitives that present a number of interesting open questions that currently represent active areas of research. These include, among other things, non-interactive ZK proofs of correct decryption for FHE ciphertexts, FHE mix-nets, threshold FHE and FHE hashing techniques. Finally, providing a proof of security for our scheme in the quantum random oracle model would be another natural follow-up work.

\stopminichaptoc{\minichaptocenabled}

% \printbibliography
\bibliographystyle{unsrt}
\bibliography{references.bib}

\appendix
\chapter{Appendix}

% \appendices
\section{Proof of Lemma \ref{TH:6.3}}\label{app:deletions}
\begin{proof}
The proof consists of two steps: first we show that $A-B > 0$ for all $k_1 \geq 1$ when $k_2 = \cdots = k_\ell = 1$; then we show that $\nabla(A-B)$ is positive along all directions others than the first one, so that an increase in any of the $k_i$ with $i\geq 2$, results in an increase of $A-B$.
We start by simplifying the expression.
To do so, we introduce the function $e(x) = -x \log_2 x$. We also use the fact that $e(xy) = xe(y)+ye(x)$, and develop the binomial coefficients: $e\left(\binom{a+b}{2} \right) = \binom{a+b}{2} + e((a+b)(a+b-1))$. Then we match the sum indexes. We also introduce the notation $e_i = e(k_i + 1)$.

Thus we can write:
\begin{align*}
A ={}
& \sum_{2 \leq i \leq l} e((k_1+1)(k_i+1)) +
 e\left(\binom{k_1+2}{2}\right) + e\left(\binom{k_2+2}{2}\right)\\
& +e(k_1+1) \sum_{1 \leq i \leq l}\widetilde{k_i}+e(k_1+k_2+1) \\
& + e(k_2+1) \sum_{1 \leq i \leq l}(\widetilde{k_i}-1)\\
& + e(k_2 + k_3+1) \\
={}
& \sum_{3 \leq i \leq \ell-1} e((k_1+1)(k_i+1)) + e((k_1+1)(k_2+1)) + e((k_1+1)(k_\ell+1)) \\
& + \binom{k_1+1}{2} + e((k_1+1)(k_1+2))\\
& + \binom{k_2+1}{2} + e((k_2+1)(k_2+2))\\
& +e(k_1+1) \sum_{3 \leq i \leq \ell-1}\widetilde{k_i}+e(k_1+k_2+1) + e(k_1+1)\widetilde{k_1} + e(k_1+1)\widetilde{k_2} + e(k_1+1)\widetilde{k_\ell}\\
& + e(k_2+1) \sum_{3 \leq i \leq \ell- 1}\widetilde{k_i}- \ell e(k_2+1) + e(k_2+1)\widetilde{k_1} + e(k_2+1)\widetilde{k_2} + e(k_2+1)\widetilde{k_\ell}\\
& + e(k_2 + k_3+1)\\
={}
& (k_1+1)\sum_{3 \leq i \leq \ell-1}e_i + e_1\sum_{3 \leq i \leq \ell-1}(k_i+1)  \\
& + (k_2+1)e_1 + (k_1+1)e_2 + (k_1+1)e_\ell + (k_\ell+1)e_1\\
& + \binom{k_1+1}{2} + (k_1+2)e_1 + (k_1+1)e(k_1+2)\\
& + \binom{k_2+1}{2} + (k_2+2)e_2 + (k_2+1)e(k_2+2) \\
& +e_1 \sum_{3 \leq i \leq \ell-1} k_i  - (\ell-3)e_1 +e(k_1+k_2+1) + e_1 k_1 + e_1 k_2 - e_1 + e_1 k_\ell\\
& + e_2 \sum_{3 \leq i \leq \ell- 1} k_i - (\ell-3)e_2 - \ell e_2 + e_2 k_1 + e_2 k_2 - e_2 + e_2 k_\ell\\
& + e(k_2 + k_3+1)
\end{align*}
At this point we regroup all terms in $e_i$ together:
\begin{align*}
A ={}
& \left( 2k_1 + 2k_2 + 2k_\ell - 3 + 2\sum_{3 \leq i \leq \ell-1} k_i \right)e_1 \\
& + \left( 2k_1 + 2k_2 + k_\ell - 2\ell - 1 + \sum_{3 \leq i \leq \ell- 1} k_i \right)e_2\\
& + (k_1+1)\sum_{3 \leq i \leq \ell-1}e_i \\
& + (k_1+1)e_\ell \\
& + e(k_1+k_2+1) + (k_1+1)e(k_1+2) + (k_2+1)e(k_2+2) + e(k_2 + k_3+1) \\
& + \binom{k_1+1}{2} + \binom{k_2+1}{2}
\end{align*}
We simplify the expression for $B$ in the same fashion:
\begin{align*}
B ={}
& 2e(k_1+k_2+1)\sum_{3 \leq i \leq \ell-1}k_i + k_1 e(k_1+k_2+1) + k_2 e(k_1+k_2+1) \\
& + k_\ell e(k_1+k_2+1) + e(k_1+ k_2 + k_3+1) \\
& +  (\ell - 3)e(k_1+k_2+1) + (k_1+k_2+2) \sum_{3 \leq i \leq \ell} e_i\\
& + \binom{k_1 +k_2+2}{2} + e((k_1+k_2+1)(k_1+k_2+2)) \\
={}
& (k_1+k_2+2) \sum_{3 \leq i \leq \ell-1} e_i  \\
& +\left(2k_1 + 2k_2 + k_\ell + \ell - 1 + 2\sum_{3 \leq i \leq \ell-1}k_i \right)e(k_1+k_2+1)\\
& + e(k_1+ k_2 + k_3+1) +  (k_1+k_2+1)e(k_1+k_2+2) + (k_1+k_2+2)e_\ell\\
& + \binom{k_1 +k_2+2}{2} \\
\end{align*}
so that we can now compute the difference:
\begin{align*}
A - B
={}
& \left( - 3 + 2\sum_{1 \leq i \leq \ell} k_i \right)e_1 \\
& + \left( k_1 + k_2 - 2\ell - 1 + \sum_{1 \leq i \leq \ell} k_i \right)e_2\\
& - (k_2+1)\sum_{3 \leq i \leq \ell}e_i  \\
& + (k_1+1)e(k_1+2) + (k_2+1)e(k_2+2) \\
& + 1 - k_1 k_2 \\
& -\left(-k_\ell + \ell + 2\sum_{1 \leq i \leq \ell}k_i \right)e(k_1+k_2+1)\\
& - (k_1+k_2-1)e(k_1+k_2+2) - e(k_1+ k_2 + k_3+1)  + e(k_2 + k_3+1) \\
={}
& P(\vec k)e_1 + Q(\vec k)e_2 - (k_2+1)\sum_{i=3}^{\ell} e_i + (k_1+1)e(k_1+2) + (k_2+1)e(k_2+2) \\
& + 1 - k_1k_2 - R(\vec k)e(k_1+k_2+1) - (k_1+k_2-1)e(k_1+k_2+2) \\
& - e(k_1+ k_2 + k_3+1) + e(k_2 + k_3+1).
\end{align*}
Where
\begin{align*}
P(\vec k) & = - 3 + 2\sum_{1 \leq i \leq \ell} k_i,\\
Q(\vec k) & =  k_1 + k_2 - 2\ell - 1 + \sum_{1 \leq i \leq \ell}k_i \\
R(\vec k) & = -k_\ell + \ell + 2\sum_{1 \leq i \leq \ell}k_i .
\end{align*}
We now compute $A-B$ where $k_i=1$ for $i \geq 2$ and show that it is positive. We get:
\begin{align*}
& (3 \ell - 3 + 2k_1)(k_1+2) \log_2 (k_1+2) + (k_1+1)(k_1+3)\log_2 (k_1+3) \\
& + 4(\ell - 2)(\log_2 2) +1 - [ (2k_1+2\ell -5)(k_1+1)\log_2 (k_1+1) \\
& + (k_1+1)(k_1+2)\log_2 (k_1+2) + 2(2k_1 - \ell -1)\log_2 2 + 9 \log_2 3 + k_1] \\
& = (2k_1 + 2 \ell -5)(k_1 + 1) \log (k_1+1) [\log_2 (k_1+2) -\log_2 (k_1+2) ] \\
& + (3 \ell - 3 + 2 k_1 + (\ell + 2)(k_1+2)) \log_2 (k_1+2) \\
& + (k_1+1)(k_1+2)[ \log_2 (k_1+2) - \log_2 (k_1+1) ] + (k_1+1) \log_2 (k_1+3) \\
& + 4(\ell - 2)(\log_2 2) +1 - [ 2(2k_1 - \ell -1)\log_2 2 + 9 \log_2 3 + k_1]
\end{align*}
Since $k_1 \geq 1$ and $\ell \geq 2$ we have $2k_1 \log_2 (k_1+2) + (k_1+1) \log_2 (k_1+3) \geq 2(2k_1 - \ell -1)\log_2 2$, $3 (\ell-1)\log_2 (k_1+2) + 2(k_1+2) \log_2 (k_1+2) \geq 9 \log_2 3$ and $\ell (k_1+2) \log_2 (k_1+2) \geq k_1$. This suffices to conclude that $A-B$ is positive when $k_i=1$ for $i\geq2$.\\
We now compute the partial derivatives for $i\geq2$ and show that they are positive.
The gradient can be computed term by term thanks to linearity, observing that for any polynomial $S(\vec k)$,
\begin{align*}
\partial_i e_i & = -\log_2(k_i+1) - \frac{1}{\ln(2)} \\
\partial_i e_j & = 0 \qquad (i \neq j)\\
\nabla S(\vec k)e_j & = \left(e_j\partial_i S(\vec k)+ S(\vec k)\partial_i e_j \right)_{i=1}^{\ell}
\end{align*}
Hence, by denoting $\vec u_1, \dotsc, \vec u_\ell$ the canonical basis, we have:
\begin{align*}
\nabla P(\vec k)e_1 & =  \left(e_1\partial_i P(\vec k)+ P(\vec k)\partial_i e_1 \right)_{i=1}^{\ell} = \partial_1 e_1 P(\vec k) \vec u_1 + e_1(\partial_i P(\vec k))_{i=1}^{\ell} \\
& = \partial_1 e_1 P(\vec k) \vec u_1 + 2(\vec u_1 + \cdots + \vec u_\ell) \\
& = (2 + \partial_1 e_1 P(\vec k))\vec u_1 + 2\vec u_2 + \cdots + 2\vec u_\ell \\
\nabla Q(\vec k)e_2 & = \left(e_2\partial_i Q(\vec k)+ Q(\vec k)\partial_i e_2 \right)_{i=1}^{\ell} = \partial_2 e_2 Q(\vec k)\vec u_2 + (\partial_i S(\vec k))_{i=1}^{\ell} \\
& = \partial_2 e_2 Q(\vec k)\vec u_2 + \vec u_1 + \vec u_2 + \vec u_1 + \cdots + \vec u_\ell \\
& = 2\vec u_1 + (2 + \partial_2 e_2 Q(\vec k))\vec u_2 + \vec u_3 + \cdots + \vec u_\ell \\
-\nabla \left((k_2+1)\sum_{i=3}^\ell e_i \right) & =
-(k_2+1)\nabla\sum_{i=3}^\ell e_i -(\nabla(k_2+1))\sum_{i=3}^\ell e_i \\
& = -((k_2+1)\partial_i e_i \vec u_i)_{i=3}^{\ell}-\left(\sum_{i=3}^{\ell}e_i\right)\vec u_2 \\
\nabla \left( (k_j+1)e(k_j+2)\right) & = -\left( \log_2(k_j+2) + \frac{1}{\ln(2)}\frac{k_j+1}{k_j+2}\right)\vec u_j \\
\nabla(1 - k_1k_2) & = -k_2\vec u_1 -k_1 \vec u_2 \\
- \nabla \left( R(\vec k)e(k_1+k_2+1) \right)
& = -R(\vec k) \nabla e(k_1+k_2+1) - e(k_1+k_2+1)\nabla R(\vec k) \\
& = R(\vec k) \left(\log_2(k_1+k_2+1) + \frac{1}{\ln(2)} \right)(\vec u_1 + \vec u_2) \\
& \qquad - e(k_1+k_2+1)(\partial_i R(\vec k))_{i=1}^{\ell} \\
& = R(\vec k) \left(\log_2(k_1+k_2+1) + \frac{1}{\ln(2)} \right)(\vec u_1 + \vec u_2) \\
& \qquad - e(k_1+k_2+1)(2\vec u_1 + \cdots + 2\vec u_{\ell-1} + \vec u_\ell) \\
- \nabla (k_1+k_2-1)e(k_1+k_2+2) & = -(k_1+k_2-1)\nabla e(k_1+k_2+2) - e(k_1+k_2+2) \nabla (k_1+k_2-1) \\
& = (k_1+k_2-1)\left( \log_2(k_1+k_2+2) + \frac{1}{\ln(2)} \right)(\vec u_1 + \vec u_2)  \\
& \qquad - e(k_1+k_2+2)(\vec u_1 + \vec u_2) \\
& = \left((k_1+k_2-1)\left( \log_2(k_1+k_2+2) + \frac{1}{\ln(2)} \right) - e(k_1+k_2+2)\right)\\
& \qquad (\vec u_1 + \vec u_2)\\
- \nabla e(k_1+ k_2 + k_3+1) & = \left(\log_2(k_1 + k_2 + k_3 + 1) + \frac{1}{\ln(2)}\right)(\vec u_1 + \vec u_2 + \vec u_3) \\
\nabla e(k_2 + k_3+1) & = -\left( \log_2(k_2+k_3+1) + \frac{1}{\ln(2)}\right)(\vec u_2 + \vec u_3)
\end{align*}
As is clearly visible from the above equations, we only need to consider the components along $\vec u_2$, $\vec u_3$, $\vec u_\ell$, and along $\vec u_i$ for any $3 < i < \ell$. For the latter, we have
\begin{align*}
\left( \nabla (A-B)\right)_i
& = 2 + 1 -(k_2+1)\partial_i e_i - 2e(k_1 + k_2 + 1) \\
& = 3 + 2(k_1 + k_2 + 1)\log_2(k_1 + k_2 + 1) +(k_2+1) \left(\log_2(k_i + 1) + \frac{1}{\ln(2)} \right) \\
& > 0.
\end{align*}
Now, along the very similar $\vec u_\ell$ axis,
\begin{align*}
\left( \nabla (A-B)\right)_\ell
& = 2 + 1 - (k_2+1)\partial_\ell e_\ell -e(k_1+k_2+1) \\
& = 3 + (k_1 + k_2 + 1)\log_2(k_1 + k_2 + 1) + (k_2+1) \left(\log_2(k_\ell + 1) + \frac{1}{\ln(2)} \right) \\
& > 0.
\end{align*}
Along $\vec u_3$,
\begin{align*}
\left( \nabla (A-B)\right)_3
& = 2 + 1 - (k_2+1)\partial_3 e_3 - 2e(k_1+k_2+1) + \log_2(k_1+k_2+k_3+1) + \frac{1}{\ln(2)} \\
& \qquad - \log_2(k_2+k_3+1) + \frac{1}{\ln(2)} \\
& = 3 + 2(k_1 + k_2 + 1)\log_2(k_1 + k_2 + 1) + (k_2+1) \left(\log_2(k_3 + 1) + \frac{1}{\ln(2)} \right) \\
& \qquad + \log_2(k_1+k_2+k_3+1) - \log_2(k_2+k_3+1) \\
& > 0.
\end{align*}
Along $\vec u_2$,
\begin{align*}
\left( \nabla (A-B)\right)_2
={} & 2 + 2 + Q(\vec k) \partial_2 e_2 - (k_2+1)\partial_2 e_2 \\
& - \sum_{i=3}^\ell e_i -\left(\log_2(k_2+2) + \frac{1}{\ln(2)}\frac{k_2+1}{k_2+2} \right) - k_1 \\
&  + R(\vec k)\left( \log_2(k_1+k_2+1) + \frac{1}{\ln(2)} \right) - 2e(k_1+k_2+1) \\
&  + \left((k_1+ k_2-1)\left( \log_2(k_1+k_2+2) + \frac{1}{\ln(2)}\right) -e(k_1+k_2+2) \right) \\
&  + \log_2(k_1+k_2+k_3+1) + \frac{1}{\ln(2)} - \log_2(k_2+k_3+1) - \frac{1}{\ln(2)}  \\
={} & 4 - (Q(\vec k) - k_2-1) \left( \log_2(k_2+1) + \frac{1}{\ln(2)}\right)   - \frac{1}{\ln(2)}\frac{k_2+1}{k_2+2} - k_1 \\
&  + \sum_{i=3}^\ell (k_i+1)\log_2(k_i + 1) \\
&  + R(\vec k)\left( \log_2(k_1+k_2+1) + \frac{1}{\ln(2)} \right) - 2e(k_1+k_2+1) \\
&  + (k_1+ k_2)\left( \log_2(k_1+k_2+2) + \frac{1}{\ln(2)}\right) -e(k_1+k_2+2)  \\
&  + \log_2(k_1+k_2+k_3+2) - \log_2(k_2+k_3+2) -\log_2(k_2+2) \\
\end{align*}
\begin{lemma}\label{lem:Q1}
$\left( \nabla (A-B)\right)_2 > 0$.
\end{lemma}
\begin{proof}[Proof of \Cref{lem:Q1}]
Letting $\lambda = \frac{1}{\ln(2)}$
We first show that the following line is positive
\begin{align*}
& -(Q (\vec k) - k_2 -1)(\log_2(k_2+1) + \lambda) - \lambda \frac{k_2+1}{k_2+2} - k_1 +  \\
& + R(\vec k)(\lambda + \log_2(k_1+k_2+1)) \\
={} & \lambda \left(R(\vec k) - Q(\vec k) + k_2 + 1 - \frac{k_2+1}{k_2+2}\right) + R(\vec k)\log_2(k_1+k_2+1) \\
&  - Q(\vec k)\log_2(k_2+1) \\
={} & \lambda \left( \sum_{i=2}^{\ell-1} k_i + 3\ell + 1 - \frac{k_2+1}{k_2+2}\right) \\
& + R(\vec k)\log_2(k_1+k_2+1) - Q(\vec k)\log_2(k_2+1).
\end{align*}
The last line is positive since in particular $ R(\vec k)\log_2(k_1+k_2+1) - Q(\vec k)\log_2(k_2+1) > (R(\vec k) - Q(\vec(k))\log_2(k_2+1)>0$.
Note that $-e(k_1+k_2+2)-\log_2(k_2+2) > 0$, $\log_2(k_1+k_2+k_3+1)-\log_2(k_2+k_3+1)>0$ and the remaining quantities are positive.
\end{proof}
As a result, we have that $A - B > 0$ for all $\vec k$ such that $k_i \geq 1$, which establishes the theorem.
\end{proof}

\section{Proof of Remark \ref{REMARK:NUMSTRING-DOUBLE}}\label{app:2}
We prove that for all positive integer sequences $(k_i)_{i \in \{1, \dotsc , \ell\}}$ such that
$\sum_{i=1}^{\ell}k_i = m$ we have :
\begin{align*}
  \frac{\ell(\ell+1)}{2} + \sum_{1 \leq i < j \leq \ell}\widetilde{k_i}\widetilde{k_j} + \sum_{i=1}^\ell \binom{\widetilde{k_i} + 1}{2} + 1 + \ell(m-\ell -2)
  = \binom{m+2}{m} + \binom{m+2}{m+1} + \binom{m+2}{m+2}
\end{align*}
We fix $\ell$ and $m$, then proceed by induction on the sequences of $(k_i)_{i \in \{1, \dotsc , \ell\}}$. We first show the equality for $k_1 = m- \ell +1$, and $k_i = 1$ for all $i> 1$.
\begin{proof}
We have on the left hand side:
\begin{align*}
&\frac{\ell(\ell+1)}{2} + \sum_{1 \leq i < j \leq \ell}\widetilde{k_i}\widetilde{k_j} + \sum_{i=1}^\ell \binom{\widetilde{k_i} + 1}{2} + 1 + \ell(m-\ell -2) \\
& = \frac{\ell(\ell+1)}{2} + 1 + \ell(m-\ell -2) + (m- \ell + 1) \sum_{j=2}^{\ell} \widetilde{k_j} + \binom{2}{2} + \binom{m-\ell+2}{2} \\
& = \frac{1}{2}\left(\ell (\ell +1) + (m-\ell+2)(m-\ell+1)\right) + (\ell +1)(m-\ell+2) + 1 \\
& = \frac{1}{2}(m^2 + 3m + 2) + m + 3 \\
& = \binom{m+2}{m} + \binom{m+2}{m+1} + \binom{m+2}{m+2}
\end{align*}
which concludes the initialization.
\end{proof}
We now fix a sequence $(k_i)_{i \in \{1,\dotsc , \ell\}}$, and $i_0 \in \{1, \dotsc , \ell\}$. We assume that the equality holds for this sequence and show that it is true for the sequence $(k'_i)_{i \in \{1, \dotsc , \ell\}}$ defined as $k'_i = k_i$ if $i \neq i_0$ and $i \neq i_0 + 1$, $k'_{i_0} = k_{i_0} - 1$ and $k'_{i_0+1} = k_{i_0+1} + 1$.
\begin{proof}
We first note that only a part of the formula on the left hand side depends on $(k_i)_{i \in \{1, \dotsc , \ell\}}$. Letting $F\left((k_i)_{i \in \{1, \dotsc , \ell \}}\right) =  \sum_{1 \leq i < j \leq \ell}\widetilde{k_i}\widetilde{k_j} + \sum_{i=1}^\ell \binom{\widetilde{k_i} + 1}{2}$, we just have to prove that
\begin{equation*}
F\left((k_i)_{i \in \{1, \dotsc , \ell\}}\right) - F\left((k'_i)_{i \in \{1, \dotsc , \ell\}}\right) = 0.
\end{equation*}
Expanding the above difference, we have:
\begin{align*}
&(\widetilde{k_{i_0}} - \widetilde{k'_{i_0}})(\sum_{i=1}^{i_0-1}\widetilde{k_i})
+  (\widetilde{k_{i_0 + 1}} - \widetilde{k'_{i_0 + 1}})(\sum_{i=1}^{i_0-1}\widetilde{k_i}) + \widetilde{k_{i_0}} \widetilde{k_{i_0+1}} - \widetilde{k'_{i_0}}\widetilde{k'_{i_0+1}}
\\
& \qquad +
(\widetilde{k_{i_0}} - \widetilde{k'_{i_0}})(\sum_{i=i_0+1}^{\ell}\widetilde{k_i})
+  (\widetilde{k_{i_0 + 1}} - \widetilde{k'_{i_0 + 1}})(\sum_{i=i_0+2}^{\ell}\widetilde{k_i}) \\
& \qquad  + \binom{\widetilde{k_{i_0}}+1}{2} - \binom{\widetilde{k'_{i_0}}+1}{2} + \binom{\widetilde{k_{i_0+1}}+1}{2} - \binom{\widetilde{k_{i_0+1}}+1}{2}
\end{align*}
This is equal to
\begin{align*}
& \widetilde{k_{i_0}} \widetilde{k_{i_0+1}} - (\widetilde{k_{i_0}}-1)(\widetilde{k_{i_0+1}}+1) + \widetilde{k_{i_0+1}} + \frac{1}{2}\Big(\widetilde{k_{i_0}} (\widetilde{k_{i_0}}+1) - (\widetilde{k_{i_0}} +1)\widetilde{k_{i_0}}\Big) \\
& \qquad - \frac{1}{2}\Big(\widetilde{k_{i_0+1}} (\widetilde{k_{i_0+1}}+1) - (\widetilde{k_{i_0+1}} + 2)\widetilde{k_{i_0+1}}+1\Big) \\
& = - \widetilde{k_{i_0}} + \widetilde{k_{i_0+1}} + 1 + \frac{1}{2}(2\widetilde{k_{i_0}} - 2 \widetilde{k_{i_0+1}} - 2) \\
&= 0.
\end{align*}
This concludes the proof.
\end{proof}

\section{Proof of Remark \ref{REMARK:NUMWEIGHT:DOUBLE}}\label{app:3}
As in \Cref{app:2} we proceed by induction to show that
if there exist positive integers $(k_i)_{i \in \{1, \dotsc, \ell \} }$ such that $m = \sum_{i=1}^{\ell} k_i$, then we have
\begin{align*}
\sum_{i=1}^{\ell}&\binom{k_i+2}{2} + \sum_{1 \leq i < j \leq \ell}(k_i +1)(k_j+1) + \sum_{1\leq i < j \leq \ell}\widetilde{k_i}\widetilde{k_j} + \sum_{i=1}^{\ell} \binom{\widetilde{k_i}+1}{2} \\+ &\sum_{i=1}^{\ell}\left[(m-\ell+2)\times (k_i+1)+k_i+k_{i+1}+1\right] = 4 \binom{m+2}{m}
\end{align*}

We fix $\ell$ and $m$, then proceed by induction on the sequences of $(k_i)_{i \in \{1, \dotsc , \ell\}}$. We first show the equality for $k_1 = m- \ell +1$, and $k_i = 1$ for all $i> 1$.
\begin{proof}
We have on the left hand side of the equation:
\begin{align*}
& (m-\ell+1)+\binom{m-\ell+2}{2}+1+\binom{m-\ell+3}{2}+(\ell-1)*3+(\ell-1)(m-\ell+2)*2\\
& +2(\ell-2)(\ell-1)+(m-\ell+2)^2 + m-\ell+3+(\ell-1)(2(m-\ell+1)+3)-1 \\
& = 2m^2 + 6 m + 4 \\
& = 4 \binom{m+2}{m}
\end{align*}
which concludes the initialization.
\end{proof}
We now fix a sequence $(k_i)_{i \in \{1,\dotsc , \ell\}}$, and $i_0 \in \{1, \dotsc , \ell\}$. We assume that the equality holds for this sequence and show that it is true for the sequence $(k'_i)_{i \in \{1, \dotsc , \ell\}}$ defined as $k'_i = k_i$ if $i \neq i_0$ and $i \neq i_0 + 1$, $k'_{i_0} = k_{i_0} - 1$ and $k'_{i_0+1} = k_{i_0+1} + 1$.
\begin{proof}
First notice that as in \Cref{app:2} we can ignore all the terms that do not depends on the $k_i$. Furthermore we can reuse the result of \Cref{app:2} to remove
$\sum_{1 \leq i < j \leq \ell}\widetilde{k_i}\widetilde{k_j} + \sum_{i=1}^\ell \binom{\widetilde{k_i} + 1}{2}$.
We also note that
\[
\sum_{i=1}^{\ell}\left[(m-\ell+2)\times (k_i+1)+k_i+k_{i+1}+1\right] =
\sum_{i=1}^{\ell}\left[(m-\ell+2)\times (k'_i+1)+k'_i+k'_{i+1}+1\right]
\]
since $\sum_{i=1}^{\ell} k_i = \sum_{i=1}^{\ell}k'_i$. We therefore define
\[
F\left((k_i)_{i \in \{1, \dotsc , \ell \}}\right) = \sum_{i=1}^{\ell}\binom{k_i+2}{2} + \sum_{1 \leq i < j \leq \ell}(k_i +1)(k_j+1)
\]
and show that
\begin{equation*}
F\left((k_i)_{i \in \{1, \dotsc , \ell\}}\right) - F\left((k'_i)_{i \in \{1, \dotsc , \ell\}}\right) = 0.
\end{equation*}
Expanding the difference we get
\begin{align*}
 &\binom{k_{i_0}+2}{2} - \binom{k'_{i_0}+2}{2}  + \binom{k_{i_0+1}+2}{2} - \binom{k'_{i_0+1}+2}{2} \\
&  + k_{i_0}\sum_{j>i_0}^{\ell}(k_j+1) - k'_{i_0}\sum_{j>i_0}^{\ell}(k'_j+1)
+ k_{i_0+1}\sum_{j>i_0+1}^{\ell}(k_j+1) - k'_{i_0+1}\sum_{j>i_0+1}^{\ell}(k'_j+1) \\
& = k_{i_0}+1 - (k_{i_0+1} + 2) + k_{i_0}S - (k_{i_0}-1) \left(S+1 \right) +
k_{i_0+1} S' - (k_{i_0+1} +1) S' \\
& = 1 - k_{i_0} + S - 1 - S' \\
& = 0
\end{align*}
where $S = \sum_{j>i_0}^{\ell}(k_j+1)$ and $S'=\sum_{j>i_0+1}^{\ell}(k_j+1)$. This concludes the proof.
\end{proof}

\section{CSP Code for Clustering Analysis}\label{app:csp-code}

\begin{verbatim}

-- Authors: A. W. (Bill) Roscoe and Arash Atashpendar

wb(N,s,k) = if k+#s > N then 0 else wb'(N,s,k)

wb'(N,<>,k) = C(N,k)    -- just choose N from k

--all extra bits = 0
-- first bit of <0>^s is 0, so must that of S
wb'(N,<0>^s,0) = wb(N-1,s,0)
wb'(N,<1>^s,0) = wb(N-1,s,0)   -- first bit of S is 1
                + wb(N-1,<1>^s,0)   -- first bit of S is 0

-- otherwise, we have the following cases
-- first bit of S is 0,  means we only have to find s in S'
-- first bit of S is 1, must be part of padding
wb'(N,<0>^s,k) = wb(N-1,s,k) + wb(N-1,<0>^s,k-1)

wb'(N,<1>^s,k) = wb(N-1,s,k)     -- first bit of S is 1
                + wb(N-1,<1>^s,k)

C(N,M) = if 2*M<=N then C'(N,M) else C'(N,N-M)

C'(N,M) = F(N,N-M+1)/F(M,2)

F(N,M) = if M>N then 1 else N*F(N-1,M)

-- base cases
Y(0,s,h) =wb(0,s,h)
Y(N,<>,h) = wb(N,<>,h)
Y(N,s,0) = wb(N,s,0)

--Y(N,<x>^s,h) = Y(N-1,<x>^s,h) + Y(N-1,<x>^s,h-1) + Y(N-1,s,h-1)
Y(N,<x>^s,h) = Y(N-1,<x>^s,h) + Y(N-1,<x>^s,h-1) + Y(N-1,s,h)

-- Easier recursion for the proof

-- base cases
LC(0,s,r) = wb(0,s,r)
LC(N,<>,r) = wb(N,<>,r)
LC(N,s,0) = wb(N,s,0)

LC(N, <x>^s, r) = LC(N-1, <x>^s, r-(1-w(<x>))) + LC(N-1, s, r)

ILC(N, s, g) = if #s == N then 1 else ILC'(N-g, s, 0)

ILC'(N, s, g) = if #s > N then 0 else ILC'(N, s, g)
ILC'(N, <>, g) = if N == 0 then 1 else 0
--ILC'(N, s, g, z) = wb(N, s, 0)

ILC'(N, <x>^s, 0) = ILC(N-1, s, 0)+ ILC(N-1, <x>^s, 0)

--ILC'(N, <x>^s, g, z) = ILC(N-1, s, g, z) + ILC(N-1, <x>^s, g, z)

w(s) = head(s) + w'(tail(s))

w'(<>) = 0
w'(s) = w(s)
--t(a) = if a == 1 then 1 else 0

SCC(N,s,k) = let a = w(s)
                 b = #s - a
                 c = N - a - b - k
             within
               if c<0 then 0 else
               LS(<C(r-1,a-1)*C(N-r,k) | r <- <a..(a+c)>>)

LS(<>) = 0
LS(<x>^xs) = x + LS(xs)


-- Count classes

cl(n, s, r) = if r+#s > n then 0 else cl'(n, s, r)

cl'(n, <0>^s, 0) = cl(n-1, s, 0)
cl'(n, <1>^s, 0) = cl(n-1, s, 0) + cl(n-1, <1>^s, 0)

cl'(n, <>, r) = 1

--cl(n, s, r) = if r+#s > n then 0 else cl'(n, s, r)

--cl'(n, <0>^s, r) = cl(n-1, s, r) + cl(n-1, <0>^s, r-1)
--cl'(n, <1>^s, r) = cl(n-1, s, r) + cl(n-1, <1>^s, r)


--otherwise
cl'(N, <x>^s, r) = cl(N-1, <x>^s, r-(1-w(<x>))) + cl(N-1, s, r)

--cl'(N,<0>^s,k) = cl(N-1,s,k) + cl(N-1,<0>^s,k-1)

--cl'(N,<1>^s,k) = cl(N-1,s,k) + cl(N-1,<1>^s,k)

\end{verbatim}

\section{Software Toolkit for Binary Sequences}\label{app:python-code}

The source code of the ``BinSeqPy'' software toolkit can be found in a separate file\footnote{For more details, see \url{http://hdl.handle.net/10993/38864}}.

\end{document}